\documentclass[aos]{imsart-arxiv} 

\RequirePackage[authoryear,sort]{natbib} 
\RequirePackage{graphicx}
\RequirePackage{amsfonts,amssymb,amsmath,amsthm,epsfig,euscript,bbm,amsthm}

\usepackage{mathtools}
 
\makeatletter
\DeclareRobustCommand\widecheck[1]{{\mathpalette\@widecheck{#1}}}
\def\@widecheck#1#2{%
 \setbox\z@\hbox{\m@th$#1#2$}%
 \setbox\tw@\hbox{\m@th$#1%
 \widehat{%
 \vrule\@width\z@\@height\ht\z@
 \vrule\@height\z@\@width\wd\z@}$}%
 \dp\tw@-\ht\z@
 \@tempdima\ht\z@ \advance\@tempdima2\ht\tw@ \divide\@tempdima\thr@@
 \setbox\tw@\hbox{%
 \raise\@tempdima\hbox{\scalebox{1}[-1]{\lower\@tempdima\box
\tw@}}}%
 {\ooalign{\box\tw@ \cr \box\z@}}}
\makeatother

\usepackage{bbm}
 \usepackage{mathrsfs}
 \usepackage{mathtools}
\usepackage{bm}
\usepackage{enumerate}
\usepackage[titletoc,title]{appendix}
\usepackage[section]{placeins}

 \usepackage[table]{xcolor}
 \definecolor {processblue}{cmyk}{0.96,0.21,0.41,0.11}
 \definecolor{orange}{RGB}{230,170,120}
 \definecolor{green}{RGB}{100,200,100}
 \usepackage{colorspace}
 
 \definespotcolor{pinkpantone}{PANTONE 17 1926 TCX RGB}[RGB]{207 92 120}
 
 \definespotcolor{greenpantone}{PANTONE 16 4728 TPX RGB}[RGB]{0 164 180}
 
 \definespotcolor{bluepantone}{PANTONE 19 4151 TPX RGB}[RGB]{0 192 259}

\usepackage{arydshln}
\usepackage{threeparttable,subcaption}
 \usepackage{booktabs}

\usepackage{rotating,multirow,makecell,array}
\usepackage{colortbl}

\usepackage {tikz}
\usetikzlibrary{arrows}
\usetikzlibrary {positioning}

\usepackage{algorithm}
\usepackage{algpseudocode}

\usepackage{threeparttable}
\usepackage{pifont}
\newcommand{\cmark}{\ding{51}}%
\newcommand{\xmark}{\ding{55}}%

\startlocaldefs
\theoremstyle{plain}
\newtheorem{theorem}{Theorem}
\newtheorem{lemma}{Lemma}

\newtheorem{proposition}{Proposition}
\theoremstyle{remark}
\newtheorem{assumption}{Assumption}
\newtheorem{remark}{Remark}
\newtheorem{example}{Example}

\endlocaldefs

\def\bX{\mathbf{X}}

\def\bu{\mathbf{u}}
\def\bU{\mathbf{U}}
\def\bv{\mathbf{v}}
\def\bV{\mathbf{V}}

\def\bW{\mathbf{W}}

\def\balpha{\boldsymbol{\alpha}}
\def\bbeta{\boldsymbol{\beta}}

\def\bdelta{\boldsymbol{\delta}}
\def\etabold{\boldsymbol{\eta}}

\def\balphahat{\widehat{\balpha}}
\def\bbetahat{\widehat{\bbeta}}

\def\Yhat{\widehat{Y}}

\def\bdeltatil{\widetilde{\bdelta}}
\def\bdeltahat{\widehat{\bdelta}}

\def\Deltahat{\widehat{\Delta}}

\def\atil{\widetilde{a}}
\def\Atil{\widetilde{A}}

\def\bzero{\mathbf{0}}

\def\bS{\mathbf{S}}
\def\bs{\mathbf{s}}

\def\bSbar{\overline{\bS}}
\def\bsbar{\overline{\bs}}

\def\S{\mathbb{S}}

\def\bDelta{\boldsymbol{\Delta}}

\def\bgamma{\boldsymbol{\gamma}}

\def\R{\mathbb{R}}

\def\S{\mathbb{S}}

\def\Ytil{\widetilde{Y}}

\def\abar{\overline{a}}
\def\Abar{\overline{A}}

\def \hs2{\hspace{2mm}}

\numberwithin{table}{section}
\numberwithin{equation}{section}

\definecolor{jcolor}{RGB}{041,122,000}
\definecolor{darkred}{RGB}{100,000,000}
\definecolor{purple}{RGB}{200,000,200}

\def\boxit#1{\vbox{\hrule\hbox{\vrule\kern6pt  \vbox{\kern6pt#1\kern6pt}\kern6pt\vrule}\hrule}}

\def\muhat{\widehat{\mu}}

\def\muhat{\widehat{\mu}}

\def\muhat{\widehat{\mu}}

\def\be{\mathbf{e}}

\def\bS{\mathbf{S}}

\def\pihat{\widehat{\pi}}
\def\rhohat{\widehat{\rho}}

\def\nuhat{\widehat{\nu}}

\def\bgammahat{\widehat{\bgamma}}

\def\Ytil{\widetilde{Y}}

\def\sigmahat{\widehat{\sigma}}

\def\thetahat{\widehat{\theta}}

\usepackage{memhfixc} 
\usepackage{xr-hyper}
\usepackage[pdftex,bookmarks,colorlinks,breaklinks]{hyperref} 

\DeclareMathOperator*{\argminn}{argmin}

\begin{document}

\begin{frontmatter}
\title{High-dimensional inference for dynamic treatment effects}
\runtitle{Dynamic Treatments: Double Robustness}

\begin{aug}
\author[A]{\fnms{Jelena} \snm{Bradic}\ead[label=e1]{jbradic@ucsd.edu}},
\author[B]{\fnms{Weijie} \snm{Ji}\ead[label=e2]{w6ji@ucsd.edu}}
\and
\author[C]{\fnms{Yuqian} \snm{Zhang}\ead[label=e3]{yuqianzhang@ruc.edu.cn}}
\address[A]{Department of Mathematics and Halicioglu Data Science Institute, University of California San Diego \printead{e1}}
\address[B]{Department of Mathematics, University of California San Diego \printead{e2}}
\address[C]{Institute of Statistics and Big Data, Renmin University of China 
\printead{e3}}
\end{aug}
\begin{abstract}

Estimating dynamic treatment effects is a crucial endeavor in causal inference, particularly when confronted with high-dimensional confounders. Doubly robust (DR) approaches have emerged as promising tools for estimating treatment effects due to their flexibility. However, we showcase that the traditional DR approaches that only focus on the DR representation of the expected outcomes may fall short of delivering optimal results. In this paper, we propose a novel DR representation for intermediate conditional outcome models that leads to superior robustness guarantees. The proposed method achieves consistency even with high-dimensional confounders, as long as at least one nuisance function is appropriately parametrized for each exposure time and treatment path. Our results represent a significant step forward as they provide new robustness guarantees. The key to achieving these results is our new DR representation, which offers superior inferential performance while requiring weaker assumptions. Lastly, we confirm our findings in practice through simulations and a real data application.

\end{abstract}



\end{frontmatter}

\section{Introduction}

The complexity of a given disease or economic policy often manifests in the diversity and size of the personal characteristics pertaining to each individual or economy under consideration, causing a considerable degree of heterogeneity in observed outcomes. However, the utility of randomized control trials (RCTs), especially over time, is frequently curtailed by prohibitive costs or ethical concerns. In contrast, the accessibility of time-varying observational studies has burgeoned of late. The ubiquity of data-driven decision-making is evident in various aspects of daily life, such as the continuous monitoring of individuals' health using mobile devices and consequential medical interventions, tracking of online presence and real-time measurement of economic and social policies implemented to enhance public health. 
The present study contributes novel insights to the literature by proposing a novel framework to construct confidence intervals pertaining to dynamic treatment effects amid high-dimensional observations. In a Job Corps real-data analysis, our novel framework provides more accurate estimates of the long-term impact of additional schooling over time on wages, which has important practical implications for designing effective policies aimed at increasing educational attainment and improving economic outcomes.


In light of intricate notational complexities, we exemplify our ideas and findings for two-stage trials while affirming that the same theoretical framework and methodology developed are extensible to multiple-stage trials; see, e.g., Section \ref{sec:multi-stage}.
Consider a two-stage series of binary treatment assignments, denoted by $A_1$ and $A_2$, and an outcome of interest, $Y \in \mathbb{R}$. Alongside this, a set of possibly high-dimensional sequential pre-treatment covariates $\bS_1 \in \mathbb{R}^{d_1}$ and $\bS_2 \in \mathbb{R}^{d_2}$, possibly of different dimensions, are also observed. 
The potential or counterfactual outcomes, $Y(a)$, refer to the outcome that a participant would have experienced had they followed a particular treatment sequence, $a = (a_1,a_2) \in \{0,1\}^2$, which may differ from the treatment they were observed with. Our parameter of interest is the dynamic treatment effect (DTE) between two treatment paths, $a$ and $a' $, which is defined as follows:
\begin{align}
\theta := E[Y(a)] - E[Y(a')]= \theta_a - \theta_{a'},\;\;\text{with}\;\;\theta_{a}:= E[Y(a)].\label{def:DTE}
\end{align}

Estimating the DTE is a challenging task when there are multiple exposures involved. The influence of past treatments on future confounders and treatment choices complicates the identifiability of $\theta$ \citep{rosenbaum1983central}. 
Adjusting for confounders may not have a causal interpretation, even when all confounders are measured and the regression is correctly specified \citep{daniel2013methods}. 
In this context, alternative methods such as Sequential Multiple Randomized Control Trials (SMART) \citep{hernan2016specifying}, Structural Nested Mean (SNM) \citep{Robins1997causal}, and Marginal Structural Mean (MSM) models \citep{murphy2001marginal} have become the gold standard for addressing these challenges. This paper contributes to the field by establishing robust MSM model estimations with new effective rates.

\subsection{The doubly robust representations}\label{rep:DR}

Throughout this work, we assume that any treatment-specific variable can only be affected by past treatments or past covariates; and not the future. This is sometimes called temporal ordering. We also assume a ``no interference'' setting and Assumption \ref{assumption_1} below \citep{robins1987addendum, robins2000marginal,murphy2003optimal}. 

\begin{assumption}\label{assumption_1} 
	(a) (Sequential Ignorability)
	$
	Y(a_{1},a_{2})\perp \!\!\! \perp A_{1}\mid\bS_{1}$ and $ Y(a_{1},a_{2}) \perp \!\!\! \perp A_{2}\mid\bS,A_{1}=a_{1}$ where $\bS=(\bS_1^\top,\bS_2^\top)^\top\in \mathbb{R}^{d}$ with $d:=d_1+d_2$. 
	(b) (Consistency of potential outcomes) 
	$Y=Y(A_{1},A_{2}).$
	(c) (Overlap) Let $c_0 \in (0,1/2)$ be a positive constant, such that 
	$P(c_0 \leq \pi_{a} (\bS_{1})\leq1-c_0) =1, $ and 
	$P(c_0 \leq \rho_{a}(\bS)\leq1-c_0) =1$. Here, the propensity scores are defined as $\pi_{a} (\bs_{1}):= P[A_{1}=a_1|\bS_{1}=\bs_{1}]$ and $\rho_{a}(\bs ):= P[A_{2}=a_2|\bS =\bs,A_{1}=a_{1}]$.
\end{assumption}


The following lemma provides a doubly robust (DR) representation of $\theta_a$. This result is consistent with previous studies in the literature, including works by \cite{van2012targeted,orellana2010dynamic,murphy2001marginal,
bang2005doubly}. We consider the MSM models where we adjust for confounding variables that may affect both the treatment assignment and the outcome of interest. In an MSM, the treatment assignment and the outcome of interest are modeled separately using  propensity scores 
$\pi_{a} (\bs_{1})$ and $\rho_{a}(\bs )$ together with the first-time and second-time conditional means,   $\mu_{a}(\bs_{1}):= E[Y(a)|\bS_{1}=\bs_{1}]$ and $ \nu_{a}(\bs ):= E[Y(a)|\bS =\bs,A_1=a_1]$. 
   Throughout this work, we use $\pi_{a}^* (\cdot)$ and $\rho_{a}^*(\cdot)$ as well as $\mu_a^*(\cdot)$ and $\nu_a^*(\cdot)$ to refer to the working models, i.e., the population-level approximations of the propensity scores and conditional means, respectively. 


\begin{lemma}[A DR representation of $\theta_a$]\label{lem_theta}
Let Assumption \ref{assumption_1} hold. Suppose that at least one of $\mu_{a}^{*}(\cdot)$ and $\pi_{a}^{*}(\cdot)$ is correctly specified, and at least one of $\nu_{a}^{*}(\cdot)$ and $\rho_{a}^{*}(\cdot)$ is correctly specified, i.e, (a) either $\mu_{a}^{*}(\cdot)=\mu_{a}(\cdot)$ or $\pi_{a}^{*}(\cdot)=\pi_{a}(\cdot)$, but not necessarily both and (b) either $\nu_{a}^{*}(\cdot)=\nu_{a}(\cdot)$ or $\rho_{a}^{*}(\cdot)=\rho_{a}(\cdot)$, but not necessarily both. Then
\begin{align}\label{rep:DR-theta}
	\theta_{a} =E\left[\mu_{a}^*(\bS_{1})+\mathbbm1_{\{A_1=a_1\}}\frac{ \nu_{a}^*(\bS )-\mu_{a}^*(\bS_{1}) }{\pi_{a} ^* (\bS_{1})}+\mathbbm1_{\{A_1=a_1, A_2=a_2\}}\frac{ Y-\nu_{a}^*(\bS )}{\pi_{a} ^* (\bS_{1})\rho_{a}^*(\bS)}\right].
\end{align}
\end{lemma}


Based on Lemma \ref{lem_theta}, consistent estimates of $\theta_{a}$ are expected as long as at least one nuisance model is correctly parametrized at each exposure time. However, this goal has not been achieved yet; see \cite{babino2019multiple} for an overview. The main obstacle is the estimation of interlocking nuisance functions, especially the first-time conditional mean, as it cannot be identified directly through the observable variables as $\mu_{a}(\bs_{1})=E[Y(a)|\bS_{1}=\bs_{1}]\neq E[Y|\bS_{1}=\bs_{1},A_1=a_1]$.
Under Assumption \ref{assumption_1}, existing DTE literature typically considers the following nested representation of $\mu_{a}(\cdot)$,
\begin{align}
\mu_{a}(\bs_{1})=E[Y(a )|\bS_{1}=\bs_{1}]=E[\nu_{a}(\bS )|\bS_{1}=\bs_{1},A_1=a_1],\label{rep:nested}
\end{align}
and suggests a nested regression (NR)  of the conditional means 
-- as long as an estimate $\nuhat_a(\cdot)$ of $\nu_a(\cdot)$ is obtained, one can use $\nuhat_a(\bS_i)$ as the imputed outcomes and perform regression to construct $\muhat_{a,{\mbox{\tiny NR}}}(\cdot)$; see, e.g., \cite{murphy2001marginal}. We formalize these properties under high-dimensional linear working models, naming the resulting DTE estimator the ``dynamic treatment Lasso'' (DTL) estimator. We show that the nested-regression approach faces certain limitations and fails to attain the DR property equivalent to Lemma \ref{lem_theta}. 
Among the multiple factors contributing to this, the biggest one is arising from a peculiar model misspecification
that we identified arising from the nested representation in Equation \eqref{rep:nested}. In the event of a misspecified linear working model $\nu_a^*(\cdot)$, the corresponding $\mu_a^*(\cdot)$ will inevitably be misspecified as well, leading to $\mu_a^*(\cdot) \neq \mu_a(\cdot)$, even when $\mu_a(\cdot)$ is itself linear. Besides the linearity of $\mu_a(\cdot)$, additional conditions on $\nu_a(\cdot)$ are necessary for the correctness of the nested-regression-based linear working model, as discussed in Section \ref{sec:mu-correct}.

This issue necessitates the use of specialized methods for which we propose a new DR representation of the first-time conditional mean function $\mu_a(\cdot)$; see \eqref{def:DR-mu} below. It provides tools to quantify the DR property of the resulting DTE estimate and to develop correction techniques that can mitigate the DR gap by achieving the estimation under model conditions equivalent to Lemma \ref{lem_theta}. 

\begin{theorem}[A DR representation of $\mu_a(\cdot)$]\label{thm:DR-mu}
Suppose that either $\nu_a^*(\cdot)=\nu_a(\cdot)$ or $\rho_a^*(\cdot)=\rho_a(\cdot)$ holds. Let Assumption \ref{assumption_1} holds. Then, for any $\bs_1\in\R^{d_1}$,
\begin{align}\label{def:DR-mu}
\mu_a(\bs_1)=E\left[\nu_a^*(\bS)+\mathbbm1_{\{A_2=a_2\}} \frac{Y-\nu_a^*(\bS)}{\rho_a^*(\bS)}\mid\bS_1=\bs_1,A_1=a_1\right].
\end{align}
\end{theorem}

Utilizing the two DR representations \eqref{rep:DR-theta} and \eqref{def:DR-mu} simultaneously, we propose a \emph{sequential doubly robust Lasso} (S-DRL) estimator. The proposed estimator is consistent as long as either the conditional mean function is truly linear or the propensity score function is truly logistic (or both) for each exposure time. To the best of our knowledge, this is the first estimator that matches Lemma \ref{lem_theta} conditions empirically. The inverse probability weighting (IPW) methods \citep{robins1986new,robins2000marginal,hernan2001marginal,robins2004optimal} require all the propensity score models to be correctly parametrized. The covariate balancing methods \citep{kallus2018optimal,yiu2018covariate,viviano2021dynamic} require all the conditional mean models to be correctly parametrized. Perhaps unexpectedly, the standard low-dimensional DR methods \citep{robins2000robust,murphy2001marginal,bang2005doubly,yu2006double} and the targeted maximum likelihood estimation (TMLE) \citep{van2012targeted} require either all the propensity score functions or all the conditional mean (or density) functions to be correctly parametrized. 
The ``multiply robust'' (MR) estimator of \cite{babino2019multiple} reaches better robustness than all of the aforementioned methods. In general $T$-stage trials, they allow for the first $t$ conditional mean models and the last $T-t$ propensity score models to be correctly parametrized for any $t$. The DTL estimator allows the first $t$ propensity score models and the last $T-t$ conditional mean models to be correctly parametrized. Our S-DRL estimator is strictly more robust in terms of consistency; see Table \ref{table:consistency} and Remark \ref{remark:compare_lowdim} for further details.

\begin{table}
\begin{threeparttable} 
\caption{Consistency of the S-DRL, DTL, and MR estimators in two-stage trials. 
}\label{table:consistency}
\renewcommand{\arraystretch}{1.5} 
\centering
\begin{tabular}{ c | c | c | c | c | c | c }
\Xhline{2\arrayrulewidth}
\multicolumn{4}{c}{Nuisance models}&\multicolumn{3}{c}{Consistency}\\
\hline
logistic $\rho_a(\cdot)$&logistic $\pi_{a}(\cdot)$&linear $\mu_{a}(\cdot)$&linear $\nu_a(\cdot)$&S-DRL&DTL&MR$\;$ \\[-2pt]
\hline
\cmark&\cmark&\cmark&\cmark&\cmark&\cmark&\cmark\\[-2pt]
\hline
\xmark&\cmark&\cmark&\cmark&\cmark&\cmark&\cmark\\[-2pt]
\hline
\cmark&\xmark&\cmark&\cmark&\cmark&\cmark&\cmark\\[-2pt]
\hline
\cmark&\cmark&\xmark&\cmark&\cmark&\cmark&\cmark\\[-2pt]
\hline
\cmark&\cmark&\cmark&\xmark&\cmark&\cmark&\cmark\\[-2pt]
\hline
\xmark&\xmark&\cmark&\cmark&\cmark&\cmark&\cmark\\[-2pt]
\hline
\xmark&\cmark&\xmark&\cmark&\cmark&\cmark&\xmark\\[-2pt]
\hline
\cmark&\xmark&\cmark&\xmark&\cmark&\xmark$\;$
&\cmark\\[-2pt]
\hline
\cmark&\cmark&\xmark&\xmark&\cmark&\cmark&\cmark\\[-2pt]
\Xhline{2\arrayrulewidth}
\end{tabular}

\end{threeparttable}
\end{table}

The S-DRL estimator demonstrates superior estimation rates in high-dimensional contexts when compared to the DTL estimator; see Table \ref{table:consistency-rate}
as well as Remark \ref{remark:compare-muhat}. Root-sample-size inference based on the S-DRL estimator is guaranteed when two product-sparsity conditions are satisfied, whereas the DTL method requires three product-sparsity conditions, as demonstrated in Theorems \ref{thm:dr_theta} and \ref{thm_sparse}. 
 The errors in nuisance estimation at different stages have a parallel effect on the estimation; see the consistency rate in Theorem \ref{thm:misspecified}. 
 


The estimation of the in-between outcome models is intrinsically linked to regression with imputed outcomes.
We have developed a novel cone-set analysis of imputed Lasso estimates that is of independent interest to other imputed, high-dimensional regressions. 
 Existing Lasso proof techniques provide conservative bounds only; see Section \ref{sec:nuisance}.
Our results are adaptive to the imputation error and can be used to guide the selection of tuning parameters in high-dimensional regression models with imputed outcomes.

In the multi-stage exposure setting, we extend our method and develop DR representations to identify both the expected potential outcomes and conditional means, as shown in Section \ref{sec:multi-stage}. While the consistency rate and asymptotic normality require intricate proofs, we anticipate they hold analogously to those in the two-stage case. It is worth noting that Theorem \ref{DR_multi_2} provides new DR representations that are independent of any specific parametric models, allowing the sequential doubly robust (S-DR) method to be utilized with non-parametric nuisance estimates, which enhances its versatility.

\subsection{Organization of the paper}

In Section \ref{sec:DR-est}, we introduce the DR estimators of the DTE, including the proposed S-DRL estimator, the DTL estimator, and a general DR estimator. The theoretical properties of the considered DTE estimators are established in Section \ref{sec:asymp}. In Section \ref{sec:nuisance}, we formalize the supporting theoretical discoveries, including a general theory for imputed Lasso estimation and the consistency results of the nuisance estimates. We further extend our setting to the case of multi-stage treatments and provide general DR representations for the intermediate conditional means in Section \ref{sec:multi-stage}. Section \ref{sec:num} presents numerical results, including simulation studies and an application to the National Job Corps Study. Further discussion is provided in Section \ref{sec:dis}.

\subsection{Notation}
For any $\alpha>0$, let $\psi_{\alpha}(\cdot)$ denote the function given by $\psi_{\alpha}(x):=\exp(x^\alpha)-1$, $\forall x>0$. Then the $\psi_{\alpha}$-Orlicz norm $\|\cdot\|_{\psi_{\alpha}}$ of a random variable $X$ is defined as
$
	\|X\|_{\psi_{\alpha}}
	:=\inf\{c>0:E[\psi_{\alpha}(|X|/c)]\leq 1\}.
$
Two special cases of finite $\psi_{\alpha}-$Orlicz norm are given by 
$\psi_{2}(x)=\exp(x^2)-1$ and $\psi_{1}(x)=\exp(x)-1$, which correspond to sub-Gaussian and sub-exponential random variables, respectively.
The notation $a_N\ll b_N$ denotes $a_N=o(b_N)$, and $a_N\gg b_N$ denotes $b_N\ll a_N$ as $N \to \infty$. The notation $a_N\asymp b_N$ denotes $cb_N\leq a_N \leq Cb_N$ for all $N \geq 1$ and with constants $c,C>0$. 
Define $g(u)= \exp(u)/\{1+\exp(u)\}$ as the logistic function and $\phi(u)=\log(1+\exp(u))$ as the corresponding link function throughout. 

\section{The doubly robust estimators}\label{sec:DR-est}

We observe a collection of independent and identically distributed (i.i.d.) samples $\mathcal D:=\{W_{i}\}_{i=1}^{N}=(Y_{i},\bS_{1i},A_{1i},\bS_{2i},A_{2i})_{i=1}^{N}$, drawn from the same distribution as $(Y,\bS_1,A_1,\bS_2,A_2)$.
In the following subsections, we present three DTE estimators: the new sequential doubly robust Lasso (S-DRL) estimator, the dynamic treatment Lasso (DTL) estimator, and the general DR estimator.

\subsection{The sequential doubly robust Lasso (S-DRL) estimator}\label{sec:DDRL}

We focus on the high-dimensional scenario, and consider linear (working) models for the conditional means $\mu_a(\cdot)$ and $\nu_a(\cdot)$, along with logistic (working) models for the propensities $\pi_a(\cdot)$ and $\rho_a(\cdot)$. 
The population minimizer approximating $\pi_{a} (\bs_1)$ is defined as $\pi_{a}^{*}(\bs_{1}) = g(\bv^{\top}\bm{\gamma}_{a}^{*})$ with $\bv=(1, \bs_{1}^{\top})^{\top}$, whereas that of approximating $\rho_{a}(\bs )$ is $ \rho_{a}^{*}(\bs )=g(\bu^{\top}\bm{\delta}_{a}^{*})$ with $\bu=(1, \bs^{\top})^{\top}$. Here
\begin{align}\label{eq:gamma_star}
	\bm{\gamma}_{a}^{*}&=\argminn_{\bm{\gamma}\in \mathbb{R}^{d_1+1}}E \left[\phi(\bV^{\top}\bm{\gamma})-\mathbbm1_{\{A_{1}=a_1\}} \bV^{\top}\bm{\gamma}\right], \ \bV=(1, \bS_{1}^{\top})^{\top} \in \mathbb{R}^{d_1+1} \mbox{ and }\\
	\bm{\delta}_a^{*}&=\argminn_{\bm{\delta}\in \mathbb{R}^{d+1}}E\left[\mathbbm1_{\{A_{1}=a_1\}}\left[\phi(\bU^{\top}\bm{\delta})-\mathbbm1_{\{A_2=a_2\}} \bU^{\top}\bm{\delta}\right]\right], \ \bU=(1, \bS)^{\top} \in \mathbb{R}^{d+1}. \label{eq:delta_star}
\end{align}
One can also consider a feature map $\varphi(\bs_1)$ (e.g., a polynomial basis) and a working model $\pi_a^*(\bs_1)=g(\varphi(\bs_1)^\top\bgamma_a^*)$ with some $\bgamma_a^*$ defined correspondingly. We focus on $\varphi(\bs_1)=\bv$, although the results apply more broadly. 
The above working models can be estimated with many regularizations. Throughout this work, we focus on the $\ell_1$-regularization, albeit the theoretical developments apply more broadly. With a subset of training data $\mathcal D_\mathcal J=\{W_i\}_{i\in\mathcal J}\subset\mathcal D$, where $\mathcal J\subset\{1,\dots,N\}$, we define
\begin{equation}
\widehat{\bm{\gamma}}_{a}:=\widehat{\bm{\gamma}}_{a}(\mathcal D_\mathcal J)=\argminn_{\bm{\gamma}\in \mathbb{R}^{d_1+1}} \frac{1}{|\mathcal J|} \sum_{i\in\mathcal J}
\bigl[\phi(\bV_{i}^{\top}\bm{\gamma})-\mathbbm1_{\{A_{1i}=a_1\}} \bV_{i}^{\top}\bm{\gamma}\bigl] +\lambda_{\bm{\gamma}}\|\bm{\gamma}\|_1,\label{207}
\end{equation}
\vspace{-0.5em}
\begin{equation}
\widehat{\bm{\delta}}_{a}:=\widehat{\bm{\delta}}_{a}(\mathcal D_\mathcal J)=\argminn_{\bm{\delta}\in \mathbb{R}^{d+1}} \frac{1}{|\mathcal J|} \sum_{i\in\mathcal J} \mathbbm1_{\{A_{1i}=a_1\}}\bigl[\phi(\bU_{i}^{\top}\bm{\delta})-\mathbbm1_{\{A_{2i}=a_2\}}\bU_{i}^{\top}\bm{\delta}\bigl]+{\lambda}_{\bm{\delta}} \|\bm{\delta}\|_1,\label{226}
\end{equation}
with tuning parameters $ \lambda_{\bm{\gamma}},{\lambda}_{\bm{\delta}} >0$. Observe that for $ \widehat{\bm{\gamma}}_{a}$, we utilize all of the observations regardless of its treatment path, whereas for $ \widehat{\bm{\delta}}_{a}$, only those whose treatment path matches $a_1$ regardless of what $a_2$ is. 
The best linear working model for the second-time conditional mean $\nu_a(\cdot)=E[Y|\bS,A_1=a_1,A_2=a_2]$ is denoted as 
\begin{align} 
	\label{eq:alpha_star}
	\nu_{a}^{*}(\bs )=\bu^{\top} \bm{\alpha}_{a}^{*},\quad
	\balpha_a^*:=\argminn_{\balpha\in\R^{d+1}} E\left[\mathbbm1_{\{A_{1}=a_1,A_{2}=a_2\}}(Y-\bU^\top\balpha)^2\right].
\end{align} 
An estimator of \eqref{eq:alpha_star} can be obtained similarly with $ {\lambda}_{\bm{\alpha}}> 0$:
\begin{align}\label{51}
\widehat{\bm{\alpha}}_{a}:=\widehat{\bm{\alpha}}_{a}(\mathcal D_\mathcal J)=\argminn_{\bm{\alpha}\in\mathbb{R}^{d+1}} \frac{1}{|\mathcal J|} \sum_{i\in\mathcal J} \mathbbm1_{\{A_{1i}=a_1,A_{2i}=a_2\}}({Y}_i-{\bU}_i^{\top} \bm{\alpha})^2+ {\lambda}_{\balpha}\|\bm{\alpha}\|_1 .
\end{align}

\begin{algorithm}\caption{Sequential Double Robust Lasso (S-DRL)}\label{alg:S-DRL}
\begin{algorithmic}[1]
\Require Observations $\mathcal D:=\{W_{i}\}_{i=1}^{N}=(Y_{i},\bS_{1i},A_{1i},\bS_{2i},A_{2i})_{i=1}^{N}$, treatment path $a$, and control $a'$.
\State For any $K\geq2$, let $\mathcal{K}=\{1,2,\dots,K \}$. Randomly split $\mathcal I=\{1,\dots,N\}$ into $K$ equal-sized $|\mathcal I_{k}|=n$. Define $\mathcal I_{-k}:=\mathcal I\backslash\mathcal I_{k}$, and further split $\mathcal I_{-k}$ into two equal-sized sets $\mathcal I_{-k,1}$ and $\mathcal I_{-k,2}$. 
\State Let $\mathcal W_{-k}:=\{W_i\}_{i\in\mathcal I_{-k}}$ and $\mathcal W_{-k,j}:=\{W_i\}_{i\in\mathcal I_{-k,j}}$ for each $j\in\{1,2\}$.
\For{$k=1,2,...,K$} 
\For{$c \in \{a,a'\}$}
\State Using $\mathcal W_{-k}$, construct estimates $\widehat{\bgamma}_c$, $\widehat{\bdelta}_c$, and $\widehat{\balpha}_c$ through \eqref{207}, \eqref{226}, and \eqref{51}, respectively. 
\State Using $\mathcal W_{-k,1}$, construct estimates $\widetilde{\bdelta}_c$ and $\widetilde{\balpha}_c$ through \eqref{226} and \eqref{51}, respectively.
\State For each $i\in\mathcal I_{-k,2}$, set $\Yhat_i^{\mbox{\tiny DR}}$ as defined in \eqref{eq:DRY} with $\widetilde{\bdelta}_c$ and $\widetilde{\balpha}_c$ from Step 6.
\State Compute $\widehat{\bm{\beta}}_{c,1}$ through \eqref{eq:beta:DR} based on the training samples $\mathcal W_{-k,2}$.
\State Exchange $\mathcal W_{-k,1}$ and $\mathcal W_{-k,2}$, repeat Steps 6-8 and obtain $\widehat{\bm{\beta}}_{c,2}$ analogously. Compute
\begin{align}\label{def:betahat-avg}
\widehat{\bm{\beta}}_{c}=(\widehat{\bm{\beta}}_{c,1}+\widehat{\bm{\beta}}_{c,2})/2.
\end{align}
\EndFor
\State Let $\hat\eta_c=(\widehat{\balpha}_c,\widehat{\bbeta}_c,\widehat{\bgamma}_c,\widehat{\bdelta}_c)$. Using the DR score \eqref{def:DR-score}, compute $\check{\theta}^{(k)}$ as 
\begin{align*}
\check{\theta}^{(k)}=|\mathcal I_{k}|^{-1}\sum_{i\in\mathcal I_{k}}\left[\psi_a(W_i;\hat\eta_a)-\psi_{a'}(W_i;\hat\eta_{a'})\right].
\end{align*}
\EndFor \quad 
\Return The S-DRL estimator and the variance estimate 
\begin{align}
\widehat{\theta}:= K^{-1}\sum_{k\in \mathcal{K}}\check{\theta}^{(k)},\quad\widehat{\sigma}^2:= {N}^{-1}\sum_{k\in \mathcal{K}, i\in\mathcal I_k} \left[\psi_a(W_i;\hat{\eta}_a)-\psi_{a'}(W_i;\hat{\eta}_{a'})-\widehat{\theta}\right]^2.\label{def:sigma-hat}
\end{align} 
\end{algorithmic}
\end{algorithm}

\paragraph*{Estimation of the first-time conditional mean $\mu_a(\cdot)$}

Recall that $\bdelta_a^*$ and $\balpha_a^*$ are defined in Equations \eqref{eq:delta_star} and \eqref{eq:alpha_star}, respectively. We propose the following DR imputed outcome 
\begin{align*}
Y^{\mbox{\tiny DR}}:= {\bU}^{\top} {\bm{\alpha}}_{a}^*+\mathbbm1_{\{A_{2}=a_2\}} \frac{Y-{\bU}^{\top} {\bm{\alpha}}_{a}^* }{g(\bU^{\top}\bdelta_a^*)}.
\end{align*}
With this in mind, we consider a linear working model for the first-time conditional mean
\begin{align}
\mu_{a}^*(\bs_{1})=\bv^{\top}\bm{\beta}_{a}^{*},\quad
\bm{\beta}_{a}^{*}:=\argminn_{\bm{\beta}\in \mathbb{R}^{d_1+1}}E\left[\mathbbm1_{\{A_{1}=a_1\}}(Y^{\mbox{\tiny DR}}-{\bV}^{\top}\bm{\beta})^2\right].\label{def:beta-star}
\end{align}
To estimate the best linear slope $\bbeta_a^*$ based on a subset of training data $\mathcal D_\mathcal J\subset\mathcal D$, we consider an additional sample splitting with $\mathcal D_\mathcal J=\mathcal D_{\mathcal J_1}\cup\mathcal D_{\mathcal J_2}$, where $\mathcal J_1$ and $\mathcal J_2$ are disjoint subsets of $\mathcal J$. 
Using the first half of the subsamples $\mathcal D_{\mathcal J_1}$, we first obtain the second-time nuisance estimates $\bdeltatil_a:=\widehat{\bdelta}_{a}(\mathcal D_{\mathcal J_1})$ and $\widetilde{\balpha}_a:=\balphahat_a(\mathcal D_{\mathcal J_1})$ as \eqref{226} and \eqref{51}, respectively.
Then, for each $i\in\mathcal J_2$, we construct a DR imputed outcome
\begin{align}\label{eq:DRY}
	\widehat Y^{\mbox{\tiny DR}}_i&:={\bU}_i^{\top} \widetilde{\balpha}_a+\mathbbm1_{\{A_{2i}=a_2\}}\frac{Y_i-{\bU}_i^{\top}\widetilde{\balpha}_a}{g(\bU_i^\top\bdeltatil_a)}.
\end{align}
Based on the DR imputed outcomes $\widehat Y^{\mbox{\tiny DR}}_{\mathcal J_2}:=\{\widehat Y^{\mbox{\tiny DR}}_i\}_{i\in\mathcal J_2}$, we propose a DR estimate:
\begin{align}\label{eq:beta:DR}
\bbetahat_a:=\bbetahat_a(\mathcal D_{\mathcal J_2},\widehat Y^{\mbox{\tiny DR}}_{\mathcal J_2})=\argminn_{\bbeta\in\R^{d_1+1}} \frac{1}{|\mathcal J_2|} \sum_{i\in\mathcal J_2}\mathbbm1_{\{A_{1i}=a_1\}}\left(\widehat Y^{\mbox{\tiny DR}}_i-{\bV}_i^{\top}\bm{\beta} \right)^2+ {\lambda}_{\bm{\beta}}\|\bm{\beta}\|_1 , 
\end{align}
where $ {\lambda}_{\bm{\beta}}>0$. To regain full sample size efficiency, we can always swap the samples $\mathcal D_{\mathcal J_1}$ and $\mathcal D_{\mathcal J_2}$, repeat the procedure, and average the results. 

\paragraph*{The S-DRL estimator of the DTE}
For each $c\in\{a,a'\}$ and for any $\eta=(\balpha,\bbeta,\bgamma,\bdelta)$, define the DR score function based on the DR representation \eqref{rep:DR-theta}:
\begin{align}\label{def:DR-score}
\psi_c(W;\eta):=\bV^{\top}\bbeta+\mathbbm1_{\{A_{1}=c_1\}}\frac{\bU^{\top}\balpha-\bV^{\top}\bbeta}{g(\bV^{\top}\bgamma)}+\mathbbm1_{\{A_{1}=c_1,A_{2}=c_2\}}\frac{Y-\bU^{\top}\balpha}{g(\bV^{\top}\bgamma)g(\bU^{\top}\bdelta_c)}.
\end{align}
We propose the \emph{sequential doubly robust Lasso} (S-DRL) estimator of $\theta$: 
\begin{align*}
\thetahat:=N^{-1}\sum_{i=1}^N[\psi_a(W_i;\hat\eta_a)-\psi_{a'}(W_i;\hat\eta_{a'})],
\end{align*}
where $\hat\eta_c:=(\widehat{\balpha}_c,\widehat{\bbeta}_c,\widehat{\bgamma}_c, \widehat{\bdelta}_c)$ are the nuisance estimates of \eqref{51}, \eqref{eq:beta:DR}, \eqref{207}, and \eqref{226}, respectively. A cross-fitting technique is used. The details are provided in Algorithm \ref{alg:S-DRL}; see \cite{chernozhukov2018double,smucler2019unifying} where the cross-fitting leads to weaker sparsity restrictions than those without it, such as \cite{farrell2015robust,tan2020model}.

\subsection{The dynamic treatment Lasso (DTL) estimator}\label{sec:DTL}
\vspace{-1em}
\begin{algorithm}\caption{Dynamic Treatment Lasso (DTL)}\label{alg:DTL}
\begin{algorithmic}[1]
\Require Observations $\{W_i\}_{i=1}^N$, number of cross-fitting subsets $K\geq 2$, treatment path $a$, and control $a'$.
\State For any $K\geq2$, let $\mathcal{K}=\{1,2,\dots,K \}$.  Randomly split $\mathcal I=\{1,\dots,N\}$ into $K$ equal-sized $|\mathcal I_{k}|=n$ with $\mathcal I_{-k}=\mathcal I\backslash\mathcal I_{k}$ and $\mathcal W_{-k}=\{W_i\}_{i\in\mathcal I_{-k}}$.
\For{$k=1,2,...,K$} 
\For{$c \in \{a,a'\}$}
\State Using $\mathcal W_{-k}$, construct estimates $\widehat{\bgamma}_c$, $\widehat{\bdelta}_c$, and $\widehat{\balpha}_c$ through \eqref{207}, \eqref{226}, and \eqref{51}, respectively. 
\State Compute $\widehat{\bm{\beta}}_{c,{\mbox{\tiny NR}}}$ as \eqref{52} based on the training samples $\mathcal W_{-k}$ and the nuisance estimate $\widehat{\balpha}_c$.
\EndFor
\State Using the DR score \eqref{def:DR-score} and let $\hat\eta_{c,{\mbox{\tiny NR}}}=(\widehat{\balpha}_c,\widehat{\bbeta}_{c,{\mbox{\tiny NR}}},\widehat{\bgamma}_c,\widehat{\bdelta}_c)$, compute $\check{\theta}_{\tiny\mbox{DTL}}^{(k)}$ as 
\begin{align*}
\check{\theta}_{\tiny\mbox{DTL}}^{(k)}=|\mathcal I_{k}|^{-1}\sum_{i\in\mathcal I_{k}}\left[\psi_a(W_i;\hat\eta_{a,{\mbox{\tiny NR}}})-\psi_{a'}(W_i;\hat\eta_{a',{\mbox{\tiny NR}}})\right].
\end{align*} 
\EndFor \qquad 
\Return The DTL estimator and the  variance estimate
\begin{align}
\widehat{\theta}_{\tiny\mbox{DTL}}:=K^{-1}\sum_{k\in \mathcal{K}}\check{\theta}_{\tiny\mbox{DTL}}^{(k)},\;\;\;\widehat{\sigma}_{\tiny\mbox{DTL}}^2:= {N}^{-1}\sum_{k\in \mathcal{K}, i\in\mathcal I_k} \left[\psi_{a}(W_i;\hat{\eta}_{a,{\mbox{\tiny NR}}})-\psi_{a'}(W_i;\hat{\eta}_{a',{\mbox{\tiny NR}}})-\widehat{\theta}_{\tiny\mbox{DTL}}\right]^2.\label{def:sigma-hat-seq}
\end{align} 
\end{algorithmic}
\end{algorithm} 
In this section, we formally define a dynamic treatment Lasso (DTL) estimator based on the DR score \eqref{rep:DR-theta} of $\theta_a$ and the nested representation \eqref{rep:nested} of $\mu_a(\cdot)$. 
Here, $\ell_1$-regularized nuisance estimates $\widehat{\bgamma}_a$, $\widehat{\bdelta}_a$, and $\widehat{\balpha}_a$ will be the same as before; see \eqref{207}, \eqref{226}, and \eqref{51} above. Estimation of the first-time conditional mean model is different. 
 Based on the best linear approximation $\nu_{a}^{*}(\bs)=\bu^{\top} \bm{\alpha}_{a}^{*}$, \eqref{eq:alpha_star}, of $\nu_a(\cdot)$, we introduce the following nested ``best linear working model'':
\begin{align}
\mu_{a,{\mbox{\tiny NR}}}^*(\bs_{1})=\bv^{\top}\bm{\beta}_{a,{\mbox{\tiny NR}}}^{*},\quad
\bm{\beta}_{a,{\mbox{\tiny NR}}}^{*}:=\argminn_{\bm{\beta}\in \mathbb{R}^{d_1+1}}E\left[\mathbbm1_{\{A_{1}=a_1\}}(\bU^\top\balpha_a^*-{\bV}^{\top}\bm{\beta})^2\right].\label{def:beta-star-seq}
\end{align}
Note that the two linear working models $\mu_{a,{\mbox{\tiny NR}}}^*(\cdot)$ and $\mu_a^*(\cdot)$ are not necessarily the same; see Section \ref{sec:mu-correct} for detailed comparisons.
 We consider the following imputed Lasso estimate of $\bm{\beta}_{a,{\mbox{\tiny NR}}}^{*}$, defined as 
 $\widehat{\bm{\beta}}_{a,{\mbox{\tiny NR}}}:=\widehat{\bm{\beta}}_{a,{\mbox{\tiny NR}}}(\mathcal D_\mathcal J,\balphahat_a)$ with
\begin{align}\label{52}
\widehat{\bm{\beta}}_{a,{\mbox{\tiny NR}}}:= \argminn_{\bm{\beta}\in \mathbb{R}^{d_1+1}} \frac{1}{|\mathcal J|} \sum_{i\in\mathcal J} \mathbbm1_{\{A_{1i}=a_1\}}\bigl ( {\bU}_i^{\top} \widehat{\bm{\alpha}}_{a}-{\bV}_i^{\top}\bm{\beta} \bigl)^2+ {\lambda}_{\bbeta} \|\bm{\beta}\|_1.
\end{align}
Now we introduce the \emph{dynamic treatment Lasso} (DTL) estimator of $\theta$:
\begin{align*}
\thetahat{\tiny\mbox{DTL}}:=N^{-1}\sum_{i=1}^N[\psi_a(W_i;\hat\eta_{a,{\mbox{\tiny NR}}})-\psi_{a'}(W_i;\hat\eta_{a',{\mbox{\tiny NR}}})],
\end{align*}
where $\psi_c(\cdot;\cdot)$ is defined in \eqref{def:DR-score} and $\hat\eta_{c,{\mbox{\tiny NR}}}:=(\widehat{\balpha}_c,\widehat{\bbeta}_{c,{\mbox{\tiny NR}}},\widehat{\bgamma}_c, \widehat{\bdelta}_c)$ are the nuisance estimates as in \eqref{51}, \eqref{52}, \eqref{207}, and \eqref{226}, respectively; see Algorithm \ref{alg:DTL} for details.


\subsection{Comparisons between the first-time working models $\mu_a^*(\cdot)$ and $\mu_{a,{\mbox{\tiny NR}}}^*(\cdot)$}\label{sec:mu-correct}
In the dynamic treatment setting, the relationship between the linear conditional mean function, $\mu_a(\cdot)$, and its corresponding approximations, $\mu_a^*(\cdot)$ and $\mu_{a,{\mbox{\tiny NR}}}^*(\cdot)$, obtained via different identification strategies, is not straightforward. Specifically, a linear $\mu_a(\cdot)$ is only a necessary condition for correctly specified linear working models; it does not guarantee
 equality. Additional conditions are required to ensure that $\mu_a^*(\cdot)=\mu_a(\cdot)$ and $\mu_{a,{\mbox{\tiny NR}}}^*(\cdot)=\mu_a(\cdot)$. In the following, we will discuss these necessary conditions in detail.
\vspace{0.5em}

\paragraph*{The working model $\mu_a^*(\cdot)$}
The proposed S-DRL approach utilizes the following identification $\mu_a(\bs_1)=E[Y^{\mbox{\tiny DR}}\mid\bS_1=\bs_1,A_1=a_1]$. This representation remains valid for a linear $\mu_a(\cdot)$ as long as either $\rho_a(\cdot)$ is truly logistic or $\nu_a(\cdot)$ is truly linear and not necessarily both. In this case, 
we also have the following equivalent expressions of $\bbeta_a^*$:
\begin{align*}
\bm{\beta}_{a}^{*} &=\argminn_{\boldsymbol{\beta}\in \mathbb{R}^{d_1+1}}E\left[(Y(a)-{\bV}^{\top}\bm{\beta})^2\mid A_1=a_1\right]=\argminn_{\bm{\beta}\in \mathbb{R}^{d_1+1}}E\left[(\mu_a(\bS_1)-{\bV}^{\top}\bm{\beta})^2 \mid A_1=a_1\right] .
\end{align*}
That is, $\mu_a^*(\bs_1)=\bv^{\top}\bbeta_a^*$ satisfies $\mu_a^*(\cdot)=\mu_a(\cdot)$. Hence, $\mu_a^*(\cdot)=\mu_a(\cdot)$ whenever 
(a) $\mu_a(\cdot)$ is a linear function and
(b) either $\rho_a(\cdot)$ is a logistic function or $\nu_a(\cdot)$ is a linear function.
It is worth noting that Condition (b) is already a prerequisite for the identification of DTE, as stated in Lemma \ref{lem_theta}. Consequently, there is no need to introduce any other conditions beyond those outlined in Condition (a).

\paragraph*{The working model $\mu_{a,{\mbox{\tiny NR}}}^*(\cdot)$} The DTL estimator relies on the nested-regression identification for which Conditions (a) and (b) above are insufficient for $\mu_{a,{\mbox{\tiny NR}}}^*(\cdot)=\mu_a(\cdot)$; see Example \ref{ex:mis-mu-seq} below. Additional conditions are needed. For instance, $\mu_{a,{\mbox{\tiny NR}}}^*(\cdot)=\mu_a(\cdot)$ if we further assume the following:
\begin{itemize}
\item[(c)] Under the treatment groups $(A_1,A_2)=(a_1,a_2)$ and $A_1=a_1$, the best linear slopes are the same while regressing $Y(a)$ on $\bU$, i.e., $\balpha_a^*=\bar{\bm{\alpha}}_{a}^{*}$, with $\balpha_a^*$
defined in \eqref{eq:alpha_star} and
\begin{align*}
\bar\balpha_a^*:=\argminn_{\balpha\in\R^{d+1}}E\left[(Y(a)-\bU^\top\balpha)^2|A_{1}=a_1\right].
\end{align*}
\end{itemize}
One sufficient but not necessary condition for (c) is that the second-time conditional mean function $\nu_a(\cdot)$ is also linear; see further justifications in Section A of Supplementary Materials \citep{bradic2023supplement}.


\paragraph*{The misspecification errors of $\mu_a^*(\cdot)$ and $\mu_{a,{\mbox{\tiny NR}}}^*(\cdot)$}
Now we consider the case where $\mu_a(\cdot)$ is possibly non-linear and compare the misspecification (approximation) errors of the linear working models $\mu_a^*(\cdot)$ and $\mu_{a,{\mbox{\tiny NR}}}^*(\cdot)$. As long as Condition (b) above holds, we have
\begin{align*}
\mathrm{Err}_{\mbox{\tiny NR}}~:=~&E\left[(\mu_a(\bS_1)-\bV^{\top}\bbeta_{a,{\mbox{\tiny NR}}}^*)^2|A_{1}=a_1\right]\\
=~&E\left[(\mu_a(\bS_1)-\bV^{\top}\bbeta_a^*)^2|A_{1}=a_1\right]+E\left[(\bV^{\top}(\bbeta_{a,{\mbox{\tiny NR}}}^*-\bbeta_{a}^*))^2|A_1=a_1\right]\\
\geq~&\mathrm{Err}_{\mbox{\tiny DR}}:=E\left[(\mu_a(\bS_1)-\bV^{\top}\bbeta_a^*)^2|A_{1}=a_1\right].
\end{align*}
Hence, we have the following conclusions.
(1) When $\nu_a(\cdot)$ is linear, $\mu_{a,{\mbox{\tiny NR}}}^*(\cdot)=\mu_a^*(\cdot)$, and both of them are the best linear approximations of the true conditional mean $\mu_a(\cdot)$ among the group $A_1=a_1$, i.e., $\mathrm{Err}_{\mbox{\tiny NR}}=\mathrm{Err}_{\mbox{\tiny DR}}$.
(2) When $\nu_a(\cdot)$ is non-linear and $\rho_a(\cdot)$ is logistic, $\mu_{a,{\mbox{\tiny NR}}}^*(\cdot)\neq\mu_a^*(\cdot)$ and $\bbeta_{a,{\mbox{\tiny NR}}}^*\neq\bbeta_{a}^*$ in general. If Assumption \ref{assumption_U} holds, the strict inequality above holds in that $\mathrm{Err}_{\mbox{\tiny NR}}>\mathrm{Err}_{\mbox{\tiny DR}}$ as long as $\bbeta^*_{a,{\mbox{\tiny NR}}}\neq\bbeta_{a}^*$. That is, the nested ``best linear approximation'', $\mu_{a,{\mbox{\tiny NR}}}^*(\cdot)$, is in general sub-optimal. If we further consider the case where $\mu_a(\cdot)$ is linear, then we have $\mathrm{Err}_{\mbox{\tiny DR}}=0$ and possibly $\mathrm{Err}_{\mbox{\tiny NR}}>0$; see also an illustration in Example \ref{ex:mis-mu-seq} below.

\begin{example}[A misspecified linear model when the truth is indeed linear]\label{ex:mis-mu-seq}
In what follows, we turn our attention to the nested-regression approach and offer an illustrative example that demonstrates how $\mu_{a,{\mbox{\tiny NR}}}^*(\cdot)\neq\mu_a(\cdot)$ for a linear $\mu_a(\cdot)$, a logistic $\rho_a(\cdot)$, a non-linear $\nu_a(\cdot)$, and a non-logistic $\pi_a(\cdot)$. To facilitate our discussion, we consider the case where both $\bS_1$ and $\bS_2$ are one-dimensional covariates with supports in $\R$. We assume $\bS_1$ follows a uniform distribution on $[-1,1]$, and consider an independent $\delta$ that follows a Bernoulli distribution with a probability of success $0.5$.
We further assume 
\begin{align*}
&P(A_1=a_1\mid\bS_1)=\bS_1/2+0.5,\quad\bS_2=\bS_1(2\delta-1),\\
&P(A_2=a_2\mid\bS,A_1=a_1)=g(\bS_2)=\exp(\bS_2)/[1+\exp(\bS_2)],\quad Y(a)=2\delta-1.
\end{align*}
It is not difficult to verify that $\mu_a(\cdot)\equiv0$ and $\mu_{a,{\mbox{\tiny NR}}}^*(\bs_1)=\bv^{\top}\bbeta_{a,{\mbox{\tiny NR}}}^*$, where $\bbeta_{a,{\mbox{\tiny NR}}}^*\approx(-0.09,0.3)^{\top}$. Notably, the nested-regression approach based on a linear working model $\mu_{a,{\mbox{\tiny NR}}}^*(\cdot)$ is misspecified, even though the true conditional mean is linear. Furthermore, because both the first-time conditional mean and propensity score models are misspecified, the DR representation of $\theta_a$ in Equation \eqref{rep:DR-theta} is no longer valid when the nested-regression-based working model $\mu_{a,{\mbox{\tiny NR}}}^*(\cdot)$ is considered. However, since the second-time propensity score $\rho_a(\cdot)$ is truly logistic, the S-DRL approach leads to $\mu_a^*(\cdot)=\mu_a(\cdot)\equiv0$ and the incidental validity of the DR representation of $\theta_a$ through Equation \eqref{rep:DR-theta}.
\end{example}

\subsection{The general DR DTE estimator}\label{sec:DR-gen}


In this section, we present a general doubly robust (DR) estimator of the DTE. We assume that we have access to estimators $\nuhat_a(\cdot)$, $\muhat_a(\cdot)$, $\pihat_a(\cdot)$, and $\rhohat_a(\cdot)$ of $\nu_a(\cdot)$, $\mu_a(\cdot)$, $\pi_a(\cdot)$, and $\rho_a(\cdot)$, respectively. The functions $\nu_a(\bs) $, $\pi_a(\bs_1) $, and $\rho_a(\bs) $ can be directly estimated using observable variables, while the remaining nuisance function $\mu_a(\cdot)$ can be identified using either the proposed DR representation \eqref{def:DR-mu} or the usual nested representation \eqref{rep:nested}. We consider flexible estimation strategies for all nuisance functions, including both parametric and non-parametric methods.
Using the DR representation of $\theta_a$ given by \eqref{rep:DR-theta}, we propose a general DR estimator of the DTE through a cross-fitting procedure. For any $K\geq2$, randomly split $\mathcal I=\{1,\dots,N\}$ into $K$ equal-sized parts with $|\mathcal I_{k}|=n=N/K$. For the sake of simplicity, we consider $n$ as an integer.
 Based on the training samples $\mathcal W_{-k}$, construct $\nuhat_{c,-k}(\cdot)$, $\muhat_{c,-k}(\cdot)$, $\pihat_{c,-k}(\cdot)$, and $\rhohat_{c,-k}(\cdot)$ as estimates of the nuisance functions $\nu_c(\cdot)$, $\mu_c(\cdot)$, $\pi_c(\cdot)$, and $\rho_c(\cdot)$, respectively. For each $c\in\{a,a'\}$, let
\begin{align}
\widehat{\psi}_{c,-k}(W)\hskip-2pt&:=\muhat_{c,-k}(\bS_{1})\hskip-2pt+\hskip-2pt\mathbbm1_{\hskip-1pt\tiny{\{A_{1}=c_1\hskip-1pt\}}}\hskip-2pt\frac{\nuhat_{c,-k}(\bS)-\muhat_{c,-k}(\bS_{1}) }{\pihat_{c,-k}(\bS_{1})}\label{def:psi-hat}\hskip-2pt+\hskip-2pt\mathbbm1_{\hskip-1pt\tiny{\{A_{1}=c_1\hskip-1pt,A_{2}=c_2\hskip-1pt\}}}\hskip-2pt\frac{Y-\nuhat_{c,-k}(\bS)}{\pihat_{c,-k}(\bS_{1})\rhohat_{c,-k}(\bS)}. \end{align}
The general DR DTE estimator and the corresponding variance estimate are then defined with $\widehat \Delta_{-k}(\cdot) = \widehat{\psi}_{a,-k} (\cdot) -\widehat{\psi}_{a',-k} (\cdot) $ and $\mathcal{K}=\{1,\cdots, K\}$ as
\begin{align}
\thetahat_\mathrm{gen} :=\frac{1}{N}\sum_{k\in \mathcal{K}, i\in\mathcal I_k} \widehat \Delta_{-k} (W_i) ,\label{def:theta-gen} 
\qquad \sigmahat_\mathrm{gen}^2 :=\frac{1}{N}\sum_{k\in \mathcal{K}, i\in\mathcal I_k} [\widehat \Delta_{-k} (W_i) -\thetahat_\mathrm{gen} ]^2.
\end{align}


\section{Asymptotic properties}\label{sec:asymp}
Here we establish consistency and asymptotic normality of the S-DRL, DTL, and the general DR estimator.
 \subsection{Properties of the S-DRL estimator}\label{sec:theory-DDRL}
We use $s_{\bm{\alpha}_{a}}:=\|\balpha_a^*\|_0$, $s_{\bm{\beta}_{a}}:=\|\bbeta_a^*\|_0$, $s_{\bm{\gamma}_{a}}:=\|\bgamma_a^*\|_0$, and $s_{\bm{\delta}_{a}}:=\|\bdelta_a^*\|_0$ to denote sparsity levels of the nuisance parameters as defined in \eqref{eq:alpha_star} , \eqref{def:beta-star}, \eqref{eq:gamma_star} and \eqref{eq:delta_star}, respectively. The number of covariates, $d_1$ and $d$, are possibly much larger than $N$; for simplicity, we consider $d_1\asymp d_2\asymp d:=d_1+d_2$.
\begin{assumption}\label{assumption_residual}
Define
$ \zeta_a:=\mathbbm1_{\{A_1=a_1,A_2=a_2\}} (Y(a)-\nu_{a}^{*}(\bS ) ) $, 
$ \varepsilon_a:=\mathbbm1_{\{A_1=a_1\}}$$ (\nu_{a}^{*}(\bS )-\mu_{a}^{*}(\bS_{1}) )$ and let $\zeta:=\zeta_a+\zeta_{a'}$, $\varepsilon:=\varepsilon_a+\varepsilon_{a'}$.
Suppose that there exist positive constants $\sigma_{\zeta}< \infty$ and $\sigma_{\varepsilon}< \infty$, such that $\zeta$ and $\varepsilon$ are sub-Gaussian, with $\|\zeta \|_{\psi_2} \leq \sigma\sigma_{\zeta}$, $\|\varepsilon\|_{\psi_2} \leq \sigma\sigma_{\varepsilon}$, and
\begin{align}
\sigma^{2}:=E[\psi_a(W;\eta_a^*)-\psi_{a'}(W;\eta_{a'}^*)-\theta]^2,\label{def:sigma}
\end{align}
where $\psi_c(\cdot;\cdot)$ is defined in \eqref{def:DR-score} and $\eta_c^*:=(\balpha_c^*,\bbeta_c^*,\bgamma_c^*,\bdelta_c^*)$.
\end{assumption}
\begin{assumption}\label{assumption_U}
Let $\bU$ be a sub-Gaussian vector such that $\|\bm{x}^\top\bU \|_{\psi_2}\leq \sigma_{u}\|\bm{x}\|_2$ for $\bm{x}\in \mathbb{R}^{d+1}$ and $\sigma_u>0$. Let $\lambda_{\min}(E[\bU\bU^{\top}\mathbbm{1}_{\{A_1=a_1\}}])\geq \kappa_l$ for any $a_1\in\{0,1\}$, with $\kappa_l>0$. 
\end{assumption}

 Assumptions \ref{assumption_residual} and \ref{assumption_U} are fairly general even among the high-dimensional literature. As $N \to \infty$, we allow $\psi_2$-norm bounds of $\zeta$ and $\varepsilon$ to diverge or to shrink to zero. When all the nuisance models are correctly specified, under the overlap condition in Assumption \ref{assumption_1}, 
 $\sigma^2\asymp E[\zeta^2]+E[\varepsilon^2]+E[\xi^2]\geq\max\{E[\zeta^2],E[\varepsilon^2]\},$ 
where $\xi :=\mu_{a}(\bS_{1})-\mu_{a'}(\bS_{1})-\theta$ denotes the centered conditional effect at the first exposure.
A sufficient condition for Assumption \ref{assumption_residual} is 
$\|\zeta/\sqrt{E[\zeta^2]}\|_{\psi_2} \leq\sigma_\zeta \mbox{ and } \|\varepsilon/\sqrt{E[\varepsilon^2]}\|_{\psi_2}\leq\sigma_\varepsilon,$ i.e., the ``normalized'' residuals have constant ${\psi_2}$-norms.
Note that, we allow $\sigma=\sigma_N$ to be dependent on $N$ while assuming $\sigma_\zeta$ and $\sigma_\varepsilon$ to be constants independent of $N$; $\sigma\to0$ and $\sigma\to\infty$ are both allowed as $N\to\infty$. The following Assumption \ref{assumption_prospensity} is an overlap condition for the working propensity score models, which is additionally required only when model misspecification occurs. 
\begin{assumption}\label{assumption_prospensity}
 Let $\pi_{a}^{*}(\cdot)$ and $\rho_{a}^{*}(\cdot)$ be such that 
$
 P(c_0 \leq \pi_{a}^{*}(\bS_{1}) \leq1-c_0)=1,\ 
 P(c_0 \leq \rho_{a}^{*}(\bS) \leq1-c_0)=1,
$
for a fixed constant $c_0 >0$.
\end{assumption}
The following theorem characterizes the consistency rate of the S-DRL estimator of $\theta$.

\begin{theorem} [Consistency of the S-DRL]\label{thm:con-DDRL}
Suppose that at least one of $\mu_{a}^{*}(\cdot)$ and $\pi_{a}^{*}(\cdot)$ is correctly specified, and at least one of the models $\nu_{a}^{*}(\cdot)$ and $\rho_{a}^{*}(\cdot)$ is correctly specified. Let Assumptions \ref{assumption_1}-\ref{assumption_prospensity} hold. Assume that ${\max\{s_{\bm{\alpha}_{a}}, s_{\bm{\beta}_{a}}, s_{\bm{\gamma}_{a}},s_{\bm{\delta}_{a}}\}\log(d)} =o(N)$, and either (a) $\|\bS_1\|_\infty\leq C$ almost surely, with a constant $C>0$, or (b) $s_{\bdelta_a}\log^2(d)=O(N)$. Then the sequential DR Lasso (S-DRL) estimator, $\widehat{\theta}$, as defined in Algorithm \ref{alg:S-DRL}, satisfies 
\begin{align*}
	\widehat{\theta}-\theta=O_p\left(\sigma\frac{s_1\log(d)}{N}+\sigma\sqrt{\frac{s_2\log(d)}{N}}+\frac{1}{\sqrt{N}}\sigma\right),
\end{align*}
as $N,d\to\infty$, with $s_1:=\max\{\sqrt{s_{\bm{\alpha}_{a}}s_{\bm{\delta}_{a}}}, \sqrt{s_{\bm{\beta}_{a}}s_{\bm{\gamma}_{a}}}\}$ and 
$s_2':=\max\{s_{\bm{\alpha}_{a}}\mathbbm{1}_{\{\pi_{a}^{*}\neq\pi_{a}\;\text{or}\;\rho_{a}^{*}\neq\rho_{a}\}},$ $s_{\bm{\beta}_{a}}\mathbbm{1}_{\{\pi_{a}^{*}\neq\pi_{a}\}}, s_{\bm{\gamma}_{a}}\mathbbm{1}_{\{\mu_{a,{\mbox{\tiny NR}}}^*\neq\mu_{a}\}}, s_{\bm{\delta}_{a}}\mathbbm{1}_{\{\nu_{a}^{*}\neq\nu_{a}\}}\}.$
\end{theorem}
We categorize bounded and unbounded covariate support and add a $\log(d)$ restriction to $s_{\bdelta_a}$ for the latter. The DR imputed outcome \eqref{eq:DRY} has an unbounded $\psi_\alpha$-Orlicz norm for any $\alpha>0$. Yet, if $\bS_1$ has bounded support, no extra sparsity condition is required as the inverse probability weighting is stable and the DR imputation has a well-behaved tail distribution.

 \begin{remark}[Comparison with low-dimensional DR DTE estimators]\label{remark:compare_lowdim}
 \cite{lewis2021double} applied debiased machine learning and g-estimation techniques in the framework of SNM models. However, the ``blip functions'' $\gamma_1(\bs_1,a_1)$ and $\gamma_2(\bs,a_1,a_2)$ -- which are defined as $E[Y(a_1,a_2)-Y(0,a_2)|A_1=a_1,\bS_1=\bs_1]$ and $E[Y(a_1,a_2)-Y(a_1,0)|A_1=a_1, A_2=a_2, \bS=\bs]$, respectively --  are considered low-dimensional and correctly specified for consistent estimation.  
 MSM models with low-dimensional confounders have been studied extensively, with significant theoretical advancements made in the seminal work of \cite{tchetgen2012semiparametric}. Additionally, \cite{robins2000robust}, \cite{murphy2001marginal}, \cite{bang2005doubly}, and \cite{yu2006double} explored DR DTE estimation with low-dimensional nuisances, proposing consistent and asymptotically normal DTE estimators given that either (a) all conditional mean models are correctly parametrized or (b) all propensity score models are correctly parametrized. More recently, \cite{babino2019multiple} proposed a multiple robust (MR) estimator that allows for an additional model misspecification scenario (c), where only the first-time conditional mean model and second-time propensity score model are correctly parametrized. However, all of the aforementioned approaches require low-dimensional, parametric nuisance estimates that are root-$N$ consistent.
 
 
 Our proposed method accommodates high-dimensional and possibly non-parametric nuisance estimates, which may not necessarily be root-$N$ consistent. This approach allows for a consistent estimate of the DTE even in challenging scenarios where only the second-time conditional mean and first-time propensity score models are correctly specified (misspecification scenario (d)). 
This scenario is more common in practice due to the difficulty of identifying the first-time conditional mean model. When $\nu_a(\bs)=\nu_a^*(\bs)=\bu^\top\balpha_a^*$ is linear, as per \eqref{rep:nested}, a linear $\mu_a(\cdot)$ would require  $E[\bU^\top\balpha_a^*|\bS_1,A_1=a_1]$ to be linear in $\bS_1$ -- an unlikely scenario if any of the $\bS_2$ are binary or discrete. 
\end{remark}

When all the nuisance models are correctly specified, we further establish asymptotic normality results and the corresponding rate DR property of the S-DRL estimator.

\begin{theorem}[Asymptotic normality of the S-DRL]\label{thm:dr_theta}
Suppose that all the nuisance models $\mu_a^*(\cdot)$, $\nu_a^*(\cdot)$, $\pi_a^*(\cdot)$, and $\rho_a^*(\cdot)$ are correctly specified. Let Assumptions \ref{assumption_1}-\ref{assumption_U} hold. Assume that ${\max\{s_{\bm{\alpha}_{a}}, s_{\bm{\beta}_{a}}, s_{\bm{\gamma}_{a}},s_{\bm{\delta}_{a}}\}\log(d)} =o(N)$, and either (a) $\|\bS_1\|_\infty\leq C$ almost surely, with some constant $C>0$, or (b) $s_{\bdelta_a}\log^2(d)=O(N)$. Additionally, assume the following product-rate condition:
 \begin{align} \label{assumption_5'}
\max\{s_{\bm{\gamma}_{a}}s_{\bm{\beta}_{a}},s_{\bm{\delta}_{a}}s_{\bm{\alpha}_{a}}\}\log^2(d)=o(N).
\end{align}
Then, with $\sigma^2$ in \eqref{def:sigma} and $\widehat{\sigma}^2$ in \eqref{def:sigma-hat}, the S-DRL estimator satisfies $\sigma^{-1}\sqrt{N}(\widehat{\theta}-\theta)\leadsto N(0,1)$ and ${\hat\sigma}^{-1}\sqrt{N}(\widehat{\theta}-\theta)\leadsto N(0,1)$ as $N,d\to\infty$.
\end{theorem}

Theorem \ref{thm:dr_theta},
as per \cite{bang2005doubly}, indicates that the S-DRL estimator achieves the semiparametric efficiency bound when all nuisance models are correctly specified.

\begin{remark}[Comparison with static ATE estimators] \label{rem:2}
We compare sparsity conditions in Theorem \ref{thm:dr_theta} with the literature on estimating ATE through DR for a single exposure.
The ATE can be seen as a special case of the DTE where we assume that $\bS_{1}$ and $A_1$ are completely random. This allows root-$N$ estimation of $\mu_a(\cdot)$ and $\pi_{a}(\cdot)$. Consequently, Theorem \ref{thm:dr_theta} requires $s_{\balpha_a}+s_{\bm{\delta}_a}=o(N/\log(d))$ and $s_{\balpha_a}s_{\bm{\delta}_a}=o(N/\log^2(d))$, which are less restrictive than \cite{farrell2015robust,tan2020model,dukes2021inference,dukes2020doubly,avagyan2021high} and are aligned with \cite{chernozhukov2018double,smucler2019unifying}.
\end{remark}

\subsection{Properties of the DTL estimator}\label{sec:theory-DTL}

With slight abuse of notation, let $s_{\bbeta_a}=\|\bbeta_{a,{\mbox{\tiny NR}}}^*\|_0$. 
\vskip -10pt
\begin{theorem} [Consistency of the DTL]\label{thm:con-DTL}
Suppose that at least one of $\mu_{a,{\mbox{\tiny NR}}}^*(\cdot)$ and $\pi_{a}^{*}(\cdot)$ is correctly specified, and at least one of the models $\nu_{a}^{*}(\cdot)$ and $\rho_{a}^{*}(\cdot)$ is correctly specified.
Let Assumptions \ref{assumption_1}- \ref{assumption_prospensity} hold with $\mu_a^*(\cdot)$ and $\bbeta_a^*$ replaced by $\mu_{a,{\mbox{\tiny NR}}}^*(\cdot)$ and $\bbeta_{a,{\mbox{\tiny NR}}}^*$. Assume that ${\max\{s_{\bm{\alpha}_{a}}, s_{\bm{\beta}_{a}}, s_{\bm{\gamma}_{a}},s_{\bm{\delta}_{a}}\}\log(d)} =o(N)$.
Then the DTL estimator satisfies, as $N,d\to\infty$,
\begin{align} \label{rate:rateDR}
\widehat{\theta}_{\tiny\mbox{DTL}}-\theta=O_p\left(\sigma\frac{s_1'\log(d)}{N}+\sigma\sqrt{\frac{s_2'\log(d)}{N}}+\frac{1}{\sqrt{N}}\sigma\right),
\end{align}
with \ $s_1':=\max\{\sqrt{s_{\bm{\alpha}_{a}}s_{\bm{\gamma}_{a}}},\sqrt{s_{\bm{\alpha}_{a}}s_{\bm{\delta}_{a}}}, \sqrt{s_{\bm{\beta}_{a}}s_{\bm{\gamma}_{a}}}\}$
and $s_2':=\max\{s_{\bm{\alpha}_{a}}\mathbbm{1}_{\{\pi_{a}^{*}\neq\pi_{a}\;\text{or}\;\rho_{a}^{*}\neq\rho_{a}\}},$ $s_{\bm{\beta}_{a}}\mathbbm{1}_{\{\pi_{a}^{*}\neq\pi_{a}\}}, s_{\bm{\gamma}_{a}}\mathbbm{1}_{\{\mu_{a,{\mbox{\tiny NR}}}^*\neq\mu_{a}\}}, s_{\bm{\delta}_{a}}\mathbbm{1}_{\{\nu_{a}^{*}\neq\nu_{a}\}}\}.$
\end{theorem}

\begin{theorem}[Asymptotic normality of the DTL]\label{thm_sparse}
Suppose that all the nuisance models $\mu_{a,{\mbox{\tiny NR}}}^*(\cdot)$, $\nu_a^*(\cdot)$, $\pi_a^*(\cdot)$, and $\rho_a^*(\cdot)$ are correctly specified. Let Assumptions \ref{assumption_1}-\ref{assumption_U} hold with $\mu_a^*(\cdot)$ and $\bbeta_a^*$ replaced by $\mu_{a,{\mbox{\tiny NR}}}^*(\cdot)$ and $\bbeta_{a,{\mbox{\tiny NR}}}^*$. Assume that ${\max\{s_{\bm{\alpha}_{a}}, s_{\bm{\beta}_{a}}, s_{\bm{\gamma}_{a}},s_{\bm{\delta}_{a}}\}\log(d)} =o(N)$. Additionally, assume the following product-rate condition:
\begin{align} 
	 {\max\{
	s_{\bm{\gamma}_{a}}s_{\bm{\beta}_{a}}, 
	s_{\bm{\delta}_{a}}s_{\bm{\alpha}_{a}},
	s_{\bm{\gamma}_{a}}s_{\bm{\alpha}_{a}}
	\}\log^2(d)}&=o(N). \label{3.43}
\end{align}
 Then, with $\sigma^2$ in \eqref{def:sigma} and $\widehat{\sigma}_{\tiny\mbox{DTL}}^2$ in \eqref{def:sigma-hat-seq}, the DTL estimator satisfies ${\sigma}^{-1}\sqrt{N}(\widehat{\theta}_{\tiny\mbox{DTL}}-\theta)\leadsto N(0,1)$ and ${\hat\sigma_{\tiny\mbox{DTL}}}^{-1}\sqrt{N}(\widehat{\theta}_{\tiny\mbox{DTL}}-\theta)\leadsto N(0,1)$ as $N,d\to\infty$.
\end{theorem}

\begin{remark}[Comparisons between the S-DRL and DTL estimators]\label{sec:compare}
 \ 
 
\vspace{0.2em}
\noindent{\it \underline{Consistency.}} 
Theorems \ref{thm:con-DDRL} and \ref{thm:con-DTL} establish the consistency of two distinct estimators, S-DRL and DTL. While both estimators necessitate correct specification of at least one of $\nu_{a}^{}(\cdot)$ and $\rho_{a}^{}(\cdot)$, their requirements on $\mu_{a}^{}(\cdot)$ and $\pi_{a}^{}(\cdot)$ differ due to their distinct conditional mean models. Specifically, S-DRL mandates 
either $\mu_a^*(\cdot)=\mu_a(\cdot)$ or $\pi_a^*(\cdot)=\pi_a(\cdot)$, 
which can be guaranteed by the linearity of $\mu_a(\cdot)$ or the logistic form of $\pi_a(\cdot)$. In contrast, DTL imposes the stricter condition of 
either $\mu_{a,{\mbox{\tiny NR}}}^*(\cdot)=\mu_a(\cdot)$ or $\pi_a^*(\cdot)=\pi_a(\cdot)$, 
which may not be fulfilled even when $\mu_a(\cdot)$ is linear, as discussed in Section \ref{sec:mu-correct}.
 For a comprehensive summary, see Table \ref{table:consistency}.

\vspace{0.2em}
 \noindent{\it \underline{Rate of estimation.}}
Different rate of estimation of S-DRL and DTL are presented in Table \ref{table:consistency-rate}. 
Table \ref{table:consistency-rate} reveals a symmetrical pattern in the rates of S-DRL, whereas DTL exhibits an asymmetric behavior -- the sparsity levels $s_{\balpha_a}$ and $s_{\bgamma_a}$ appear to be more influential than $s_{\bbeta_a}$ and $s_{\bdelta_a}$. When either $\rho_a(\cdot)$ or $\mu_a(\cdot)$ are misspecified, there is no difference in the rates. 
 However, when they are both correctly specified, the rate of DTL contains additional terms that involve the sparsity level $s_{\balpha_a}$ (and $s_{\bgamma_a}$ under certain circumstances). 
 Notably, if $s_{\balpha_a}$ is relatively large, S-DRL exhibits a faster consistency rate than DTL.

\vspace{0.2em}
\noindent{\it \underline{Asymptotic normality.}}
The S-DRL and DTL estimators are both asymptotically normal, as proven in Theorems \ref{thm:dr_theta} and \ref{thm_sparse}. When $\mu_{a,{\mbox{\tiny NR}}}^*(\cdot)=\mu_a(\cdot)=\mu_a^*(\cdot)$, their asymptotic efficiency is the same. However, they require different sparsity conditions. The DTL estimator requires three product-sparsity conditions, specifically: (a) the first-time conditional mean $\mu_a(\cdot)$ and first-time propensity score $\pi_a(\cdot)$, (b) the second-time conditional mean $\nu_a(\cdot)$ and second-time propensity score $\rho_a(\cdot)$, and (c) the second-time conditional mean $\nu_a(\cdot)$ and first-time propensity score $\pi_a(\cdot)$, which are given in equation \eqref{3.43}. On the other hand, the S-DRL estimator only requires two product-sparsity conditions, as defined in \eqref{assumption_5'}. These correspond to (a) and (b) above, with (c) becoming irrelevant.
The S-DRL estimator is known as (sequential) rate DR because it is asymptotically normal when the product-sparsity of the nuisance parameters is $o(N/\log^2(d))$ for each exposure time; see more details in Remark \ref{rem:2}.
 
\begin{table}[]
\begin{threeparttable}[b]
\caption{Consistency rates of $\thetahat$ and $\thetahat_{\tiny\mbox{DTL}}$ under various misspecification settings when the conditions in Theorems \ref{thm:con-DDRL} and \ref{thm:con-DTL} are satisfied with $\sigma\asymp1$. Misspecified and correctly specified models are denoted by \xmark$\;$and \cmark, respectively. LHS below corresponds to the consistency rate provided to the left of that position. }\label{table:consistency-rate}
\renewcommand{\arraystretch}{1} 
\centering
\begin{tabular}{ c | c | c | c | c | c }
\Xhline{2\arrayrulewidth}
\multicolumn{4}{c|}{Correctly specified models}&\multirow{2}{*}{Consistency rate of $\thetahat$}&\multirow{2}{*}{Consistency rate of $\thetahat_{\tiny\mbox{DTL}}$}\\
\cline{1-4}
$\rho_a(\cdot)$&$\pi_{a}(\cdot)$&$\mu_{a}(\cdot)$&$\nu_a(\cdot)$&&\\
\hline
\cmark&\cmark&\cmark&\cmark&{\fontsize{5}{6}\selectfont $ \frac{1}{\sqrt N}+\frac{\sqrt{{s_{\balpha_a}s_{\bdelta_a}}}\log d}{N}+\frac{\sqrt{s_{\bbeta_a}s_{\bgamma_{a}}}\log d}{N}$} &LHS{\fontsize{5}{6}\selectfont $\;+\;\frac{\sqrt{s_{\balpha_a}s_{\bgamma_a}}\log d}{N}$}\\
\hline
\xmark&\cmark&\cmark&\cmark&{\fontsize{5}{6}\selectfont $\frac{\sqrt{s_{\bbeta_a}s_{\bgamma_{a}}}\log d}{N}+\sqrt\frac{s_{\balpha_a}\log d}{N}$}&LHS\\ 
\hline
\cmark&\xmark&\cmark&\cmark&{\fontsize{5}{6}\selectfont $\frac{\sqrt{s_{\balpha_a}s_{\bdelta_a}}\log d}{N}+\sqrt\frac{s_{\bbeta_{a}}\log d}{N}$}&{\fontsize{5}{6}\selectfont $\sqrt{\frac{{s_{\balpha_a}\log d}}{ N}}+\sqrt\frac{s_{\bbeta_a}\log d}{N}$}\\
\hline
\cmark&\cmark&\xmark&\cmark&{\fontsize{5}{6}\selectfont $\frac{\sqrt{s_{\balpha_a}s_{\bdelta_a}}\log d}{N}+\sqrt\frac{s_{\bgamma_{a}}\log d}{N}$}&LHS\\
\hline
\cmark&\cmark&\cmark&\xmark&{\fontsize{5}{6}\selectfont $\frac{\sqrt{s_{\bbeta_a}s_{\bgamma_{a}}}\log d}{N}+\sqrt\frac{s_{\bdelta_a}\log d}{N}$}&LHS{\fontsize{5}{6}\selectfont $\;+\;\frac{\sqrt{s_{\balpha_a}s_{\bgamma_{a}}}\log d}{N}\;$}\tnote{$^*$}\\
\hline
\xmark&\xmark&\cmark&\cmark&{\fontsize{5}{6}\selectfont $\sqrt\frac{s_{\balpha_a}\log d}{N}+\sqrt\frac{s_{\bbeta_a}\log d}{N}$}&LHS\\
\hline
\xmark&\cmark&\xmark&\cmark&{\fontsize{5}{6}\selectfont $\sqrt\frac{s_{\balpha_a}\log d}{N}+\sqrt\frac{s_{\bgamma_a}\log d}{N}$}&LHS\\
\hline
\cmark&\xmark&\cmark&\xmark&{\fontsize{5}{6}\selectfont $\sqrt\frac{s_{\bbeta_a}\log d}{N}+\sqrt\frac{s_{\bdelta_a}\log d}{N}$}&LHS{\fontsize{5}{6}\selectfont $\;+\;\sqrt\frac{s_{\balpha_a}\log d}{N}\;$}\tiny\tnote{$^{**}$}\\
\hline
\cmark&\cmark&\xmark&\xmark&{\fontsize{5}{6}\selectfont $\sqrt\frac{s_{\bgamma_{a}}\log d}{N}+\sqrt\frac{s_{\bdelta_{a}}\log d}{N}$}&LHS\\
\Xhline{2\arrayrulewidth}
\end{tabular}
\begin{tablenotes}
\item [\tiny *] \tiny This consistency rate requires $\mu_{a,{\mbox{\tiny NR}}}^*(\cdot)=\mu_a(\cdot)$. Without it, the DTL's rate is equal to the last row of the above table; see Section \ref{sec:mu-correct}.
\item [\tiny **] This consistency rate requires $\mu_{a,{\mbox{\tiny NR}}}^*(\cdot)=\mu_a(\cdot)$. Otherwise, the DTL is inconsistent; see Table \ref{table:consistency}.
\end{tablenotes}
\end{threeparttable}
\end{table}
 
\end{remark}

\subsection{Properties of the general DR estimator}\label{sec:theory-gen}

In this section, we provide a new consistency result of the general DR DTE estimator.
Here we consider arbitrary working models $\pi_a^*(\cdot)$, $\rho_a^*(\cdot)$, $\mu_a^*(\cdot)$, and $\nu_a^*(\cdot)$, which may not follow the logistic or linear forms as before. For each $c\in\{a,a'\}$, define the corresponding DR score function as
\begin{align}
{\psi}_c^*(W):=\mu_{c}^*(\bS_{1})+\mathbbm1_{\{A_{1}=c_1\}}\frac{ \nu_{c}^*(\bS)-\mu_{c}^*(\bS_{1}) }{\pi_{c}^*(\bS_{1})}+\mathbbm1_{\{A_{1}=c_1, A_{2}=c_2\}}\frac{ Y-\nu_{c}^*(\bS)}{\pi_{c}^* (\bS_{1})\rho_{c}^*(\bS)},\label{def:psi-gen-star}
\end{align}
with
\begin{align}\label{def:sigma-gen}
\sigma^2:=E\left[{\psi}_a^*(W)-{\psi}_{a'}^*(W) - \theta\right]^2.
\end{align}

\begin{assumption}\label{assumption_4}
For positive sequences $a_N=o(\sigma)$, $b_N=o(\sigma)$, $c_N=o(1)$, and $d_N=o(1)$, let 
$
	 E[\widehat{\nu}_{a}(\bS )-\nu_{a}^{*}(\bS )]^2 =O_p(a_N^2), \ 
	 E[\widehat{\mu}_{a}(\bS_{1})-\mu_{a}^{*}(\bS_{1})]^2 =O_p(b_N^2),
$ $
	 E[\widehat{\pi}_a(\bS_{1})-\pi_{a}^{*}(\bS_{1})]^2 =O_p(c_N^2), \ $ and $
	 E[\widehat{\rho}_{a}(\bS )-\rho_{a}^{*}(\bS )]^2 =O_p(d_N^2).
$
Moreover, for $c_0 \in (0,1/2)$, $P(c_0 \leq \widehat{\pi}_a(\bS_{1}) \leq1-c_0)=1$ and $P(c_0 \leq \widehat{\rho}_{a}(\bS ) \leq1-c_0)=1$ with probability approaching one. 
For $\zeta$ and $\varepsilon$ defined in Assumption \ref{assumption_residual},
$\max \{ {E|\zeta|^q}/{[E|\zeta|^2]^{{q}/{2}}}$, ${E|\varepsilon|^q}/{[E|\varepsilon|^2]^{{q}/{2}}}$, ${E|\xi|^q}/{[E|\xi|^2]^{{q}/{2}}} \} \leq C$, $P(E[\zeta^{2}|\bS ]\leq CE[\zeta^2]) =1$, and $P(E[\varepsilon^{2}|\bS_{1}]\leq CE[\varepsilon^2]) =1$, for constants $C>0$ and $q>2$.
\end{assumption}


The probability measures and corresponding expectations above are with respect to a fresh draw $\bS$ (or $\bS_1$). Note that Assumption \ref{assumption_4} allows for $\rho_a^*(\cdot)$ to differ from $\rho_a(\cdot)$ while requiring a overlap condition 
 consistent with the existing literature; see, e.g., \cite{chernozhukov2018double}.
The $\max$ condition, satisfied by sub-Gaussian random variables, controls
 the tails of $\zeta$, $\varepsilon$, and $\xi$. 
The last two conditions of Assumption \ref{assumption_4} aim to ensure the interpretability of the results by bounding the ``normalized'' conditional second moments. 

\begin{theorem} (Consistency of the general DR estimator) \label{thm:misspecified}
Suppose that at least one of $\mu_{a}^{*}(\cdot)$ and $\pi_{a}^{*}(\cdot)$ is correctly specified, and at least one of $\nu_{a}^{*}(\cdot)$ and $\rho_{a}^{*}(\cdot)$ is correctly specified. 
Let Assumptions \ref{assumption_1}, \ref{assumption_prospensity}, and \ref{assumption_4} hold. Additionally, let $E[\mathbbm1_{\{A_1=a_1\}}(\mu_a(\bS_1)-\mu_a^*(\bS_1))^2]\leq C_\mu\sigma^2$, for some $C_\mu>0$. Then the general DR estimator, $\widehat{\theta}_\mathrm{gen}$, satisfies
$\widehat{\theta}_\mathrm{gen}-\theta=O_p(q_N)$ as $N,d\to\infty$, where 
$q_N=b_Nc_N+a_Nd_N+b_N\mathbbm{1}_{\{\pi_{a}^{*}\neq\pi_{a}\}}+a_N\mathbbm{1}_{\{\rho_{a}^{*}\neq\rho_{a}\}}+c_N\sigma\mathbbm{1}_{\{\mu_{a}^{*}\neq\mu_{a}\}}$$+d_N\sigma\mathbbm{1}_{\{\nu_{a}^{*}\neq\nu_{a}\}}$$+{\sigma/\sqrt{N}}$. 
\end{theorem} 
 
 The aforementioned theorem yields two distinct conclusions that warrant discussion. The first pertains to the conditions that are necessary for achieving root-$N$ consistency, while the second relates to the issue of consistency under model misspecification. If all the models are correctly specified, $\thetahat_\mathrm{gen}-\theta=O_p(b_Nc_N+a_Nd_N+\sigma N^{-1/2})$ and root-$N$ consistency happens as long as $b_Nc_N+a_Nd_N=O(N^{-1/2})$ and $\sigma=O(1)$.

 If, on the other hand, at least one of the nuisance models is correctly specified at each exposure time, $\thetahat_\mathrm{gen}$ is a consistent estimator as long as $\sigma=O(1)$. Model misspecification can take an asymmetric form in terms of estimation rates. Specifically, while $q_N$ is symmetric in the rates themselves, the dependence of $b_N$ on $a_N$ and/or $d_N$ can introduce potential asymmetries. For example, when performing $\ell_1$-regularized nested regression, Theorem \ref{cor_mu1} indicates that $b_N$ depends additively on $a_N$ as $b_N=b_N^*+a_N$, where $\{b_N^*\}^2= {s_{\bbeta_a}\log(d)/N}$ is the estimation error of $\mu_a(\cdot)$ when $\nu_a(\cdot)$ is known. As a result, the consistency rate of $\thetahat_\mathrm{gen}$ includes an additional term $a_Nc_N+a_N\mathbbm1_{\{\pi_a^*\neq\pi_a\}}$, as illustrated by the DTL estimator in \eqref{rate:rateDR}.

On the other hand, if we consider a new DR approach based on the DR representation \eqref{def:DR-mu} to estimate $\mu_a(\cdot)$, the corresponding $b_N$ will depend on both $a_N$ and $d_N$. For instance, when all the nuisance models are correctly specified, Theorem \ref{thm:dr_mu} indicates that $\ell_1$-regularized DR estimation leads to a symmetric rate with $b_N=b_N^*+a_Nd_N$, resulting in $\widehat{\theta}_\mathrm{gen}-\theta=O_p(b_N^*c_N+a_Nd_N+1/\sqrt N)$ if $\sigma\asymp1$ and $a_N,b_N^*,c_N,d_N=o(1)$. The approach used to estimate the first-time conditional mean $\mu_a(\cdot)$ determines the persistence of the symmetry.




\begin{theorem} (Asymptotic normality of the general DR estimator) \label{thm:DML}
Suppose that all the nuisance models $\mu_a^*(\cdot)$, $\nu_a^*(\cdot)$, $\pi_a^*(\cdot)$, and $\rho_a^*(\cdot)$ are correctly specified. Whenever Assumptions \ref{assumption_1}, \ref{assumption_4} hold and the rates of estimation satisfy the following product conditions
\begin{equation}\label{eq:rdr}
b_Nc_N=o(\sigma {N}^{-1/2}), \qquad a_Nd_N=o(\sigma {N}^{-1/2}),
\end{equation}
then the estimator $\widehat{\theta}_\mathrm{gen}$ satisfies
$\sigma^{-1}\sqrt{N}(\widehat{\theta}_\mathrm{gen}-\theta)\leadsto N(0,1)$ and $\sigmahat_\mathrm{gen}^{-1}\sqrt{N}(\widehat{\theta}_\mathrm{gen}-\theta)\leadsto N(0,1)$
as $N\to\infty$ (and potentially $d\to\infty$), where $\sigma^2$ and $\widehat{\sigma}_\mathrm{gen}^2$ are defined in \eqref{def:sigma-gen} and \eqref{def:theta-gen}, respectively.
\end{theorem}

\begin{remark}[Rate double robustness]\label{remark:RDR}

The topic of rate double robustness in the presence of multiple exposures has been addressed in \cite{bodory2022evaluating}, but the authors require three product-rate conditions, including $a_Nc_N=o( {N}^{-1/2})$ as stated in their Assumption 3.1, in addition to the two product-rates \eqref{eq:rdr}. As a result, this does not allow for relatively large values of $a_N$ and $c_N$, which is permissible in our setting. For example, consider a special case where $A_2$ is completely random and $\mu_a(\cdot)$ is a constant function. In conjunction with the sequential ignorability condition of Assumption \ref{assumption_1}, we have $\mu_a(\bs_1)=E[Y|\bS_1=\bs_1,A_1=a_1,A_2=a_2]$, allowing $\mu_a(\cdot)$ to be identified directly through observable variables. In this scenario, both $\mu_a(\cdot)$ and $\rho_a(\cdot)$ can be estimated with a root-$N$ rate. Therefore, Theorem \ref{thm:DML} only requires $a_N+c_N=o(1)$, i.e., $\nu_a(\cdot)$ and $\pi_a(\cdot)$ are consistently estimated.
In comparison, \cite{bodory2022evaluating} additionally require $a_Nc_N=o(N^{-1/2})$, which may not be feasible when $\nu_a(\cdot)$ and $\pi_a(\cdot)$ are only known to be Lipschitz continuous, and the covariate dimensions satisfy $d\geq d_1>2$. 

Our proof relies on a nuanced decomposition of the second-order estimation bias resulting from the estimation errors of $\nuhat_{c,-k}(\cdot)$ and $\pihat_{c,-k}(\cdot)$. Leveraging the Neyman orthogonality of the DR score \eqref{def:psi-gen-star}, we reformulate the second-order bias as the product $E[\mathbbm1_{\{A_1=c_1\}}(1-\mathbbm1_{\{A_2=c_2\}}/\rho_c^*(\bS))(\nuhat_{c,-k}(\bS)/\pihat_{c,-k}(\bS_1)-\nu_c^*(\bS)/\pi_c^*(\bS_1))]$. In our analysis, we then examine this product collectively rather than as separate terms, resulting in a cohesive flow of the arguments. We show that the population effect of this term is exactly zero whenever the model $\rho_c^*(\cdot)$ is correctly specified -- a condition fullfilled when discussing asymptotic normality in high-dimensional regimes. We then showcase that the sample equivalent is negligible and does not contribute to the estimation error.

\end{remark}

\section{Supporting theoretical discoveries}\label{sec:nuisance}

This section presents supplementary findings that, while not the primary focus of the research, may nonetheless be informative or valuable.

\subsection{An adaptive theory for imputed Lasso with high-dimensional covariates}\label{sec:imp-gen}
Let $\S:=(Y_i^*,\bX_i)_{i=1}^M$ be i.i.d. observations and let $(Y^*,\bX)$ be an independent copy with $Y^*\in\R$ and $\bX\in\R^d$. 
Suppose that there exists, possibly random, $\Yhat_i\in\R$. Note that for some, and possibly all observations, outcomes $Y^*$ are imputed, i.e., estimated using $\hat Y_i$. The true population slope is defined as 
$\bbeta^*=\mathrm{argmin}_{\bbeta\in\R^d}E [Y^*-\bX^\top\bbeta ]^2.$ Then its estimator is 
\begin{align}
	\bbetahat&~:=~\mathrm{argmin}_{\bbeta\in\R^d}\left\{M^{-1}\sum_{i=1}^M[\Yhat_i-\bX_i^\top\bbeta]^2+\lambda_M\|\bbeta\|_1\right\}\label{def:betahat_imp},
\end{align}
for $\lambda_M>0$.
The following result delineates properties of such imputed-Lasso estimator, $\bbetahat$.

\begin{theorem}[General imputed Lasso estimators]\label{imputation}
	
	Let $s=\|\bbeta^*\|_0$ and $\varepsilon_i:=Y_i^*-\bX_i^\top\bbeta^*$.
	Suppose that $\|\boldsymbol{a}^\top\bX\|_{\psi_2}\leq\sigma_\bX\|\boldsymbol{a}\|_2$ for $\boldsymbol{a}\in\R^d$, $\lambda_{\min}(E[\bX\bX^\top])>\lambda_\bX$, and $\|\varepsilon\|_{\psi_2}\leq\sigma$ with $\sigma_\bX,\lambda_\bX>0$ and $\sigma=\sigma_M>0$ potentially dependent on $M$. For some $\delta_M>0$, define the event
$\mathcal E_1:=\{M^{-1}\sum_{i=1}^M[\Yhat_i-Y_i^*]^2<\delta_M^2\}.$
	For any $t>0$, let $\lambda_M:=16\sigma\sigma_\bX(\sqrt{\log(d)/M}+t)$. Then on the event $\mathcal E_1$, when 
	$M>\max\{\log(d),100\kappa_2^2s\log(d)\},$
	we have
	\begin{align*}
		\|\bbetahat-\bbeta^*\|_2&\leq\max\left(\frac{5\kappa_2\delta_M^2}{4\sigma\sigma_{\bX}}+4\kappa_1^{-1/2}\delta_M,8\kappa_1^{-1}\sqrt s\lambda_M\right),
	\end{align*}
	with probability at least $1-2\exp(-\frac{4Mt^2}{1+2t+\sqrt{2t}})-c_1\exp(-c_2M)$, where $\kappa_1,\kappa_2,c_1,c_2>0$ are constants independent of $M$ and $d$. 
	Moreover, if $\delta_M=o(\sigma)$, $P(\mathcal E_1)=1-o(1)$, and $M\gg s\log(d)$, then with $\lambda_M\asymp\sigma\sqrt{\log(d)/M}$, as $M,d\to\infty$,
	\begin{align}\label{rate:ourthm1}
		\|\bbetahat-\bbeta^*\|_2&=O_p\Bigl(\sigma\sqrt\frac{s\log(d)}{M}+\delta_M\Bigl).
	\end{align}
\end{theorem} 
The above result is of independent interests as it provides a general theory for any Lasso estimators based on imputed outcomes. It contributes to the literature in three specific aspects: (a) The ``imputation error'', $\Yhat_i-Y_i^*$, can be dependent on and even possibly correlated with covariates $\bX_i$; (b) We allow every $\Yhat_i$ to be fitted using the same set of observations $(X_i,Y_i)_{i=1}^M$, i.e., $\Yhat_i$s are also possibly dependent on each other; and (c) The tuning parameter $\lambda_M$ is of the same order as the one chosen for the fully observed data and is independent of the imputation error or any sparsity parameter.

Compared with the existing literature, Theorem \ref{imputation} requires weaker sparsity assumptions and provides better rates of estimation. Imputed Lasso of \cite{zhu2019proper} requires $s=o(M)$, $\log(d)=o(\sqrt M)$, and $\sqrt s\delta_M=o(1)$. That of \cite{lewis2021double} requires an ultra-sparse setup $s^2\log(d)=o(M)$ and $s\delta_M=o(1)$. In contrast, we only require $s\log(d)=o(M)$ and $\delta_M=o(1)$. 
Additionally, \cite{zhu2019proper} choose a tuning parameter $\lambda_M\gg\sqrt s\delta_M$ and  provide $\| \bbetahat -\bbeta^* \|_2 =O_p(\sqrt s\delta_M+\sqrt{s/M})$ upon requiring strict conditions for assuring model selection consistency; see Theorem 2 therein. 
\cite{lewis2021double} take $\lambda_M\asymp\sqrt{\log(d)/M}+\delta_M$ and establish 
$\| \bbetahat -\bbeta^* \|_2 =O_p(s\sqrt{\log (d)/M}+ s\delta_M)$; see Theorem 13 therein.
In contrast, we allow $\lambda_M\asymp\sqrt{\log(d)/M}$. The imputation error $\delta_M$ only appears in our final estimation rate \eqref{rate:ourthm1} additively, and its effect does not explode as the sparsity level grows. 

Theorem \ref{imputation} requires development of new proof techniques: the standard Lasso inequality followed by the cone-set reduction are not valid in this instance. In fact, the error, $\bbetahat-\bbeta^* $, no longer belongs to the accustomed cone set, $\mathcal C(S,k):=\{\bDelta\in\R^d:\|\bDelta_{S^c}\|_1\leq k\|\bDelta_S\|_1\}$. Instead, we identify a new set, $\widetilde{\mathcal C}(S):=\{\bDelta\in\R^d:\|\bDelta_{S^c}\|_1\leq 16 \lambda_M^{-1} \delta_M^2, \|\bDelta_{S}\|_1\leq 4\lambda_M^{-1} \delta_M^2\}$, and show that the error vector belongs to the union of the above two sets. This enables us to avoid choosing a tuning parameter dependent on the imputation error, as is done in the above literature. 
Moreover, our results are adaptive to the imputation error in that 
when there is no imputation, i.e., $\delta_M=0$, our result reaches the standard consistency rate in the high-dimensional statistics literature, e.g., \cite{bickel2009simultaneous,negahban2012unified,wainwright2019high}. 


\subsection{Theoretical characteristics of nuisance estimators with imputed outcomes}\label{sec:DTL'}
 
As a result of constraints on the length of the main file, we have included the theoretical properties of the nuisance estimates $\balphahat_a$, $\bgammahat_a$, and $\widehat{\bdelta}_a$ as defined by equations \eqref{51}, \eqref{207}, and \eqref{226} respectively, in the Supplementary Materials {\citep{bradic2023supplement} where we show $\|\widehat{\bm{\alpha}}_a-\bm{\alpha}_a^*\|_2=O_p(\sigma\sqrt{s_{\bm{\alpha}_{a}}\log(d)/N})$, $\|\widehat{\bm{\gamma}}_a-\bm{\gamma}_a^*\|_2=O_p(\sqrt{s_{\bm{\gamma}_{a}}\log(d_1)/N})$, and $\|\widehat{\bm{\delta}}_a-\bm{\delta}_a^*\|_2=O_p(\sqrt{s_{\bm{\delta}_{a}}\log(d)/N})$.
Now we establish the properties of the first-time conditional mean model estimates, where imputation is required. We first consider the DR-imputation-based estimator $\bbetahat_a$ defined as \eqref{def:betahat-avg} and the corresponding conditional mean estimate $\muhat_a(\bs_1)=\bv^{\top}\bbetahat_a$. 

\begin{theorem}\label{thm:dr_mu}
Let Assumptions \ref{assumption_1}-\ref{assumption_prospensity} hold. Assume that $\max\{s_{\balpha_a}\log(d)$, $s_{\bbeta_a}\log(d_1)$, $s_{\bdelta_a}\log(d)\}=o(N)$, and either (a) $\|\bS_1\|_\infty\leq C$ almost surely, with some constant $C>0$, or (b) $s_{\bdelta_a}\log(d_1)\log(d)=O(N)$. Choose some ${\lambda}_{\bm{\alpha}} \asymp \sigma\sqrt{\log(d)/N}$, $ {\lambda}_{\bm{\beta}} \asymp \sigma\sqrt{\log(d_1)/N}$, and ${\lambda}_{\bm{\delta}} \asymp \sqrt{\log(d)/N}$. Then for any constant $r\geq1$, as $N,d\to \infty$, we have
\[
\|\bbetahat_a -\bbeta_a^* \|_2 + \left\{E[\muhat_a (\bS_1)-\mu_{a}^{*} (\bS_1)]^r\right\}^{1/r}= O_p\left(r_n\right), 
\] 
with $r_n$ being determined as follows 
(a) whenever $\rho_a^*(\cdot)=\rho_a(\cdot)$ and $\nu_a^*(\cdot)=\nu_a(\cdot)$, $r_n=\sigma\sqrt\frac{s_{\bbeta_a}\log(d_1)}{N}+\frac{\sigma\sqrt{s_{\bdelta_a}s_{\balpha_a}}\log(d)}{N}$, 
(b) whenever $\rho_a^*(\cdot)=\rho_a(\cdot)$, $r_n= \sigma\sqrt\frac{s_{\bbeta_a}\log(d_1)}{N}+\sigma\sqrt\frac{s_{\bdelta_a}\log(d)}{N}$, 
(c) or whenever $\nu_a^*(\cdot)=\nu_a(\cdot)$, $r_n=\sigma\sqrt\frac{s_{\bbeta_a}\log(d_1)}{N}+\sigma\sqrt\frac{s_{\balpha_a}\log(d)}{N}$.
\end{theorem}

Theorem \ref{thm:dr_mu} elucidates that the consistency rate of $\bbetahat_a$ is subject to the fidelity of the second-time nuisance models $\rho_a^*(\cdot)$ and $\nu_a^*(\cdot).$ More specifically, when both models are accurately specified, the DR imputation error contributes multiplicatively to the consistency rate in Theorem \ref{thm:dr_mu}(a).
In contrast, when only one of $\rho_a^*(\cdot)$ and $\nu_a^*(\cdot)$ is correctly specified, the estimation error of the correctly specified model contributes additively to the rates presented in Theorem \ref{thm:dr_mu}(b) and Theorem \ref{thm:dr_mu}(c). It is noteworthy that these results do not rely on the correctness of the first-time conditional mean model per se.

In the following, we also provide the consistency results of the nested-regression-based estimator, $\bbetahat_{a,{\mbox{\tiny NR}}}$, defined in Equation \eqref{52}, and the corresponding conditional mean estimate $\muhat_{a,{\mbox{\tiny NR}}}(\bs_1)=\bv^{\top}\bbetahat_{a,{\mbox{\tiny NR}}}$.

\begin{theorem}\label{cor_mu1}
Let Assumptions \ref{assumption_1}-\ref{assumption_U} hold. 
Assume that $\max\{s_{\bm{\alpha}_{a}}\log(d)$, $s_{\bm{\beta}_{a}}\log(d_1)\}=o(N)$. Choose some $ {\lambda}_{\bm{\alpha}} \asymp \sigma\sqrt{\log(d)/N}$ and $ {\lambda}_{\bm{\beta}} \asymp \sigma\sqrt{\log(d_1)/N}$. Then for any constant $r\geq1$, as $N,d\to \infty$, we have with $r_n$ as in Theorem \ref{thm:dr_mu}(c),
\begin{align}
\|\widehat{\bm{\beta}}_{a,{\mbox{\tiny NR}}}-\bm{\beta}_{a,{\mbox{\tiny NR}}}^*\|_2+\{E[\widehat{\mu}_{a,{\mbox{\tiny NR}}}(\bS_{1})-\mu_{a,{\mbox{\tiny NR}}}^{*}(\bS_{1})]^r\}^{1/r}&=O_p\left(r_n\right).\label{rate:betahat}
\end{align}
\end{theorem}

\begin{remark}[Comparison between $\bbetahat_a$ and $\bbetahat_{a,{\mbox{\tiny NR}}}$]\label{remark:compare-muhat}
The present remark compares the consistency rates of two estimators, $\bbetahat_a$ and $\bbetahat_{a,{\mbox{\tiny NR}}}$, in different scenarios. 
(a) In the case where $\nu_a(\cdot)$ is nonlinear and $\rho_a(\cdot)$ is logistic,
 estimators converge to distinct targets, $\bbeta_a^*$ and $\bbeta_{a,{\mbox{\tiny NR}}}^*$, respectively. Here, $\bbeta_a^*$ represents the optimal linear slope approximating the true conditional mean function $\mu_a(\cdot)$, 
while $\bbeta_{a,{\mbox{\tiny NR}}}^*$ is the optimal linear slope approximating the misspecified model $\nu_a^*(\cdot)$; see discussions in Section \ref{sec:mu-correct} above. When the first-time conditional mean $\mu_a(\cdot)$ is linear, $\bbetahat_a$ converges to the true linear slope, and a consistent estimate of $\mu_a(\cdot)$ is obtained. However, $\bbetahat_{a,{\mbox{\tiny NR}}}$ typically converges to some $\bbeta_{a,{\mbox{\tiny NR}}}^*$ that differs from the true linear slope, resulting in an inconsistent estimate of $\widehat{\mu}_{a,{\mbox{\tiny NR}}}(\cdot)$. 
(b) When $\nu_a(\cdot)$ is linear and $\rho_a(\cdot)$ is logistic, $\bbeta_a^* =\bbeta_{a,{\mbox{\tiny NR}}}^*$. However, in this case, $\bbetahat_a$ exhibits a faster consistency rate than $\bbetahat_{a,{\mbox{\tiny NR}}}$. This can be attributed to the fact that the DR imputation error contributes to the consistency rate of $\bbetahat_a$ in a product form Theorem \ref{thm:dr_mu}(a), while the imputation error from nested regression contributes in an additive form Theorem \ref{thm:dr_mu}(c), which then dominates if $s_{\bbeta_a}=o(s_{\balpha_a})$. 
Consequently, the S-DRL estimator constructed based on $\bbetahat_a$ has a faster convergence rate than the DTL estimator, which is constructed based on $\bbetahat_{a,{\mbox{\tiny NR}}}$; see Table \ref{table:consistency-rate}. The enhanced convergence rate exhibited by the S-DRL estimator implies a reduction in the requisite level of sparsity conditions necessary for the inferential guarantees.
(c) In the scenario where $\nu_a(\cdot)$ is linear and $\rho_a(\cdot)$ is non-logistic, the targets $\bbeta_a^*$ and $\bbeta_{a,{\mbox{\tiny NR}}}^*$ are identical and $\bbetahat_a$ and $\bbetahat_{a,{\mbox{\tiny NR}}}$ have the same consistency rates, as seen in Theorem \ref{thm:dr_mu}(c) and \eqref{rate:betahat}.
\end{remark}


In general, $\mu_a(\cdot)-\mu_{a'}(\cdot)$ can be seen as a conditional average treatment effect (CATE) parameter through the well established nested representation \eqref{rep:nested}. Outside of dynamic settings, DR approaches for CATE estimation typically rely on DR influence function representation of the conditional means. When those conditional means independently are not smooth enough, \cite{kennedy2020towards} proposes to instead use DR imputations for the joint estimation of the difference of the conditional means. Here, the nested structure of $\mu_a(\cdot)$, where the true outcome is never observed, prevents direct influence function approaches. 
Instead, our approach leverages cases when $\mu_a(\cdot)$ has {\it sparser} structure than  $\nu_a(\cdot)$.

\section{Advancing multi-stage treatment estimation with DR methods}\label{sec:multi-stage}

The objective of this section is to expand upon the methodology of sequential doubly robust estimation by considering its application in multi-stage settings.
Consider $T\geq2$ exposure times and suppose that we observe i.i.d. samples $\{W_{T,i}\}_{i=1}^N= (\bS_{1i},A_{1i},\dots,\bS_{Ti},A_{Ti},Y)_{i=1}^N$. Let $W _T:= (\bS_1,A_{1},\dots,\bS_T,A_T,Y)$ be an independent copy of $W_{T,i}$. For each $t\leq T$, let $\bS_t\in\R^{d_t}$ and $A_t\in\{0,1\}$ denote the covariate vector and the treatment assignment at the $t$-th exposure time, respectively. Let $Y \in \R$ denote the observed outcome variable at the final stage. Denote $\bSbar_t:=(\bS_1,\dots,\bS_t)$ and $\Abar_t:=(A_1,\dots,A_t)$ for any $1\leq t\leq T$. 
Let $Y(\abar_T)$ be the counterfactual outcome corresponding to the treatment path $\abar_T=(a_1,\dots,a_T)\in\{0,1\}^T$. The DTE between any treatment paths $\abar_T,\abar_T'\in\{0,1\}^T$ is now defined as
\begin{align*}
\theta:=E[Y(\abar_T)]-E[Y(\abar_T')]=\theta_{\abar_T}-\theta_{\abar_T'}.
\end{align*}
We define the conditional mean and propensity score functions as
\begin{align}
 \mu_{t}(\bsbar_{t},\abar_{T})&:=E[Y(\abar_T)\mid \bSbar_{t}=\bsbar_{t}, \Abar_{t-1}=\abar_{t-1}],\quad\text{for }\;\;1\leq t\leq T+1,\label{def:mu-k}\\
 \pi_{t}(\bsbar_{t},\abar_{t})&:=P[A_t=a_t\mid \bSbar_{t}=\bsbar_{t}, \Abar_{t-1}=\abar_{t-1}],\quad\text{for }\;\;1\leq t\leq T,\label{def:pi-k}
\end{align}
where for the sake of simplicity, we denote with $\bar A_0=\abar_0=\emptyset$ and $\bSbar_{T+1}:=(\bS_1,\dots,\bS_T,Y)$. For each $1\leq t\leq T$, we denote $\mu_{t}^*(\bsbar_{t},\abar_{T})$ and $\pi_{t}^*(\bsbar_{t},\abar_{t})$ as the working models for the conditional mean and propensity score, respectively. Additionally, with $\bSbar_0=\bsbar_0=\emptyset$, we set $\mu_{0}(\bsbar_{0},\abar_{T}):=E[Y(\abar_T)|\bSbar_0=\bsbar_0]=\theta_{\abar_T}$ and $\mu_{T+1}^{*}(\bsbar_{T+1},\abar_{T}):=s_{T+1}$. Note that, under the Assumption \ref{assumption_8}(b) below, we have $\mu_{T+1}^*(\bSbar_{T+1},\abar_{T})=\mu_{T+1}(\bSbar_{T+1},\abar_{T})=Y$.
To identify $\theta_{\abar_T}=E[Y(\abar_T)]$ for any $\abar_T\in\{0,1\}^T$, we assume a multi-stage version of Assumption \ref{assumption_1} in the following; see also, e.g., \cite{murphy2003optimal,robins2000marginal,robins1987addendum}. 
\begin{assumption}\label{assumption_8} 
	(a) (Sequential Ignorability)
	$Y(\abar_T)\perp \!\!\! \perp A_t\mid \bSbar_t, \Abar_{t-1}=\abar_{t-1}
	$ for each $1\leq t\leq T$.
	(b) (Consistency of potential outcomes) 
	$Y=Y(\Abar_T).$
	(c) (Overlap) Let $c_0 \in (0,1/2)$ be a positive constant, such that 
	$P(c_0 \leq	\pi_{t}(\bsbar_{t},\abar_{t})\leq1-c_0) =1 $ for each $1\leq t\leq T$.
\end{assumption}

The following proposition presents a well-known DR representation of $E[Y(\abar_t)]$ under the multi-stage dynamic setting; see, e.g., \cite{bang2005doubly, murphy2001marginal}. 
\begin{proposition}\label{DR_multi_1}
Let Assumption \ref{assumption_8} hold. For $t \leq T$ suppose that at least one of $\mu_{t}^{*}(\cdot,\abar_{T})$ and $\pi_{t}^{*}(\cdot,\abar_{t})$ is correctly specified, i.e., either $\mu_{t}^{*}(\cdot,\abar_{T})=\mu_{t}(\cdot,\abar_{T})$ or $\pi_{t}^{*}(\cdot,\abar_{t})=\pi_{t}(\cdot,\abar_{t})$. Then
\begin{align}\label{rep:DR-multi-theta}
\theta_{\abar_T} =E\left[\sum_{t=1}^{T}\frac{\mathbbm{1}_{\{\Abar_{t} = \abar_t\}}}{\prod_{l=1}^{t}\pi_{l}^{*}(\bSbar_l,\abar_{l})}(\mu_{t+1}^{*}(\bSbar_{t+1},\abar_{T})-\mu_{t}^{*}(\bSbar_{t},\abar_{T}))+\mu_{1}^{*}(\bS_1,\abar_T)\right].
\end{align}
\end{proposition}
According to Proposition \ref{DR_multi_1}, a consistent estimate of $\theta_{\abar_T}$ should be achievable as long as we can consistently estimate at least one of the nuisance functions $\mu_{t}(\cdot,\abar_{T})$ or $\pi_{t}(\cdot,\abar_{t})$ for each exposure time $t$. 
In the present context, the propensity score functions of \eqref{def:pi-k} are identifiable via observable variables.
Additionally, by Assumption \ref{assumption_8}, $\mu_T(\bsbar_T,\abar_{T})= E[Y|\bSbar_T=\bsbar_T,\Abar_{T}=\abar_{T}]$, thereby facilitating its estimation using the corresponding samples. However, the remaining conditional mean functions for stages $ t\leq T-1$ cannot be identified directly. To address this challenge, we propose DR representations of these intermediate conditional means, as an alternative to the conventional nested representation of \eqref{rep:nested-multi-mu}.

\begin{theorem}\label{DR_multi_2}
Let Assumption \ref{assumption_8} hold. For $t \leq T-1$ and $t+1\leq r \leq T$, suppose that either $\pi_{r}^{*}(\cdot,\abar_{r})$ or $\mu_{r}^{*}(\cdot,\abar_{T})$ is correctly specified, i.e., either $\pi_{r}^{*}(\cdot,\abar_{r})=\pi_{r}(\cdot,\abar_{r})$ or $\mu_{r}^{*}(\cdot,\abar_{T})=\mu_{r}(\cdot,\abar_{T})$. Then $\mu_{t}(\bsbar_{t},\abar_{T})=E[\psi^*(W_T,\abar_T)\mid\bSbar_{t}=\bsbar_{t},\Abar_t=\abar_{t}]$, where 
\begin{align*}
\psi^*(W_T,\abar_T):=\sum_{r=t+1}^{T}\frac{\prod_{l=t+1}^{r}\mathbbm{1}_{\{A_l = a_l\}}}{\prod_{l=t+1}^{r}\pi_{l}^{*}(\bSbar_l,\abar_{l})}(\mu_{r+1}^{*}(\bSbar_{r+1},\abar_{T})-\mu_{r}^{*}(\bSbar_{r},\abar_{T}))+\mu_{t+1}^{*}(\bSbar_{t+1},\abar_{T}).
\end{align*}
\end{theorem}
 
Theorem \ref{DR_multi_2} can be regarded as an overarching, comprehensive, umbrella result that subsumes a range of components, particularly encompassing Proposition \ref{DR_multi_1} as a specific case when $t=0$. Indeed, $\theta_{\abar_T}=E[Y(\abar_T)|\bSbar_0=\bsbar_0]=\mu_0(\bsbar_0,\abar_T)$ is a conditional mean function at ``stage zero''. 
Theorem \ref{DR_multi_2} indicates that $\mu_{t}(\cdot,\abar_{T})$ can be identified through a DR representation using all the conditional means and propensity scores at later stages. 
Therefore, $\mu_{t}(\cdot,\abar_{T})$ can be estimated sequentially backward in time based on the DR imputations. 
 
 
 For example, if we use linear working models for the conditional means and logistic models for the propensity scores, and either the true conditional mean $\mu_{r}(\cdot,\abar_{T})$ is linear or the true propensity score $\pi_{r}(\cdot,\abar_{r})$ is logistic at each later stage $r\geq t+1$, we can get a consistent estimate of the $\mu_{t}(\cdot,\abar_{T})$ using DR imputed linear regression. 
 By repeating this process backwards, we conclude that if either the conditional means or propensity scores are correctly parametrized at every stage $t$, we can estimate all nuisance functions consistently, leading to a consistent estimate of $\theta_{\abar_T}$. 
An alternative approach to our proposed sequential doubly robust method is the nested estimator \citep{murphy2001marginal}. This approach represents all conditional means using the following equation: \begin{align}\label{rep:nested-multi-mu}
\mu_{t}(\bsbar_{t},\abar_{T})=E[\mu_{t+1}(\bSbar_{t+1},\abar_{T})\mid \bSbar_{t}=\bsbar_{t}, \Abar_{t}=\abar_{t}].
\end{align}
However, in order to ensure the consistency of the nested estimator for $\mu_{t}(\cdot,\abar_{T})$, it is essential that all subsequent conditional mean functions exhibit true linearity. Interestingly, even the multiply robust approach presented by \cite{babino2019multiple} falls short in achieving the same level of robustness as the S-DRL method. Further insights can be found in the comments following Theorem \ref{thm:DR-mu} and Table \ref{table:consistency}.
Our method demonstrates a growing advantage as the number of exposure times increases.
For instance, in the case of $T$ exposure times, consider all the cases including correctly or incorrectly parametrized $\pi_t^*(\cdot,\abar_t)$ and $\mu_t^*(\cdot,\abar_T)$ for each $t$, our method enables $3^T$ out of $4^T$ possible cases. In contrast, the nested-regression-based and multiply robust approaches only allow for $(T+2)2^{T-1}$ cases. This conclusion is independent of the particular parametrization used -- nonparametric, smooth models are permissible -- and extends to the difference of means $\theta=\theta_{\abar_T}-\theta_{\abar'_T}$ as well.

\section{Numerical Experiments} \label{sec:num}
\subsection{Simulation studies}\label{sec:sim}
We illustrate the finite sample properties of the introduced estimators in several simulated experiments; auxiliary settings are relegated to Section C of the Supplementary Material \citep{bradic2023supplement}. We consider $a=(1,1)$ and $a'=(0,0)$,  and use  $\mathbf{1}_{(q)}:=(1,\dots,1)^\top\in\R^{q}$ as well as $\mathbf0_{(q)}:=(0,\ldots,0)\in\R^q$. Below we use $\zeta_i\sim^\mathrm{iid}\mathrm{Uniform}(-1,1)$ and $\{\bdelta_{i}\}_j\sim^\mathrm{iid}\mathrm{Uniform}(-1,1)$. We decompose $\balpha_a$ into two components  as $\balpha_a =(\balpha_{a,1}^\top,\balpha_{a,2}^\top)^\top$ and consider $\balpha_{a'}=-\balpha_a$ and $\etabold_{a'}=-\etabold_a$ unless specified differently. For each $c\in\{a,a'\}$, $\rho_c(\bS_{i})=g(\bU_i^\top\etabold_c)$ and $A_{2i}|(\bS_{i},A_{1i}=c_1)\sim\mathrm{Bernoulli}(\rho_{c}(\bS_{i}))$.

\paragraph*{M1: Correctly parametrized models}
Consider $\bS_{1i}\sim^\mathrm{iid} N_{d_1}(\mathbf{0},\mathbf{I}_{d_1})$ and  $A_{1i}|\bS_{1i}\sim\mathrm{Bernoulli}$ \ $(\pi_{a} (\bS_{1i}))$, with 
$\pi_{a} (\bS_{1i})=g(\bV_i^\top\bgamma_{a})$. 
Let $\delta_{1i}\sim^\mathrm{iid} N(0,1)$, 
$\bdelta_{1i}\sim^\mathrm{iid} N_{d_1}(0,\mathbf{I}_{d_1})$, and 
$$\bS_{2i}=\bS_{1i}+A_{1i}(1+\delta_{1i})\mathbf{1}_{(d_1)}+\bdelta_{1i}.$$ 
The outcomes are $Y_i(c)=\bU_i^\top\balpha_c+ N(0,1)$ with parameters
$\balpha_{a} =(-1,-1,1,-1,\mathbf0_{(d_1-3)},$ $-1,-1,1,\mathbf0_{(d_2-3)})^\top,$ $\balpha_{a'} =(1,1,1,-1,\mathbf0_{(d_1-3)},1,1,1,\mathbf0_{(d_2-3)})^\top,$ $\bgamma_{a} =(0,1,1,1,\mathbf0_{(d_1-3)})^\top,$ $\etabold_a =(0,1,1,\mathbf0_{(d_1-2)},1,-1,\mathbf0_{(d_2-2)})^\top,$ and $\etabold_{a'} =(0,0.5,0,-0.5,\mathbf0_{(d_1-3)},0.5,0,0.5,\mathbf0_{(d_2-3)})^\top$.  

\paragraph*{M2: Weakly sparse $\nu_c(\cdot)$ and dense $\pi_c(\cdot)$}
Let  $D_i\sim^\mathrm{iid} \mathrm{Bernoulli}(0.5)$ and 
$$\{\bS_{1i}\}_{j}\sim D_i\cdot\mathrm{Uniform}(-1,-0.5)+(1-D_i)\cdot\mathrm{Uniform}(0.5,1).$$ Define 
$\{\overline{W}\}_j=0.5\cdot 0.9^j$ and let $\bS_{2i}=\bS_{1i}^\top\overline{W}\mathbf{1}_{(d_2)}+\bdelta_{i}$. Let $Y_i(c)=\bU_i^\top\balpha_c+\zeta_i$. The parameter  $\balpha_{a,1} =(-1,\mathbf0_{(d_1)})^\top$ and  $\{\balpha_{a,2}\}_j=-0.3\cdot 0.99^{j-1}$, 
$\etabold_a =(3,0.1,\mathbf0_{(d_1+d_2-1)})^\top$. 

\paragraph*{M3: Non-linear $\nu_c(\cdot)$ and non-logistic $\pi_{c}(\cdot)$}
Consider $\{\bS_{1i}\}_{j}\sim^\mathrm{iid}\mathrm{Uniform}(-1,1)$. 
Let $\pi_{a}(\bS_{1i})=\bar g(\bV_i^\top\bgamma_{a})$, where 
$$\bar g (u)=(|u|/(|u|+1))\mathbbm{1}_{\{u>0\}}+(1/(|u|+1))\mathbbm{1}_{\{u<0\}}.$$ Define 
$\{\widetilde{W}(a)\}_j=0.7 \cdot 0.8^j$ and $\{\widetilde{W}(a')\}_j=0.5\cdot 0.9^j$. Let $\bW_{2i}=\{\widetilde W(A_{1i})\}^\top\bS_{1i}\mathbf{1}_{(d_2)}+\bdelta_{i}$ and $\{\bS_{2i}\}_j=\sqrt{|\{\bW_{2i}\}_j|}$. The parameter  $\balpha_a$ has the same $\balpha_{a,1}$ as M2  and 
 $\{\balpha_{a,2}\}_j=-0.3\cdot 0.9^{j-1}$.
The parameter
 $\bgamma_{a} =5\cdot\mathbf1_{(10)}^\top,$  and $\etabold_a =(2,0.1,\mathbf0_{(d_1-1)},0.1,\mathbf0_{(d_2-1)})^\top$. Let 
$$Y_i(c)=\bV_i^\top\balpha_{c,1}+\sum_{j=1}^{d_2}\{\alpha_{c,2}\}_j\mathrm{sgn}(\{\bW_{2i}\}_j)\{\bS_{2i}\}_j^2+\zeta_i.$$
 

For M1, $d_1=d_2=100$ and $N\in\{1000,4000\}$; for M2, $d_2=500$, $(N,d_1)$ are chosen from $(2000,20)$, $(4000,20)$, and $(4000,50)$; for M3, $d_1=20$, $(N,d_2)$ are chosen from $(500,500)$, $(1000,500)$, $(2000,500)$, $(4000,500)$, $(1000,1000)$, and $(2000,1000)$. We replicate settings 500 times. We report S-DRL as well as  a version S-DRL', which has $\widehat{\bm{\beta}}_{c}$ constructed with $\widehat{\bdelta}_c$ and $\balphahat_c$ build on the whole sub-sample $\mathcal W_{-k}=\mathcal W_{-k,1}\cup\mathcal W_{-k,2}$.}
We also present (a) DTL, Algorithm \ref{alg:DTL}, (b) IPW with 
$\ell_1$-regularized logistic PS, 
(c) an empirical difference estimator (empdiff), $\thetahat_{\mbox{\tiny empdiff}}:=\sum_{i=1}^NA_{1i}A_{2i}Y_i/\sum_{i=1}^NA_{1i}A_{2i}-\sum_{i=1}^N(1-A_{1i})(1-A_{2i})Y_i/\sum_{i=1}^N(1-A_{1i})(1-A_{2i})$, and (d) an oracle DR estimator constructed with the true nuisances. All methods use $10$-fold cross validation for selection of tuning parameters. 

\begin{table}[]
\centering
\caption{Setting M1. Bias: empirical bias; RMSE: root mean square error; Length: average length of the $95\%$ confidence intervals; Coverage: average coverage of the $95\%$ confidence intervals; ESD: empirical standard deviation; ASD: average of estimated standard deviations. All the reported values (except Coverage) are based on robust (median-type) estimates. $\mathrm{Err}_a$: average estimation error of $\mu_a(\cdot)$; $\mathrm{Err}_{a'}$: average estimation error of $\mu_{a'}(\cdot)$. $N_1$ and $N_0$ are the expected number of observations in groups $(1,1)$ and $(0,0)$. } \label{table:M1}
\begin{tabular}{@{\extracolsep{5pt}}llccccccc}
	\toprule
	Method&Bias&RMSE&Length&Coverage&ESD&ASD&$\mathrm{Err}_a$&$\mathrm{Err}_{a'}$\\
\hline
\multicolumn{1}{c}{} & \multicolumn{8}{c}{ \cellcolor{gray!50} $N=1000,N_1=293,N_0=282,d_1=100,d_2=100$}\\
\cline{2-9}
empdiff&0.734&0.734&0.274&0.004&0.234&0.070&NA&NA\\
oracle&0.003&0.220&1.091&0.954&0.325&0.278&0.000&0.000\\
IPW&0.864&0.865&1.342&0.346&0.319&0.342&NA&NA\\
DTL&0.124&0.189&0.876&0.894&0.264&0.223&0.141&0.216\\
S-DRL&0.131&0.202&0.880&0.880&0.271&0.224&0.227&0.337\\
S-DRL'&0.126&0.188&0.876&0.894&0.259&0.223&0.135&0.193\\
\hline
\multicolumn{1}{c}{} & \multicolumn{8}{c}{ \cellcolor{gray!50} $N=4000,N_1=1178,N_0=1128,d_1=100,d_2=100$}\\
\cline{2-9}
empdiff&0.731&0.731&0.137&0.000&0.111&0.035&NA&NA\\
oracle&-0.006&0.121&0.602&0.956&0.178&0.153&0.000&0.000\\
IPW&0.534&0.538&0.959&0.454&0.287&0.245&NA&NA\\
DTL&0.033&0.097&0.488&0.930&0.136&0.125&0.032&0.052\\
S-DRL&0.031&0.098&0.489&0.930&0.142&0.125&0.050&0.070\\
S-DRL'&0.028&0.098&0.489&0.932&0.138&0.125&0.033&0.044\\
\bottomrule
\end{tabular}
\end{table}

Tables \ref{table:M1} and \ref{table:M7} show the estimation and inference results for the DTE estimators, while Table \ref{table:M6} focuses on estimation performances, as valid inference is unlikely with misspecified models. Our summarized findings, shown in Tables \ref{table:M1}-\ref{table:M6}, reveal that the naive empirical difference estimator has large biases due to confounding between outcome and treatment assignments. The IPW method also performs poorly, with large biases and RMSEs, and confidence interval coverages far from the desired $95\%$. The DTL, S-DRL, and S-DRL' estimators behave similarly in Table \ref{table:M1} (under M1), with correctly specified nuisance models and relatively low sparsity levels. The S-DRL method's additional sample splitting in Algorithm \ref{alg:S-DRL} (Steps 5-7) leads to larger estimation errors in the first-time conditional mean estimates than those in the DTL and S-DRL' methods. Consequently, when $N=1000$, the S-DRL estimator's RMSE is slightly larger than that of the DTL and S-DRL' estimators, but they have similar RMSEs when $N=4000$. In terms of inference behaviors, the corresponding confidence interval coverages are below the desired $95\%$ when $N=1000$. However, increasing the total sample size to $N=4000$ brings the coverages closer to $95\%$. Estimating $\nu_c(\cdot)$ under M2 is more challenging than estimating $\rho_c(\cdot)$. As a result, the DR estimates of $\mu_c(\cdot)$ in the S-DRL and S-DRL' methods have significantly smaller estimation errors compared to the nested regression used in the DTL method, as shown in Table \ref{table:M7}, leading to smaller RMSEs and coverages closer to $95\%$. Moving on to M3, both $\nu_c(\cdot)$ and $\pi_c(\cdot)$ are misspecified. Table \ref{table:M6} shows that the estimation errors of $\mu_c(\cdot)$ with the S-DRL and S-DRL' are substantially smaller than those of the DTL. Consequently, we see an improvement of the RMSEs in the S-DRL and S-DRL' estimators.

\begin{table}[]
\centering
\caption{Setting M2. The rest of the caption details remain the same as those in Table \ref{table:M1}.} \label{table:M7}
	\begin{tabular}{@{\extracolsep{5pt}}llccccccc}
		\toprule
		Method&Bias&RMSE&Length&Coverage&ESD&ASD&$\mathrm{Err}_a$&$\mathrm{Err}_{a'}$\\
		\hline
		\multicolumn{1}{c}{} & \multicolumn{8}{c}{ \cellcolor{gray!50} $N=2000,N_1=954,N_0=951,d_1=20,d_2=500$}\\
		\cline{2-9}
		oracle&-0.021&0.691&4.146&0.950&1.029&1.058&0.000&0.000\\
empdiff&-0.204&0.601&1.576&0.636&0.848&0.402&NA&NA\\
IPW&-0.153&2.542&12.747&0.970&3.718&3.252&NA&NA\\
DTL&-0.013&0.714&4.151&0.950&1.060&1.059&0.361&0.365\\
S-DRL&-0.023&0.686&4.144&0.952&1.008&1.057&0.139&0.136\\
S-DRL'&-0.027&0.689&4.145&0.952&1.019&1.057&0.126&0.122\\
		\hline
		\multicolumn{1}{c}{} & \multicolumn{8}{c}{ \cellcolor{gray!50} $N=4000,N_1=1909,N_0=1901,d_1=20,d_2=500$}\\
		\cline{2-9}
oracle&-0.029&0.572&2.942&0.936&0.835&0.751&0.000&0.000\\
empdiff&-0.099&0.408&1.116&0.630&0.598&0.285&NA&NA\\
IPW&-0.101&2.510&11.928&0.978&3.643&3.043&NA&NA\\
DTL&-0.053&0.569&2.950&0.934&0.843&0.753&0.163&0.161\\
S-DRL&-0.026&0.554&2.942&0.940&0.838&0.750&0.040&0.040\\
S-DRL'&-0.029&0.560&2.942&0.940&0.843&0.751&0.040&0.042\\
\hline
\multicolumn{1}{c}{} & \multicolumn{8}{c}{ \cellcolor{gray!50} $N=4000,N_1=1908,N_0=1901,d_1=50,d_2=500$}\\
\cline{2-9}
oracle&-0.030&0.608&2.990&0.948&0.893&0.763&0.000&0.000\\
empdiff&-0.156&0.485&1.260&0.608&0.683&0.322&NA&NA\\
IPW&-0.086&2.099&10.088&0.984&3.136&2.574&NA&NA\\
DTL&-0.019&0.607&2.983&0.940&0.903&0.761&0.273&0.275\\
S-DRL&-0.013&0.565&2.988&0.948&0.854&0.762&0.082&0.083\\
S-DRL'&-0.012&0.574&2.987&0.948&0.863&0.762&0.080&0.081\\
		\bottomrule
	\end{tabular}
\end{table}


\begin{table}[]
	\centering
	\caption{Setting M3. The rest of the caption details remain the same as those in Table \ref{table:M1}.} \label{table:M6}
	\begin{tabular}{lcccccccccc}
		\toprule
		Method&Bias&RMSE&$\mathrm{Err}_a$&$\mathrm{Err}_{a'}$&&Bias&RMSE&$\mathrm{Err}_a$&$\mathrm{Err}_{a'}$\\
		\hline
		\multicolumn{1}{c}{ } & \multicolumn{4}{c}{\cellcolor{gray!50} $N=500,d_1=20,d_2=500$}
		&&\multicolumn{4}{c}{ \cellcolor{gray!50} $N=1000,d_1=20,d_2=500$}
		\\
		\cline{2-5}\cline{7-10}
oracle&0.013&0.119&0.000&0.000&&0.007&0.098&0.000&0.000\\
empdiff&0.149&0.162&NA&NA&&0.157&0.159&NA&NA\\
IPW&-0.035&0.465&NA&NA&&-0.164&0.493&NA&NA\\
DTL&-0.003&0.273&0.151&0.196&&0.012&0.199&0.076&0.077\\
S-DRL&0.009&0.224&0.087&0.096&&0.024&0.166&0.036&0.040\\
S-DRL'&-0.007&0.207&0.067&0.074&&0.016&0.170&0.031&0.033\\
\hline
\multicolumn{1}{c}{ } & \multicolumn{4}{c}{\cellcolor{gray!50} $N=2000,d_1=20,d_2=500$}
&&\multicolumn{4}{c}{ \cellcolor{gray!50} $N=4000,d_1=20,d_2=500$}
\\
\cline{2-5}\cline{7-10}		oracle&0.005&0.062&0.000&0.000&&0.003&0.047&0.000&0.000\\		empdiff&0.142&0.142&NA&NA&&0.142&0.142&NA&NA\\
IPW&-0.276&0.525&NA&NA&&-0.390&0.487&NA&NA\\
DTL&0.010&0.126&0.038 &0.040 &&-0.001&0.101&0.021&0.021\\
S-DRL&0.013&0.108&0.016&0.016&&0.000&0.089&0.007&0.007\\
S-DRL'&0.004&0.106&0.015&0.015&&-0.007&0.090&0.007&0.007\\
\hline
\multicolumn{1}{c}{ } & \multicolumn{4}{c}{\cellcolor{gray!50} $N=1000,d_1=20,d_2=1000$}
&&\multicolumn{4}{c}{ \cellcolor{gray!50} $N=2000,d_1=20,d_2=1000$}
\\
\cline{2-5}\cline{7-10}	oracle&0.001&0.096&0.000&0.000&&0.011&0.059&0.000&0.000\\
empdiff&0.149&0.149&NA&NA&&0.143&0.143&NA&NA\\
IPW&-0.221&0.469&NA&NA&&-0.249&0.478&NA&NA\\
DTL&-0.028 &0.212 &0.090 &0.091&&0.010&0.135&0.048&0.049\\
S-DRL&-0.012&0.160&0.037&0.039&&0.011&0.113&0.016&0.016\\
S-DRL'&-0.029&0.165&0.031&0.032&&0.002&0.116&0.015&0.015\\
		\bottomrule
	\end{tabular}
\end{table}



\subsection{Application to National Job Corps Study (NJCS)}
Job Corps (JC) is the largest and most comprehensive federal job training program in the US for disadvantaged youth between 16 and 24 years old. Each year, about 50,000 participants receive vocational training and academic education at JC centers to improve their job prospects. On average, a JC student spends 8 months at a local center, completing around 1,100 hours of instruction, which is roughly equivalent to one year of high school. For a more detailed description, refer to \cite{schochet2008does} and \cite{schochet2001national}.

Numerous studies have investigated the effects of Job Corps on wages. \cite{lee2009training} highlighted sample selection issues in their analysis. \cite{zhang2008evaluating} separated the causal effects of JC enrollment on wages from those on employment. \cite{flores2012estimating} found that longer exposure to JC training is associated with higher future earnings. \cite{chen2015bounds} separated the effects of sample selection from noncompliance, while \cite{huber2020direct} distinguished between the causal direct and indirect effects in the presence of mediators.
In addition to studying the effects of Job Corps in single-time treatment settings, researchers have also explored the dynamic treatment setting offered by Job Corps. \cite{bodory2022evaluating} investigated the effects of JC's educational and training programs and found positive impacts on fourth-year employment compared to no program participation. Meanwhile, \cite{singh2021kernel} analyzed the total, direct, and indirect dynamic dose response of job training on employment. Their study concluded that a few class hours in the first and second years significantly increase employment in the fourth year. In this section, we will evaluate the effects of sequential job training programs on wages using the S-DRL and DTL methods, as defined in Algorithms \ref{alg:S-DRL} and \ref{alg:DTL}.

\newcommand{\grayline}{\arrayrulecolor{gray!70}\hline\arrayrulecolor{black}}
\begin{table}[h]
\centering
\caption{ Job Corps estimates. Here, $N_a$ and $N_{a'}$ denote the number of observations in the treatment groups $a$ and $a'$, respectively. SE: the standard error; CI: $95\%$ confidence interval; p-value: $H_0: \theta=0$ vs. $H_1: \theta\neq0$.} \label{table:DTE-JC}
	\begin{tabular}{@{\extracolsep{1pt}}ccccccccccc}
		\toprule
		$z$ & $z'$&$N_z$&$N_{z'}$& Method&$\thetahat_{z}$&$\thetahat_{z'}$&$\thetahat$&SE&CI&p-value\\
		\hline
		\multirow{2}{*}{$(3,3)$}&\multirow{2}{*}{$(1,1)$}&\multirow{2}{*}{568}&\multirow{2}{*}{315}&DTL&6.201&5.652&0.549&0.270&[0.020, 1.078]&0.042\\
		& & & &S-DRL&6.208&5.641&0.567&0.273&[0.032, 1.102]&0.037\\
\grayline
	 \multirow{2}{*}{$(3,3)$}&\multirow{2}{*}{$(2,2)$}&\multirow{2}{*}{568}&\multirow{2}{*}{336}&DTL&6.200&5.424&0.776&0.313&[0.163, 1.389]&0.013\\
		& & & &S-DRL&6.209&5.390&0.819&0.314&[0.204, 1.434]&0.009\\
		\grayline
		\multirow{2}{*}{$(2,2)$}&\multirow{2}{*}{$(1,1)$}&\multirow{2}{*}{336}&\multirow{2}{*}{315}&DTL&5.410&5.639&-0.229&0.335&[-0.886, 0.428]&0.493\\
 & & & &S-DRL&5.371&5.626&-0.255&0.337&[-0.916, 0.406]&0.450 \\
		\bottomrule
	\end{tabular}
\end{table}

We analyze a dataset of 11,313 individuals, with 6,828 assigned to the Job Corps and 4,485 not. They are interviewed 1, 2, and 4 years post-randomization. For each year $t\in{1,2}$, $Z_t\in\{0,1,2,3\}$ represents the treatment assignment in the $t$-th year. We assign $Z_t=0$ for non-enrollment, $Z_t=1$ for enrollment without program participation, $Z_t=2$ for high-school-level education, and $Z_t=3$ for vocational training. The baseline covariate vector, $\bS_1$, has 909 characteristics, while $\bS_2$ includes 1,427 characteristics. In total, there are 2,336 covariates. The outcome is the log-transformed wage $\Ytil=\log(\mathrm{wage}+1)\in\R$. We exclude 2,610 individuals with missing treatment stages that are missing completely at random \citep{Schochet2003national}
and an additional 133 with missing covariates or outcomes, resulting in a final sample of 8,570 individuals.

\begin{figure}[ht]
\centering
\captionsetup{justification=centering, font=small}
\begin{subfigure}[]{0.3\linewidth}
\centering
 \includegraphics[height=0.8\linewidth,width=0.9\linewidth]{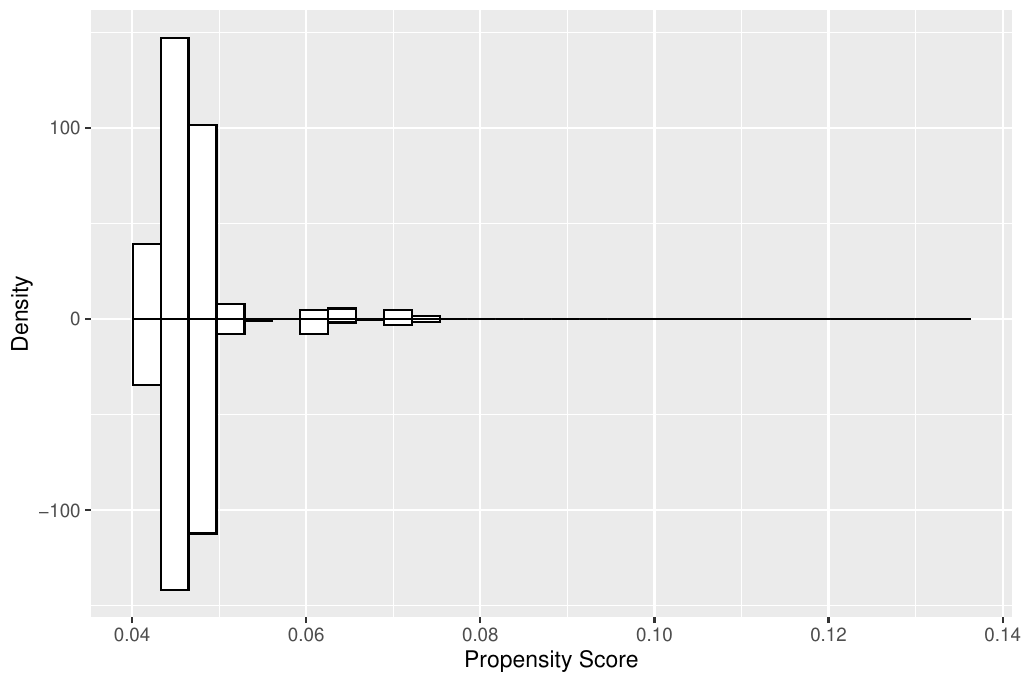}
\subcaption*{\tiny (a) Estimates of $P[Z_{1i}=1|\bS_i]$.\\
Above: within the subgroup $Z_{1i}=1$; \\
Below: within the subgroup $Z_{1i}\neq1$.}
\end{subfigure}%
\begin{subfigure}[]{0.3\linewidth}
\centering
\includegraphics[height=0.8\linewidth,width=0.9\linewidth]{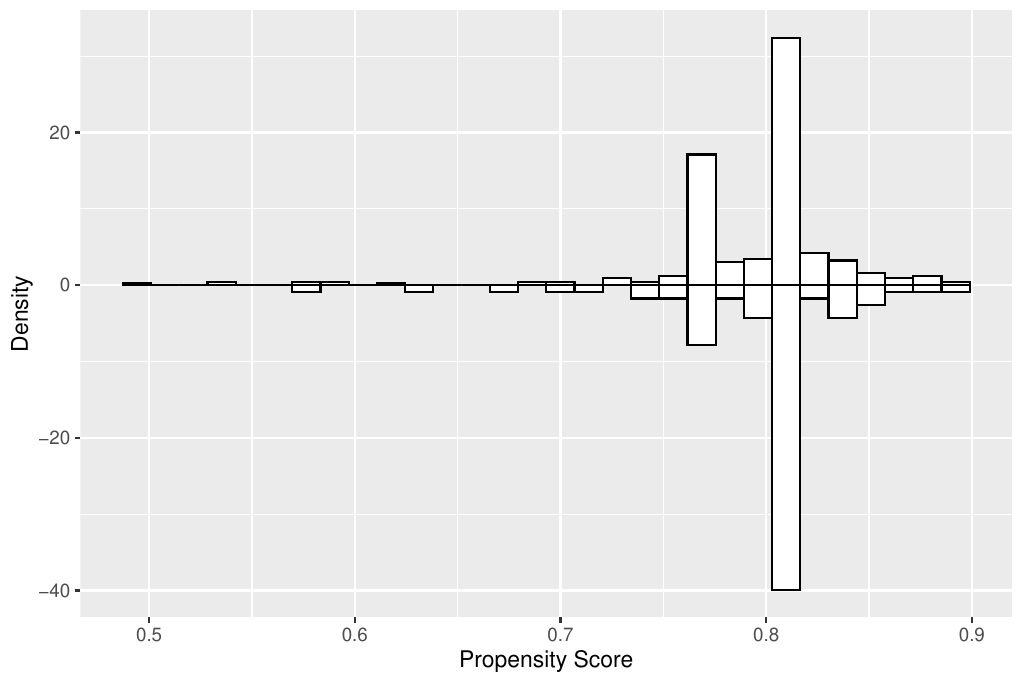}
\subcaption*{\tiny (b) Estimates of $P[Z_{2i}=1|\bS_i, Z_{1i}=1]$.\\
Above: within the subgroup $Z_{1i}=Z_{2i}=1$; \\
Below: within the subgroup $Z_{1i}=1,Z_{2i}\neq1$.}
\end{subfigure}
\begin{subfigure}[]{0.3\linewidth}
\centering
\includegraphics[height=0.8\linewidth,width=0.9\linewidth]{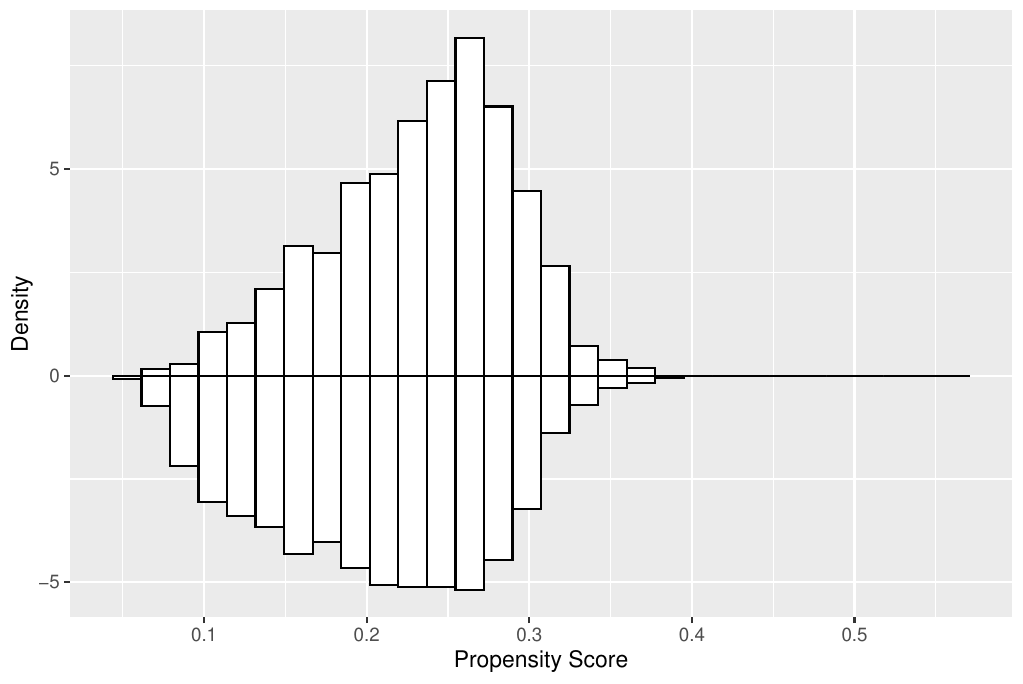}
\subcaption*{\tiny (c) Estimates of $P[Z_{1i}=2|\bS_i]$.\\
Above: within the subgroup $Z_{1i}=2$; \\
Below: within the subgroup $Z_{1i}\neq2$.}
\end{subfigure}
\\ 
\begin{subfigure}[]{0.3\linewidth}
\centering
\includegraphics[height=0.8\linewidth,width=0.9\linewidth]{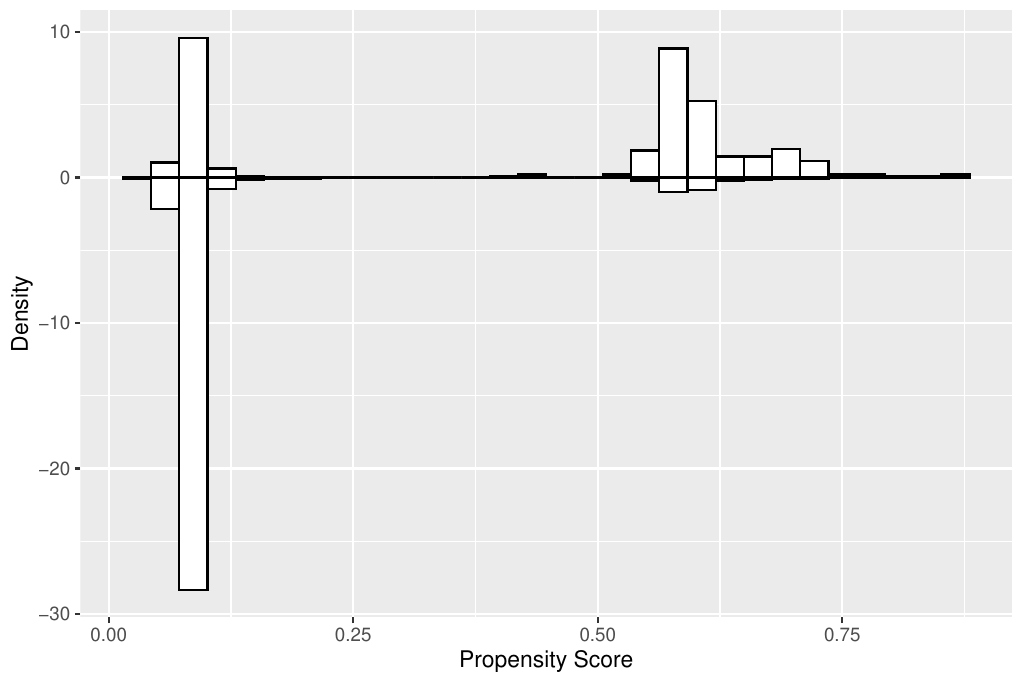}
\subcaption*{\tiny (d) Estimates of $P[Z_{2i}=2|\bS_i, Z_{1i}=2]$.\\
Above: within the subgroup $Z_{1i}=Z_{2i}=2$;\\
Below: within the subgroup $Z_{1i}=2,Z_{2i}\neq2$.}
\end{subfigure}
\begin{subfigure}[]{0.3\linewidth}
\centering		\includegraphics[height=0.8\linewidth,width=0.9\linewidth]{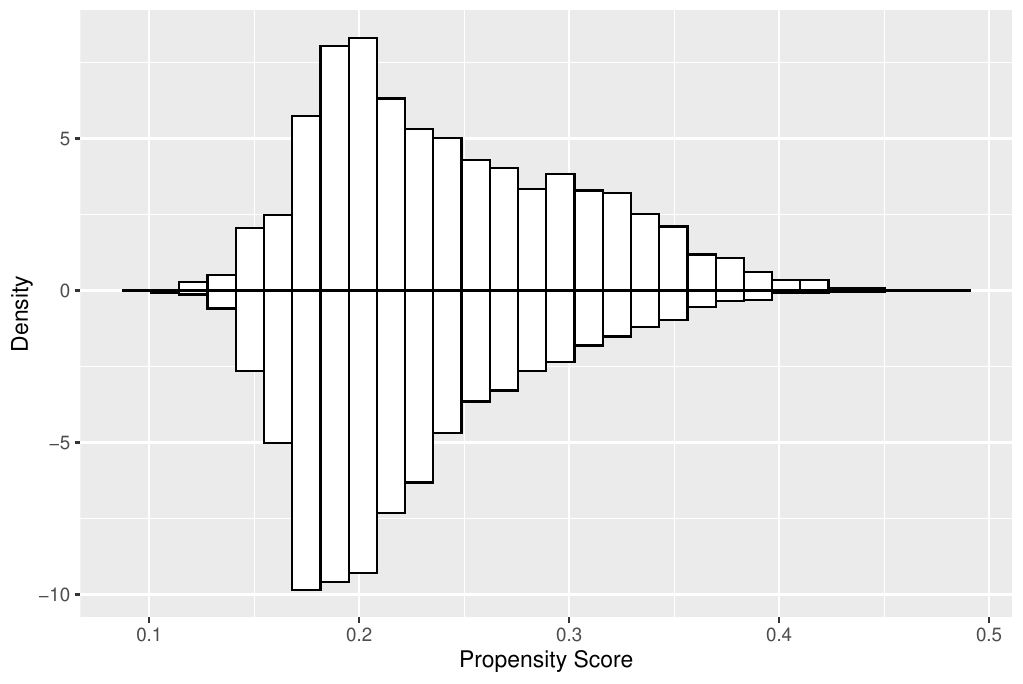}
\subcaption*{\tiny (e) Estimates of $P[Z_{1i}=3|\bS_i]$.\\
Above: within the subgroup $Z_{1i}=3$; \\
Below: within the subgroup $Z_{1i}\neq3$.}
\end{subfigure}%
\begin{subfigure}[]{0.3\linewidth}
\centering
\includegraphics[height=0.8\linewidth,width=0.9\linewidth]{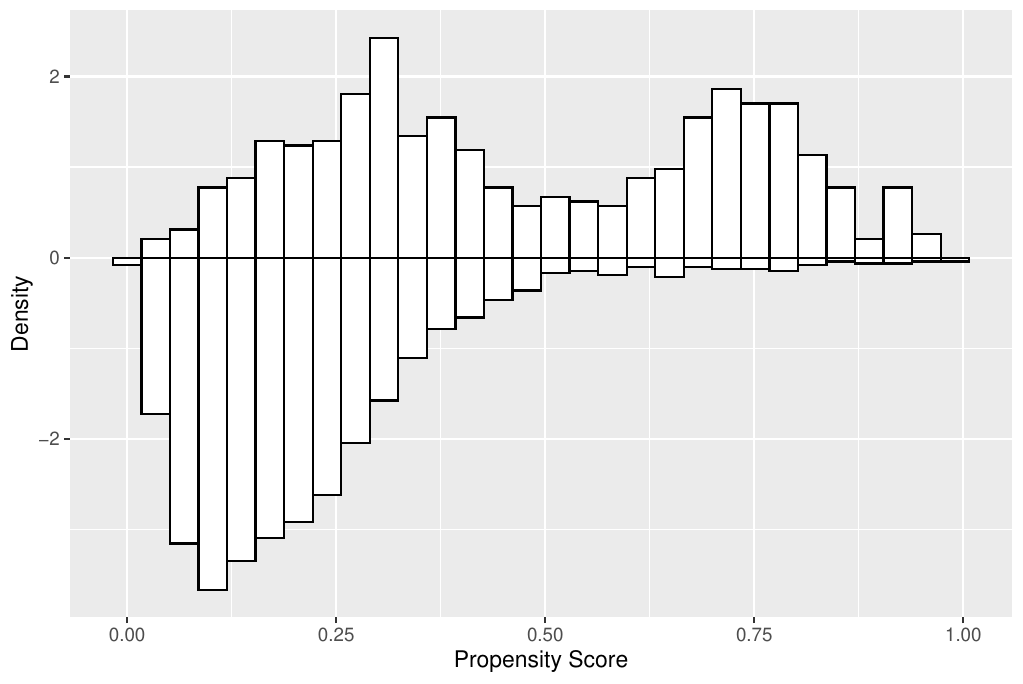}
\subcaption*{\tiny (f) Estimates of $P[Z_{2i}=3|\bS_i, Z_{1i}=3]$.\\
Above: within the subgroup $Z_{1i}=Z_{2i}=3$; \\
Below: within the subgroup $Z_{1i}=3,Z_{2i}\neq3$.}
\end{subfigure}
\caption{Mirror histograms of propensity score overlaps for $(z_1,z_2)\in\{(3,3),(2,2),(1,1)\}$.}\label{fig:overlap}
\end{figure}

 Table \ref{table:DTE-JC} shows estimated DTEs between treatment paths $(3,3)$ vs. $(1,1)$, $(3,3)$ vs. $(2,2)$, and $(2,2)$ vs. $(1,1)$.
 Both the S-DRL and DTL methods suggest that vocational training has a positive impact on achieving higher wages by showing non-zero effects between the first two paths.
On the other hand, the estimates between paths $(2,2)$ and $(1,1)$ are negative, and their corresponding confidence intervals contain zero, making it impossible to determine if academic education is beneficial or detrimental. However, our analysis does suggest that individuals seeking higher-paying jobs would benefit more from vocational training compared to academic education, which only provides high school-level education without any significant vocational training. The S-DRL estimates have a slightly greater distance from zero compared to DTL's, with similar standard errors, leading to slightly smaller p-values.
 Figure \ref{fig:overlap} 
 examines the overlap of estimated propensity scores, displaying mirror histograms of estimated propensity scores within the treatment and control groups. The substantial overlap seen in the mirror histograms indicates that the inverse propensity score weights are relatively stable. Figure \ref{fig:overlap}(d) displays bimodal patterns in the histograms due to a binary confounder variable that significantly influences the propensity score estimate of $P[Z_{2}=2\mid\bS,Z_{1}=2]$, with a close association between a participant's decision to enroll in second-year education and their attendance in the class during the final weeks of the first year. 

 \section{Discussion}\label{sec:dis}
 

This paper aims to enhance the understanding of estimating causal parameters in multi-stage settings. While prior DR literature has recognized the importance of the stage-zero DR representation for the expected potential outcome, it has overlooked the fact that all the intermediate conditional mean functions can also be identified in a DR manner. This approach leads to better theoretical guarantees and greater flexibility in modeling dynamic dependencies, which can be complex and involve multiple time exposures. Furthermore, our findings have significant practical implications beyond parametric models, especially in situations where doctors or policymakers cannot rely on randomized treatments or simplistic treatment rules. With the ability to model dynamic treatment effects using robust principles, new avenues of discovery are emerging, including optimal treatment rules and determining the best treatment times. Our approach also enables the exploration of further important issues, such as mitigating network spillover effects through robustness perspectives and enriching balancing methods with better robustness properties. The significance of our work lies in the fact that it allows researchers to estimate treatment effects in complex settings more accurately and provides a valuable tool for policymakers seeking to make informed decisions based on robust causal inference methods.

\section{Acknowledgement}
This work was supported in part by NSF awards CNS-1730158, ACI-1540112, ACI-1541349, OAC-1826967, the University of California Office of the President, and the University of California San Diego's California Institute for Telecommunications and Information Technology/Qualcomm Institute.
Jelena Bradic's work has been supported by the NSF grand DMS-1712481.
The majority of this work was done while Yuqian Zhang was with the Department of Mathematics, University of California San Diego.

\begin{supplement}
 \stitle{Supplementary Materials for the ``High-dimensional inference for dynamic treatment effects'', \cite{bradic2023supplement}}
 
\sdescription{This supplementary document contains additional justifications and the proofs of the theoretical results presented in the main document. All the results and notation are numbered and used as in the main text unless stated otherwise. Statements introduced in the Supplementary Materials only are numbered using an alphanumerical scheme. Supplementary Materials includes further discussions on the nuisance models, additional numerical results, Auxiliary Lemmas S.1-S.18 with their proofs used for establishing the main results, Theorems \ref{thm:DR-mu}-\ref{DR_multi_2}.} 
\end{supplement}

\appendix

\bibliographystyle{imsart-nameyear} 
\bibliography{ref}

\renewcommand{\thetheorem}{S.\arabic{theorem}}
\renewcommand{\thelemma}{S.\arabic{lemma}}


\clearpage\newpage  
\par\bigskip         
\begin{center}
\textbf{\uppercase{Supplementary Materials for the ``High-dimensional inference for dynamic treatment effects''}}
\end{center}

\par\medskip
This supplementary document contains additional justifications and the proofs of the theoretical results presented in the main document. All the results and notation are numbered and used as in the main text unless stated otherwise. Statements introduced in the Supplementary Materials only are numbered using an alphanumerical scheme. Supplementary Materials includes further discussions on the nuisance models, comparison with the oracle inverse propensity score estimator, additional numerical results, Auxiliary Lemmas S.1-S.18 with their proofs used for establishing the main results, Theorems 1-11.

To simplify the exposition, we begin by listing some shorthand notations used throughout the supplementary document. We let $\bU=(1, \bS)^{\top} \in \mathbb{R}^{d+1}$ and $\bV=(1, \bS_{1}^{\top})^{\top} \in \mathbb{R}^{d_1+1}$. In the following it is important to follow the individuals with pre-specified treatment plan. For that purpose we introduce the following shorthand notation:
$\widetilde{Y}_{a}=Y\mathbbm{1}_{\{\bm A=a\}}$, $ \widetilde{\bU}_{a}=\bU\mathbbm{1}_{\{ \bm A=a\}}$ where $ \bm A=(A_{1},A_{2})=a$. Additionally, we use 
$\bar Y_{a}=Y\mathbbm{1}_{\{A_{1}=a_1\}}$,
$\bar{\bU}_{a}=\bU\mathbbm{1}_{\{A_{1}=a_1\}}$,
$ \bar{\bV}_{a}=\bV\mathbbm{1}_{\{A_{1}=a_1\}}$ to denote individuals who have taken the treatment $a_1$ regardless of which treatment they received at the second exposure. 
Where possible, we suppress the sub-index $a$.

\section{Further discussions on the nuisance models} \label{sec:a}

\subsection{Model correctness of DTL}\label{subsec:model_mis}
We illustrate when will the two working outcome models $\nu_a^*(\bs)=\bu^{\top}\bm{\alpha}^*_{a}$ and $\mu_{a,{\mbox{\tiny NR}}}^*(\bs_{1})=\bv^{\top}\bm{\beta}_{a,{\mbox{\tiny NR}}}^{*}$, the models used for DTL, be correctly specified. If the model $\nu_a^*(\cdot)$ is misspecified, then the model $\mu_{a,{\mbox{\tiny NR}}}^*(\cdot)$ is also very likely to be misspecified, but there are no guarantees either way. A few comments are in order as the relationship between the two nested models is often masked. The following four cases are of potential interest. Their justifications are provided in Section \ref{subsection-justification} below. 

\begin{itemize}
\item[(i)] If we assume that the true outcome model, $\nu_a(\cdot)$ is linear in that 
\begin{align}
	\nu_a(\bS)=E[Y(a)|\bS,A_1=a_1, A_2=a_2]=\bU^{\top}\bm{\alpha}_{a}\label{nu specified}
\end{align}
holds for some vector $\bm{\alpha}_{a}\in\R^{d+1}$,
then it follows that $\bm{\alpha}_{a}^{*}=\bm{\alpha}_{a}$ and hence
$\nu_{a}^*(\cdot)=\nu_{a}(\cdot)$, i.e., $\nu_{a}^*(\cdot)$ is correctly specified.
\item[(ii)] Otherwise, if we assume that (only) the true outcome model, $\mu_a(\cdot)$, is linear in that
\begin{align}
	\mu_a(\bS_1)=E[Y(a)|\bS_1,A_{1}=a_1]&=\bV^{\top}\bm{\beta}_{a}\label{mu specified}
\end{align}
holds for some vector $\bm{\beta}_{a}\in\R^{d_1+1}$,
then it is possible that the working model is still not linear, i.e., $\mu_{a,{\mbox{\tiny NR}}}^*(\cdot)\neq\mu_{a}(\cdot)$ making $\mu_{a,{\mbox{\tiny NR}}}^*(\cdot)$ potentially misspecified.
\item[(iii)] Now, if the true outcome model \eqref{mu specified} holds and in addition $\bm{\alpha}_{a}^{*}$, (3.2), is equal to $\bar{\bm{\alpha}}_{a}^{*}$, with $\bar{\bm{\alpha}}_{a}^{*}$ defined as 
\begin{align*}
	\bar\balpha_a^*:=
	&\arg\min_{\balpha\in\R^{d+1}}E\left[(Y(a)-\bU^\top\balpha)^2|A_{1}=a_1\right]
	=\left[E[\bar \bU\bar \bU^\top]\right]^{-1}E[\bar \bU Y(a)],
\end{align*}
then, we have 
$\bm{\beta}_{a,{\mbox{\tiny NR}}}^{*}=\bm{\beta}_{a}$ and $\mu_{a,{\mbox{\tiny NR}}}^*(\cdot)=\mu_{a}(\cdot)$, i.e., $\mu_{a,{\mbox{\tiny NR}}}^*(\cdot)$ is correctly specified.
\item[(iv)] Lastly, if both of the true outcome models are linear, i.e., \eqref{nu specified} and \eqref{mu specified} hold simultaneously, then, both $\nu_a^*(\cdot)$ and $\mu_{a,{\mbox{\tiny NR}}}^*(\cdot)$ are correctly specified. Case (iv) is equivalent to requiring $E(\bS_{2}^\top\balpha_{a,2}|\bm S_1)$ to be linear in $\bm S_1$; here, $\balpha_a=(\balpha_{a,1},\balpha_{a,2})^\top$ where $\balpha_{a,1}\in\R^{d_1+1}$ and $\balpha_{a,2}\in\R^{d_2}$. This, in turn, occurs for any closed class of spherical distributions, including normal and Student-$t$ distributions, or any linear time-series models of covariate dependence.
\end{itemize}
Some discussions are provided below. We can see that the correctness of the model $\mu_{a,{\mbox{\tiny NR}}}^*(\cdot)$ also depends on $\balpha^*_a$, the slope parameter of $\nu_a^*(\cdot)$.
A true linear outcome model $\mu_a(\cdot)$ does not guarantee a correctly specified $\mu_{a,{\mbox{\tiny NR}}}^*(\cdot)$; however, if the true outcome model $\nu_a(\cdot)$ is also linear, then $\mu_{a,{\mbox{\tiny NR}}}^*(\cdot)$ is correctly specified. Moreover, a linear $\nu_a(\cdot)$ and $\mu_a(\cdot)$ are sufficient for a correctly specified $\nu_a^*(\cdot)$, but they are not required. 
Case (iii) provides an illustration where a correctly specified $\mu_{a,{\mbox{\tiny NR}}}^*(\cdot)$ does not require a correctly specified $\nu_a^*(\cdot)$. This occurs, for example, whenever $\bm{\alpha}_{a}^{*}=\bar{\bm{\alpha}}_{a}^*$.

For an illustration, consider $a=(1,1)$ and $S_1,S_2,Z\sim^\mathrm{iid}\mathrm{Uniform}(-1,1)$ with a nonlinear outcome model $\nu_a(\cdot)$, $Y(a)=S_1+S_2^3+Z$. Let the treatment assignments satisfy
$$\pi_{a}(s_1)=|s_1|, \mbox{ and }\rho_a(s_1,s_2)=\exp(s_1+s_2)/\{1+\exp(s_1+s_2)\},$$
for all $s_1,s_2\in\R$. Then, $\bm{\alpha}_{a}^{*}=\bar{\bm{\alpha}}_{a}^{*}$ and therefore guaranteeing correctness of the linear working model $\mu_{a,{\mbox{\tiny NR}}}^*(\cdot)$. Here, $\pi_{a}^*(\cdot)$ and $\nu_a^*(\cdot)$ are misspecified, $\rho_a^*(\cdot)$ and $\mu_{a,{\mbox{\tiny NR}}}^*(\cdot)$ are correctly specified.

\subsection{Model correctness of S-DRL}\label{subsec:model_mis'}
Note that we consider the same working models $\nu_a^*(\cdot)$, $\pi_{a}^*(\cdot)$, and $\rho_a^*(\cdot)$ for the DTL and S-DRL estimators; only the outcome model at time two differs among these two estimators. In the following, we illustrate when will the doubly robust working model $\mu_a^*(\bs_1)=\bv^\top\bbeta_a^*$ be misspecified and when will $\bbeta_a^*=\bm{\beta}_{a,{\mbox{\tiny NR}}}^{*}$.

\begin{itemize}
\item[(i)] Let either $\nu_a^*(\cdot)=\nu_a(\cdot)$ or $\rho_a^*(\cdot)=\rho_a(\cdot)$. Then, as long as $\mu_a(\cdot)$ is linear in that \eqref{mu specified} holds for some vector $\bm{\beta}_{a}\in\R^{d_1+1}$, we have $\bbeta_a^*=\bbeta$ and hence $\mu_a^*(\bs_1)=\bv^\top\bbeta_a^*=\mu_a(\bs_1)$, i.e., the doubly robust working model $\mu_a^*(\cdot)$ is correctly specified. 
\item[(ii)] If $\nu_a^*(\cdot)=\nu_a(\cdot)$, then $\bbeta_a^*=\bm{\beta}_{a,{\mbox{\tiny NR}}}^{*}$.
\end{itemize}
The justifications of cases (i) and (ii) are also provided in Section \ref{subsection-justification} below.

\subsection{Justifications}\label{subsection-justification}

\subsubsection{Justifications of cases (i)-(iv) in Section \ref{subsec:model_mis}}
For (i), under Assumption 1 and by the tower rule, we have
\begin{align*}
	\bm{\alpha}_{a}^{*}&=\left[E\left[\widetilde{\bU}\widetilde{\bU}^{\top}\right]\right]^{-1}E\left[\widetilde{\bU}\widetilde{Y}\right]=\left[E\left[\widetilde{\bU}\widetilde{\bU}^{\top}\right]\right]^{-1}E\left[\mathbbm1_{\{A_1=a_1, A_2=a_2\}}\bU Y(a)\right]\\
	&=\left[E\left[\widetilde{\bU}\widetilde{\bU}^{\top}\right]\right]^{-1}E\left[\bU E\left[Y(a)|\bU,A_1=a_1, A_2=a_2\right]P\left[A_1=a_1, A_2=a_2|\bU\right]\right]\\
	&=\left[E\left[\widetilde{\bU}\widetilde{\bU}^{\top}\right]\right]^{-1}E\left[\bU\bU^\top\balpha_aE\left[\mathbbm1_{\{A_1=a_1, A_2=a_2\}}|\bU\right]\right]\\
	&=\left[E\left[\widetilde{\bU}\widetilde{\bU}^{\top}\right]\right]^{-1}E\left[\widetilde{\bU}\widetilde{\bU}^{\top}\right]\balpha_a=\bm{\alpha}_{a}.
\end{align*}
It follows that
\begin{align*}
	\nu_{a}(\bS)=\bU^{\top}\bm{\alpha}_{a}=\bU^{\top}\bm{\alpha}_{a}^*=\nu_{a}^*(\bS).
\end{align*}
Therefore, if the model \eqref{nu specified} holds, the working model $\nu_{a}^*(\bS)$ is correctly specified.

\vspace{0.5em}

For (ii), it suffices to prove a counterexample. We refer to Example 1 in Section 2.3. 

\vspace{0.5em}

For (iii), if we assume that $\bar{\bm{\alpha}}_{a}^{*}=\bm{\alpha}_{a}^{*}$, we have 
\begin{align*}
	\bm{\beta}_{a,{\mbox{\tiny NR}}}^{*}&=\left[E\left[\bar{\bV}\bar{\bV}^{\top}\right]\right]^{-1}E\left[\bar{\bV}\bar{\bU}^{\top}\right]\bm{\alpha}_{a}^{*}=\left[E\left[\bar{\bV}\bar{\bV}^{\top}\right]\right]^{-1}E\left[\bar{\bV}\bar{\bU}^{\top}\right]\bar{\bm{\alpha}}_{a}^*\\
	&=\left[E\left[\bar{\bV}\bar{\bV}^{\top}\right]\right]^{-1}E\left[\bar{\bV}\bar{\bU}^{\top}\right]\left[E\left[\bar{\bU}\bar{\bU}^{\top}\right]\right]^{-1}E\left[\bar{\bU}Y(a)\right].
\end{align*}
By the fact that $\bU=(\bV^\top, \bS_2^\top)^\top$, we can write
\begin{align}
	\bV=\mathbf{Q}\bU\;\;\text{where}\;\;\mathbf{Q}=\begin{pmatrix}\bm{I}_{d_1+1} & \bm{0}_{(d_1+1)\times d_2}\end{pmatrix},\label{72}
\end{align}
and hence $\bar\bV=\mathbf{Q}\bar\bU$, which implies that
\begin{align*}
	&E\left[\bar{\bV}\bar{\bU}^{\top}\right]\left[E\left[\bar{\bU}\bar{\bU}^{\top}\right]\right]^{-1}E\left[\bar{\bU}Y(a)\right]=\mathbf{Q}E\left[\bar\bU\bar\bU^{\top}\right]\left[E\left[\bar{\bU}\bar{\bU}^{\top}\right]\right]^{-1}E\left[\bar{\bU}Y(a)\right]\\
	&\qquad=\mathbf{Q}E\left[\bar{\bU}Y(a)\right]=E\left[\bar\bV Y(a)\right].
\end{align*}
Therefore, 
\begin{align*}
	\bm{\beta}_{a,{\mbox{\tiny NR}}}^{*}
	&=\left[E\left[\bar{\bV}\bar{\bV}^{\top}\right]\right]^{-1}E\left[\bar{\bV}Y(a)\right]=\left[E\left[\bar{\bV}\bar{\bV}^{\top}\right]\right]^{-1}E\left[\mathbbm1_{\{A_{1}=a_1\}}\bV Y(a)\right].
\end{align*}
By the tower rule,
\begin{align}
	\bm{\beta}_{a,{\mbox{\tiny NR}}}^{*}
	&=\left[E\left[\bar{\bV}\bar{\bV}^{\top}\right]\right]^{-1}E\left[\bV E\left[Y(a)|\bV,A_{1}=a_1\right]E\left[\mathbbm1_{\{A_{1}=a_1\}}|\bV\right]\right]\nonumber\\
	&=\left[E\left[\bar{\bV}\bar{\bV}^{\top}\right]\right]^{-1}E\left[\mathbbm1_{\{A_{1}=a_1\}}\bV\bV^\top\bbeta_a\right]=\bm{\beta}_{a}.\nonumber
\end{align}
It follows that
\begin{align*}
	\mu_a(\bS_1)=\bV^{\top}\bm{\beta}_{a}=\bV^{\top}\bm{\beta}_{a,{\mbox{\tiny NR}}}^{*}=\mu_{a,{\mbox{\tiny NR}}}^*(\bS_{1}).
\end{align*}
Therefore, if the model \eqref{mu specified} holds and $\bar{\bm{\alpha}}_{a}^{*}=\bm{\alpha}_{a}^{*}$, the working model $\mu_{a,{\mbox{\tiny NR}}}^*(\bS_{1})$ is correctly specified.

\vspace{0.5em}

Regarding (iv), based on the results in (i), we have $\balpha_a^*=\balpha_a$. Under Assumption 1 and \eqref{nu specified}, we have 
$$\nu_{a}(\bS)=E\left[Y(a)|\bS,A_{1}=a_1\right]
=\bU^{\top}\bm{\alpha}_{a}.$$
Hence, we also have
\begin{align*}
\bar\balpha_a^*&=\left[E\left[\bar{\bU}\bar{\bU}^{\top}\right]\right]^{-1}E\left[\bar{\bU}Y(a)\right]=\left[E\left[\bar{\bU}\bar{\bU}^{\top}\right]\right]^{-1}E\left[\mathbbm1_{\{A_{1}=a_1\}}\bU Y(a)\right]\\
&=\left[E\left[\bar{\bU}\bar{\bU}^{\top}\right]\right]^{-1}E\left[\bU E\left[Y(a)|\bU,A_{1}=a_1\right]P\left[A_{1}=a_1|\bU\right]\right]\\
&=\left[E\left[\bar{\bU}\bar{\bU}^{\top}\right]\right]^{-1}E\left[\bU\bU^\top\balpha_aE\left[\mathbbm1_{\{A_{1}=a_1\}}|\bU\right]\right]\\
&=\left[E\left[\bar{\bU}\bar{\bU}^{\top}\right]\right]^{-1}E\left[\bar{\bU}\bar{\bU}^{\top}\right]\balpha_a=\bm{\alpha}_{a}.
\end{align*}
Therefore, 
$$\bm{\alpha}_{a}^{*}=\bar{\bm{\alpha}}_{a}^{*}=\bm{\alpha}_{a}.$$
Together with the results in (iii), we conclude that $\mu_{a,{\mbox{\tiny NR}}}^*(\cdot)$ is correctly specified.

\subsubsection{Justifications of cases (i)-(ii) in Section \ref{subsec:model_mis'}}
For (i), by the definition of $\bbeta_a^*$ and the KKT condition, we have
\begin{align*}
	&\bbeta_a^*=\left[E\left[\bar{\bV}\bar{\bV}^{\top}\right]\right]^{-1}E\left[\bar{\bV}Y^{\mbox{\tiny DR}}\right]\\
	&\;\;=\left[E\left[\bar{\bV}\bar{\bV}^{\top}\right]\right]^{-1}E\left[\mathbbm1_{\{A_{1}=a_1\}}\bV\left[\nu_a^*(\bS)+\mathbbm1_{\{A_2=a_2\}} \frac{Y-\nu_a^*(\bS)}{\rho_a^*(\bS)}\right]\right]\\
	&\;\;\overset{(i)}{=}\left[E\left[\bar{\bV}\bar{\bV}^{\top}\right]\right]^{-1}E\left[E\left[\mathbbm1_{\{A_{1}=a_1\}}\mid\bS_1\right]\bV E\left[\nu_a^*(\bS)+\mathbbm1_{\{A_2=a_2\}} \frac{Y-\nu_a^*(\bS)}{\rho_a^*(\bS)}\mid\bS_1,A_{1}=a_1\right]\right]\\
	&\;\;\overset{(ii)}{=}\left[E\left[\bar{\bV}\bar{\bV}^{\top}\right]\right]^{-1}E\left[E[\mathbbm1_{\{A_{1}=a_1\}}\mid\bS_1]\bV\mu_a(\bS_1)\right]=\left[E\left[\bar{\bV}\bar{\bV}^{\top}\right]\right]^{-1}E\left[\bar\bV\mu_a(\bS_1)\right]\\
	&\;\;\overset{(iii)}{=}\left[E\left[\bar{\bV}\bar{\bV}^{\top}\right]\right]^{-1}E\left[\bar\bV\bar\bV^\top\bbeta_a\right]=\bbeta_a,
\end{align*}
where (i) holds by the tower rule; (ii) holds by Theorem 1 in Section 1.1; (iii) holds since $\bar\bV\mu_a(\bS_1)=\bar\bV\bV^\top\bbeta_a=\bar\bV\bar\bV^\top\bbeta_a$.

\vspace{0.5em}

For (ii), we observe that
\begin{align*}
	&\bbeta_a^*-\bm{\beta}_{a,{\mbox{\tiny NR}}}^{*}=\left[E\left[\bar{\bV}\bar{\bV}^{\top}\right]\right]^{-1}E\left[\bar{\bV}Y^{\mbox{\tiny DR}}-\bar\bV\bar\bU^\top\balpha_a^*\right]\\
	&\quad=\left[E\left[\bar{\bV}\bar{\bV}^{\top}\right]\right]^{-1}E\left[\mathbbm1_{\{A_1=a_1, A_2=a_2\}}\bV\frac{Y(a)-\nu_a^*(\bS)}{\rho_a^*(\bS)}\right]\\
	&\quad\overset{(i)}{=}\left[E\left[\bar{\bV}\bar{\bV}^{\top}\right]\right]^{-1}E\left[E[A_{1}=a_1\mid\bS]\bV E\left[\mathbbm1_{\{A_2=a_2\}}\frac{Y(a)-\nu_a^*(\bS)}{\rho_a^*(\bS)}\mid\bS,A_{1}=a_1\right]\right]\\
	&\quad\overset{(ii)}{=}\left[E\left[\bar{\bV}\bar{\bV}^{\top}\right]\right]^{-1}E\left[E[A_{1}=a_1\mid\bS]\bV E\left[\rho_a(\bS)\frac{\nu_a(\bS)-\nu_a^*(\bS)}{\rho_a^*(\bS)}\mid\bS,A_{1}=a_1\right]\right]=0,
\end{align*}
as long as $\nu_a^*(\cdot)=\nu_a(\cdot)$. Here, (i) holds by the tower rule; (ii) holds since $Y(a)\perp \!\!\! \perp A_2\mid\bS,A_{1}=a_1$ under Assumption 1, $\rho_a(\bS)=E[\mathbbm1_{\{A_2=a_2\}}\mid\bS,A_{1}=a_1]$, and $\nu_a(\bS)=E[Y(a)\mid\bS,A_{1}=a_1]$. Hence, $\bbeta_a^*=\bm{\beta}_{a,{\mbox{\tiny NR}}}^{*}$ as long as $\nu_a^*(\cdot)=\nu_a(\cdot)$.

\section{Additional numerical experiments}
In this section, we present additional simulation results under different data generating processes (DGPs) where all the nuisance functions are correctly parametrized. For each $i\leq N$, generate
$\bS_{1i}\sim^\mathrm{iid} N_{d_1}(\mathbf{0},\mathbf{I}_{d_1})$ and $A_{1i}|\bS_{1i}\sim\mathrm{Bernoulli}(\pi_{a} (\bS_{1i}))$, where $\pi_{a} (\bS_{1i})=g(\bV_i^\top\bgamma_{a})$. Let $\delta_{1i}\sim^\mathrm{iid} N(0,1)$, $\bdelta_{1i}\sim^\mathrm{iid} N_{d_1}(0,\mathbf{I}_{d_1})$ and $\bdelta_{2i}\sim^\mathrm{iid} N_{d_2}(0,\mathbf{I}_{d_2})$. The following models on $\bS_{2i}|(\bS_{1i},A_{1i})$ are considered:
\begin{itemize}
\item[M4.] (Sparse linear) $\bS_{2i}=W_s(A_{1i})\bS_{1i}+A_{1i}(1+\delta_{1i})\mathbf{1}_{d_2\times 1}+\bdelta_{2i}$. 
\item[M5.] (Dense linear) $\bS_{2i}=W_d(A_{1i})\bS_{1i}+A_{1i}(1+\delta_{1i})\mathbf{1}_{d_2\times 1}+\bdelta_{2i}$.
\item[M6.] (Dense quadratic) $\bS_{2i}=0.5\widetilde{W}_d(A_{1i})(\bS_{1i}^2-1)+W_d(A_{1i})\bS_{1i}+A_{1i}(1+\delta_{1i})\mathbf{1}_{d_2\times 1}+\bdelta_{2i}$, where $\bS_{1i}^2\in\R^{d_1}$ is the coordinate-wise square of $\bS_{1i}$.
\end{itemize}

For each $c=(c_1,c_2)\in\{a,a'\}$, the matrices $W_s(c),W_d(c),\widetilde{W}_d(c)\in\R^{d_2\times d_1}$ are defined as the following: for each $i\leq d_2$ and $j\leq d_1$,
\begin{align*}
\{W_s(a)\}_{i,j}&=0.8^{|i-j|}\mathbbm1\{|i-j|\leq1\},\quad\{W_d(a)\}_{i,j}=0.8^{|i-j|},\\
\{W_s(a')\}_{i,j}&=0.7^{|i-j|}\mathbbm1\{|i-j|\leq2\},\quad\{W_d(a')\}_{i,j}=0.7^{|i-j|},\\
\{\widetilde{W}_d(c)\}_{i,j}&=\{W_d(c)\}_{i,j}\mathbbm1\{j>3\}\;\;\text{for each}\;\;c\in\{a,a'\}.
\end{align*}
The treatment indicators at time $t=2$ are generated as
\begin{align*}
&A_{2i}|(\bS_{i},A_{1i}=c_1)\sim\mathrm{Bernoulli}(\rho_{c}(\bS_{i})),\;\;\text{with}\\
&\rho_c(\bS_{i})=g(c_1\bU_i^\top\etabold_a+(1-c_1)\bU_i^\top\etabold_{a'}),\;\;\text{for each}\;\;c=(c_1,c_2)\in\{a,a'\}.
\end{align*}
The outcome variables are generated as
\begin{align*}
Y_i=Y_i(A_{1i},A_{2i}),\;\;Y_i(c)&=\bU_i^\top\balpha_c+\zeta_i,\;\;\text{for each}\;\;c\in\{a,a'\},\;\;\text{where}\;\;\zeta_i\sim^\mathrm{iid}N(0,1). 
\end{align*}

The parameter values are chosen as
$
\balpha_c =(\balpha_{c,1}^\top,\balpha_{c,2}^\top)^\top$, for each $c\in\{a,a'\},$
$\balpha_{a,1} =(-1,-1,1,-1,\mathbf0_{(d_1-3)})^\top,$ $\balpha_{a,2}=(-1,-1,1,\mathbf0_{(d_2-3)})^\top,$
$\balpha_{a',1} =(1,1,1,-1,\mathbf0_{(d_1-3)})^\top,$ $\balpha_{a',2}=(1,1,1,\mathbf0_{(d_2-3)})^\top,$
$\bgamma_{a} =(0,1,1,1,\mathbf0_{(d_1-3)})^\top,$
$\etabold_a =(0,1,1,\mathbf0_{(d_1-2)},1,-1,\mathbf0_{(d_2-2)})^\top,$ and 
$\etabold_{a'} =(0,0.5,0,-0.5,$ $\mathbf0_{(d_1-3)},0.5,0,0.5,\mathbf0_{(d_2-3)})^\top,
$
where $\mathbf0_q:=(0,\ldots,0)\in\R^q$ for any $q\geq1$. Under the above DGPs, we have the following nuisance functions: for each $c\in\{a,a'\}$,
\begin{align}
\nu_c(\bS)&=E[Y(c)|\bS,A_1=c_1]=\bU^\top\balpha_c,\label{sim:nu}\\
\mu_c(\bS_{1})&=E[Y(c)|\bS_{1},A_1=c_1]=\bV^\top\balpha_{c,1}+E[\bS_2^\top\balpha_{c,2}|\bS_1,A_1=c_1]=\bV^\top\bbeta_c,\label{sim:mu}
\end{align}
where $\bbeta_c$ varies for different models on $\bS_{2i}|(\bS_{1i},A_{1i})$ as follows:
\begin{itemize}
\item[M4.] $\bbeta_c=\balpha_{c,1}+(\sum_{j=1}^{d_2}\balpha_{a',2}\mathbbm1\{c=a'\},(W_s(c)\balpha_{c,2})^\top)^\top$ with $\|\bbeta_a\|_0=4$ and $\|\bbeta_{a'}\|_0=5$.
\item[M5-6.] $\bbeta_c=\balpha_{c,1}+(\sum_{j=1}^{d_2}\balpha_{a',2}\mathbbm1\{c=a'\},(W_d(c)\balpha_{c,2})^\top)^\top$ is weakly sparse in that $\|\bbeta_a\|_0=\|\bbeta_{a'}\|_0=d_1+1$, $\|\bbeta_a\|_1<5.23$, and $\|\bbeta_{a'}\|_1<7.24$.
\end{itemize}

\begin{table}
\centering
\caption{Simulation under M4. Bias: empirical bias; RMSE: root mean square error; Length: average length of the $95\%$ confidence intervals; Coverage: average coverage of the $95\%$ confidence intervals; ESD: empirical standard deviation; ASD: average of estimated standard deviations. All the reported values (except Coverage) are based on robust (median-type) estimates. Denote $N_1$ and $N_0$ as the expected number of observations in the treatment groups $(1,1)$ and $(0,0)$, respectively.} \label{table:M4}
\begin{tabular}{@{\extracolsep{10pt}}ccccccc}
\toprule
Estimator&Bias&RMSE&Length&Coverage&ESD&ASD\\
\hline
\multicolumn{7}{c}{ \cellcolor{gray!50} $N=1000,N_1=279,N_0=312,d_1=100,d_2=50$}\\
\hline
empdiff&2.485&2.485&1.258&0.000&0.318&0.321\\
oracle&-0.035&0.243&1.305&0.972&0.350&0.333\\
DTL&0.063&0.218&1.121&0.934&0.326&0.286\\
S-DRL&0.133&0.202&0.881&0.874&0.262&0.225\\
S-DRL'&0.121&0.195&0.876&0.898&0.262&0.224\\
\hline
\multicolumn{7}{c}{ \cellcolor{gray!50} $N=4000,N_1=1115,N_0=1248,d_1=100,d_2=50$}\\
\hline
empdiff&2.484&2.484&0.627&0.000&0.162&0.160\\
oracle&0.003&0.125&0.706&0.946&0.185&0.180\\
DTL&0.029&0.119&0.600&0.928&0.171&0.153\\
S-DRL&0.031&0.122&0.601&0.926&0.169&0.153\\
S-DRL'&0.030&0.119&0.598&0.922&0.173&0.153\\
\bottomrule
\end{tabular}
\end{table}

 The following choices of parameters are implemented:
$N\in\{1000,4000\}$, $d_1=100$, and $d_2=d_1/2=50$. For each of the DGPs, we repeat the simulation for 500 times. For each replication, we consider the oracle estimator, the empirical difference estimator (empdiff), the DTL estimator, the S-DRL estimator, and the S-DRL' estimator as in Section 6.1. The results are reported in Tables \ref{table:M4}-\ref{table:M6}. 

The considered DGPs are only different on the procedure of generating $\bS_{2}$ based on $\bS_{1}$ and $A_1$. Under M4, we consider a sparse linear dependence that $\bS_{2}$ is linearly dependent on $\bS_{1}$ through a sparse and dense matrix operator, where the corresponding coefficient $\bbeta_c$ is a sparse vector. Under M5, we consider a dense linear dependence that the corresponding coefficient $\bbeta_c$ is only weakly sparse that it's $\|\cdot\|_1$ norm is bounded. Under M6, we consider a dense quadratic dependence between $\bS_{2}$ and $\bS_{1}$ but the nuisance function $\mu_c(\cdot)$ is still linear - we can see that the nuisance function can be linear even when $\bS_{2}$ is not linearly dependent on $\bS_{1}$. Note that, although $E(\bS_{2}|\bS_{1},A_1=c_1)$ is quadratic in $\bS_{1}$, $E(\bS_{2}^\top\alpha_{c,2}|\bS_{1},A_1=c_1)$ is still linear on $\bS_{1}$ and hence the linear working models $\mu_c^*(\cdot)=\mu_{c,{\mbox{\tiny NR}}}^*(\cdot)$ are both correctly specified as the second-time conditional mean $\nu_c(\cdot)$ is also linear. 

\begin{table}[h]
\centering
\caption{Simulation under M5. The rest of the caption details remain the same as those in Table \ref{table:M4}.} \label{table:M5}
\begin{tabular}{@{\extracolsep{10pt}}ccccccc}
\toprule
Estimator&Bias&RMSE&Length&Coverage&ESD&ASD\\
\hline
\multicolumn{7}{c}{ \cellcolor{gray!50} $N=1000,N_1=296,N_0=310,d_1=100,d_2=50$}\\
\hline
empdiff&2.921&2.921&1.239&0.000&0.317&0.316\\
oracle&0.002&0.245&1.346&0.962&0.364&0.343\\
DTL&0.084&0.219&1.139&0.920&0.322&0.291\\
S-DRL&0.097&0.225&1.140&0.922&0.323&0.291\\
S-DRL'&0.084&0.224&1.138&0.926&0.321&0.290\\
\hline
\multicolumn{7}{c}{ \cellcolor{gray!50} $N=4000,N_1=1184,N_0=1240,d_1=100,d_2=50$}\\
\hline
empdiff&2.922&2.922&0.619&0.000&0.159&0.158\\
oracle&-0.006&0.137&0.710&0.946&0.202&0.181\\
DTL&0.019&0.113&0.608&0.934&0.166&0.155\\
S-DRL&0.024&0.114&0.609&0.930&0.166&0.155\\
S-DRL'&0.019&0.114&0.609&0.932&0.166&0.155\\
\bottomrule
\end{tabular}
\end{table}

\begin{table}[h]
\centering
\caption{Simulation under M6. The rest of the caption details remain the same as those in Table \ref{table:M4}.} \label{table:M6}
\begin{tabular}{@{\extracolsep{10pt}}ccccccc}
\toprule
Estimator&Bias&RMSE&Length&Coverage&ESD&ASD\\
\hline
\multicolumn{7}{c}{ \cellcolor{gray!50} $N=1000,N_1=296,N_0=310,d_1=100,d_2=50$}\\
\hline
empdiff&2.921&2.921&1.239&0.000&0.317&0.316\\
oracle&0.002&0.245&1.346&0.962&0.364&0.343\\
DTL&0.083&0.225&1.141&0.924&0.318&0.291\\
S-DRL&0.077&0.228&1.139&0.914&0.320&0.290\\
S-DRL'&0.076&0.213&1.135&0.920&0.320&0.289\\
\hline
\multicolumn{7}{c}{ \cellcolor{gray!50} $N=4000,N_1=1184,N_0=1240,d_1=100,d_2=50$}\\
\hline
empdiff&2.922&2.922&0.619&0.000&0.159&0.158\\
oracle&-0.006&0.137&0.710&0.946&0.202&0.181\\
DTL&0.019&0.114&0.610&0.936&0.166&0.156\\
S-DRL&0.021&0.115&0.610&0.928&0.166&0.156\\
S-DRL'&0.020&0.112&0.608&0.932&0.166&0.155\\
\bottomrule
\end{tabular}
\end{table}

We first focus on the DTL, S-DRL, and S-DRL' estimators and compare their behaviors. We can see that when the model is relatively easy (under M4), and the total sample size is relatively small ($N=1000$), the DTL method provides better coverage but with a worse RMSE than the S-DRL and S-DRL' methods; see Table \ref{table:M4}. This is because, although the DTL estimator has a smaller bias, it also has a larger ESD compared with the S-DRL and S-DRL' estimators. If we further increase the sample size ($N=4000$), we can see that the coverages based on DTL, S-DRL, and S-DRL' estimators are close to each other and also overall acceptable. When the estimation of the first-time conditional mean is relatively hard as its linear parameter is only weakly sparse (under M4 and M5), we can see that the RMSEs and confidence intervals' coverages of the DTL, S-DRL, and S-DRL' methods are relatively close to each other for both $N=1000$ and $N=4000$; see Tables \ref{table:M5} and \ref{table:M6}. In addition, we can see that the considered estimators have very similar behaviors among DGPs M5 and M6. Note that the nuisance functions, including the propensity score and conditional mean functions, are the same under M5 and M6, although the conditional densities of $\bS_2$ given $(\bS_1,A_1)$ are different. This observation indicates that the considered estimators' behavior mainly relies on the conditional means of the potential outcomes and treatment variables instead of the conditional densities. Lastly, we can also see that the naive empirical difference estimator, $\thetahat_{\mbox{\tiny empdiff}}$, is not even consistent because of the appearance of confounders.

\section{Proof of the results for the doubly robust representation}
\begin{proof}[Proof of Lemma 1]
By the tower rule and $Y=Y(A_{1},A_2)$ under Assumption 1,
\begin{align*}
	&E\left[\mathbbm1_{\{A_1=a_1, A_2=a_2\}}\frac{Y-\nu_{a}^{*}(\bS)}{\pi_{a}^{*}(\bS_1)\rho_{a}^{*}(\bS)}\right]=E\left[\mathbbm1_{\{A_1=a_1, A_2=a_2\}}\frac{Y(a)-\nu_{a}^{*}(\bS)}{\pi_{a}^{*}(\bS_1)\rho_{a}^{*}(\bS)}\right]\\
	&\qquad=E\left[E\left[\mathbbm1_{\{A_1=a_1, A_2=a_2\}}\frac{Y(a)-\nu_{a}^{*}(\bS)}{\pi_{a}^{*}(\bS_1)\rho_{a}^{*}(\bS)}\mid\bS,A_{1}=a_1\right]P(A_{1}=a_1\mid\bS)\right].
\end{align*}
By $Y(a)\perp \!\!\! \perp A_2\mid\bS,A_{1}=a_1$ under Assumption 1,
\begin{align*} 		
	&E\left[E\left[\mathbbm1_{\{A_1=a_1, A_2=a_2\}}\frac{Y(a)-\nu_{a}^{*}(\bS)}{\pi_{a}^{*}(\bS_1)\rho_{a}^{*}(\bS)}\mid\bS,A_{1}=a_1\right]P(A_{1}=a_1\mid\bS)\right]\\
	&\qquad=E\left[\frac{E[\mathbbm1_{\{A_2=a_2\}}\mid\bS,A_{1}=a_1](E[Y(a)\mid\bS,A_{1}=a_1]-\nu_{a}^{*}(\bS))}{\pi_{a}^{*}(\bS_1)\rho_{a}^{*}(\bS)}E[\mathbbm1_{\{A_1=a_1\}}\mid\bS]\right]\\
	&\qquad\overset{(i)}{=}E\left[\frac{\rho_{a}(\bS)(\nu_{a}(\bS)-\nu_{a}^{*}(\bS))}{\pi_{a}^{*}(\bS_1)\rho_{a}^{*}(\bS)}E[\mathbbm1_{\{A_1=a_1\}}\mid\bS]\right]\\
	&\qquad\overset{(ii)}{=}E\left[\mathbbm1_{\{A_1=a_1\}}\frac{\rho_{a}(\bS)(\nu_{a}(\bS)-\nu_{a}^{*}(\bS))}{\pi_{a}^{*}(\bS_1)\rho_{a}^{*}(\bS)}\right],
\end{align*} 	
where (i) holds since $\rho_a(\bS)=P[A_2=a_2\mid\bS,A_{1}=a_1]$ and $\nu_a(\bS)=E\left[Y(a)\mid\bS,A_{1}=a_1\right]$; (ii) holds by the tower rule. Hence,
\begin{align*} 		
	E\left[\mathbbm1_{\{A_1=a_1, A_2=a_2\}}\frac{Y-\nu_{a}^{*}(\bS)}{\pi_{a}^{*}(\bS_1)\rho_{a}^{*}(\bS)}\right]=E\left[\mathbbm1_{\{A_1=a_1\}}\frac{\rho_{a}(\bS)(\nu_{a}(\bS)-\nu_{a}^{*}(\bS))}{\pi_{a}^{*}(\bS_1)\rho_{a}^{*}(\bS)}\right].
\end{align*} 	
Observe that
\begin{align}
&E\left[\mathbbm1_{\{A_1=a_1, A_2=a_2\}}\frac{ Y-\nu_{a}^*(\bS )}{\pi_{a} ^* (\bS_{1})\rho_{a}^*(\bS)}+\mathbbm1_{\{A_1=a_1\}}\frac{ \nu_{a}^*(\bS )-\mu_{a}^*(\bS_{1}) }{\pi_{a} ^* (\bS_{1})}+\mu_{a}^*(\bS_{1})\right]-\theta_{a}\nonumber \\
&\qquad=E\biggl[\mathbbm1_{\{A_1=a_1\}}\frac{\rho_{a}(\bS)(\nu_{a}(\bS)-\nu_{a}^{*}(\bS))}{\pi_{a}^{*}(\bS_1)\rho_{a}^{*}(\bS)}+\mathbbm1_{\{A_1=a_1\}}\frac{\nu_{a}^{*}(\bS)-\mu_{a}^{*}(\bS_1)}{\pi_{a}^{*}(\bS_1)}\nonumber\\
&\qquad\qquad+\mu_{a}^{*}(\bS_1)-\mu_{a}(\bS_1)\biggl]\overset{(i)}{=}G_1+G_2+G_3,\label{lem1_G1-G3}
\end{align}
where
\begin{align*}
	G_{1}:=&E\left[\mathbbm1_{\{A_1=a_1\}}\frac{\nu_{a}^*(\bS)-\nu_{a}(\bS)}{\pi_{a}^*(\bS_1)}\left(1-\frac{\rho_{a}(\bS)}{\rho_{a}^*(\bS)}\right)\right],\\
	G_{2}:=&E\left[(\mu_{a}^*(\bS_1)-\mu_{a}(\bS_1))\left(1-\frac{\mathbbm1_{\{A_1=a_1\}}}{\pi_{a}^*(\bS_1)}\right)\right],\\
	G_{3}:=&E\left[\mathbbm1_{\{A_1=a_1\}}\frac{\nu_{a}(\bS)-\mu_{a}(\bS_1)}{\pi_{a}^*(\bS_1)}\right].
\end{align*}
In the above, (i) holds by rearranging the terms after the following decomposition
$$ \nu_{a}^{*}(\bS)-\mu_{a}^{*}(\bS_1)= { (\nu_{a}^{*}(\bS)-\nu_{a}(\bS))} + { (\nu_{a}(\bS)-\mu_{a}(\bS_1))} + { (\mu_{a}(\bS)-\mu_{a}^{*}(\bS_1))} .$$
By assumption, either $\nu_{a}^*(\cdot)=\nu_{a}(\cdot)$ or $\rho_{a}^*(\cdot)=\rho_{a}(\cdot)$, we have
\begin{align} 
 G_1=0. \label{lem1_G1}
\end{align}
For $G_2$, by the tower rule, 
\begin{align}
 G_2&=E\left[E\left[(\mu_{a}(\bS_1)-\mu_{a}^*(\bS_1))\left(1-\frac{\mathbbm1_{\{A_1=a_1\}}}{\pi_{a}^*(\bS_1)}\right)\mid \bS_1\right]\right]\nonumber\\
 &\overset{(i)}{=} E\left[(\mu_{a}(\bS_1)-\mu_{a}^*(\bS_1))\left(1-\frac{\pi_{a}(\bS_1)}{\pi_{a}^*(\bS_1)}\right)\right]\overset{(ii)}{=}0, \label{lem1_G2}
\end{align}
where (i) holds since $\pi_{a} (\bS_{1})= P[A_{1}=a_1|\bS_{1}]$; (ii) holds since $\pi_{a} (\bS_{1})= P[A_{1}=a_1|\bS_{1}]$; (ii) holds since, by assumption, either $\mu_{a}^*(\cdot)=\mu_{a}(\cdot)$ or $\pi_{a}^*(\cdot)=\pi_{a}(\cdot)$.
For $G_3$, by the tower rule, 
\begin{align}
 G_3&=E\left[E\left[\mathbbm1_{\{A_1=a_1\}}\frac{\nu_{a}(\bS)-\mu_{a}(\bS_1)}{\pi_{a}^*(\bS_1)}\mid\bS_1,A_{1}=a_1\right]P(A_{1}=a_1\mid\bS_1)\right]\nonumber\\
	&\overset{(i)}{=}E\left[\frac{\pi_{a}(\bS_1)}{\pi_{a}^*(\bS_1)}\left[E\left[\nu_{a}(\bS)\mid\bS_1,A_{1}=a_1\right]-\mu_{a}(\bS_1)\right]\right]\overset{(ii)}{=}0,\label{lem1_G3}
\end{align}
where (i) holds since $\pi_{a} (\bS_{1})= P[A_{1}=a_1|\bS_{1}]$; (ii) holds since $\mu_{a}(\bS_1)=E[\nu_{a}(\bS)\mid\bS_1,A_{1}=a_1]$.
Combining \eqref{lem1_G1}-\eqref{lem1_G3} with \eqref{lem1_G1-G3}, we have
$$\theta_{a}=E\left[\mathbbm1_{\{A_1=a_1, A_2=a_2\}}\frac{ Y-\nu_{a}^*(\bS )}{\pi_{a} ^* (\bS_{1})\rho_{a}^*(\bS)}+\mathbbm1_{\{A_1=a_1\}}\frac{ \nu_{a}^*(\bS )-\mu_{a}^*(\bS_{1}) }{\pi_{a} ^* (\bS_{1})}+\mu_{a}^*(\bS_{1})\right].$$
\end{proof}

\begin{proof}[Proof of Theorem 1]
Observe that
\begin{align*}
	&E\left[\nu_a^*(\bS)+\mathbbm1_{\{A_2=a_2\}} \frac{Y-\nu_a^*(\bS)}{\rho_a^*(\bS)}\mid\bS_1,A_{1}=a_1\right]-\mu_a(\bS_1)\\
	&\qquad\overset{(i)}{=}E\left[\nu_a^*(\bS)-\nu_a(\bS)+\mathbbm1_{\{A_2=a_2\}} \frac{Y(a)-\nu_a^*(\bS)}{\rho_a^*(\bS)}\mid\bS_1,A_{1}=a_1\right]\\
	&\qquad=E\left[\left(\nu_a^*(\bS)-\nu_a(\bS)\right)\left(1-\frac{\mathbbm1_{\{A_2=a_2\}}}{\rho_a^*(\bS)}\right)+\mathbbm1_{\{A_2=a_2\}} \frac{Y(a)-\nu_a(\bS)}{\rho_a^*(\bS)}\mid\bS_1,A_{1}=a_1\right],
\end{align*}
where (i) holds since $\mu_a(\bS_1)=E[\nu_a(\bS )\mid\bS_1,A_{1}=a_1]$ and $Y=Y(A_{1},A_2)$ under Assumption 1. By the tower rule,
\begin{align*}
	&E\left[\left(\nu_a^*(\bS)-\nu_a(\bS)\right)\left(1-\frac{\mathbbm1_{\{A_2=a_2\}}}{\rho_a^*(\bS)}\right)\mid\bS_1,A_{1}=a_1\right]\\
	&\qquad=E\left[E\left[\left(\nu_a^*(\bS)-\nu_a(\bS)\right)\left(1-\frac{\mathbbm1_{\{A_2=a_2\}}}{\rho_a^*(\bS)}\right)\mid\bS,A_{1}=a_1\right]\mid\bS_1,A_{1}=a_1\right]\\
	&\qquad\overset{(i)}{=}E\left[\left(\nu_a^*(\bS)-\nu_a(\bS)\right)\left(1-\frac{\rho_a(\bS)}{\rho_a^*(\bS)}\right)\mid\bS_1,A_{1}=a_1\right]\overset{(ii)}{=}0,
\end{align*}
where (i) holds since $\rho_a(\bS)=P[A_2=a_2\mid\bS,A_{1}=a_1]$; (ii) holds since either $\nu_a^*(\cdot)=\nu_a(\cdot)$ or $\rho_a^*(\cdot)=\rho_a(\cdot)$ by Assumption. In addition, by the tower rule, we also have
\begin{align*}
	&E\left[\mathbbm1_{\{A_2=a_2\}} \frac{Y(a)-\nu_a(\bS)}{\rho_a^*(\bS)}\mid\bS_1,A_{1}=a_1\right]\\
	&\qquad=E\left[E\left[\mathbbm1_{\{A_2=a_2\}} \frac{Y(a)-\nu_a(\bS)}{\rho_a^*(\bS)}\mid\bS,A_{1}=a_1\right]\mid\bS_1,A_{1}=a_1\right]\\
	&\qquad\overset{(i)}{=}E\left[E\left[\mathbbm1_{\{A_2=a_2\}}\mid\bS,A_{1}=a_1\right]\frac{E\left[Y(a)\mid\bS,A_{1}=a_1\right]-\nu_a(\bS)}{\rho_a^*(\bS)}\mid\bS_1,A_{1}=a_1\right]\overset{(ii)}{=}0,
\end{align*}
where (i) holds since $Y(a)\perp \!\!\! \perp A_2\mid\bS,A_{1}=a_1$ under Assumption 1; (ii) holds since $\nu_a(\bS)=E\left[Y(a)\mid\bS,A_{1}=a_1\right]$. Therefore, for any $\bs_1\in\R^{d_1}$,
\begin{align*}
	\mu_a(\bs_1)=E\left[\nu_a^*(\bS)+\mathbbm1_{\{A_2=a_2\}} \frac{Y-\nu_a^*(\bS)}{\rho_a^*(\bS)}\mid\bS_1=\bs_1,A_1=a_1\right].
\end{align*}
\end{proof}

\section{Convergence rates for nuisance estimators}
\subsection{Auxiliary Lemmas}
The following lemmas will be helpful in our proofs.
\begin{lemma}[Selection of Lemma D.1 of \cite{chakrabortty2019high}] \label{lemma:psi2norm}
Let $X,Y\in \R$ be a random variable. Then $\|cX\|_{\psi_2}=|c|\|X\|_{\psi_2} \ \forall c\in\R $. If $|X|\leq|Y|$ a.s., then $\|X\|_{\psi_2}\leq\|Y\|_{\psi_2}$. 
Moreover, for $X$ and $Y$ sub-Gaussian, $\| XY\|_{\psi_1}\leq \| X\|_{\psi_2} \| Y\|_{\psi_2}$. 
If $X$ is bounded, i.e, $|X|\leq C$ a.s. for some constant $C$, then $\|X\|_{\psi_2}\leq(\log 2)^{-1/2}C$. If $\|X\|_{\psi_2}\leq\sigma$, then $E(|X|^m)\leq2\sigma^m \Gamma(m/2+1)\ \forall \ m\geq1 $, where $\Gamma(a):=\int_0^\infty x^{a-1}\exp(-x)dx$ $\forall a>0$ denotes the Gamma function. Hence, $E(|X|)\leq \sigma\sqrt\pi$ and $E(|X|^m)\leq2\sigma^m(m/2)^{m/2} \ \forall \ m\geq2 $. Let $\{X_i\}_{i=1}^n$ be random variables (possibly dependent) with $\max_{1\leq i \leq n}\|X_i\|_{\psi_2}\leq\sigma$, then $\|\max_{1\leq i \leq n}|X_i|\|_{\psi_2}\leq\sigma(\log n+2)^{1/2}$.
\end{lemma}
\begin{lemma}\label{lem_orlicz}
Let $X\in\R$ be a random variable. If $E(|X|^{2k})\leq 2\sigma^{2k}\Gamma(k+1)$ for any $k\in \mathbb{N}$, then $\|X\|_{\psi_2}\leq 2\sigma$. 
\end{lemma}
The following lemma provides the same type of results as used in the Assumption 3 but now for covariates at different exposure time and different treatment paths.
\begin{lemma}\label{lem_eigen} 
	Let the overlap conditions of Assumption 1 and Assumption 3 hold. Consider the constants $c_0,\kappa_l,\sigma_u$ defined as in Assumptions 1 and 3. Then, the following statements hold:
	
	(a) $0 < c_0\kappa_{l} \leq \lambda_{\min}(E[\widetilde{\bU}\widetilde{\bU}^{\top}]) \leq \lambda_{\max}(E[\widetilde{\bU}\widetilde{\bU}^{\top}]) \leq 2\sigma_{u}^2<\infty$
	and $\widetilde{\bU}$ is sub-Gaussian with 
	$\|\bm{x}^\top\widetilde{\bU} \|_{{\psi}_2} \leq 2\sigma_{u}\|\bm{x}\|_2$ for any $\bm{x}\in \mathbb{R}^{d+1}$; 
	
	(b) $0 < \kappa_{l} \leq \lambda_{\min}(E[\bar{\bU}\bar{\bU}^{\top}]) \leq\lambda_{\max}(E[\bar{\bU}\bar{\bU}^{\top}]) \leq 2\sigma_{u}^2<\infty$
	and $\bar{\bU}$ is sub-Gaussian with 
	$\|\bm{x}^\top\bar{\bU} \|_{{\psi}_2} \leq 2\sigma_{u}\|\bm{x}\|_2$ for any $\bm{x}\in \mathbb{R}^{d+1}$;
	
	c) $0 < \kappa_{l} \leq \lambda_{\min}(E[{\bU}{\bU}^{\top}]) \leq \lambda_{\max}(E[{\bU}{\bU}^{\top}]) \leq 2\sigma_{u}^2<\infty$ and ${\bU}$ is sub-Gaussian with $\|\bm{x}^\top{\bU} \|_{{\psi}_2} \leq \sigma_{u}\|\bm{x}\|_2$ for any $\bm{x}\in \mathbb{R}^{d_1+1}$;
	
	d) $0 < \kappa_{l} \leq \lambda_{\min}(E[\bar{\bV}\bar{\bV}^{\top}]) \leq \lambda_{\max}(E[\bar{\bV}\bar{\bV}^{\top}]) \leq 2\sigma_{u}^2<\infty$ and $\bar{\bV}$ is sub-Gaussian with $\|\bm{x}^\top\bar{\bV} \|_{{\psi}_2} \leq 2\sigma_{u}\|\bm{x}\|_2$ for any $\bm{x}\in \mathbb{R}^{d_1+1}$;
	
	e) $0 < \kappa_{l} \leq \lambda_{\min}(E[{\bV}{\bV}^{\top}]) \leq \lambda_{\max}(E[{\bV}{\bV}^{\top}]) \leq 2\sigma_{u}^2<\infty$ and ${\bV}$ is sub-Gaussian with $\|\bm{x}^\top{\bV} \|_{{\psi}_2} \leq 2\sigma_{u}\|\bm{x}\|_2$ for any $\bm{x}\in \mathbb{R}^{d_1+1}$.
\end{lemma}

\subsubsection{The second-time conditional mean model }
The following lemma characterizes the estimation error of the second-time conditional mean model, $\nu_a^*(\cdot)$. The corresponding conditional mean estimator is defined as $\widehat{\nu}_{a}(\bS)=\bU^\top \balphahat_a $.
\begin{lemma}\label{cor_mu2}
Let Assumptions 1-3 hold. For any $t>0$, let $ {\lambda}_{\bm{\alpha}}:=32\sigma\sigma_{u}\sigma_{\zeta}(t+\sqrt{\log(d+1)/ |\mathcal J|})$. Let
$|\mathcal J|\geq\max\{\log(d+1),100 \kappa_2^2 s_{\bm{\alpha}_{a}}\log(d+1)\}.$
Then $\balphahat_a$, (2.6), satisfies
\begin{align}
	&\|\widehat{\bm{\alpha}}_a-\bm{\alpha}_a^*\|_2~\leq~8\kappa_1^{-1} {\lambda}_{\bm{\alpha}}\sqrt{s_{\bm{\alpha}_a}},\quad
	\frac{1}{|\mathcal J|} \sum_{i\in\mathcal J}[\widetilde{\bU}_i^{\top}(\widehat{\bm{\alpha}}_a-\bm{\alpha}_a^*)]^2~\leq~32\kappa_1^{-1} {\lambda}_{\bm{\alpha}}^2s_{\bm{\alpha}_a},\label{eq:insample'}
\end{align}
with probability at least $1-2\exp\left(-\frac{4|\mathcal J|t^2}{1+2t+\sqrt{2t}}\right)-c_1\exp(-c_2 |\mathcal J|)$ and some constants $c_1,c_2,\kappa_1,\kappa_2>0$. In addition, assume $|\mathcal J|\asymp N$ and $N \gg s_{\bm{\alpha}_{a}}\log(d)$.
Choose some $ {\lambda}_{\bm{\alpha}} \asymp \sigma\sqrt{\log(d)/N}$. Then for any constant $r\geq1$, as $N,d\to \infty$, we have
\begin{align}
\|\widehat{\bm{\alpha}}_a-\bm{\alpha}_a^*\|_2&=O_p\left(\sigma\sqrt{s_{\bm{\alpha}_{a}}\log(d)/N}\right),\label{Op:alphahat}\\
\{E[\widehat{\nu}_{a}(\bS)-\nu_{a}^{*}(\bS)]^r\}^{1/r}&=O_p\left(\sigma\sqrt{s_{\bm{\alpha}_{a}}\log(d)/N}\right).\label{Op:nuhat}
\end{align}
\end{lemma} 
In the above, the expectation of the left-hand side of \eqref{Op:nuhat} is taken respect to the distribution of a new observation's covariate vector $\bS$.

The nested-regression-based estimator $\bbetahat_{a,{\mbox{\tiny NR}}}$ proposed in Section 2.2 is constructed based on $\balphahat_a$ and hence we need to first control the estimation error of $\balphahat_a$. Note that, $\balphahat_a$ and $\bbetahat_{a,{\mbox{\tiny NR}}}$ are actually obtained based on overlapping but different sample groups. For $\balphahat_a$, we only utilize the samples satisfying $A_{1i}=a_1$ and $A_{2i}=a_2$; as for $\bbetahat_{a,{\mbox{\tiny NR}}}$ , we are using the samples such that $A_{1i}=a_1$ and there is no constraint on $A_{2i}$. As a result, the in-sample error \eqref{eq:insample'} is not enough for our analysis. Instead, we require an upper bound for a ``partially in-sample'' error. We show the prerequisite results in the following lemma.
\begin{lemma}\label{lemma:insample}
	Let Assumptions of Lemma \ref{cor_mu2} hold. In addition, let 
	$|\mathcal J|\geq\max\{\log(d+1),(c_3+100 \kappa_2^2) s_{\bm{\alpha}_{a}}\log(d+1)\},$
	with constant $c_3>0$. Then 
	$$\frac{1}{|\mathcal J|} \sum_{i\in\mathcal J}[\bar\bU_i^\top (\balphahat_a-\balpha_a^*)]^2\leq288\sigma_u\kappa_1^{-2}\lambda_{\balpha}^2s_{\balpha_a},$$
	with probability at least $1-2\exp\left(-\frac{4|\mathcal J|t^2}{1+2t+\sqrt{2t}}\right)-c_1\exp(-c_2|\mathcal J|)-2\exp(-c_4|\mathcal J|)$ and constants $c_1,c_2,c_4>0$.
\end{lemma}

\subsubsection{The propensity score models}
The following lemma provides asymptotic upper bounds on the estimation errors of the propensity score models, $\pi_{a}^*(\cdot)$ and $\rho_a^*(\cdot)$. The corresponding propensity score estimators are defined as $\pihat_{a}(\bS_1)=g(\bV^\top\bgammahat_a)$ and $\rhohat_{a}(\bS)=g(\bU^\top\bdeltahat_a)$, respectively.
\begin{lemma}\label{lem_pi1}
	Let the overlap conditions of Assumption 1 and Assumptions 3-4 hold. Let the sample size be such that $|\mathcal J|\asymp N$ and $N\gg \max \{ s_{\bm{\gamma}_{a}}\log(d_1), s_{\bm{\delta}_{a}}\log(d)\}$. Then, as $N,d\to \infty$, (a) the logistic Lasso $(2.3)$ with $\lambda_{\bm{\gamma}}\asymp \sqrt{\log(d_1)/N}$
	satisfies
	\begin{align}\label{3.36}
		\|\widehat{\bm{\gamma}}_{a}-\bm{\gamma}_{a}^{*}\|_2
		&= O_p\left(\sqrt{s_{\bm{\gamma}_{a}}\log(d_1)/N}\right),\\
		E[\widehat{\pi}_{a}(\bS_1)-\pi_{a}^{*}(\bS_1)]^2
		&=O_p\left(s_{\bm{\gamma}_{a}}\log(d_1)/N\right),\label{3.31}
	\end{align}
	whereas (b) the logistic Lasso $(2.4)$ with $\lambda_{\bm{\delta}}\asymp \sqrt{\log(d)/N}$
	satisfies 
	\begin{align}\label{3.38}
		\|\widehat{\bm{\delta}}_{a}-\bm{\delta}_{a}^{*}\|_2
		&= O_p\left(\sqrt{s_{\bm{\delta}_{a}}\log(d)/N}\right),\\
		E[\widehat{\rho}_{a}(\bS)-\rho_{a}^{*}(\bS)]^2
		&=O_p\left(s_{\bm{\delta}_{a}}\log(d)/N\right).\label{3.32}
	\end{align}
	In the left-hand side of \eqref{3.31} and \eqref{3.32}, the expectations are only taken w.r.t. the distribution of the new observations $\bS_1$ and $\bS$, respectively. Note that Assumption 4 holds under Assumption 1 when $\pi_{a}^{*}(\bS_1)$ and $\rho_{a}^{*}(\bS)$ are correctly specified.
\end{lemma}

\begin{lemma}\label{lem_conclusion}
	Let Assumptions of Lemma \ref{lem_pi1} hold.
	Define the event
	$\mathcal{A}:=\{\|\widehat{\bm{\gamma}}_{a}-\bm{\gamma}_{a}^*\|_2\leq1,\|\widehat{\bdelta}_a-\bdelta_a^*\|_2\leq1\}.$
	Then, as $N,d \to \infty$, 
	$
	P(\mathcal{A})=1-o(1).
	$
	Moreover, on the event $\mathcal{A}$, as $N,d \to \infty$, $\{E|\widehat{\pi}_{a}(\bS_{1})|^{-r}\}^{1/r}$ and $\{E|\widehat{\rho}_{a}(\bS)|^{-r}\}^{1/r}$ are both bounded uniformly by some constants independent of $N$ and for $r >2$,
	\begin{align} 
		\left\{E\left| {\widehat{\pi}_{a}^{-1}(\bS_{1})}- {{\pi_{a}^{*}}^{-1}(\bS_{1})}\right|^{r}\right\}^{{1}/{r}}&=O_p\left(\sqrt{s_{\bm{\gamma}_{a}}\log(d_1)/N}\right), \nonumber
		\\
		\left\{E\left| {\widehat{\rho}^{-1}_{a}(\bS)}- {{\rho_{a}^{*}}^{-1}(\bS)}\right|^{r}\right\}^{{1}/{r}}&=O_p\left(\sqrt{s_{\bm{\delta}_{a}}\log(d)/N}\right), \nonumber
	\end{align}
	$$
	\left\{E\left| {\widehat{\pi}_{a}^{-1}(\bS_{1})\widehat{\rho}_{a}^{-1}(\bS)}- {{\pi_{a}^{*}}^{-1}(\bS_{1}){\rho_{a}^{*}}^{-1}(\bS)}\right|^{r}\right\}^{{1}/{r}}=O_p\left(\sqrt{(s_{\bm{\gamma}_{a}}\log(d_1)+s_{\bm{\delta}_{a}}\log(d))/N}\right).\label{rate:pihatrhohat}
	$$
\end{lemma} 
In the left-hand side of the equations above, the expectations are only taken w.r.t. the distribution of the new observations $\bS_1$ or $\bS$.

\subsection{Convergence rate for the general imputed Lasso estimator}
\begin{proof}[Proof of Theorem 8]
By the definition of $\widehat{\bm{\beta}}$, we have
$$\frac{1}{M} \sum_{i=1}^{M}[\Yhat_i-\bX_i^\top\bbetahat]^2+\lambda_M\|\bbetahat\|_1\leq\frac{1}{M} \sum_{i=1}^{M}[\Yhat_i-\bX_i^\top\bbeta^*]^2+\lambda_M\|\bbeta^*\|_1,$$
or, expanding and rearranging, 
\begin{align}
	&\frac{1}{M} \sum_{i=1}^{M}[\bX_i^{\top}(\bbetahat-\bbeta^*)]^2+\lambda_M\|\bbetahat\|_1\nonumber\\
	&\qquad\leq \frac{2}{M} \sum_{i=1}^{M}[\Yhat_i-\bX_i^\top\bbeta^*]\bX_i^{\top}(\bbetahat-\bbeta^*)+\lambda_M\|\bbeta^*\|_1\nonumber\\
	&\qquad=\frac{2}{M} \sum_{i=1}^{M}\varepsilon_{i}\bX_i^{\top}(\bbetahat-\bbeta^*)+\frac{2}{M} \sum_{i=1}^{M}[\Yhat_i-Y_i^*]\bX_i^{\top}(\bbetahat-\bbeta^*)+\lambda_M\|\bbeta^*\|_1.\label{144'}
\end{align}
For any $t>0$, let $\lambda_M:=16\sigma\sigma_\bX(\sqrt\frac{\log(d)}{M}+t)$. Define the event 
\begin{align}
	\mathcal E_2:=\left\{\max_{1\leq j \leq d} \left|\frac{1}{M} \sum_{i=1}^{M}\bX_{i,j}\varepsilon_{i}\right|\leq\frac{\lambda_M}{4}\right\},\nonumber
\end{align}
where $\bX_{i,j}$ represents the $j$-th component of $\bX_{i}$. Note that
\begin{align}
	&P\left(\max_{1\leq j \leq d}\left|\frac{1}{M} \sum_{i=1}^{M}\bX_{i,j}\varepsilon_{i}\right|\geq\frac{\lambda_M}{4}\right)=P\left(\bigcup_{j=1}^{d} \left\{\left|\frac{1}{M} \sum_{i=1}^{M}\bX_{i,j}\varepsilon_{i}\right|\geq\frac{\lambda_M}{4}\right\}\right)\nonumber\\
	&\qquad\leq \sum_{j=1}^{d}P\left(\left|\frac{1}{M} \sum_{i=1}^{M}\bX_{i,j}\varepsilon_{i}\right|\geq\frac{\lambda_M}{4}\right).\label{151'}
\end{align}
Let $\boldsymbol e_j\in \mathbb{R}^d$ be the vector whose $j$-th element is $1$ and other elements are $0$s, for each $1\leq j\leq d$. Since $\|\boldsymbol e_j^\top\bX\|_{\psi_2}\leq\sigma_\bX$ and $\|\varepsilon\|_{\psi_2}\leq\sigma$, by Lemma \ref{lemma:psi2norm},
$$\|\boldsymbol e_j^\top\bX\varepsilon\|_{\psi_1}\leq\|\boldsymbol e_j^\top\bX\|_{\psi_2}\cdot\|\varepsilon\|_{\psi_2}\leq \sigma\sigma_\bX.$$
Note that, here we do not make any assumption on the sample gram matrix $\hat{\bm{\Sigma}}:=M^{-1}\sum_{i=1}^M\bX_i\bX_i^\top$, e.g., $\sup_{1\leq j\leq d}\hat{\bm{\Sigma}}_{j,j}\leq1$ as required in \cite{wainwright2019high,negahban2012unified}. Instead, we consider $\boldsymbol e_j^\top\bX\varepsilon$ as a sub-exponential random variable, and the Bernstein's inequality is applied in the following to control \eqref{151'}. 
Recall the definition of $\bbeta^*$, we have $E[\bX\varepsilon]=0$. By Bernstein's inequality, for each $1\leq j\leq d$,
\begin{align}
	P\left(\left|\frac{1}{M} \sum_{i=1}^{M}\bX_{i,j}\varepsilon_{i}\right|\geq2\sigma\sigma_\bX\epsilon+\sigma\sigma_\bX\epsilon^2\right) \leq 2 \exp \left(-M \epsilon^2\right),\quad\text{for any}\;\;\epsilon>0.\label{bound:Xepsj}
\end{align}
Set $\epsilon=\sqrt\frac{\log(d)}{M}+\frac{\sqrt{1+8t}-1}{2}$ for any $t>0$. When $M>\log(d)$, we have 
\begin{align*}
	2\epsilon+\epsilon^2&\leq2\sqrt\frac{\log(d)}{M}+\sqrt{1+8t}-1+\left(\sqrt\frac{\log(d)}{M}+\frac{\sqrt{1+8t}-1}{2}\right)^2\\
	&\leq2\sqrt\frac{\log(d)}{M}+\sqrt{1+8t}-1+\frac{2\log(d)}{M}+2\left(\frac{\sqrt{1+8t}-1}{2}\right)^2\\
	&=2\sqrt\frac{\log(d)}{M}+\sqrt{1+8t}-1+2\sqrt\frac{\log(d)}{M}\cdot\sqrt\frac{\log(d)}{M}+1+4t-\sqrt{1+8t}\\
	&\leq4\sqrt\frac{\log(d)}{M}+4t,
\end{align*}
and hence 
\begin{align}
	2\sigma\sigma_\bX\epsilon+\sigma\sigma_\bX\epsilon^2\leq4\sigma\sigma_\bX\left(\sqrt\frac{\log(d)}{M}+t\right)=\frac{\lambda_M}{4}.\label{bound:eps_1+2}
\end{align}
Additionally, we also have 
\begin{align*}
	\epsilon^2&=\left(\sqrt\frac{\log(d)}{M}+\frac{\sqrt{1+8t}-1}{2}\right)^2\geq\frac{\log(d)}{M}+\frac{1+4t-\sqrt{1+8t}}{2}\\
	&=\frac{\log(d)}{M}+\frac{8t^2}{1+4t+\sqrt{1+8t}}\geq\frac{\log(d)}{M}+\frac{4t^2}{1+2t+\sqrt{2t}}.
\end{align*}
Together with \eqref{bound:Xepsj} and \eqref{bound:eps_1+2}, we have, for each $1\leq j\leq d$,
\begin{align*}
	&P\left(\left|\frac{1}{M} \sum_{i=1}^{M}\bX_{i,j}\varepsilon_{i}\right|\geq\frac{\lambda_M}{4}\right)\leq P\left(\left|\frac{1}{M} \sum_{i=1}^{M}\bX_{i,j}\varepsilon_{i}\right|\geq2\sigma\sigma_\bX\epsilon+\sigma\sigma_\bX\epsilon^2\right)\\
	&\qquad\leq 2 \exp \left(-M \epsilon^2\right)\leq\frac{2}{d}\exp\left(-\frac{4Mt^2}{1+2t+\sqrt{2t}}\right).
\end{align*}
Together with \eqref{151'},
\begin{align}\label{bound:E2_1}
	P(\mathcal E_2)=P\left(\max_{1\leq j \leq d}\left|\frac{1}{M} \sum_{i=1}^{M}\bX_{i,j}\varepsilon_{i}\right|\leq \frac{\lambda_M}{4}\right)\geq1-2\exp\left(-\frac{4Mt^2}{1+2t+\sqrt{2t}}\right).
\end{align}
On the event $\mathcal E_2$, we have 
\begin{align}
	\left|\frac{2}{M}\sum_{i=1}^M\varepsilon_i\bX_i^\top(\bbetahat-\bbeta^*)\right|\leq2\|\bbetahat-\bbeta^*\|_1\max_{1\leq j \leq d}\left|\frac{1}{M} \sum_{i=1}^{M}\bX_{i,j}\varepsilon_{i}\right|\leq\lambda_M\|\bbetahat-\bbeta^*\|_1/2.\label{bound:epsXDbeta}
\end{align}
As for the second term of \eqref{144'}, by the fact that $2ab\leq a^2+b^2$ for any $a,b\in\R$, and we set $a=\sqrt2[\Yhat_i-Y_i^*]$, $b=\bX_i^{\top}(\bbetahat-\bbeta^*)/\sqrt2$, we have 
\begin{align}
	&\left|\frac{2}{M} \sum_{i=1}^{M}[\Yhat_i-Y_i^*]\bX_i^{\top}(\bbetahat-\bbeta^*)\right|\leq\frac{2}{M} \sum_{i=1}^{M}[\Yhat_i-Y_i^*]^2+\frac{1}{2M} \sum_{i=1}^{M}\left[\bX_i^{\top}(\bbetahat-\bbeta^*)\right]^2\nonumber\\
	&\qquad\leq2\delta_M^2+\frac{1}{2M} \sum_{i=1}^{M}\left[\bX_i^{\top}(\bbetahat-\bbeta^*)\right]^2,\label{bound:DYXDbeta}
\end{align}
on the event $\mathcal E_1=\{M^{-1}\sum_{i=1}^{M}[\Yhat_i-Y_i^*]^2<\delta_M^2\}$. Multiplying the left-hand side and right-hand side of \eqref{144'} by 2, we have
\begin{align*}
	&\frac{2}{M} \sum_{i=1}^{M}[\bX_i^{\top}(\bbetahat-\bbeta^*)]^2+2\lambda_M\|\bbetahat\|_1\\
	&\qquad\leq\frac{4}{M} \sum_{i=1}^{M}\varepsilon_{i}\bX_i^{\top}(\bbetahat-\bbeta^*)+\frac{4}{M} \sum_{i=1}^{M}[\Yhat_i-Y_i^*]\bX_i^{\top}(\bbetahat-\bbeta^*)+2\lambda_M\|\bbeta^*\|_1.
\end{align*}
Together with \eqref{bound:epsXDbeta} and \eqref{bound:DYXDbeta}, on the event $\mathcal E_1\cap\mathcal E_2$, we have
\begin{align}
	&\frac{2}{M} \sum_{i=1}^{M}[\bX_i^{\top}(\bbetahat-\bbeta^*)]^2+2\lambda_M\|\bbetahat\|_1\nonumber\\
	&\qquad\leq\lambda_M\|\bbetahat-\bbeta^*\|_1+\frac{1}{M} \sum_{i=1}^{M}[\bX_i^{\top}(\bbetahat-\bbeta^*)]^2+4\delta_M^2+2\lambda_M\|\bbeta^*\|_1.\label{bound:to-be-used}
\end{align}
Hence, 
\begin{align}
	&\frac{1}{M} \sum_{i=1}^{M}[\bX_i^{\top}(\bbetahat-\bbeta^*)]^2+2\lambda_M\|\bbetahat\|_1\leq\lambda_M\|\bbetahat-\bbeta^*\|_1+2\lambda_M\|\bbeta^*\|_1+4\delta_M^2\nonumber\\
	&\qquad=\lambda_M\|\bbetahat_S-\bbeta_S^*\|_1+\lambda_M\|\bbetahat_{S^c}\|_1+2\lambda_M\|\bbeta_S^*\|_1+4\delta_M^2,\label{163'}
\end{align}
where $S:=\{j\leq d:\bbeta_j^*\neq0\}$ and note that $s=|S|$, $\|\bbetahat-\bbeta^*\|_1=\|\bbetahat_S-\bbeta_S^*\|_1+\|\bbetahat_{S^c}-\bbeta_{S^c}^*\|_1=\|\bbetahat_S-\bbeta_S^*\|_1+\|\bbetahat_{S^c}\|_1$, and $\|\bbeta^*\|_1=\|\bbeta_S^*\|_1$. By the triangle inequality,
\begin{align}\label{164'}
	\|\bbetahat\|_1=\|\bbetahat_S\|_1+\|\bbetahat_{S^c}\|_1\geq\|\bbeta_S^*\|_1- \|\bbetahat_S-\bbeta_S^*\|_1+\|\bbetahat_{S^c}\|_1.
\end{align}
By \eqref{163'} and \eqref{164'}, on the event $\mathcal E_1\cap\mathcal E_2$, we get that 
\begin{align}\label{167'}
	\frac{1}{M} \sum_{i=1}^{M}[\bX_i^{\top}(\bbetahat-\bbeta^*)]^2+\lambda_M\|\bbetahat_{S^c}\|_1
	&\leq 3\lambda_M\|\bbetahat_S-\bbeta_S^*\|_1+4\delta_M^2.
\end{align}
By Lemma 4.5 of \cite{zhang2021double}, 
there exist constants $\kappa_1,\kappa_2>0$, such that
\begin{align}\label{73'}
	\frac{1}{M} \sum_{i=1}^{M}(\bX_i^{\top}\bDelta)^2\geq\kappa_1\|\bDelta\|_2\left\{\|\bDelta\|_2-\kappa_2\sqrt{\frac{\log(d)}{M}}\|\bDelta\|_1 \right\} \quad \text{for all}\;\;\|\bDelta\|_2\leq 1,
\end{align}
with probability at least $1-c_1\exp(-c_2M)$ and some constants $c_1,c_2>0$. Note that Lemma 4.5 of \cite{zhang2021double} discusses logistic loss but applies more broadly and does include the least squares loss as well. 

Let $\bdelta=\bbetahat-\bbeta^*$ and define
\begin{equation}\label{ineq:delta}
	\mathcal E_3:=\left\{\frac{1}{M} \sum_{i=1}^{M}(\bX_i^{\top}\bdelta)^2\geq\kappa_1\|\bdelta\|_2^2-\kappa_1\kappa_2\sqrt{\frac{\log(d)}{M}}\|\bdelta\|_1\|\bdelta\|_2\right\}.
\end{equation}
Let $\bDelta=\bdelta/\|\bdelta\|_2$. Then, $\|\bDelta\|_2=1$ and hence by \eqref{73'}, 
$$P(\mathcal E_3)\geq1-c_1\exp(-c_2M).$$
We now condition on the event $\mathcal E_1\cap\mathcal E_2\cap\mathcal E_3$ and introduce two cases need to be separately analyzed.

\textbf{Case 1.} Case of $\|\bdelta_S\|_1<4\lambda_M^{-1}\delta_M^2$. Then, by \eqref{167'},
$$\|\bdelta_{S^c}\|_1\leq3\|\bdelta_S\|_1+4\lambda_M^{-1}\delta_M^2\leq16\lambda_M^{-1}\delta_M^2.$$
Hence,
$$\|\bdelta\|_1=\|\bdelta_S\|_1+\|\bdelta_{S^c}\|_1\leq20\lambda_M^{-1}\delta_M^2,$$
and
$$\frac{1}{M} \sum_{i=1}^{M}(\bX_i^{\top}\bdelta)^2\leq3\lambda_M\|\bdelta_S\|_1+4\delta_M^2\leq16\delta_M^2.$$
In addition, on the event $\mathcal E_3$,
$$\kappa_1\|\bdelta\|_2^2-\kappa_1\kappa_2\sqrt{\frac{\log(d)}{M}}\|\bdelta\|_1\|\bdelta\|_2\leq\frac{1}{M} \sum_{i=1}^{M}(\bX_i^{\top}\bdelta)^2\leq16\delta_M^2.$$
It follows that,
\begin{align*}
	\|\bdelta\|_2&\leq\frac{\kappa_1\kappa_2\sqrt{\frac{\log(d)}{M}}\|\bdelta\|_1+\sqrt{\kappa_1^2\kappa_2^2\frac{\log(d)}{M}\|\bdelta\|_1^2+64\kappa_1\delta_M^2}}{2\kappa_1}\\
	&\leq\kappa_2\sqrt{\frac{\log(d)}{M}}\|\bdelta\|_1+4\kappa_1^{-1/2}\delta_M\leq20\kappa_2\sqrt{\frac{\log(d)}{M}}\lambda_M^{-1}\delta_M^2+4\kappa_1^{-1/2}\delta_M\\
	&\leq\frac{5\kappa_2\delta_M^2}{4\sigma\sigma_{\bX}}+4\kappa_1^{-1/2}\delta_M,
\end{align*}
since $\lambda_M=16\sigma\sigma_\bX(\sqrt\frac{\log(d)}{M}+t)\geq16\sigma\sigma_\bX\sqrt\frac{\log(d)}{M}$.

\textbf{Case 2.} Case of $\|\bdelta_S\|_1\geq4\lambda_M^{-1}\delta_M^2$. Then, by \eqref{167'},
\begin{align}\label{eq:basic'}
	\frac{1}{M} \sum_{i=1}^{M}(\bX_i^{\top}\bdelta)^2+\lambda_M\|\bdelta_{S^c}\|_1\leq\lambda_M(3\|\bdelta_S\|_1+4\lambda_M^{-1}\delta_M^2)\leq4\lambda_M\|\bdelta_S\|_1,
\end{align}
and hence
\begin{align}\label{eq:cone}
	\|\bdelta_{S^c}\|_1\leq 4\|\bdelta_S\|_1.
\end{align}
Notice that, $\|\bdelta_S\|_1\leq\sqrt s\|\bdelta_S\|_2$. It follows that
$$\|\bdelta\|_1=\|\bdelta_S\|_1+\|\bdelta_{S^c}\|_1\leq5\|\bdelta_S\|_1\leq5\sqrt s\|\bdelta_S\|_2\leq5\sqrt s\|\bdelta\|_2.$$
Hence, under the event $\mathcal E_3$, when $M>100{\kappa_2}^2s\log(d)$,
\begin{align}
	\frac{1}{M} \sum_{i=1}^{M}(\bX_i^{\top}\bdelta)^2&\geq\kappa_1\|\bdelta\|_2^2-5\kappa_1\kappa_2\sqrt{\frac{s\log(d)}{M}}\|\bdelta\|_2^2\nonumber\\
	&\geq\frac{\kappa_1}{2}\|\bdelta\|_2^2\geq\frac{\kappa_1}{2}\|\bdelta_S\|_2^2\geq\frac{\kappa_1}{2s}\|\bdelta_S\|_1^2.\label{eq:deltaS2norm}
\end{align}
Together with \eqref{eq:basic'}, we have
$$\frac{\kappa_1}{2s}\|\bdelta_S\|_1^2\leq\frac{1}{M} \sum_{i=1}^{M}(\bX_i^{\top}\bdelta)^2\leq4\lambda_M\|\bdelta_S\|_1.$$ 
Hence, on the event $\mathcal E_1\cap\mathcal E_2\cap\mathcal E_3$,
\begin{align}\label{bound:deltaS}
	\|\bdelta_S\|_1\leq8\kappa_1^{-1}s\lambda_M.
\end{align}
By \eqref{eq:cone},
$$\|\bdelta\|_1\leq\|\bdelta_S\|_1+\|\bdelta_{S^c}\|_1\leq5\|\bdelta_S\|_1\leq40\kappa_1^{-1}s\lambda_M.$$
Besides, by \eqref{eq:basic'} and \eqref{bound:deltaS},
$$\frac{1}{M} \sum_{i=1}^{M}(\bX_i^{\top}\bdelta)^2\leq4\lambda_M\|\bdelta_S\|_1\leq32\kappa_1^{-1}s\lambda_M^2.$$
Additionally, by \eqref{eq:deltaS2norm}, when $M>100{\kappa_2}^2s\log(d)$,
$$\|\bdelta\|_2\leq\sqrt{\frac{2}{\kappa_1M} \sum_{i=1}^{M}(\bX_i^{\top}\bdelta)^2}\leq8\kappa_1^{-1}\sqrt s\lambda_M.$$
To sum up, on the event $\mathcal E_1\cap\mathcal E_2\cap\mathcal E_3$ and when $M>\max\{\log(d),100{\kappa_2}^2s\log(d)\}$, 
\begin{align}
	\|\bbetahat-\bbeta^*\|_2&\leq\max\left(\frac{5\kappa_2\delta_M^2}{4\sigma\sigma_{\bX}}+4\kappa_1^{-1/2}\delta_M,8\kappa_1^{-1}\sqrt s\lambda_M\right),\label{bound:finite1}\\
	\|\bbetahat-\bbeta^*\|_1&\leq\max\left(20\lambda_M^{-1}\delta_M^2,40\kappa_1^{-1}s\lambda_M\right),\label{bound:finite2}\\
	\frac{1}{M} \sum_{i=1}^{M}(\bX_i^{\top}\bdelta)^2&\leq\max\left(16\delta_M^2,32\kappa_1^{-1}s\lambda_M^2\right).\label{bound:finite3}
\end{align}
Here,
$$P(\mathcal E_2\cap\mathcal E_3)\geq 1-P(\mathcal E_2^c)-P(\mathcal E_3^c)=1-2\exp\left(-\frac{4Mt^2}{1+2t+\sqrt{2t}}\right)-c_1\exp(-c_2M).$$ 
The remaining claims follow by noticing that for some $\lambda_M\asymp\sigma\sqrt\frac{\log(d)}{M}$ and
$\delta_M=o(\sigma)$, $P(\mathcal E_1)=1-o(1)$, and with $M\gg s\log(d)$ as $M,d\to\infty$,
$$P(\mathcal E_1\cap\mathcal E_2\cap\mathcal E_3)\geq1-2\exp\left(-\frac{4Mt^2}{1+2t+\sqrt{2t}}\right)-c_1\exp(-c_2M)-o(1).$$
\end{proof}

\subsection{Convergence rates for nuisance estimators with imputed outcomes}
\subsubsection{DR-imputation-based estimator}
\begin{proof}[Proof of Theorem 9]
We first consider the DR-imputation-based estimator $\bbetahat_{a,2}=\bbetahat_a(\mathcal D_{\mathcal J_2},\widehat Y^{\mbox{\tiny DR}}_{\mathcal J_2})$ defined as (2.11).
In this case, the expectations are taken w.r.t. the samples in $\mathcal D_{\mathcal J_2}$; with a slight abuse of notation, $\bdeltatil_a:=\widehat{\bdelta}_{a}(\mathcal D_{\mathcal J_1})$ and $\widetilde{\balpha}_a:=\balphahat_a(\mathcal D_{\mathcal J_1})$ are fitted using samples in $\mathcal D_{\mathcal J_1}$ and are treated as fixed or condition on. Repeat the same procedure as in \eqref{144'}, we have
\begin{align*}
	&\frac{1}{|\mathcal J_2|}\sum_{i\in \mathcal J_2}(\bar\bV_i^\top\bDelta_{\bbeta})^2+\lambda_{\bbeta}\|\bbetahat_{a,2}\|_1\leq\frac{2}{|\mathcal J_2|}\sum_{i\in \mathcal J_2}\mathbbm1_{\{A_{1i}=a_1\}}(\Yhat_i^{\mbox{\tiny DR}}-\bV_i^\top\bbeta_a^*)\bar\bV_i^\top\bDelta_{\bbeta}+\lambda_{\bbeta}\|\bbeta_a^*\|_1\\
	&\;\;=\frac{2}{|\mathcal J_2|}\sum_{i\in \mathcal J_2}(\Delta_{1i}+\Delta_{2i}+\Delta_{3i}+\Delta_{4i}+\Delta_{5i}+\Delta_{6i})\bar\bV_i^\top\bDelta_{\bbeta}+\lambda_{\bbeta}\|\bbeta_a^*\|_1\\
	&\;\;\leq\sum_{l=1}^3\left\|\frac{2}{|\mathcal J_2|}\sum_{i\in \mathcal J_2}\Delta_{li}\bar\bV_i\right\|_\infty\|\bDelta_{\bbeta}\|_1+\frac{2}{|\mathcal J_2|}\sum_{i\in \mathcal J_2}\left(\sum_{l=4}^6\Delta_{li}\right)^2\\
	&\qquad+\frac{1}{2|\mathcal J_2|}\sum_{i\in \mathcal J_2}(\bar\bV_i^\top\bDelta_{\bbeta})^2+\lambda_{\bbeta}\|\bbeta_a^*\|_1,
\end{align*}
where $\bDelta_{\bbeta}:=\bbetahat_{a,2}-\bbeta_a^*$,
\begin{align*}
	\Delta_{1i}:=&\bar\bU_i^\top\balpha_a^*+\frac{\Ytil_i-\widetilde{\bU}_i^\top\balpha_a^*}{g(\bU_i^\top\bdelta_a^*)}-\bar\bV_i^\top\bbeta_a^*,\\
	\Delta_{2i}:=&\left\{1-\frac{\mathbbm1_{\{A_{2i}=a_2\}}}{g(\bU_i^\top\bdelta_a^*)}\right\}\bar\bU_i^\top(\widetilde{\balpha}_a-\balpha_a^*)\mathbbm{1}_{\{\rho_{a}^{*}(\cdot)=\rho_{a}(\cdot)\}},\\
	\Delta_{3i}:=&\left\{\frac{1}{g(\bU_i^\top\bdeltatil_a)}-\frac{1}{g(\bU_i^\top\bdelta_a^*)}\right\} (\Ytil_i-\widetilde{\bU}_i^\top\balpha_a^*)\mathbbm{1}_{\{\nu_{a}^{*}(\cdot)=\nu_{a}(\cdot)\}},\\
	\Delta_{4i}:=&-\left\{\frac{1}{g(\bU_i^\top\bdeltatil_a)}-\frac{1}{g(\bU_i^\top\bdelta_a^*)}\right\}\widetilde{\bU}_i^\top(\widetilde{\balpha}_a-\balpha_a^*),\\
	\Delta_{5i}:=&\left\{1-\frac{\mathbbm1_{\{A_{2i}=a_2\}}}{g(\bU_i^\top\bdelta_a^*)}\right\}\bar\bU_i^\top(\widetilde{\balpha}_a-\balpha_a^*)\mathbbm{1}_{\{\rho_{a}^{*}(\cdot)\neq\rho_{a}(\cdot)\}},\\
	\Delta_{6i}:=&\left\{\frac{1}{g(\bU_i^\top\bdeltatil_a)}-\frac{1}{g(\bU_i^\top\bdelta_a^*)}\right\} (\Ytil_i-\widetilde{\bU}_i^\top\balpha_a^*)\mathbbm{1}_{\{\nu_{a}^{*}(\cdot)\neq\nu_{a}(\cdot)\}},
\end{align*}
and $g(u)=\exp(u)/\{1+\exp(u)\}$ is the logistic function. Let $\Delta_l$ be an independent copy of $\Delta_{li}$ for $1\leq l\leq 6$. We first show that $\Delta_{l}\bar\bV$ are zero mean vectors for each $l\in\{1,2,3\}$. By the definition of $\bbeta_a^*$, we have $E[\Delta_{1}\bar\bV]=\bzero$. By the tower rule, we have
\begin{align*}
&E[\Delta_{2}\bar\bV]\\
&\;\;=E\left[P(A_{1}=a_1\mid\bU)E\left[1-\frac{\mathbbm1_{\{A_{2}=a_2\}}}{\rho_a^*(\bS)}\mid\bU,A_{1}=a_1\right]\bU^\top(\widetilde{\balpha}_a-\balpha_a^*)\bV\mathbbm{1}_{\{\rho_{a}^{*}(\cdot)=\rho_{a}(\cdot)\}}\right]\\
&\;\;=E\left[P(A_{1}=a_1\mid\bU)\left\{1-\frac{\rho_a(\bS)}{\rho_a^*(\bS)}\right\}\bU^\top(\widetilde{\balpha}_a-\balpha_a^*)\bV\mathbbm{1}_{\{\rho_{a}^{*}(\cdot)=\rho_{a}(\cdot)\}}\right]=\bzero.
\end{align*}
Similarly, we also have
\begin{align*}
&E[\Delta_{3}\bar\bV]\\
&\;\;=E\biggl[ P(A_{1}=a_1\mid\bU)\left\{\frac{1}{g(\bU^\top\bdeltatil_a)}-\frac{1}{g(\bU^\top\bdelta_a^*)}\right\}\\
&\qquad\qquad\cdot E[(Y(a)-\bU^\top\balpha_a^*)\mathbbm1_{\{A_{2}=a_2\}}\mid\bU,A_{1}=a_1]\bV\mathbbm{1}_{\{\nu_{a}^{*}(\cdot)=\nu_{a}(\cdot)\}}\biggl]\\
&=E\biggl[P(A_{1}=a_1\mid\bU)\left\{\frac{1}{g(\bU^\top\bdeltatil_a)}-\frac{1}{g(\bU^\top\bdelta_a^*)}\right\}\\
&\qquad\qquad\cdot P(A_{2}=a_2\mid\bU,A_{1}=a_1)[\nu_a(\bS)-\nu_a^*(\bS)]\bV\mathbbm{1}_{\{\nu_{a}^{*}(\cdot)=\nu_{a}(\cdot)\}}\biggl]=\bzero.
\end{align*}
Let $\boldsymbol e_j\in \mathbb{R}^d$ be the vector whose $j$-th element is $1$ and other elements are $0$s, for each $1\leq j\leq d_1+1$. Under Assumption 4, we have $|\Delta_{1}\bar\bV^\top\boldsymbol e_j|=|(\varepsilon_a+g^{-1}(\bU^\top\bdelta_a^*)\zeta_a)\bar\bV^\top \boldsymbol e_j|\leq(|\varepsilon_a|+c_0^{-1}|\zeta_a|)|\bar\bV^\top \boldsymbol e_j|$ for each $1\leq j\leq d_1+1$. By Lemma \ref{lemma:psi2norm},
\begin{align*}
	\|\Delta_{1}\bar\bV^\top\boldsymbol e_j\|_{\psi_1}\leq (\|\varepsilon_a\|_{\psi_2}+c_0^{-1}\|\zeta_a\|_{\psi_2})\|\bar\bV^\top \boldsymbol e_j\|_{\psi_2}\overset{(i)}{\leq} \sigma(\sigma_{\zeta}+c_0^{-1}\sigma_{\varepsilon})\sigma_u,
\end{align*}
where (i) holds by Assumptions 2 and 3. By Lemma D.4 of \cite{chakrabortty2019high}, for each $1\leq j\leq d_1+1$ and any $t>0$,
\begin{align*}
&P\left(\left|\frac{1}{|\mathcal J_2|}\sum_{i\in \mathcal J_2}\Delta_{1i}\bar\bV_i^\top\boldsymbol e_j\right|>h(t)\right)\leq2\exp(-t-\log (d_1+1)).
\end{align*}
where $h(t)=\sigma(\sigma_{\zeta}+c_0^{-1}\sigma_{\varepsilon})\sigma_u \left(2\sqrt\frac{t+\log (d_1+1)}{|\mathcal J_2|}+\frac{t+\log (d_1+1)}{|\mathcal J_2|}\right)$. It follows that,
\begin{align*}
&P\left(\left\|\frac{1}{|\mathcal J_2|}\sum_{i\in \mathcal J_2}\Delta_{1i}\bar\bV_i\right\|_\infty>h(t)\right)\leq\sum_{j=1}^{d_1+1}P\left(\left|\frac{1}{|\mathcal J_2|}\sum_{i\in \mathcal J_2}\Delta_{1i}\bar\bV_i^\top\boldsymbol e_j\right|>h(t)\right)\\
&\qquad\leq 2(d_1+1)\exp(-t-\log (d_1+1))=2\exp(-t).
\end{align*}
Therefore, by $|\mathcal J|\asymp N$, we have
\begin{align}\label{bound:Delta1}
	\left\|\frac{1}{|\mathcal J_2|}\sum_{i\in \mathcal J_2}\Delta_{1i}\bar\bV_i\right\|_\infty=O_p\left( \sigma \sqrt\frac{\log (d_1)}{N}\right).
\end{align}
In addition, note that $|\Delta_2\bar\bV^T\be_j|\leq(1+c_0^{-1})|\bU_i^T(\widetilde{\balpha}_a-\balpha^*)\bV^T\be_j|$ under Assumption 4. By Lemma \ref{lemma:psi2norm}, conditional on $\mathcal D_{\mathcal J_1}$ , we have
\begin{align*}
	\|\Delta_2\bar\bV^T\be_j\|_{\psi_1}\leq(1+c_0^{-1})\|\bU_i^T(\widetilde{\balpha}_a-\balpha^*)\|_{\psi_2}\|\bV^T\be_j\|_{\psi_2}\overset{(i)}{\leq}(1+c_0^{-1})\|\widetilde{\balpha}_a-\balpha^*\|_2\sigma_u^2.
\end{align*}
where (i) holds by Assumption 3. By Lemma D.4 of \cite{chakrabortty2019high} and the union bound, for any $t>0$, we have
\begin{align*}
\left\|\frac{1}{|\mathcal J_2|}\sum_{i\in \mathcal J_2}\Delta_{2i}\bar\bV_i\right\|_\infty>(1+c_0^{-1})\|\widetilde{\balpha}_a-\balpha^*\|_2\sigma_u^2\left(2\sqrt\frac{t+\log (d_1+1)}{|\mathcal J_2|}+\frac{t+\log (d_1+1)}{|\mathcal J_2|}\right)
\end{align*}
with probability at most $2\exp(-t)$. Therefore, by $|\mathcal J|\asymp N$, we have
\begin{align}\label{bound:Delta2i}
	\left\|\frac{1}{|\mathcal J_2|}\sum_{i\in \mathcal J_2}\Delta_{2i}\bar\bV_i\right\|_\infty=O_p\left( \|\widetilde{\balpha}_a-\balpha^*\|_2 \sqrt\frac{\log (d_1)}{N}\right)\overset{(i)}{=}o_p\left( \sigma \sqrt\frac{\log (d_1)}{N}\right),
\end{align}
where (i) holds by Lemma \ref{cor_mu2} with $s_{\balpha_a}\log(d)=o(N)$.
Besides, by Corollary 2.3 of \cite{dumbgen2010nemirovski}, we have
\begin{align*}
	E\left[\left\|\frac{1}{|\mathcal J_2|}\sum_{i\in \mathcal J_2}\Delta_{3i}\bar\bV_i\right\|_\infty^2\right]\leq\frac{(2e\log(d_1)-e)E[\|\Delta_{3i}\bar\bV_i\|_\infty^2]}{|\mathcal J_2|}.
\end{align*}

(a) If $\nu_{a}^{*}(\cdot)=\nu_{a}(\cdot)$ and $\|\bS_1\|_\infty\leq C$ almost surely, we have $\|\bar\bV_i\|_\infty\leq\|\bV_i\|_\infty\leq \max\{1,C\}$, which implies
\begin{align*}
	E\left[\left\|\frac{1}{|\mathcal J_2|}\sum_{i\in \mathcal J_2}\Delta_{3i}\bar\bV_i\right\|_\infty^2\right]=O\left(\frac{E[\Delta_{3}^2]\log(d_1)}{N}\right),
\end{align*}
since $|\mathcal J|\asymp N$.
Under Assumption 2, by Lemma \ref{lemma:psi2norm}, we have $E[\zeta^8]\leq2^9\sigma^8\sigma_\zeta^8$. Together with Lemma \ref{lem_conclusion}, we have
\begin{align}\label{bound:Delta3}
	E[\Delta_3^2]\leq\sqrt{E[\Delta_{3}^4]}\leq\left\{E[\zeta^8]E[g^{-1}(\bU_i^\top\bdeltatil_a)-g^{-1}(\bU_i^\top\bdelta_a^*)]^8\right\}^{1/4}=O_p\left(\frac{\sigma^2s_{\bdelta_a}\log(d)}{N}\right).
\end{align}
Hence,
\begin{align}\label{bound:Delta3i'}
	E\left[\left\|\frac{1}{|\mathcal J_2|}\sum_{i\in \mathcal J_2}\Delta_{3i}\bar\bV_i\right\|_\infty^2\right]=O_p\left(\frac{\sigma^2s_{\bdelta_a}\log(d)}{N}\cdot\frac{\log(d_1)}{N}\right)=o_p\left(\frac{\sigma^2\log(d)}{N}\right),
\end{align}
since $s_{\bdelta_a}\log(d)=o(N)$.

(b) If $\nu_{a}^{*}(\cdot)=\nu_{a}(\cdot)$ and $s_{\bm{\delta}_{a}}\log(d_1)\log(d)=O(N)$, by Lemma \ref{lemma:psi2norm}, we have $\|\|\bar\bV_i\|_\infty\|_{\psi_2}\leq\sigma_u\sqrt{\log(d_1+1)+2}$, which implies that $E\|\bar\bV_i\|_\infty^4\leq8\sigma_u^4\{\log(d_1+1)+2\}^2$ through the moment bound of Lemma \ref{lemma:psi2norm}. By H\"older's inequality with $|\mathcal J|\asymp N$, 
\begin{align*}
	E\left[\left\|\frac{1}{|\mathcal J_2|}\sum_{i\in \mathcal J_2}\Delta_{3i}\bar\bV_i\right\|_\infty^2\right]=O\left(\frac{\sqrt{E[\Delta_{3}^4]E\|\bar\bV_i\|_\infty^4}\log(d_1)}{N}\right)=O\left(\frac{\sqrt{E[\Delta_{3}^4]}\log^2(d_1)}{N}\right).
\end{align*}
Hence, we also have
\begin{align*}
	E\left[\left\|\frac{1}{|\mathcal J_2|}\sum_{i\in \mathcal J_2}\Delta_{3i}\bar\bV_i\right\|_\infty^2\right]=O_p\left(\frac{\sigma^2s_{\bdelta_a}\log(d)\log^2(d_1)}{N^2}\right)=O_p\left(\frac{\sigma^2\log (d_1)}{N}\right),
\end{align*}
since $s_{\bm{\delta}_{a}}\log(d_1)\log(d)=O(N)$. 

Together with \eqref{bound:Delta3i'}, we conclude that $E[\||\mathcal J_2|^{-1}\sum_{i\in \mathcal J_2}\Delta_{3i}\bar\bV_i\|_\infty^2]=O_p(\sigma\sqrt{\log (d_1)/N})$ when $\nu_{a}^{*}(\cdot)=\nu_{a}(\cdot)$ and either (a) $\|\bS_1\|_\infty\leq C$ almost surely or (b) $s_{\bm{\delta}_{a}}\log(d_1)\log(d)=O(N)$. By Markov's inequality, we have
\begin{align}\label{bound:Delta3i}
	\left\|\frac{1}{|\mathcal J_2|}\sum_{i\in \mathcal J_2}\Delta_{3i}\bar\bV_i\right\|_\infty=O_p\left( \sigma \sqrt\frac{\log (d_1)}{N}\right).
\end{align}
Together with \eqref{bound:Delta1} and \eqref{bound:Delta2i},
\begin{align*}
	\sum_{l=1}^3\left\|\frac{1}{|\mathcal J_2|}\sum_{i\in \mathcal J_2}\Delta_{li}\bar\bV_i\right\|_\infty=O_p\left(\sigma \sqrt\frac{\log (d_1)}{N}\right).
\end{align*}
That is, for any $t>0$, there exists some $\lambda_{\bbeta}\asymp\sigma\sqrt{\log(d_1)/N}$ such that $\mathcal E_5$ occurs with probability at least $1-t$, where
$$\mathcal E_5:=\left\{\sum_{l=1}^3\left\|\frac{1}{|\mathcal J_2|}\sum_{i\in \mathcal J_2}\Delta_{li}\bar\bV_i\right\|_\infty\leq\frac{\lambda_{\bbeta}}{4}\right\}.$$
Condition on the event $\mathcal E_5$. Then, now we have
\begin{align*}
	&\frac{1}{|\mathcal J_2|}\sum_{i\in \mathcal J_2}(\bar\bV_i^\top\bDelta_{\bbeta})^2+\lambda_{\bbeta}\|\bbetahat_{a,2}\|_1\\
	&\qquad\leq\frac{\lambda_{\bbeta}}{2}\|\bDelta_{\bbeta}\|_1+\frac{2}{|\mathcal J_2|}\sum_{i\in \mathcal J_2}\left(\sum_{l=4}^6\Delta_{li}\right)^2+\frac{1}{2|\mathcal J_2|}\sum_{i\in \mathcal J_2}(\bar\bV_i^\top\bDelta_{\bbeta})^2+\lambda_{\bbeta}\|\bbeta_a^*\|_1.
\end{align*}
Since $\left(\sum_{l=4}^6\Delta_{li}\right)^2\leq3\sum_{l=4}^6\Delta_{li}^2$, we have
\begin{align*}
	\frac{1}{|\mathcal J_2|}\sum_{i\in \mathcal J_2}(\bar\bV_i^\top\bDelta_{\bbeta})^2+2\lambda_{\bbeta}\|\bbetahat_{a,2}\|_1\leq\lambda_{\bbeta}\|\bDelta_{\bbeta}\|_1+\frac{12}{|\mathcal J_2|}\sum_{i\in \mathcal J_2}\sum_{l=4}^6\Delta_{li}^2+2\lambda_{\bbeta}\|\bbeta_a^*\|_1,
\end{align*}
which reaches \eqref{bound:to-be-used} in the proof of Theorem 1. Repeat the remaining steps therein, when $N\gg s_{\bm{\beta}_{a}}\log(d_1)$, with $ {\lambda}_{\bm{\beta}} \asymp \sigma\sqrt{\log(d_1)/N}$, we have
\begin{align*}
	\|\bbetahat_{a,2}-\bbeta_a^*\|_2=\|\bDelta_{\bbeta}\|_2=O_p\left(\sigma\sqrt\frac{s_{\bbeta_a}\log(d_1)}{N}+\left(\frac{1}{|\mathcal J_2|}\sum_{i\in \mathcal J_2}\sum_{l=4}^6\Delta_{li}^2\right)^{1/2}\right).
\end{align*}
In the following, we further control the term $|\mathcal J_2|^{-1}\sum_{i\in \mathcal J_2}\Delta_{li}^2$ for each $l\in\{4,5,6\}$. 
By Lemmas \ref{cor_mu2} and \ref{lem_conclusion} with the H{\"o}lder's inequality, we have
\begin{align*}
&E\left[\frac{1}{|\mathcal J_2|}\sum_{i\in \mathcal J_2}\Delta_{4i}^2\right]=E\left[\Delta_{4}^2\right]\\
&\qquad\leq\left\{E\left[\left\{\frac{1}{g(\bU^\top\bdeltatil_a)}-\frac{1}{g(\bU^\top\bdelta_a^*)}\right\}\right]^4E\left[\bU^\top(\widetilde{\balpha}_a-\balpha_a^*)\right]^4\right\}^{1/2}\\
&\qquad=O_p\left(\frac{\sigma^2s_{\bdelta_a}s_{\balpha_a}\log^2(d)}{N^2}\right).
\end{align*}
Under Assumption 4, by Lemma \ref{cor_mu2} with the H{\"o}lder's inequality, we have
\begin{align*}
	E\left[\frac{1}{|\mathcal J_2|}\sum_{i\in \mathcal J_2}\Delta_{5i}^2\right]&=E\left[\left[\left\{1-\frac{\mathbbm1_{\{A_{2i}=a_2\}}}{g(\bU_i^\top\bdelta_a^*)}\right\}\bar\bU_i^\top(\widetilde{\balpha}_a-\balpha_a^*)\right]^2\mathbbm{1}_{\{\rho_{a}^{*}(\cdot)\neq\rho_{a}(\cdot)\}}\right]\\
	&=O_p\left(\frac{\sigma^2s_{\balpha_a}\log(d)}{N}\mathbbm{1}_{\{\rho_{a}^{*}(\cdot)\neq\rho_{a}(\cdot)\}}\right).
\end{align*}
Under Assumption 2, by Lemma \ref{lem_conclusion} with the H{\"o}lder's inequality, we have
\begin{align*}
	E\left[\frac{1}{|\mathcal J_2|}\sum_{i\in \mathcal J_2}\Delta_{6i}^2\right]&=E\left[\left[\left\{\frac{1}{g(\bU_i^\top\bdeltatil_a)}-\frac{1}{g(\bU_i^\top\bdelta_a^*)}\right\} (\Ytil_i-\widetilde{\bU}_i^\top\balpha_a^*)\right]^2\mathbbm{1}_{\{\nu_{a}^{*}(\cdot)\neq\nu_{a}(\cdot)\}}\right]\\
	&=O_p\left(\frac{\sigma^2s_{\bdelta_a}\log(d)}{N}\mathbbm{1}_{\{\nu_{a}^{*}(\cdot)\neq\nu_{a}(\cdot)\}}\right).
\end{align*}
By Markov's inequality,
\begin{align*}
	&\frac{1}{|\mathcal J_2|}\sum_{i\in \mathcal J_2}\sum_{l=4}^6\Delta_{li}^2\\
	&\quad=O_p\left(\frac{\sigma^2s_{\bdelta_a}s_{\balpha_a}\log^2d}{N^2}+\frac{\sigma^2s_{\balpha_a}\log(d)}{N}\mathbbm{1}_{\{\rho_{a}^{*}(\cdot)\neq\rho_{a}(\cdot)\}}+\frac{\sigma^2s_{\bdelta_a}\log(d)}{N}\mathbbm{1}_{\{\nu_{a}^{*}(\cdot)\neq\nu_{a}(\cdot)\}}\right).
\end{align*}
Therefore, we have
\begin{align*}
\|\bbetahat_{a,2}-\bbeta_a^*\|_2=O_p(r_n),
\end{align*}
with $r_n=\sigma\sqrt\frac{s_{\bbeta_a}\log(d_1)}{N}+\frac{\sigma\sqrt{s_{\bdelta_a}s_{\balpha_a}}\log(d)}{N}+\sigma\sqrt\frac{s_{\balpha_a}\log(d)}{N}\mathbbm{1}_{\{\rho_{a}^{*}\neq\rho_{a}\}}+\sigma\sqrt\frac{s_{\bdelta_a}\log(d)}{N}\mathbbm{1}_{\{\nu_{a}^{*}\neq\nu_{a}\}}$. Similarly, we consider the DR-imputation-based estimator $\bbetahat_{a,1}=\bbetahat_a(\mathcal D_{\mathcal J_1},\widehat Y^{\mbox{\tiny DR}}_{\mathcal J_1})$ defined as (2.11).
In this case, the expectations are taken w.r.t. the samples in $\mathcal D_{\mathcal J_1}$; $\widehat{\bdelta}_{a}(\mathcal D_{\mathcal J_2})$ and $\balphahat_a(\mathcal D_{\mathcal J_2})$ are fitted using samples in $\mathcal D_{\mathcal J_2}$ and are treated as fixed or condition on. We also have the same consistency rate for $\bbetahat_{a,1}$. Therefore, for $\bbetahat_a=(\bbetahat_{a,1}+\bbetahat_{a,2})/2$, we have
\begin{align*}
\|\bbetahat_a-\bbeta_a^*\|_2=O_p(r_n).
\end{align*}
By Lemma \ref{lem_eigen}, $\|\bV^\top(\bbetahat_a-\bbeta_a^*)\|_{\psi_2}\leq2\sigma_u\|\bbetahat_a-\bbeta_a^*\|_2$. By Lemma \ref{lemma:psi2norm}, for any $r\geq1$,
\begin{align*}
	&\left\{E[\muhat_a(\bS_1)-\mu_{a}^{*}(\bS_1)]^r\right\}^{1/r}=\left\{E[\bV^\top(\bbetahat_a-\bbeta_a^*)]^r\right\}^{1/r}=O(\|\bbetahat_a-\bbeta_a^*\|_2)=O_p(r_n).
\end{align*}
\end{proof}

\subsubsection{Nested-regression-based estimator}
\begin{proof}[Proof of Theorem 10]
Let $\Yhat=\bar{\bU}^\top\balphahat_a$, $Y^*=\bar{\bU}^\top\balpha_a^*$, $\bX=\bar{\bV}$, $\S=(\bar{\bV}_i)_{i\in J}$, $M=|\mathcal J|$, and $\delta_M^2=288\sigma_u\kappa_1^{-2}\lambda_{\balpha}^2s_{\balpha_a}$. For any $t>0$, let $ {\lambda}_{\bm{\alpha}}:=32\sigma\sigma_{u}\sigma_{\zeta}(\sqrt{\log(d+1)/|\mathcal J|}+t)$ and $ {\lambda}_{\bm{\beta}}:=32\sigma\sigma_{u}\sigma_{\varepsilon}(\sqrt{\log(d_1+1)/ |\mathcal J|}+t)$. Suppose that
$|\mathcal J|\geq\max\{\log(d+1),(c_3+100 \kappa_2^2) s_{\bm{\alpha}_{a}}\log(d+1),100 \kappa_2^2 s_{\bm{\beta}_{a}}\log(d_1+1)\}.$ Now, for the event $\mathcal E_1:=\{|\mathcal J|^{-1}\sum_{i\in\mathcal J}[\Yhat_i-Y_i^*]^2<\delta_M^2\}$, by Lemma \ref{lemma:insample}, we have
$$P(\mathcal E_1)\geq1-2\exp\left(-\frac{4|\mathcal J|t^2}{1+2t+\sqrt{2t}}\right)-c_1\exp(-c_2|\mathcal J|)-2\exp(-c_4|\mathcal J|).$$
By Lemma \ref{lem_eigen}, $\lambda_{\min}(E[\bar{\bV}\bar{\bV}^{\top}])\geq \kappa_{l}$
and $\bar{\bV}$ is sub-Gaussian with 
$\|\bm{x}^\top\bar{\bV} \|_{{\psi}_2} \leq 2\sigma_{u}\|\bm{x}\|_2$, for any $\bm{x}\in \mathbb{R}^{d_1+1}$. Additionally, under Assumption 2, $\|\varepsilon\|_{\psi_2} \leq \sigma\sigma_{\varepsilon}$.
Here, $\kappa_l$, $\sigma_u$, $\sigma_\varepsilon$, and $\sigma$, defined in Assumptions 2 and 3, are positive constants independent of $N$ and $d$.
Hence, the estimation rates of $\widehat{\bm{\beta}}_{a,{\mbox{\tiny NR}}}$ in Theorem 10 follow from Theorem 8. To show the esitmation rate of $\widehat{\mu}_{a,{\mbox{\tiny NR}}}(\cdot)$, by Lemma \ref{lemma:psi2norm}, for any $r\geq1$,
\begin{align*}
&\{E[\widehat{\mu}_{a,{\mbox{\tiny NR}}}(\bS_1)-\mu_{a,{\mbox{\tiny NR}}}^{*}(\bS_1)]^r\}^{1/r}=\{E[\bV^\top(\widehat{\bm{\beta}}_{a,{\mbox{\tiny NR}}}-\bm{\beta}_{a,{\mbox{\tiny NR}}}^*)]^r\}^{1/r}\\
&\qquad=O(\|\widehat{\bm{\beta}}_{a,{\mbox{\tiny NR}}}-\bm{\beta}_{a,{\mbox{\tiny NR}}}^*\|_2)=O_p\left(\sigma\sqrt\frac{s_{\bbeta_a}\log(d_1)}{N}+\sigma\sqrt\frac{s_{\balpha_a}\log(d)}{N}\right),
\end{align*}
since $\|\bV^\top(\widehat{\bm{\beta}}_{a,{\mbox{\tiny NR}}}-\bm{\beta}_{a,{\mbox{\tiny NR}}}^*)\|_{\psi_2}\leq2\sigma_u\|\widehat{\bm{\beta}}_{a,{\mbox{\tiny NR}}}-\bm{\beta}_{a,{\mbox{\tiny NR}}}^*\|_2$ by Lemma \ref{lem_eigen}.
Here, the expectation is only taken w.r.t. the distribution of the new observation $\bS_1$.	
\end{proof}

\subsection{Proofs of Auxiliary Lemmas}
\begin{proof}[Proof of Lemma \ref{lem_orlicz}]
By the definition of $\|X\|_{\psi_2}=\inf\{c>0:E[\exp(X^2/c^2)]\leq 2\}$ and
\begin{align*}
 E\left[\exp\left(\frac{X^2}{4\sigma^2}\right)\right]
 = E\left[\sum_{k=0}^{\infty}\frac{X^{2k}}{k!(4\sigma^2)^{k}}\right]
 \leq \sum_{k=0}^{\infty}\frac{2^k\sigma^{2k}\Gamma(k+1)}{k!(4\sigma^2)^{k}}
 =\sum_{k=0}^{\infty}\frac{1}{2^{k}}=2,
\end{align*}
therefore, leading to
$
 \|X\|_{\psi_2}\leq 2\sigma.
$
\end{proof}

\begin{proof}[Proof of Lemma \ref{lem_eigen}]
(a) Observe that
\begin{align}
&\lambda_{\min}(E[\widetilde{\bU}\widetilde{\bU}^{\top}])=\min_{\bm{x}\in\R^{d+1}:\|\bm{x}\|_2=1}\bm{x}^{\top}E[\bU\bU^{\top}\mathbbm{1}_{\{A_1=a_1, A_2=a_2\}}]\bm{x} \nonumber\\
&\qquad\overset{(i)}{=}\min_{\bm{x}\in\R^{d+1}:\|\bm{x}\|_2=1}E[E[(\bU^{\top}\bm{x})^{2}\mathbbm{1}_{\{A_1=a_1, A_2=a_2\}}|\bU,A_{1}=a_1]P[A_{1}=a_1|\bU]]\nonumber\\
&\qquad=\min_{\bm{x}\in\R^{d+1}:\|\bm{x}\|_2=1}E[(\bU^{\top}\bm{x})^{2}\cdot P[A_2=a_2|\bU,A_{1}=a_1]E[\mathbbm{1}_{\{A_{1}=a_1\}}|\bU]]\nonumber\\
&\qquad\overset{(ii)}{=}\min_{\bm{x}\in\R^{d+1}:\|\bm{x}\|_2=1}E[(\bU^{\top}\bm{x})^{2}\mathbbm{1}_{\{A_{1}=a_1\}}\cdot P[ A_2=a_2|\bU,A_{1}=a_1]],\label{81} 
\end{align}
where (i) and (ii) hold by the tower rule.
Under the overlap conditions of Assumption 1,
\begin{align*}
 P(c_0 \leq P[A_2=a_2|\bU, A_{1}=a_1] \leq 1-c_0)=1.
\end{align*}
Together with \eqref{81}, under Assumption 3, we have
\begin{align*}
 \lambda_{\min}(E[\widetilde{\bU}\widetilde{\bU}^{\top}])
 \geq c_0 \min_{\bm{x}\in\R^{d+1}:\|\bm{x}\|_2=1}E[(\bU^{\top}\bm{x})^{2}\mathbbm{1}_{\{A_{1}=a_1\}}]
 \geq c_0 \kappa_{l}
 >0.
\end{align*}
Additionally, we also have
\begin{align*}
&\lambda_{\max}(E[\widetilde{\bU}\widetilde{\bU}^{\top}])=\max_{\bm{x}\in\R^{d+1}:\|\bm{x}\|_2=1}\bm{x}^{\top}E[\bU\bU^{\top}\mathbbm{1}_{\{A_1=a_1, A_2=a_2\}}]\bm{x}\\
&\qquad\leq\max_{\bm{x}\in\R^{d+1}:\|\bm{x}\|_2=1}\bm{x}^{\top}E[\bU\bU^{\top}]\bm{x}=\lambda_{\max}(E[\bU\bU^{\top}])\overset{(i)}{\leq}2\sigma_u^2,
\end{align*}
where (i) holds since, by Assumption 3 and Lemma \ref{lemma:psi2norm},
\begin{align}
\lambda_{\max}(E[\bU\bU^{\top}])=\max_{\|\bm{x}\|_2=1}E[(\bm{x}^{\top}\bU)^2]\leq\max_{\|\bm{x}\|_2=1} 2 \sigma_{u}^2\|\bm{x}\|_2^2=2\sigma_{u}^2<\infty.\label{bound:EUUT}
\end{align}
Besides, for any $\bm{x}\in\R^{d+1}$ and $k\in\mathbb{N}$,
\begin{align*}
E[|\bm{x}^\top\widetilde{\bU}|^{2k}]=E[|\bm{x}^\top\bU|^{2k}\mathbbm{1}_{\{A_1=a_1, A_2=a_2\}}]\leq E[|\bm{x}^\top\bU|^{2k}]\overset{(i)}{\leq}2(\sigma_{u}\|\bm{x}\|_2)^{2k}\Gamma(k+1),
\end{align*}
where (i) holds by Assumption 3 and Lemma \ref{lemma:psi2norm}. By Lemma \ref{lem_orlicz}, we have
\begin{align*}
 \|\bm{x}^\top\widetilde{\bU} \|_{{\psi}_2}&\leq 2\sigma_{u}\|\bm{x}\|_2, \quad \text{for any} \ \bm{x}\in \mathbb{R}^{d+1}.
\end{align*}

(b) Under Assumption 3, we also have
\begin{align}
\lambda_{\min}(E[\bar{\bU}\bar{\bU}^{\top}])=\min_{\bm{x}\in\R^{d+1}:\|\bm{x}\|_2=1}E[(\bU^{\top}\bm{x})^{2}\mathbbm{1}_{\{A_{1}=a_1\}}]\geq \kappa_{l}>0,\label{bound:lamminUbar}
\end{align}
and by \eqref{bound:EUUT},
\begin{align}
&\lambda_{\max}(E[\bar{\bU}\bar{\bU}^{\top}])=\max_{\bm{x}\in\R^{d+1}:\|\bm{x}\|_2=1}\bm{x}^{\top}E[\bU\bU^{\top}\mathbbm{1}_{\{A_{1}=a_1\}}]\bm{x}\nonumber\\
&\qquad\leq \max_{\bm{x}\in\R^{d+1}:\|\bm{x}\|_2=1}\bm{x}^{\top}E[\bU\bU^{\top}]\bm{x}\leq 2\sigma_{u}^2<\infty.\label{bound:lammaxUbar}
\end{align}
In addition, for any $\bm{x}\in\R^{d+1}$ and $k\in\mathbb{N}$,
\begin{align}
E[|\bm{x}^\top\bar{\bU}|^{2k}]=E[|\bm{x}^\top\bU|^{2k}\mathbbm{1}_{\{A_{1}=a_1\}}]\leq E[|\bm{x}^\top\bU|^{2k}]\overset{(i)}{\leq}2(\sigma_{u}\|\bm{x}\|_2)^{2k}\Gamma(k+1),\label{bound:EU2k}
\end{align}
where (i) holds by Assumption 3 and Lemma \ref{lemma:psi2norm}. By Lemma \ref{lem_orlicz}, we have
\begin{align*}
\|\bm{x}^\top\bar{\bU} \|_{{\psi}_2}&\leq 2\sigma_{u}\|\bm{x}\|_2, \quad \text{for any} \ \bm{x}\in \mathbb{R}^{d+1}.
\end{align*} 

(c) Note that
\begin{align*}
	&\lambda_{\min}(E[\bU\bU^{\top}])=\min_{\bm{x}\in\R^{d_1+1}:\|\bm{x}\|_2=1}\bm{x}^{\top}E[\bU\bU^{\top}]\bm{x}\\
	&\qquad\geq\min_{\bm{x}\in\R^{d_1+1}:\|\bm{x}\|_2=1}\bm{x}^{\top}E[\bU\bU^{\top}\mathbbm{1}_{\{A_{1}=a_1\}}]\bm{x}=\lambda_{\min}(E[\bar{\bU}\bar{\bU}^{\top}])\overset{(i)}{\geq}\kappa_l>0,
\end{align*}
where (i) holds by \eqref{bound:lamminUbar}. By \eqref{bound:EUUT}, we know
$\lambda_{\max}(E[\bU\bU^{\top}])\leq2\sigma_{u}^2<\infty$. By Assumption 3, we have
\begin{align*}
	\|\bm{x}^\top\bU \|_{{\psi}_2}&\leq \sigma_{u}\|\bm{x}\|_2, \quad \text{for any} \ \bm{x}\in \mathbb{R}^{d+1}.
\end{align*} 

(d) Recall the representation \eqref{72}, we also have
\begin{align}
&\lambda_{\min}(E[\bar{\bV}\bar{\bV}^{\top}])=\min_{\bm{x}\in\R^{d_1+1}:\|\bm{x}\|_2=1}\bm{x}^{\top}E[\bV\bV^{\top}\mathbbm{1}_{\{A_{1}=a_1\}}]\bm{x}\nonumber\\
&\qquad=\min_{\bm{x}\in\R^{d_1+1}:\|\bm{x}\|_2=1}\bm{x}^{\top}E[\mathbf{Q}\bU\bU^\top\mathbf{Q}^\top\mathbbm{1}_{\{A_{1}=a_1\}}]\bm{x}\nonumber\\
&\qquad\overset{(i)}{\geq}\min_{\bm{x}\in\R^{d+1}:\|\bm{x}\|_2=1}\bm{x}^{\top}E[\bU\bU^{\top}\mathbbm{1}_{\{A_{1}=a_1\}}]\bm{x}=\lambda_{\min}(E[\bar{\bU}\bar{\bU}^{\top}] )\overset{(ii)}{\geq}\kappa_l>0,\label{bound:lamminVbar}
\end{align}
where (i) holds since, for every $\|\bm{x}\|_2=1$ and $\bm{x}\in\R^{d_1+1}$, $\mathbf{Q}^\top\bm{x}=(\bm{x}^\top,0,\dots,0)^\top\in\R^{d+1}$ and hence $\|\mathbf{Q}^\top\bm{x}\|_2=\|\bm{x}\|_2=1$; (ii) follows from \eqref{bound:lamminUbar}. Similarly,
\begin{align*}
&\lambda_{\max}(E[\bar{\bV}\bar{\bV}^{\top}])=\max_{\bm{x}\in\R^{d_1+1}:\|\bm{x}\|_2=1}\bm{x}^{\top}E[\bV\bV^{\top}\mathbbm{1}_{\{A_{1}=a_1\}}]\bm{x}\\
&\qquad=\max_{\bm{x}\in\R^{d_1+1}:\|\bm{x}\|_2=1}\bm{x}^{\top}E[\mathbf{Q}\bU\bU^\top\mathbf{Q}^\top\mathbbm{1}_{\{A_{1}=a_1\}}]\bm{x}\\
&\qquad\leq\max_{\bm{x}\in\R^{d+1}:\|\bm{x}\|_2=1}\bm{x}^{\top}E[\bU\bU^{\top}\mathbbm{1}_{\{A_{1}=a_1\}}]\bm{x}=\lambda_{\max}(E[\bar{\bU}\bar{\bU}^{\top}] )\overset{(i)}{\leq}2\sigma_u^2<\infty,
\end{align*}
where (i) follows from \eqref{bound:lammaxUbar}. In addition, for any $k\in\mathbb{N}$,
\begin{align*}
&\sup_{\bm{x}\in\R^{d_1+1}:\|\bm{x}\|_2=1}E[|\bm{x}^\top\bar{\bV}|^{2k}]=\sup_{\bm{x}\in\R^{d_1+1}:\|\bm{x}\|_2=1}E[|\bm{x}^\top\mathbf{Q}\bar{\bU}|^{2k}]\\
&\qquad\overset{(i)}{\leq}\sup_{\bm{x}\in\R^{d+1}:\|\bm{x}\|_2=1}E[|\bm{x}^\top\bar{\bU}|^{2k}]\overset{(ii)}{\leq}2(\sigma_{u}\|\bm{x}\|_2)^{2k}\Gamma(k+1),
\end{align*}
where (i) holds since, for every $\|\bm{x}\|_2=1$ and $\bm{x}\in\R^{d_1+1}$, $\|\mathbf{Q}^\top\bm{x}\|_2=\|\bm{x}\|_2=1$ ; (ii) follows from \eqref{bound:EU2k}. Hence, for any $\bm{x}\in\R^{d+1}$ and $k\in\mathbb{N}$,
$$E[|\bm{x}^\top\bar{\bV}|^{2k}]\leq2(\sigma_{u}\|\bm{x}\|_2)^{2k}\Gamma(k+1).$$
By Lemma \ref{lem_orlicz}, we have $\bar\bV$ is sub-Gaussian with
\begin{align*}
\|\bm{x}^\top\bar{\bV} \|_{{\psi}_2}&\leq 2\sigma_{u}\|\bm{x}\|_2, \quad \text{for any} \ \bm{x}\in \mathbb{R}^{d_1+1}.
\end{align*} 

(e) Lastly, note that
\begin{align*}
&\lambda_{\min}(E[\bV\bV^{\top}])=\min_{\bm{x}\in\R^{d_1+1}:\|\bm{x}\|_2=1}\bm{x}^{\top}E[\bV\bV^{\top}]\bm{x}\\
&\qquad\geq\min_{\bm{x}\in\R^{d_1+1}:\|\bm{x}\|_2=1}\bm{x}^{\top}E[\bV\bV^{\top}\mathbbm{1}_{\{A_{1}=a_1\}}]\bm{x}=\lambda_{\min}(E[\bar{\bV}\bar{\bV}^{\top}])\overset{(i)}{\geq}\kappa_l>0,
\end{align*}
where (i) holds by \eqref{bound:lamminVbar}. Besides, 
\begin{align*}
&\lambda_{\max}(E[{\bV}{\bV}^{\top}])=\max_{\bm{x}\in\R^{d_1+1}:\|\bm{x}\|_2=1}\bm{x}^{\top}E[\bV\bV^{\top}]\bm{x}=\max_{\bm{x}\in\R^{d_1+1}:\|\bm{x}\|_2=1}\bm{x}^{\top}E[\mathbf{Q}\bU\bU^\top\mathbf{Q}^\top]\bm{x}\\
&\qquad\leq\max_{\bm{x}\in\R^{d+1}:\|\bm{x}\|_2=1}\bm{x}^{\top}E[\bU\bU^{\top}]\bm{x}=\lambda_{\max}(E[\bU\bU^{\top}] )\overset{(i)}{\leq}2\sigma_u^2<\infty,
\end{align*}
where (i) follows from \eqref{bound:lammaxUbar}. In addition, for any $k\in\mathbb{N}$,
\begin{align*}
&\sup_{\bm{x}\in\R^{d_1+1}:\|\bm{x}\|_2=1}E[|\bm{x}^\top\bV|^{2k}]=\sup_{\bm{x}\in\R^{d_1+1}:\|\bm{x}\|_2=1}E[|\bm{x}^\top\mathbf{Q}\bU|^{2k}]\\
&\qquad\overset{(i)}{\leq}\sup_{\bm{x}\in\R^{d+1}:\|\bm{x}\|_2=1}E[|\bm{x}^\top\bU|^{2k}]\overset{(ii)}{\leq}2(\sigma_{u}\|\bm{x}\|_2)^{2k}\Gamma(k+1),
\end{align*}
where (i) holds since, for every $\|\bm{x}\|_2=1$ and $\bm{x}\in\R^{d_1+1}$, $\|\mathbf{Q}^\top\bm{x}\|_2=\|\bm{x}\|_2=1$ ; (ii) follows from \eqref{bound:EU2k}. Hence, for any $\bm{x}\in\R^{d+1}$ and $k\in\mathbb{N}$,
$$E[|\bm{x}^\top\bV|^{2k}]\leq2(\sigma_{u}\|\bm{x}\|_2)^{2k}\Gamma(k+1).$$
By Lemma \ref{lem_orlicz}, we have $\bV$ is also sub-Gaussian with
\begin{align*}
\|\bm{x}^\top\bV\|_{{\psi}_2}&\leq 2\sigma_{u}\|\bm{x}\|_2, \quad \text{for any} \ \bm{x}\in \mathbb{R}^{d_1+1}.
\end{align*} 
\end{proof}

\begin{proof}[Proof of Lemma \ref{cor_mu2}] 
Now, we consider the Lasso estimator $\balphahat_a$ defined as $(2.6)$, which is constructed using the outcome $\widetilde Y$, covariates $\widetilde{\bU}$ and training samples $\mathcal D_\mathcal J$. Note that $\balphahat_a$ is a special case of $\bbetahat$, $(4.1)$. 
Let $\Yhat=Y^*=\widetilde{Y}$, $\bX=\widetilde{\bU}$, $\S=(\bX_i)_{i\in J}$, $M=|\mathcal J|$, and $\delta_M=0$. By Lemma \ref{lem_eigen}, $\lambda_{\min}(E[\widetilde{\bU}\widetilde{\bU}^{\top}])\geq c_0\kappa_{l}$
and $\widetilde{\bU}$ is sub-Gaussian with 
$\|\bm{x}^\top\widetilde{\bU} \|_{{\psi}_2} \leq 2\sigma_{u}\|\bm{x}\|_2$, for any $\bm{x}\in \mathbb{R}^{d+1}$. Additionally, under Assumption 2, $\|\zeta \|_{\psi_2} \leq \sigma\sigma_{\zeta}$.
Here, $c_0$, $\kappa_l$, $\sigma_u$, $\sigma_\zeta$, and $\sigma$, defined in Assumptions 1, 2, and 3, are positive constants independent of $N$ and $d$.
Hence, the estimation rates of $\balphahat_a$ in Lemma \ref{cor_mu2} follows from Theorem 8. To show the estimation rate of $\widehat\nu_a(\cdot)$, by Lemma \ref{lemma:psi2norm}, for any $r\geq1$,
\begin{align*}
&\{E[\widehat{\nu}_{a}(\bS)-\nu_{a}^{*}(\bS)]^r\}^{1/r}=\{E[\bU^\top(\balphahat_a-\balpha_a^*)]^r\}^{1/r}\\
&\qquad=O(\|\balphahat_a-\balpha_a^* \|_2)=O_p\left(\sigma\sqrt{\frac{s_{\bm{\alpha}_{a}}\log(d)}{N}}\right),
\end{align*}
since $\|\bU^\top(\balphahat_a-\balpha_a^*)\|_{\psi_2}\leq\sigma_u\|\balphahat_a-\balpha_a^*\|_2$ under Assumption 3.
Here, the expectation is only taken w.r.t. the joint distribution of the new observations $\bS$.
\end{proof}

\begin{proof}[Proof of Lemma \ref{lemma:insample}]
Let $\Yhat=Y^*=\widetilde{Y}$, $\bX=\widetilde{\bU}$, $\S=(\bX_i)_{i\in J}$, $M=|\mathcal J|$, and $\delta_M=0$. Following the proof of Theorem 8, since $\delta_M=0$, we have $\|\bdelta_S\|_1\geq4\lambda^{-1}\delta_M^2$. That is, we are under Case 2. Hence, $\bdelta$ is in the cone set as in \eqref{eq:cone}. By Lemma \ref{lem_eigen}, $\|\boldsymbol{a}^\top\bar\bU\|_{\psi_2}\leq2\sigma_u\|\boldsymbol{a}\|_2$ for any $\boldsymbol{a}\in\R^{d+1}$ and $\lambda_{\min}(E[\bar\bU\bar\bU^\top])\geq \kappa_l$. Here, $\sigma_u$ and $\kappa_l$, defined in Assumption 3, are positive constants independent of $N$ and $d$. By Theorem 15 of \cite{pmlr-v23-rudelson12}, with some constants $c_3,c_4>0$, when $|\mathcal J|\geq c_3s_{\balpha_a}\log(d+1)$,
$$\frac{1}{|\mathcal J|}\sum_{i\in\mathcal J}\left\{\bar\bU_i^\top(\balphahat_a-\balpha_a^*)\right\}^2\leq1.5^2\lambda_{\max}(E[\bar\bU\bar\bU^\top])\|\balphahat_a-\balpha_a^*\|_2^2\leq4.5\sigma_u\|\balphahat_a-\balpha_a^*\|_2^2,$$
with probability at least $1-2\exp(-c_4|\mathcal J|)$. In addition, by Lemma \ref{cor_mu2}, we have
$$\|\balphahat_a-\balpha_a^*\|_2\leq8\kappa_1^{-1}\lambda_{\balpha}\sqrt{s_{\balpha_a}},$$
with probability at least $1-2\exp(-\frac{4|\mathcal J|t^2}{1+2t+\sqrt{2t}})-c_1\exp(-c_2|\mathcal J|)$. Therefore, with probability at least $1-2\exp(-\frac{4|\mathcal J|t^2}{1+2t+\sqrt{2t}})-c_1\exp(-c_2|\mathcal J|)-2\exp(-c_4|\mathcal J|)$,
$$\frac{1}{|\mathcal J|}\sum_{i\in\mathcal J}[\bar\bU_i^\top (\balphahat_a-\balpha_a^*)]^2\leq288\sigma_u\kappa_1^{-2}\lambda_{\balpha}^2s_{\balpha_a}.$$
\end{proof}


\begin{proof}[Proof of Lemma \ref{lem_pi1}]
In this Lemma, we provide estimation rates for $\bgammahat_{a}$, $\widehat\pi_{a}(\cdot)$, $\widehat\bdelta_a$, and $\widehat\rho_a(\cdot)$. We allow model misspecifications that $\pi_{a}^*(\cdot)\neq\pi_{a}(\cdot)$ and $\rho_a^*(\cdot)\neq\rho_a(\cdot)$. Note that, classical results for generalized linear models only consider correrctly specified cases; see, e.g., Corollary 9.26 of \cite{wainwright2019high} and Section 4.4 of \cite{negahban2012unified}.

(a) We first show \eqref{3.36} and \eqref{3.31}. In part (a), the expectations are only taken w.r.t. the distribution of the new observation $\bS_1$.

Consider the link function $\phi(u)=\log(1+\exp(u))$, we have
$$\phi''(\bV^{\top}\bm{\gamma}_{a}^*)
=\frac{\exp( \bV^{\top}\bm{\gamma}_{a}^*)}{(1+\exp(\bV^{\top}\bm{\gamma}_{a}^*))^2}
=\pi_{a}^*(\bS_1)(1-\pi_{a}^*(\bS_1)).$$
Under Assumption 4, we have $P(c_0^2\leq \phi''(\bV^{\top}\bm{\gamma}_{a}^*)\leq (1-c_0)^2)=1$.
By Lemma \ref{lem_eigen},
\begin{align}
\lambda_{\min}(E[\bV\bV^{\top}])\geq\kappa_l > 0,\quad\lambda_{\max}(E[\bV\bV^{\top}])\leq 2\sigma_{u}^2<\infty,\label{A.30}
\end{align}
and $\bV$ is sub-Gaussian with 
$\|\bm{x}^\top\bV \|_{\psi_2}\leq 2\sigma_{u}\|\bm{x}\|_2$ for any $\bm{x}\in \mathbb{R}^{d_1+1}$. 

Next, we control the gradient at the potentially misspecified location: recall that the underlying model may be misspecified and $\pi_{a}^*(\cdot)$ not necessarily equal to $\pi_{a}(\cdot)$; The true $\bgamma_{a}$ may not exists such that $\pihat_{a}(\cdot)$ has a logistic form. Below we ensure and discuss the Restricted Strong Convexity (RSC) as well as the properties of the gradient. 

We first consider the RSC property. Note that, the RSC property \eqref{RSC:deltalM} below only depends on the distribution of $\bS_1$ and does not depend on the distribution of $A_{1}|\bS_1$. This is because $\delta\ell_\mathcal J(\bDelta,\bgamma_{a}^*)$ defined in \eqref{def:deltaell} can be written as
$$\delta\ell_\mathcal J(\bDelta,\bgamma_{a}^*)=\frac{1}{|\mathcal J|}\sum_{i\in \mathcal J}\left[\phi(\bV_i^\top(\bgamma_{a}^*+\bDelta))-\phi(\bV_i^\top\bgamma_{a}^*)-\bDelta^\top\bV_i\phi'(\bV_i^\top\bgamma_{a}^*)\right],$$
which is function of $\bS_{1i}$s, and $A_{1i}$s are not involved above. As a result, the model misspecification for $\pi_{a}(\bS_1)=E(A_{1}|\bS_1)$ does not affect the RSC property. In below, we consider the RSC property studied by \cite{zhang2021double}.
For any $\bgamma_{a},\bDelta\in\R^{d_1+1}$, define
\begin{align}
	\ell_\mathcal J(\bgamma_{a})&:=\frac{1}{|\mathcal J|}\sum_{i\in \mathcal J}\left[\phi(\bV_{i}^{\top}\bm{\gamma_{a}})-\mathbbm1_{\{A_{1i}=a_1\}} \bV_{i}^{\top}\bgamma_{a}\right],\nonumber\\
	\delta\ell_\mathcal J(\bDelta,\bgamma_{a}^*)&:=\ell_\mathcal J(\bgamma_{a}^*+\bDelta)-\ell_\mathcal J(\bgamma_{a}^*)-\bDelta^\top\nabla\ell_\mathcal J(\bgamma_{a}^*).\label{def:deltaell}
\end{align}
By Lemma 4.5 of \cite{zhang2021double}, we have the following RSC property holds:
\begin{align}
\delta\ell_\mathcal J(\bDelta,\bgamma_{a}^*)&\geq\kappa_1\|\bDelta\|_2\left\{\|\bDelta\|_2-\kappa_2\sqrt\frac{\log(d_1+1)}{|\mathcal J|}\|\bDelta\|_1\right\}\nonumber\\
&\geq\frac{\kappa_1}{2}\|\bDelta\|_2^2-\frac{\kappa_1\kappa_2^2\log(d_1+1)}{2|\mathcal J|}\|\bDelta\|_1^2\quad\text{for all}\;\;\|\bDelta\|_2\leq1,\label{RSC:deltalM}
\end{align}
with probability at least $1-c_1\exp(-c_2|\mathcal J|)$, where $c_1,c_2,\kappa_1,\kappa_2>0$ are some constants.

Additionally, the gradient $\|\nabla\ell_\mathcal J(\bgamma_{a}^*)\|_\infty$ is controlled in the following. We allow a possibly misspecified model that $\pi_{a}^*(\cdot)\neq\pi_{a}(\cdot)$. Note that, even under model misspecification, we still have \eqref{eq:1stcond} below. Hence, $\|\nabla\ell_\mathcal J(\bgamma_{a}^*)\|_\infty$ is the maximum of zero-mean random variables.
By the union bound, we have
\begin{align}
P\left(\|\nabla\ell_\mathcal J(\bgamma_{a}^*)\|_\infty\geq\frac{\lambda_{\bgamma}}{2}\right)&=P\left(\max_{1\leq j\leq d_1+1}\left|\frac{1}{|\mathcal J|}\sum_{i\in \mathcal J}(g(\bV_i^\top\bgamma_{a}^*)-\mathbbm1_{\{A_{1i}=a_1\}})\bV_{i,j}\right|\geq\frac{\lambda_{\bgamma}}{2}\right)\nonumber\\
&\leq\sum_{j=1}^{d_1+1}P\left(\left|\frac{1}{|\mathcal J|}\sum_{i\in \mathcal J}(g(\bV_i^\top\bgamma_{a}^*)-\mathbbm1_{\{A_{1i}=a_1\}})\bV_{i,j}\right|\geq\frac{\lambda_{\bgamma}}{2}\right),\label{bound:lMgamma}
\end{align}
where $g(u)= \exp(u)/\{1+\exp(u)\}$ is the logistic function. By definition, $\bgamma_{a}^*=\\ \arg\min_{\bgamma_{a}\in\R^{d_1+1}}E[\ell(\bgamma_{a})]$, where for any $\bgamma_{a}\in\R^{d_1+1}$,
$$\ell(\bgamma_{a}):=E\left[\phi(\bV^{\top}\bm{\gamma}_{a})-\mathbbm1_{\{A_{1}=a_1\}} \bV^{\top}\bm{\gamma}_{a}\right].$$
By the first-order optimality condition, we know that
\begin{align}
\nabla E[\ell(\gamma_a^*)]=E\left[(g(\bV^\top\bgamma_{a}^*)-\mathbbm1_{\{A_{1}=a_1\}})\bV\right]=\bm{0}\in\R^{d_1+1}.\label{eq:1stcond}
\end{align} 
Additionally, since $|g(\bV^\top\bgamma_{a}^*)-\mathbbm1_{\{A_{1}=a_1\}}|\leq1$, by Lemma \ref{lemma:psi2norm} and under Assumption 3, for any $i\in \mathcal J$ and $j\leq d_1+1$,
$$\|(g(\bV_i^\top\bgamma_{a}^*)-\mathbbm1_{\{A_{1i}=a_1\}})\bV_{i,j}\|_{\psi_2}\leq\|\bV_{i,j}\|_{\psi_2}\leq\sigma_u.$$
That is, $(g(\bV_i^\top\bgamma_{a}^*)-\mathbbm1_{\{A_{1i}=a_1\}})\bV_{i,j}$ is a zero-mean sub-Gaussian random variable. Hence, by Hoeffding's inequality, for each $j\leq d_1+1$,
\begin{align*}
&P\left(\left|\frac{1}{|\mathcal J|}\sum_{i\in \mathcal J}(g(\bV_i^\top\bgamma_{a}^*)-\mathbbm1_{\{A_{1i}=a_1\}})\bV_{i,j}\right|\geq\frac{\lambda_{\bgamma}}{2}\right)\leq2\exp\left(\frac{-|\mathcal J|\lambda_{\bgamma}^2}{32\sigma_u^2}\right)\\
&\qquad\leq2\exp\left(-\log(d_1+1)-|\mathcal J|t^2\right)=\frac{2\exp(-|\mathcal J|t^2)}{d_1+1},
\end{align*}
where for any $t>0$, we set $\lambda_{\bgamma}:=4\sqrt2\sigma_u(\sqrt\frac{\log(d_1+1)}{|\mathcal J|}+t)$. Together with \eqref{bound:lMgamma}, it follows that
$$P\left(\|\ell_\mathcal J(\bgamma_{a}^*)\|_\infty\leq\frac{\lambda_{\bgamma}}{2}\right)\leq1-2\exp(-|\mathcal J|t^2).$$
Together with \eqref{RSC:deltalM}, when $|\mathcal J|\geq64\kappa_2^2s_{\bgamma_{a}}\log(d_1+1)$ and $9s_{\bgamma_{a}}\lambda_{\bgamma}^2\leq\kappa_1^2$, by Corollary 9.20 of \cite{wainwright2019high}, we conclude that
\begin{align*}
\|\bgammahat_{a}-\bgamma_{a}^*\|_2\leq\frac{3\sqrt{s_{\bgamma_{a}}}\lambda_{\bgamma}}{\kappa_1},\quad\|\bgammahat_{a}-\bgamma_{a}^*\|_1\leq\frac{6s_{\bgamma_{a}}\lambda_{\bgamma}}{\kappa_1},
\end{align*}
with probability at least $1-2\exp(-|\mathcal J|t^2)-c_1\exp(-c_2|\mathcal J|)$. Hence, when $|\mathcal J|\asymp N$ and $N\gg s_{\bgamma_{a}}\log(d_1)$, with some $\lambda_\mathcal J\asymp\sqrt\frac{\log(d_1)}{|\mathcal J|}$,
\begin{align}\label{213}
	\|\widehat{\bm{\gamma}}_{a}-\bm{\gamma}_{a}^{*}\|_2^2= O_p\left(\frac{s_{\bm{\gamma}_{a}}\log(d_1)}{N}\right).
\end{align}
Now, we show the estimation rate for $\pihat_{a}(\cdot)$. In the following, we will use Taylor's Theorem to control the estimation error of $\pihat_{a}(\cdot)$ by the estimation error of $\bgammahat_{a}$ as in \eqref{bound:Dpi_VDgamma}. Then, we apply the estimation rate \eqref{213} proved above to obtain the rate for $\pihat_{a}(\cdot)$.
 
Recall that $g(u):= \exp(u)/\{1+\exp(u)\}=\phi'(u)$ for any $u\in \mathbb{R}$. Note that, for any $u^*,\Delta\in\R$, 
\begin{align*}
\frac{d(g(u^*+t\Delta)-g(u^*))^2}{dt}&=2(g(u^*+t\Delta)-g(u^*))g'(u^*+t\Delta)\Delta,\\
\frac{d^2(g(u^*+t\Delta)-g(u^*))^2}{dt^2}&=2(g'(u^*+t\Delta))^2\Delta^2+2(g(u^*+t\Delta)-g(u^*))g''(u^*+t\Delta)\Delta^2,
\end{align*}
where, for any $u\in\R$, since $g(u)\in(0,1)$, we have
\begin{align}
g'(u)=g(u)(1-g(u))\in(0,1),\;\;g''(u)=g(u)(1-g(u))(1-2g(u))\in(-1,1).\label{bound:f'_f''}
\end{align}
Set $u^*=\bV^\top\bgamma_{a}^*$ and $\Delta=\bV^\top(\bgammahat_{a}-\bgamma_{a}^*)$. By Taylor's Theorem, with some $\widetilde t\in(0,1)$, 
\begin{align*}
&E[g(\bV^{\top} \widehat{\bm{\gamma}}_{a})-g(\bV^{\top}\bm{\gamma}_{a}^{*})]^2=E[g(u^*+1\cdot\Delta)-g(u^*)]^2\\
&\qquad=E[g(u^*+0\cdot\Delta)-g(u^*)]^2+\frac{dE(g(u^*+t\Delta)-g(u^*))^2}{dt}\biggr|_{t=0}\cdot1\\
&\qquad\qquad+\frac{d^2E(g(u^*+t\Delta)-g(u^*))^2}{2dt^2}\biggr|_{t=\widetilde t}\cdot1^2\\
&\qquad=0+E\left[2(g(u^*+0\cdot\Delta)-g(u^*))g'(u^*+0\cdot\Delta)\Delta\right]\\
&\qquad\qquad+E\left[(g'(u^*+\widetilde t\Delta))^2\Delta^2+(g(u^*+\widetilde t\Delta)-g(u^*))g''(u^*+\widetilde t\Delta)\Delta^2\right]\\
&\qquad=E\left[(g'(u^*+\widetilde t\Delta))^2\Delta^2+(g(u^*+\widetilde t\Delta)-g(u^*))g''(u^*+\widetilde t\Delta)\Delta^2\right]\\
&\qquad\overset{(i)}{\leq}2E[\Delta^2]=2E[ \bV^{\top}(\widehat{\bm{\gamma}}_{a}-\bm{\gamma}_{a}^{*})]^2,
\end{align*}
where (i) holds since, by \eqref{bound:f'_f''}, $(g'(u^*+\widetilde t\Delta))^2\leq1$ and $(g(u^*+\widetilde t\Delta)-g(u^*))g''(u^*+\widetilde t\Delta)\leq1$. Hence, 
\begin{align}
E[\widehat{\pi}_{a}(\bS_1)-\pi_{a}^{*}(\bS_1)]^2=E[g(\bV^{\top} \widehat{\bm{\gamma}}_{a})-g(\bV^{\top}\bm{\gamma}_{a}^{*})]^2\leq2E[ \bV^{\top}(\widehat{\bm{\gamma}}_{a}-\bm{\gamma}_{a}^{*})]^2.\label{bound:Dpi_VDgamma}
\end{align}
Then, from \eqref{A.30} and $\eqref{213}$, we have
\begin{align}
	E[\widehat{\pi}_{a}(\bS_1)-\pi_{a}^{*}(\bS_1)]^2
	\leq 2 \|E[\bV\bV^{\top}]\|_2 \|\widehat{\bm{\gamma}}_{a}-\bm{\gamma}_{a}^{*}\|_2^2
	=O_p\left(\frac{s_{\bm{\gamma}_{a}}\log(d_1)}{N}\right).
\end{align}

(b) Now, we show \eqref{3.38} and \eqref{3.32}. In part (b), the expectations are only taken w.r.t. the distribution of the new observations $\bS$.

By Lemma \ref{lem_eigen}, we know that the minimum and maximum eigenvalues of covariance matrix $E[\bU\bU^{\top}]$ satisfy
\begin{align*}
	\lambda_{\min}(E[\bU\bU^{\top}])\geq \kappa_{l}>0,\;\;	\lambda_{\max}(E[\bU\bU^{\top}])\leq 2\sigma_{u}^2<\infty, 
\end{align*} 
and $\bU$ is sub-Gaussian with $\|\bm{x}^\top\bU \|_{{\psi}_2}\leq \sigma_{u}\|\bm{x}\|_2$ for any $\bm{x}\in \mathbb{R}^{d+1}$.
Additionally, we also have $P(c_0^2\leq \phi''(\bU^{\top}\bm{\delta}_{a})\leq (1-c_0)^2)=1$ under Assumption 4. Repeating the same procedure as in part (a), we also have
\begin{align*}
	\|\widehat{\bm{\delta}}_{a}-\bm{\delta}_{a}^{*}\|_2^2= O_p\left(\frac{s_{\bm{\delta}_{a}}\log(d)}{N}\right),
\end{align*}
and
\begin{align*}
	E[\widehat{\rho}_{a}(\bS)-\rho_{a}^{*}(\bS)]^2
	&=E[g(\bU^{\top} \widehat{\bm{\delta}}_{a})-g(\bU^{\top}\bm{\delta}_{a}^{*})]^2 
	\leq2E[ \bU^{\top}(\widehat{\bm{\delta}}_{a}-\bm{\delta}_{a}^{*})]^2\\
	&\leq 2 \|E[\bU\bU^{\top}]\|_2\|\widehat{\bm{\delta}}_{a}-\bm{\delta}_{a}^{*}\|_2^2
	=O_p\left(\frac{s_{\bm{\delta}_{a}}\log(d)}{N}\right).
\end{align*}
\end{proof}

\begin{proof}[Proof of Lemma \ref{lem_conclusion}]
In this proof, the expectations are only taken w.r.t. the distribution of the new observations $\bS$ (or only $\bS_1$ if $\bS_2$ is not involved).
By Lemma \ref{lem_pi1}, we have $P(\mathcal{A})=1-o(1)$.
By Minkowski's inequality, we have
\begin{align*}
	\{E|\widehat{\pi}_{a}(\bS_1)|^{-r}\}^{\frac{1}{r}}
	=\{E|1+\exp(-\bV^{\top} \widehat{\bm{\gamma}}_{a})|^{r}\}^{\frac{1}{r}}
	\leq1+\{E|\exp(-\bV^{\top} \widehat{\bm{\gamma}}_{a})|^{r}\}^{\frac{1}{r}}.
\end{align*}
Under Assumption 4, we know that
\begin{align}\label{B.3}
	P\left(\frac{c_0}{1-c_0} \leq \exp(-\bV^{\top}\bm{\gamma}_{a}^*) \leq \frac{1-c_0}{c_0} \right)=1.
\end{align}
which implies that
\begin{align*}
	&\{E|\exp(-\bV^{\top} \widehat{\bm{\gamma}}_{a})|^{r}\}^{\frac{1}{r}}=\{E|\exp(-\bV^{\top}\bm{\gamma}_{a}^*)\exp(-\bV^{\top} (\widehat{\bm{\gamma}}_{a}-\bm{\gamma}_{a}^*))|^{r}\}^{\frac{1}{r}}\\
	&\qquad\leq\frac{1-c_0}{c_0}\{E|\exp(-\bV^{\top} (\widehat{\bm{\gamma}}_{a}-\bm{\gamma}_{a}^*))|^r\}^{\frac{1}{r}}.
\end{align*}
Hence, to prove $\{E|\widehat{\pi}_{a}(\bS_1)|^{-r}\}^{\frac{1}{r}}$ is bounded uniformly, i.e., bounded by a constant independent of $N$, it suffices to show $\{E|\exp(-r\bV^{\top} (\widehat{\bm{\gamma}}_{a}-\bm{\gamma}_{a}^*))|\}^{\frac{1}{r}}$ is bounded uniformly. 

Let $\mu=E[|\bV^{\top} (\widehat{\bm{\gamma}}_{a}-\bm{\gamma}_{a}^*)|]$.
By Lemma \ref{lem_eigen}, we have
\begin{align}
	\|\bV^\top(\widehat{\bm{\gamma}}_{a}-\bm{\gamma}_{a})\|_{\psi_2}\leq 2\sigma_{u}\|\widehat{\bm{\gamma}}_{a}-\bm{\gamma}_{a}\|_2.\label{bound:Vdgamma}
\end{align}	
Now, condition on the event $\mathcal{A}$, we have
\begin{align}\label{B.4}
	\mu\leq \sqrt{\pi_{a}}\sigma_{u},\quad\|\mu\|_{\psi_2}\leq (\log2)^{-1/2}\sqrt{\pi_{a}}\sigma_u,
\end{align}	
which follows from Lemma \ref{lemma:psi2norm}. Note that, in the above, the $\psi_2$-norm is defined through the probability measure of a new observation $\bS_1$.	
By basic properties of Orlicz norm $\|X+Y\|_{\psi_2}\leq \|X\|_{\psi_2}+\|Y\|_{\psi_2}$, we have
\begin{align*}
	\||\bV^{\top} (\widehat{\bm{\gamma}}_{a}-\bm{\gamma}_{a}^*)|-\mu\|_{\psi_2}
	\leq \|\bV^{\top}(\widehat{\bm{\gamma}}_{a}-\bm{\gamma}_{a}^*)\|_{\psi_2}+\|\mu\|_{\psi_2}
	\leq [1+(\log2)^{-1/2}\sqrt{\pi_{a}}]\sigma_u,
\end{align*}
and with it that the moment generating function can be bounded with
\begin{align*}
	E[\exp\{r(|\bV^{\top} (\widehat{\bm{\gamma}}_{a}-\bm{\gamma}_{a}^*)|-\mu)\}]\leq \exp\{2r^2[1+(\log2)^{-1/2}\sqrt{\pi_{a}}]^2\sigma_u^2\}.
\end{align*}	
Using \eqref{B.4}, we get that
\begin{align}\label{bound:rVDgamma}
	&\{E|\exp(-r\bV^{\top} (\widehat{\bm{\gamma}}_{a}-\bm{\gamma}_{a}^*))|\}^{\frac{1}{r}}\leq \{E|\exp(r|\bV^{\top} (\widehat{\bm{\gamma}}_{a}-\bm{\gamma}_{a}^*)|)|\}^{\frac{1}{r}}\\ 
	&\qquad \leq\exp\{\sqrt{\pi_{a}}\sigma_{u}+2r[1+(\log2)^{-1/2}\sqrt{\pi_{a}}]^2\sigma_u^2\},\nonumber
\end{align}		
which is bounded and hence $\{E|\widehat{\pi}_{a}(\bS_1)|^{-r}\}^{\frac{1}{r}}$ is bounded uniformly. Repeating the same procedure above, we can obtain that $\{E|\widehat{\pi}_{a}(\bS_1)|^{-2r}\}^{\frac{1}{2r}}$ is also bounded uniformly, which will be used later on in the proof.
By \eqref{B.3}, we have
\begin{align}
	&\left\{E\left|\frac{1}{\widehat{\pi}_{a}(\bS_1)}-\frac{1}{\pi_{a}^{*}(\bS_1)}\right|^{r}\right\}^{\frac{1}{r}}=\{E|\exp(-\bV^{\top}\bm{\gamma}_{a}^*)[\exp(-\bV^{\top} (\widehat{\bm{\gamma}}_{a}-\bm{\gamma}_{a}^*))-1]|^{r}\}^{\frac{1}{r}}\nonumber\\
	&\qquad\leq\frac{1-c_0}{c_0}\{E|\exp(-\bV^{\top} (\widehat{\bm{\gamma}}_{a}-\bm{\gamma}_{a}^*))-1|^{r}\}^{\frac{1}{r}}.\label{bound:Dpiinv}
\end{align}
For any $u\in\R$, by Taylor's theorem, $\exp(u)=1+\exp(tu)u$ with some $t\in(0,1)$. Hence, with some $t\in(0,1)$
\begin{align}
	&|\exp(-\bV^{\top} (\widehat{\bm{\gamma}}_{a}-\bm{\gamma}_{a}^*))-1|=\exp(-t\bV^{\top} (\widehat{\bm{\gamma}}_{a}-\bm{\gamma}_{a}^*))|\bV^{\top} (\widehat{\bm{\gamma}}_{a}-\bm{\gamma}_{a}^*)|\nonumber\\
	&\qquad\overset{(i)}{\leq}[1+\exp(-\bV^{\top} (\widehat{\bm{\gamma}}_{a}-\bm{\gamma}_{a}^*))]|\bV^{\top} (\widehat{\bm{\gamma}}_{a}-\bm{\gamma}_{a}^*)|,\label{bound:expVDgamma}
\end{align}
where (i) holds since for any $t\in(0,1)$ and $u\in\R$, $\exp(tu)\leq\exp(u)$ when $u>0$ and $\exp(tu)\leq\exp(0)=1$ when $u\leq0$, and it follows that $\exp(tu)\leq1+\exp(u)$.

Combining \eqref{bound:Dpiinv} and \eqref{bound:expVDgamma}, we have 
\begin{align*}
	&\left\{E\left|\frac{1}{\widehat{\pi}_{a}(\bS_1)}-\frac{1}{\pi_{a}^{*}(\bS_1)}\right|^{r}\right\}^{\frac{1}{r}}\leq\frac{1-c_0}{c_0}\{E|\exp(-\bV^{\top} (\widehat{\bm{\gamma}}_{a}-\bm{\gamma}_{a}^*))-1|^{r}\}^{\frac{1}{r}}\\
	&\qquad\leq\frac{1-c_0}{c_0}\left\{E\left|[1+\exp(-\bV^{\top} (\widehat{\bm{\gamma}}_{a}-\bm{\gamma}_{a}^*))]\bV^{\top} (\widehat{\bm{\gamma}}_{a}-\bm{\gamma}_{a}^*)\right|^r\right\}^{\frac{1}{r}}\\
	&\qquad\overset{(i)}{\leq}\frac{1-c_0}{c_0}\left\{E\left|\bV^{\top} (\widehat{\bm{\gamma}}_{a}-\bm{\gamma}_{a}^*)\right|^r\right\}^{\frac{1}{r}}\\
	&\qquad\qquad+\frac{1-c_0}{c_0}\left\{E\left|\exp(-\bV^{\top} (\widehat{\bm{\gamma}}_{a}-\bm{\gamma}_{a}^*))\bV^{\top}(\widehat{\bm{\gamma}}_{a}-\bm{\gamma}_{a}^*)\right|^r\right\}^{\frac{1}{r}}\\
	&\qquad\overset{(ii)}{\leq}\frac{1-c_0}{c_0}\left\{E\left|\bV^{\top} (\widehat{\bm{\gamma}}_{a}-\bm{\gamma}_{a}^*)\right|^r\right\}^{\frac{1}{r}}\\
	&\qquad\qquad+\frac{1-c_0}{c_0}\left\{E\left|\exp(-\bV^{\top} (\widehat{\bm{\gamma}}_{a}-\bm{\gamma}_{a}^*))\right|^{2r}\right\}^{\frac{1}{2r}}\left\{E\left|\bV^{\top}(\widehat{\bm{\gamma}}_{a}-\bm{\gamma}_{a}^*)\right|^{2r}\right\}^{\frac{1}{2r}},
\end{align*}
where (i) holds by the Minkowski inequality; (ii) holds by H{\"o}lder's inequality.

Recall the equation \eqref{bound:rVDgamma}, we know that $\{E|\exp(-\bV^{\top} (\widehat{\bm{\gamma}}_{a}-\bm{\gamma}_{a}^*))|^{2r}\}^{\frac{1}{2r}}$ is bounded uniformly.
In addition, recall the equation \eqref{bound:Vdgamma}, by Lemma \ref{lemma:psi2norm}, we have
\begin{align*}
	\{E|\bV^{\top} (\widehat{\bm{\gamma}}_{a}-\bm{\gamma}_{a}^*)|^{r}\}^{\frac{1}{r}}&=O(\|\widehat{\bm{\gamma}}_{a}-\bm{\gamma}_{a}^*\|_2) \overset{(i)}{=}O_p\left(\sqrt{\frac{s_{\bm{\gamma}_{a}}\log(d_1)}{N}}\right),
\end{align*}
where (i) holds by Lemma \ref{lem_pi1}.
Therefore, we obtain that
\begin{align}
	\left\{E\left|\frac{1}{\widehat{\pi}_{a}(\bS_1)}-\frac{1}{\pi_{a}^{*}(\bS_1)}\right|^{r}\right\}^{\frac{1}{r}}=O_p\left(\sqrt{\frac{s_{\bm{\gamma}_{a}}\log(d_1)}{N}}\right).\label{bound:Dpir}
\end{align}
Repeating the same procedure, we obtain that $\{E|\widehat{\rho}_{a}(\bS)|^{-r}\}^{\frac{1}{r}}$ is bounded uniformly and
\begin{align}
	\left\{E\left|\frac{1}{\widehat{\rho}_{a}(\bS)}-\frac{1}{\rho_{a}^{*}(\bS)}\right|^{r}\right\}^{\frac{1}{r}}=O_p\left(\sqrt{\frac{s_{\bm{\delta}_{a}}\log(d)}{N}}\right).\label{bound:Drhor}
\end{align}
Therefore, 
\begin{align*}
	&\left\{E\left|\frac{1}{\widehat{\pi}_{a}(\bS_1)\widehat{\rho}_{a}(\bS)}-\frac{1}{\pi_{a}^{*}(\bS_1)\rho_{a}^{*}(\bS)}\right|^{r}\right\}^{\frac{1}{r}}\\
	&\quad\overset{(i)}{\leq}\left\{E\left|\frac{1}{\widehat{\pi}_{a}(\bS_1)}\left(\frac{1}{\widehat{\rho}_{a}(\bS)}-\frac{1}{\rho_{a}^{*}(\bS)}\right)\right|^{r}\right\}^{\frac{1}{r}}+\left\{E\left|\frac{1}{\rho_{a}^{*}(\bS)}\left(\frac{1}{\widehat{\pi}_{a}(\bS_1)}-\frac{1}{\pi_{a}^{*}(\bS_1)}\right)\right|^{r}\right\}^{\frac{1}{r}}\\
	&\quad\overset{(ii)}{\leq}\{E|\widehat{\pi}_{a}(\bS_1)|^{-2r}\}^{\frac{1}{2r}}\left\{E\left|\frac{1}{\widehat{\rho}_{a}(\bS)}-\frac{1}{\rho_{a}^{*}(\bS)}\right|^{2r}\right\}^{\frac{1}{2r}}+\frac{1}{c_0}\left\{E\left|\frac{1}{\widehat{\pi}_{a}(\bS_1)}-\frac{1}{\pi_{a}^{*}(\bS_1)}\right|^{r}\right\}^{\frac{1}{r}}\\
	&\quad\overset{(iii)}{=}O_p\left(\sqrt{\frac{s_{\bm{\gamma}_{a}}\log(d_1)+s_{\bm{\delta}_{a}}\log(d)}{N}}\right).
\end{align*}
where (i) holds by the Minkowski inequality; (ii) holds by H{\"o}lder's inequality; (iii) holds by \eqref{bound:Dpir}, \eqref{bound:Drhor}, and the fact that $\{E|\widehat{\pi}_{a}(\bS_1)|^{-2r}\}^{\frac{1}{2r}}$ is bounded uniformly.
\end{proof}

\section{Asymptotic theory for general Dynamic Treatment Effect (DTE)}\label{sec:proof_gen}
In this section, we consider general nuisance estimators and general working models. Below we introduce some shorthand notations that increase the readability of the proofs. With a slight abuse of notation, 
$\nuhat_{c}(\cdot)=\nuhat_{c,-k}(\cdot)$, $\muhat_{c}(\cdot)=\muhat_{c,-k}(\cdot)$, $\pihat_{c}(\cdot)=\pihat_{c,-k}(\cdot)$, and $\rhohat_{c}(\cdot)=\rhohat_{c,-k}(\cdot)$ are estimates of the nuisance functions $\nu_c(\cdot)$, $\mu_c(\cdot)$, $\pi_c(\cdot)$, and $\rho_c(\cdot)$ using the training samples $\mathcal W_{-k}$. We also define $\widehat \Delta(\cdot)=\widehat \Delta_{-k}(\cdot) $ and $\widehat{\psi}_{c} (\cdot)=\widehat{\psi}_{c,-k} (\cdot)$ for each $c\in\{a,a'\}$ and $k=1,...,K$. We suppress the dependence on $k$ when possible. Note that we have $\widehat \Delta(W)= \widehat{\psi}_{a} (W) -\widehat{\psi}_{a'} (W) $, where for each $c\in\{a,a'\}$,
\begin{align}
	\widehat{\psi}_{c}(W)&:=\muhat_{c}(\bS_{1})+\mathbbm1_{\{A_{1}=c_1\}}\frac{\nuhat_{c}(\bS)-\muhat_{c}(\bS_{1}) }{\pihat_{c}(\bS_{1})}+\mathbbm1_{\{A_{1}=c_1, A_{2}=c_2\}}\frac{Y-\nuhat_{c}(\bS)}{\pihat_{c}(\bS_{1})\rhohat_{c}(\bS)}.\label{def:DR-scorehat}
\end{align}
Define $\Delta^*(W) = \psi_{a}^* (W)-\psi_{a'}^* (W)$
, where for each $c\in\{a,a'\}$, 
\begin{align}
{\psi}_c^*(W)&:=\mu_{c}^*(\bS_{1})+\mathbbm1_{\{A_{1}=c_1\}}\frac{ \nu_{c}^*(\bS)-\mu_{c}^*(\bS_{1}) }{\pi_{c}^*(\bS_{1})}+\mathbbm1_{\{A_{1}=c_1, A_{2}=c_2\}}\frac{ Y-\nu_{c}^*(\bS)}{\pi_{c}^* (\bS_{1})\rho_{c}^*(\bS)}.\label{def:psi-gen-star}
\end{align}
Define $\check{\theta}_\mathrm{gen}^{(k)}=n^{-1}\sum_{ i\in\mathcal I_k} \widehat \Delta(W_i)$, where $n=N/K=|\mathcal I_k|$. Then $\thetahat_\mathrm{gen}=K^{-1}\sum_{k=1}^{K}\check{\theta}_\mathrm{gen}^{(k)}$.
For each $k=1,...,K$, we divide $\check{\theta}_\mathrm{gen}^{(k)}-\theta$ into four terms $T_{1}, T_{2}, T_{3}, T_{4}$,
\begin{align}\label{representation}
	\check{\theta}_\mathrm{gen}^{(k)}-\theta=n^{-1}\sum_{ i\in\mathcal I_k} \widehat \Delta_{-k} (W_i)-\theta := T_{1}+T_{2}+T_{3}+T_{4},
\end{align}
where
\begin{align}
	T_{1}&:=E[\Delta^*(W)]-\theta,\label{def:T_1}\\
	T_{2}&:=T_{2}^{(k)}:=E[\widehat \Delta(W)-\Delta^*(W)],\label{def:T_2}\\
	T_{3}&:=T_{3}^{(k)}:=\frac{1}{n}\sum_{i\in \mathcal I_k}\Delta^*(W_i)-E[\Delta^*(W)],\label{def:T_3}\\
	T_{4}&:=T_{4}^{(k)}:=\frac{1}{n}\sum_{i\in \mathcal I_k}[\widehat \Delta(W_i)-\Delta^*(W_i)]-E[\widehat \Delta(W)-\Delta^*(W)].\label{def:T_4}
\end{align}

\subsection{Auxiliary Lemmas}
\begin{lemma}\label{lem_T1}
Suppose that at least one of $\mu_{a}^{*}(\cdot)$ and $\pi_{a}^{*}(\cdot)$ is correctly specified, and at least one of the models $\nu_{a}^{*}(\cdot)$ and $\rho_{a}^{*}(\cdot)$ is correctly specified. Let Assumption 1 hold. 
Then,
\begin{align}
	T_{1}=0,\label{45}
\end{align}
where $T_{1}$ is defined as \eqref{def:T_1}. 
\end{lemma}
\begin{lemma}\label{lem_T2_general}
(a) Suppose that at least one of $\mu_{a}^{*}(\cdot)$ and $\pi_{a}^{*}(\cdot)$ is correctly specified, and at least one of the models $\nu_{a}^{*}(\cdot)$ and $\rho_{a}^{*}(\cdot)$ is correctly specified. Let Assumptions 1, 4 and 5 hold. Then
\begin{align}
T_{2}&=O_p\biggr(b_Nc_N+a_Nd_N+b_N\mathbbm{1}_{\{\pi_{a}^{*}\neq\pi_{a}\}}+a_N\mathbbm{1}_{\{\rho_{a}^{*}\neq\rho_{a}\}}\label{22}\\
&\qquad\qquad+c_N\sqrt{E[\zeta^{2}+\varepsilon^{2}]}\mathbbm{1}_{\{\mu_{a}^{*}\neq\mu_{a}\}}+d_N\sqrt{E[\zeta^{2}]}\mathbbm{1}_{\{\nu_{a}^{*}\neq\nu_{a}\}}\biggr),\nonumber
\end{align}
where $T_{2}$ is defined as \eqref{def:T_2}. 	

(b) Suppose that all the nuisance models $\mu_a^*(\cdot)$, $\nu_a^*(\cdot)$, $\pi_a^*(\cdot)$, and $\rho_a^*(\cdot)$ are correctly specified. Let Assumptions 1 and 5 hold. Then
\begin{align}
T_{2}=O_p\left(b_Nc_N+a_Nd_N\right).\label{S.2}
\end{align}
\end{lemma}

\begin{lemma}\label{lem_T3}
(a) Suppose that at least one of $\mu_{a}^{*}(\cdot)$ and $\pi_{a}^{*}(\cdot)$ is correctly specified, and at least one of the models $\nu_{a}^{*}(\cdot)$ and $\rho_{a}^{*}(\cdot)$ is correctly specified. Let Assumptions 1, 4 hold. Then
\begin{align}
T_{3}=O_p\left(\frac{1}{\sqrt{N}}\left[\sqrt{E[\zeta^{2}]}+\sqrt{E[\varepsilon^{2}]}+\sqrt{E[\xi^2]}\right]\right),\label{39}
\end{align}	
where $\xi:=\mu_{a}(\bS_1)-\mu_{a'}(\bS_1)-\theta$ and $T_{3}$ is defined as \eqref{def:T_3}.

(b) Suppose that all the nuisance models $\mu_a^*(\cdot)$, $\nu_a^*(\cdot)$, $\pi_a^*(\cdot)$, and $\rho_a^*(\cdot)$ are correctly specified. Let Assumption 1 hold. Then we also have \eqref{39}.
\end{lemma} 

\begin{lemma}\label{lem_T4_general}
(a) Suppose that at least one of $\mu_{a}^{*}(\cdot)$ and $\pi_{a}^{*}(\cdot)$ is correctly specified, and at least one of the models $\nu_{a}^{*}(\cdot)$ and $\rho_{a}^{*}(\cdot)$ is correctly specified. Let Assumptions 1, 4 and 5 hold. Then
\begin{align}
T_{4}=O_p\left(\frac{1}{\sqrt{N}}\left[a_N+b_N+\sqrt{E[\zeta^{2}]}+\sqrt{E[\varepsilon^{2}]}\right]\right),\label{40}
\end{align}
where $T_{4}$ is defined as \eqref{def:T_4}. 

(b) Suppose that all the nuisance models $\mu_a^*(\cdot)$, $\nu_a^*(\cdot)$, $\pi_a^*(\cdot)$, and $\rho_a^*(\cdot)$ are correctly specified. Let Assumptions 1 and 5 hold. Then
\begin{align}
T_{4}=O_p\left(\frac{1}{\sqrt{N}}(a_N+b_N+c_N(\sqrt{E[\zeta^{2}]}+\sqrt{E[\varepsilon^{2}]})+d_N\sqrt{E[\zeta^{2}]})\right).\label{B.31}	
\end{align}
\end{lemma}

\begin{lemma}\label{lemma:sigma_mis} 
Suppose that at least one of $\mu_{a}^{*}(\cdot)$ and $\pi_{a}^{*}(\cdot)$ is correctly specified, and at least one of the models $\nu_{a}^{*}(\cdot)$ and $\rho_{a}^{*}(\cdot)$ is correctly specified. Let Assumption 1 hold. 

(a) Assume that $E[\mathbbm1_{\{A_{1}=a_1\}}(\mu_a(\bS_1)-\mu_a^*(\bS_1))^2]\leq C_\mu\sigma^2$, with some constant $C_\mu>0$. Then
\begin{align*}
	E[\zeta^2]+E[\varepsilon^2]+E[\xi^2]\leq\left(\frac{4}{c_0^2}+6C_\mu\right)\sigma^2,
\end{align*}
where $\sigma^2:=E(\Delta^*(W)-\theta)^2$.

(b) Let Assumption 2 hold. Then
$$E[\zeta^2]+E[\varepsilon^2]+E[\xi^2]\leq\left(\frac{1}{c_0^2}+2\sigma_\varepsilon^2\right)\sigma^2.$$
\end{lemma}

\begin{lemma}\label{lem_Lyap}
Suppose that all the nuisance models $\mu_a^*(\cdot)$, $\nu_a^*(\cdot)$, $\pi_a^*(\cdot)$, and $\rho_a^*(\cdot)$ are correctly specified. Let Assumption 1 hold. Then we have for some constants $t>0$ and $C_t>0$ possibly dependent with t, such that
\begin{align}
&\sigma^2:=E(\Delta^*(W)-\theta)^2=E(\Delta(W)-\theta)^2\geq E[\zeta^2]+E[\varepsilon^2]+E[\xi^2],\label{53}\\
&E|\Delta^*(W)-\theta|^{2+t}\leq \frac{2C_t}{c_0^{4+2t}}E\biggl[|\zeta|^{2+t}+|\varepsilon|^{2+t}+|\xi|^{2+t}\biggl].\label{B.36}
\end{align}
\end{lemma} 

\begin{lemma}\label{lem_sigmahat}
Suppose that all the nuisance models $\mu_a^*(\cdot)$, $\nu_a^*(\cdot)$, $\pi_a^*(\cdot)$, and $\rho_a^*(\cdot)$ are correctly specified. Let Assumption 1 hold.
Define $\sigmahat_\mathrm{gen}^2=N^{-1}\sum_{k=1}^{K}\sum_{i\in \mathcal I_k}(\widehat \Delta(W_i)-\thetahat_\mathrm{gen})^2$ and $\sigma^2:=E(\Delta^*(W)-\theta)^2$. If 
$$ \thetahat_\mathrm{gen}-\theta=O_p(\sigma/\sqrt{N}), \qquad\left[\frac{1}{n}\sum_{i\in \mathcal I_k}|\widehat \Delta(W_i)-\Delta^*(W_i)|^2\right]^{\frac{1}{2}}=o_p(\sigma),$$
for each $k\leq K$, and $[ {E|(\Delta^*(W)-\theta)|^{2+t}}]^{\frac{2}{2+t}}<C \sigma^2$ for some constant $C$, we have
\begin{align}
	\sigmahat_\mathrm{gen}^2-\sigma^2=o_p(\sigma^2).\label{consistency_varest}
\end{align}
\end{lemma}

\subsection{Proof of Theorem 6}
Recall the representation \eqref{representation}. By Lemmas \ref{lem_T1}, \ref{lem_T2_general}, \ref{lem_T3}, and \ref{lem_T4_general}, we have
\begin{align*}
	T_{1}&=0,\\
	T_{2}^{(k)}&=O_p\biggr(b_Nc_N+a_Nd_N+b_N\mathbbm{1}_{\{\pi_{a}^{*}\neq\pi_{a}\}}+a_N\mathbbm{1}_{\{\rho_{a}^{*}\neq\rho_{a}\}}\\
	&\qquad\qquad+c_N\sqrt{E[\zeta^{2}+\varepsilon^{2}]}\mathbbm{1}_{\{\mu_{a}^{*}\neq\mu_{a}\}}+d_N\sqrt{E[\zeta^{2}]}\mathbbm{1}_{\{\nu_{a}^{*}\neq\nu_{a}\}}\biggr),\nonumber\\
	T_{3}^{(k)}&=O_p\left(\frac{1}{\sqrt{N}}\left[\sqrt{E[\zeta^{2}]}+\sqrt{E[\varepsilon^{2}]}+\sqrt{E[\xi^2]}\right]\right),\\
	T_{4}^{(k)}&=O_p\left(\frac{1}{\sqrt{N}}\left[a_N+b_N+\sqrt{E[\zeta^{2}]}+\sqrt{E[\varepsilon^{2}]}\right]\right).
\end{align*}
Together with Lemma \ref{lemma:sigma_mis} and further assume that $E(\mu_a^*(\bS_1)-\mu_a(\bS_1))^2\leq C_\mu\sigma^2$ with some constant $C_\mu>0$, we obtain 
\begin{align*} 
	\thetahat_\mathrm{gen}-\theta&=K^{-1}\sum_{k=1}^K(T_{1}+T_{2}^{(k)}+T_{3}^{(k)}+T_{4}^{(k)})\\
	&=O_p\biggr(b_Nc_N+a_Nd_N+b_N\mathbbm{1}_{\{\pi_{a}^{*}\neq\pi_{a}\}}+a_N\mathbbm{1}_{\{\rho_{a}^{*}\neq\rho_{a}\}}\\
	&\qquad\qquad+c_N\sigma\mathbbm{1}_{\{\mu_{a}^{*}\neq\mu_{a}\}}+d_N\sigma\mathbbm{1}_{\{\nu_{a}^{*}\neq\nu_{a}\}}+\frac{1}{\sqrt{N}}\sigma\biggr).
\end{align*}

\subsection{Proof of Theorem 7}
In this theorem, we consider correctly specified nuisance models.
\subsubsection{Consistency} 
Recall the representation \eqref{representation}, by Lemmas \ref{lem_T1}, \ref{lem_T2_general}, \ref{lem_T3}, and \ref{lem_T4_general}, we have
\begin{align}
	T_{1}&=0,\label{rate:T_1}\\
	T_{2}^{(k)}&=O_p\left(b_Nc_N+a_Nd_N\right),\label{rate:T_2}\\
	T_{3}^{(k)}&=O_p\left(\frac{1}{\sqrt{N}}\left[\sqrt{E[\zeta^{2}]}+\sqrt{E[\varepsilon^{2}]}+\sqrt{E[\xi^2]}\right]\right),\nonumber\\
	T_{4}^{(k)}&=O_p\left(\frac{1}{\sqrt{N}}(a_N+b_N+c_N(\sqrt{E[\zeta^{2}]}+\sqrt{E[\varepsilon^{2}]})+d_N\sqrt{E[\zeta^{2}]})\right).\label{rate:T_4}
\end{align}
By assumption, $b_Nc_N+a_Nd_N=o(\sigma N^{-1/2})$. Together with Lemma \ref{lem_Lyap}, we obtain that 
\begin{align}
	\thetahat_\mathrm{gen}-\theta&=K^{-1}\sum_{k=1}^K(T_{1}+T_{2}^{(k)}+T_{3}^{(k)}+T_{4}^{(k)})\nonumber\\
	&=O_p\left(\frac{1}{\sqrt{N}}\left[\sqrt{E[\zeta^{2}]}+\sqrt{E[\varepsilon^{2}]}+\sqrt{E[\xi^2]}\right]+b_Nc_N+a_Nd_N\right)\nonumber\\
	&\qquad+O_p\left(\frac{1}{\sqrt{N}}(a_N+b_N+c_N(\sqrt{E[\zeta^{2}]}+\sqrt{E[\varepsilon^{2}]})+d_N\sqrt{E[\zeta^{2}]})\right)\nonumber\\
	&=O_p\left(\frac{1}{\sqrt{N}}\sigma\right).\label{54}
\end{align}

\subsubsection{Asymptotic Normality}
Now, we demonstrate that $\sqrt{N}\sigma^{-1}(\thetahat_\mathrm{gen}-\theta)\leadsto N(0,1)$. By \eqref{rate:T_1}, \eqref{rate:T_2}, and \eqref{rate:T_4}, under Assumption 5 and $b_Nc_N+a_Nd_N=o(\sigma N^{-1/2})$, we have
$$\sqrt{n}\sigma^{-1}(T_{1}+T_{2}^{(k)}+T_{4}^{(k)})=o_p(1)$$
for each $k\leq K$. Hence, we only need to show 
$$\sqrt{N}\sigma^{-1}K^{-1}\sum_{k=1}^KT_{3}^{(k)}=\sqrt N \sigma^{-1}\left(N^{-1}\sum_{i=1}^N\Delta^*(W_i)-\theta\right)\leadsto N(0,1),$$
where $T_{3}^{(k)}$ is defined as \eqref{def:T_3}. By Lyapunov's central limit theorem, it suffices to show the following Lyapunov's condition holds: with some $t>0$,
\begin{align}\label{D.40}
	\lim_{n \to \infty}\frac{E|\Delta^*(W)-\theta|^{2+t}}{n^{\frac{t}{2}}\sigma^{2+t}}=0.
\end{align}

\paragraph*{Step 1}
To check Lyapunov's condition, it suffices to show that for some constant $C'>0$,
\begin{align}\label{D.41}
	\frac{E|\Delta^*(W)-\theta|^{2+t}}{\sigma^{2+t}}<C'.
\end{align}
By Lemma \ref{lem_Lyap}, we have, for some constants $t>0$ and $C_t>0$,
\begin{align}
	&\frac{E|\Delta^*(W)-\theta|^{2+t}}{\sigma^{2+t}}\leq \frac{2C_t}{c_0^{4+2t}}\frac{E[|\zeta|^{2+t}+|\varepsilon|^{2+t}+|\xi|^{2+t}]}{(E[\zeta^2]+E[\varepsilon^2]+E[\xi^2])^{1+\frac{t}{2}}}\nonumber\\
	&\qquad\leq \frac{2C_t}{c_0^{4+2t}}\left(\frac{E[|\zeta|^{2+t}]}{(E[\zeta^2])^{1+\frac{t}{2}}}+\frac{E[|\varepsilon|^{2+t}]}{(E[\varepsilon^2])^{1+\frac{t}{2}}}+\frac{E[|\xi|^{2+t}]}{(E[\xi^2])^{1+\frac{t}{2}}}\right)\leq\frac{2CC_t}{c_0^{4+2t}},\label{B.39}
\end{align}
where the last inequality follows from Assumption 4.
Taking $C'=2CC_t/c_0^{4+2t}$, we get \eqref{D.40} and hence the Lyapunov's condition holds.

\paragraph*{Step 2}
By \eqref{54}, we have $ \thetahat_\mathrm{gen}-\theta=O_p(\sigma/\sqrt{N})$. Here, we show that, for each $k\leq K$,
\begin{align}\label{D.43}
	\left[\frac{1}{n}\sum_{i\in \mathcal I_k}|\widehat \Delta(W_i)-\Delta^*(W_i)|^2\right]^{\frac{1}{2}}=o_p(\sigma).
\end{align}
Note that
\begin{align}
	&E\biggl[\frac{1}{n}\sum_{i\in \mathcal I_k}|\widehat \Delta(W_i)-\Delta^*(W_i)|^2\biggl]^{\frac{1}{2}}\overset{(i)}{\leq}\biggl\{E\biggl[\frac{1}{n}\sum_{i\in \mathcal I_k}|\widehat \Delta(W_i)-\Delta^*(W_i)|^2\biggl]\biggl\}^{\frac{1}{2}}\label{E_notnew2}\\
	&\qquad\overset{(ii)}{=}[E|\widehat \Delta(W)-\Delta^*(W)|^2]^{\frac{1}{2}}\label{E_new1}\\
	&\qquad\overset{(iii)}{=}O_p\left(a_N+b_N+c_N(\sqrt{E[\zeta^{2}]}+\sqrt{E[\varepsilon^{2}]})+d_N\sqrt{E[\zeta^{2}]}\right),\nonumber
\end{align}
where in \eqref{E_notnew2}, the expectations are taken w.r.t. the joint distribution of $(W_i)_{i\in \mathcal I_k}$; in \eqref{E_new1}, the expectation is taken w.r.t. the joint distribution of a new $W$. In the above, (i) holds by Jensen's inequality; (ii) holds since the estimator of nuisance functions are independent of $\{W_i\}_{i\in \mathcal I_k}$ based on cross-fitting, $\{W_i\}_{i\in \mathcal I_k}$ are i.i.d. distributed and $W$ is an independent copy of them; (iii) holds by Lemma \ref{lem_T4_general}. 
By Markov's inequality, we have
\begin{align*}
	&\left[\frac{1}{n}\sum_{i\in \mathcal I_k}|\widehat \Delta(W_i)-\Delta^*(W_i)|^2\right]^{\frac{1}{2}}\\
	&\qquad=O_p\left(a_N+b_N+c_N(\sqrt{E[\zeta^{2}]}+\sqrt{E[\varepsilon^{2}]})+d_N\sqrt{E[\zeta^{2}]}\right)=o_p(\sigma).
\end{align*}
Together with \eqref{54}, \eqref{D.40}, \eqref{D.43}, and Lemma \ref{lem_sigmahat}, we conclude that 
$$\sigmahat_\mathrm{gen}^2-\sigma^2=o_p(\sigma^2).$$

\subsection{Proofs of Auxiliary Lemmas}

\begin{proof}[Proof of Lemma \ref{lem_T1}]
Recall the definition \eqref{def:T_1}. Since $\theta = \theta_a - \theta_{a'}$ and $ \Delta^*(W) = \psi_{a}^* (W)-\psi_{a'}^* (W)$, we have 
\begin{align*}
	T_{1}=(E[ \psi_{a}^*(W)]-\theta_a)-(E[\psi_{a'}^* (W)]-\theta_{a'}).
\end{align*}
By Lemma 1, we have $\theta_c=E[ \psi_{c}^*(W)]$ for each $c\in\{a,a'\}$. Therefore, $T_{1}=0$.
\end{proof}

\begin{proof}[Proof of Lemma \ref{lem_T2_general}]
In this proof, the expectations are taken w.r.t. the distribution of new observations $\bS$ (or only $\bS_1$ if $\bS_2$ is not involved). We only focus on the treatment paths $a=(1,1)$ and $a'=(0,0)$. We begin by decomposing $T_{2}$, \eqref{def:T_2}, as a sum of six terms
\begin{align}
	\widehat \Delta(W)-\Delta^*(W)=\sum_{i=1}^{6}Q_i,\label{eq:psi_Q}
\end{align}
where
\begin{align}
	Q_{1}&:=\frac{A_{1}A_2}{\widehat{\pi}_{a}(\bS_1)\widehat{\rho}_{a}(\bS)}(Y-\widehat{\nu}_{a}(\bS))-\frac{A_{1}A_2}{\pi_{a}^{*}(\bS_1)\rho_{a}^{*}(\bS)}(Y-\nu_{a}^{*}(\bS)),\label{Q_1}\\ 
	Q_2&:=\frac{A_{1}}{\widehat{\pi}_{a}(\bS_1)}(\widehat{\nu}_{a}(\bS)-\widehat{\mu}_{a}(\bS_1))-\frac{A_{1}}{\pi_{a}^{*}(\bS_1)}(\nu_{a}^{*}(\bS)-\mu_{a}^{*}(\bS_1)),\label{Q_2}\\
	Q_{3}&:=\widehat{\mu}_{a}(\bS_1)-\mu_{a}^{*}(\bS_1),\label{Q_3}\\
	Q_{4}&:=-\frac{(1-A_{1})(1-A_2)}{\widehat{\pi}_{a'}(\bS_1)\widehat{\rho}_{a'}(\bS)}(Y-\widehat{\nu}_{a'}(\bS))\nonumber\\
	&\quad+\frac{(1-A_{1})(1-A_2)}{\pi_{a'}^{*}(\bS_1)\rho_{a'}^{*}(\bS)}(Y-\nu_{a'}^{*}(\bS)),\label{Q_4}\\
	Q_{5}&:=-\frac{1-A_{1}}{\widehat{\pi}_{a'}(\bS_1)}(\widehat{\nu}_{a'}(\bS)-\widehat{\mu}_{a'}(\bS_1))+\frac{1-A_{1}}{\pi_{a'}^{*}(\bS_1)}(\nu_{a'}^{*}(\bS)-\mu_{a'}^{*}(\bS_1)),\label{Q_5}\\
	Q_{6}&:=-\widehat{\mu}_{a'}(\bS_1)+\mu_{a'}^{*}(\bS_1).\label{Q_6}
\end{align}
Hence, we have the following representation for $T_{2}$:
\begin{align}
	T_{2}=E[\widehat \Delta(W)-\Delta^*(W)]=\sum_{i=1}^6E[Q_i],\label{eq:T3_Q}
\end{align}
where the expecatations are only taken w.r.t. the distribution of the new obseravtion $W$.

(a) We first obtain an upper bound for $E[Q_1+Q_2+Q_3]$. By the tower rule,
\begin{align*}
	E[Q_1]&=E\left[\frac{A_{1}\rho_{a}(\bS)}{\widehat{\pi}_{a}(\bS_1)\widehat{\rho}_{a}(\bS)}(\nu_{a}(\bS)-\widehat{\nu}_{a}(\bS))-\frac{A_{1}\rho_{a}(\bS)}{\pi_{a}^{*}(\bS_1)\rho_{a}^{*}(\bS)}(\nu_{a}(\bS)-\nu_{a}^{*}(\bS))\right].
\end{align*}
Through rearranging, we have the following representation:
\begin{align}
	E[Q_1+Q_2+Q_3]=\sum_{i=1}^8R_i,\label{eq:Q_R}
\end{align}
where
\begin{align}
	R_1:=&E\left[\frac{A_{1}\rho_{a}^{*}(\bS)(\widehat{\nu}_{a}(\bS)-\nu_{a}^{*}(\bS))}{\pihat_{a}(\bS_1)}\left(\frac{1}{\rho_{a}^{*}(\bS)}-\frac{1}{\widehat\rho_{a}(\bS)}\right)\right],\label{def:R_1}\\
	R_2:=&E\left[\pi_{a}^{*}(\bS_1)(\widehat{\mu}_{a}(\bS_1)-\mu_{a}^{*}(\bS_1))\left(\frac{1}{\pi_{a}^*(\bS_1)}-\frac{1}{\widehat{\pi}_{a}(\bS_1)}\right)\right],\label{def:R_2}\\
	R_3:=&E\left[\frac{A_{1}(\rho_{a}^{*}(\bS)-\rho_{a}(\bS))(\widehat{\nu}_{a}(\bS)-\nu_{a}^{*}(\bS))}{\widehat{\pi}_{a}(\bS_1)\widehat{\rho}_{a}(\bS)}\right],\label{def:R_3}\\
	R_4:=&E\left[\frac{(\pi_{a}^*(\bS_1)-A_{1})(\widehat{\mu}_{a}(\bS_1)-\mu_{a}^{*}(\bS_1))}{\widehat{\pi}_{a}(\bS_1)}\right]\nonumber\\
	\overset{(i)}{=}&E\left[\frac{(\pi_{a}^*(\bS_1)-\pi_{a}(\bS_1))(\widehat{\mu}_{a}(\bS_1)-\mu_{a}^{*}(\bS_1))}{\widehat{\pi}_{a}(\bS_1)}\right],\label{def:R_4}\\
	R_5:=&E\left[\frac{A_{1}\rho_{a}^{*}(\bS)(\nu_{a}^{*}(\bS)-\nu_{a}(\bS))}{\widehat{\pi}_{a}(\bS_1)}\left(\frac{1}{\rho_{a}^{*}(\bS)}-\frac{1}{\widehat\rho_{a}(\bS)}\right)\right],\label{def:R_5}\\
	R_6:=&E\left[A_{1}(\mu_{a}^{*}(\bS_1)-\mu_{a}(\bS_1))\left(\frac{1}{\pi_{a}^*(\bS_1)}-\frac{1}{\widehat{\pi}_{a}(\bS_1)}\right)\right],\label{def:R_6}\\
	R_7:=&E\left[\left(\frac{A_{1}}{\widehat\pi_{a}(\bS_1)\widehat\rho_{a}(\bS)}-\frac{A_{1}}{\pi_{a}^*(\bS_1)\rho_{a}^*(\bS)}\right)(\rho_{a}^*(\bS)-\rho_{a}(\bS))(\nu_{a}^*(\bS)-\nu_{a}(\bS))\right]\nonumber\\
	\overset{(ii)}{=}&0,\label{def:R_7}\\
	R_8:=&E\left[\frac{A_{1}(\widehat{\pi}_{a}(\bS_1)-\pi_{a}^*(\bS_1))(\mu_{a}(\bS_1)-\nu_{a}(\bS))}{\widehat{\pi}_{a}(\bS_1)\pi_{a}^*(\bS_1)}\right]\overset{(iii)}{=}0.\label{def:R_8}
\end{align}
Here, (i) holds by the tower rule; (ii) holds since either $\rho_{a}^*(\cdot)=\rho_{a}(\cdot)$ or $\mu_{a}^*(\cdot)=\mu_{a}(\cdot)$ by assumption; (iii) holds by the tower rule and the fact that $\mu_{a}(\bS_{1})=E[\nu_{a}(\bS )|\bS_{1},A_1=a_1]$.
We condition on the following event
\begin{align}\label{def:E4}
	\mathcal E_4:=\left\{P(c_0\leq \widehat{\pi}_{a}(\bS_1)\leq1-c_0)=1,\;\;P(c_0\leq \widehat{\rho}_{a}(\bS)\leq1-c_0)=1\right\}.
\end{align}
Under Assumption 5, the event $\mathcal E_4$ occurs with probability approaching one. 
In the following, we use Cauchy-Schwarz inequality to obtain an upper bound $R_i$ $(i\in\{1,\dots,6\})$.
For $R_1+R_2$, on the event $\mathcal E_4$, we have 
\begin{align}
	R_1+R_2&\leq\frac{1}{c_0^2}[E(\widehat{\rho}_{a}(\bS)-\rho_{a}^{*}(\bS))^{2}]^{\frac{1}{2}}[E(\widehat{\nu}_{a}(\bS)-\nu_{a}^{*}(\bS))^{2}]^{\frac{1}{2}}\nonumber\\
	&\qquad+\frac{1}{c_0}[E(\widehat{\pi}_{a}(\bS_1)-\pi_{a}^{*}(\bS_1))^{2}]^{\frac{1}{2}}[E(\widehat{\mu}_{a}(\bS_1)-\mu_{a}^{*}(\bS_1))^{2}]^{\frac{1}{2}}\nonumber\\
	&=O_p\left(b_Nc_N+a_Nd_N\right),\label{29}
\end{align}
under Assumption 5. For $R_3+R_4$, on the event $\mathcal E_4$, we have
\begin{align*}
	R_3+R_4&\leq\frac{1}{c_0^2}[E(\rho_{a}^{*}(\bS)-\rho_{a}(\bS))^{2}]^{\frac{1}{2}}[E(\widehat{\nu}_{a}(\bS)-\nu_{a}^{*}(\bS))^{2}]^{\frac{1}{2}}\\
	&\qquad+\frac{1}{c_0}[E(\pi_{a}^{*}(\bS_1)-\pi_{a}(\bS_1))^{2}]^{\frac{1}{2}}[E(\widehat{\mu}_{a}(\bS_1)-\mu_{a}^{*}(\bS_1))^{2}]^{\frac{1}{2}}\\
	&\leq\frac{\mathbbm{1}_{\{\rho_{a}^{*}(\cdot)\neq\rho_{a}(\cdot)\}}}{c_0^2}[E(\widehat{\nu}_{a}(\bS)-\nu_{a}^{*}(\bS))^{2}]^{\frac{1}{2}}+\frac{\mathbbm{1}_{\{\pi_{a}^{*}(\cdot)\neq\pi_{a}(\cdot)\}}}{c_0}[E(\widehat{\mu}_{a}(\bS_1)-\mu_{a}^{*}(\bS_1))^{2}]^{\frac{1}{2}},
\end{align*}
since 
\begin{align*}
	&E(\rho_{a}^{*}(\bS)-\rho_{a}(\bS))^2=\mathbbm{1}_{\{\rho_{a}^{*}(\cdot)\neq\rho_{a}(\cdot)\}}E(\rho_{a}^{*}(\bS)-\rho_{a}(\bS))^2\overset{(i)}{\leq}\mathbbm{1}_{\{\rho_{a}^{*}(\cdot)\neq\rho_{a}(\cdot)\}},\\
	&E(\pi_{a}^{*}(\bS_1)-\pi_{a}(\bS_1))^2=\mathbbm{1}_{\{\pi_{a}^{*}(\cdot)\neq\pi_{a}(\cdot)\}}E(\pi_{a}^{*}(\bS_1)-\pi_{a}(\bS_1))^2\overset{(ii)}{\leq}\mathbbm{1}_{\{\pi_{a}^{*}(\cdot)\neq\pi_{a}(\cdot)\}},
\end{align*}
where (i) and (ii) hold because $\rho_{a}(\bS)=E(A_2|\bS,A_{1}=a_1)\in(0,1)$, $\pi_{a}(\bS_1)=E(A_{1}|\bS_1)\in(0,1)$, and, under Assumption 4, $\rho_{a}^*(\bS),\pi_{a}^*(\bS_1)\in(0,1)$ with probability one.
Hence, under Assumption 5, we have
\begin{align}
	R_3+R_4
	=O_p\left(b_N\mathbbm{1}_{\{\pi_{a}^{*}(\cdot)\neq\pi_{a}(\cdot)\}}+a_N\mathbbm{1}_{\{\rho_{a}^{*}(\cdot)\neq\rho_{a}(\cdot)\}}\right).\label{30}
\end{align}
As for $R_5+R_6$, similarly, we have
\begin{align}
	R_5+R_6&\leq \frac{1}{c_0^2}[E(\widehat{\rho}_{a}(\bS)-\rho_{a}^{*}(\bS))^{2}]^{\frac{1}{2}}\left[E[A_{1}(\nu_{a}^{*}(\bS)-\nu_{a}(\bS))^{2}]\right]^{\frac{1}{2}}\nonumber\\
	&\qquad+\frac{1}{c_0^2}[E(\widehat{\pi}_{a}(\bS_1)-\pi_{a}^{*}(\bS_1))^{2}]^{\frac{1}{2}}\left[E[A_{1}(\mu_{a}^{*}(\bS_1)-\mu_{a}(\bS_1))^{2}]\right]^{\frac{1}{2}}\label{30'}
\end{align}
Here, we need upper bound for $[E[A_{1}(\nu_{a}^{*}(\bS)-\nu_{a}(\bS))^{2}]]^{\frac{1}{2}}$ and $[E[A_{1}(\mu_{a}^{*}(\bS_1)-\mu_{a}(\bS_1))^{2}]]^{\frac{1}{2}}$. By definition,
\begin{align*}
	\zeta=\zeta_{a}+\zeta_{a'},\quad\varepsilon=\varepsilon_{a}+\varepsilon_{a'},\quad Y=Y(a)A_{1}A_2+Y(0,0)(1-A_{1})(1-A_2),
\end{align*} 
where
\begin{align*}
	\zeta_{a}=A_{1}A_2\left(Y(a)-\nu_{a}^{*}(\bS)\right),\quad\varepsilon_{a}=A_{1}\left(\nu_{a}^{*}(\bS)-\mu_{a}^{*}(\bS_1)\right).
\end{align*}
Hence, we have
\begin{align}
	E[\zeta^{2}]\geq E[A_{1}A_2\zeta^{2}]=E[\zeta_{a}^{2}]=E[A_{1}A_2(Y-\nu_{a}^{*}(\bS))^{2}]\label{23}
\end{align}
Note that
\begin{align*}
	&E[A_{1}A_2(Y-\nu_{a}(\bS))(\nu_{a}(\bS)-\nu_{a}^{*}(\bS))]\\
	&\quad\overset{(i)}{=}E[E[A_{1}A_2(Y(a)-\nu_{a}(\bS))(\nu_{a}(\bS)-\nu_{a}^{*}(\bS))|\bS, A_{1}=a_1]P(A_{1}=a_1|\bS)]\\
	&\quad\overset{(ii)}{=}E[E[A_2|\bS, A_{1}=a_1](E[Y(a)|\bS, A_{1}=a_1]-\nu_{a}(\bS))(\nu_{a}(\bS)-\nu_{a}^{*}(\bS))P(A_{1}=a_1|\bS)]\\
	&\quad\overset{(iii)}{=}0,
\end{align*}
where (i) holds by the tower rule and the fact that $A_{1}A_2Y=A_{1}A_2Y(a)$; (ii) holds under Assumption 1; (iii) holds since $\nu_{a}(\bS)=E[Y(a)|\bS, A_{1}=a_1]$. Therefore,
\begin{align}
	&E[A_{1}A_2(Y-\nu_{a}^{*}(\bS))^{2}]=E[A_{1}A_2[(Y-\nu_{a}(\bS))^{2}+(\nu_{a}(\bS)-\nu_{a}^{*}(\bS))^{2}]]\label{eq:Y-nustar}\\
	&\qquad\geq E[A_{1}A_2(\nu_{a}^{*}(\bS)-\nu_{a}(\bS))^{2}]\overset{(i)}{=}E[A_{1}\rho_{a}(\bS)(\nu_{a}^{*}(\bS)-\nu_{a}(\bS))^{2}]\nonumber\\
	&\qquad\overset{(ii)}{\geq} c_0E[A_{1}(\nu_{a}^{*}(\bS)-\nu_{a}(\bS))^{2}],\nonumber
\end{align}
where (i) holds by the tower rule; (ii) holds under Assumption 1. Together with \eqref{23}, we have
\begin{align}\label{25}
	E[A_{1}(\nu_{a}^{*}(\bS)-\nu_{a}(\bS))^{2}]\leq\frac{1}{c_0}E[\zeta^2].
\end{align}
Besides, note that
\begin{align*}
	&E[A_{1}(\nu_{a}(\bS)-\mu_{a}(\bS_1))(\mu_{a}(\bS_1)-\mu_{a}^*(\bS_1))]\\
	&\qquad=E[(\mu_{a}(\bS_1)-\mu_{a}^*(\bS_1))E[(\nu_{a}(\bS)-\mu_{a}(\bS_1))|\bS_1,A_{1}=a_1]P(A_{1}=a_1|\bS)]=0,
\end{align*}
since $E[\nu_{a}(\bS)|\bS_1,A_{1}=a_1]=\mu_{a}(\bS_1)$. Therefore, we have
\begin{align}
	E[A_{1}(\nu_{a}(\bS)-\mu_{a}^*(\bS_1))^2]&=E[A_{1}(\nu_{a}(\bS)-\mu_{a}(\bS_1))^2]+E[A_{1}(\mu_{a}(\bS_1)-\mu_{a}^*(\bS_1))^2]\nonumber\\
	&\geq E[A_{1}(\mu_{a}(\bS_1)-\mu_{a}^*(\bS_1))^2].\label{26}
\end{align}
Additionally, observe that
\begin{align*}
	E[A_{1}(\nu_{a}(\bS)-\mu_{a}^*(\bS_1))^2]&\leq2E[A_{1}(\nu_{a}^*(\bS)-\nu_{a}(\bS))^2]+2E[\varepsilon_{a}^2]\\
	&\overset{(i)}{\leq}\frac{2}{c_0}E[\zeta^2]+2E[A_{1}\varepsilon^2]\leq\frac{2}{c_0}E[\zeta^2]+2E[\varepsilon^2],
\end{align*}
where (i) holds by \eqref{25} and the fact that $\varepsilon_{a}^2=A_{1}\varepsilon^2$. Together with \eqref{26}, we obtain
\begin{align}
	E[A_{1}(\mu_{a}^*(\bS_1)-\mu_{a}(\bS_1))^2]\leq\frac{2}{c_0}E[\zeta^2]+2E[\varepsilon^2].\label{28}
\end{align}
Therefore, under Assumption 5,
\begin{align}
	R_5+R_6&\leq \frac{1}{c_0^2}[E(\widehat{\rho}_{a}(\bS)-\rho_{a}^{*}(\bS))^{2}]^{\frac{1}{2}}[E[A_{1}(\nu_{a}^{*}(\bS)-\nu_{a}(\bS))^{2}]]^{\frac{1}{2}}\nonumber\\
	&\qquad+\frac{1}{c_0^2}[E(\widehat{\pi}_{a}(\bS_1)-\pi_{a}^{*}(\bS_1))^{2}]^{\frac{1}{2}}[E[A_{1}(\mu_{a}^{*}(\bS_1)-\mu_{a}(\bS_1))^{2}]]^{\frac{1}{2}}\nonumber\\
	&= O_p\left(c_N\sqrt{E[\zeta^{2}+\varepsilon^{2}]}\mathbbm{1}_{\{\mu_{a}^{*}(\cdot)\neq\mu_{a}(\cdot)\}}+d_N\sqrt{E[\zeta^{2}]}\mathbbm{1}_{\{\nu_{a}^{*}(\cdot)\neq\nu_{a}(\cdot)\}}\right).\label{31}
\end{align}
Pluging \eqref{def:R_7}, \eqref{def:R_8}, \eqref{29},\eqref{30}, and \eqref{31} into \eqref{eq:Q_R}, we obtain
\begin{align*}
	E[Q_1+Q_2+Q_{3}]=&O_p\biggr(b_Nc_N+a_Nd_N+b_N\mathbbm{1}_{\{\pi_{a}^{*}(\cdot)\neq\pi_{a}(\cdot)\}}+a_N\mathbbm{1}_{\{\rho_{a}^{*}(\cdot)\neq\rho_{a}(\cdot)\}}\\
	&\qquad+c_N\sqrt{E[\zeta^{2}+\varepsilon^{2}]}\mathbbm{1}_{\{\mu_{a}^{*}(\cdot)\neq\mu_{a}(\cdot)\}}+d_N\sqrt{E[\zeta^{2}]}\mathbbm{1}_{\{\nu_{a}^{*}(\cdot)\neq\nu_{a}(\cdot)\}}\biggr).
\end{align*}
By repeating all the previous steps, we can obtain the same result for $E[Q_{4}+Q_{5}+Q_{6}]$. Therefore, \eqref{22} follows.

(b) When all the nuisance models are correct, Assumption 4 holds under Assumption 1. Hence, by part (a), we also have \eqref{22}. Since all the nuisance models are correct, we further conclude that \eqref{S.2} holds.
\end{proof}

\begin{proof}[Proof of Lemma \ref{lem_T3}]
(a) Recall the definition \eqref{def:T_3}. By Chebyshev's inequality, we have for any $t>0$,
\begin{align*}
	P(|T_{3}|>t)\leq\frac{1}{t^2} \mathrm{Var}\left(\frac{1}{n}\sum_{i\in \mathcal I_k}\Delta^*(W_i)\right) \leq\frac{1}{nt^2}E[\Delta^*(W)]^2.
\end{align*}
To prove \eqref{39}, we only need to show $[E(\Delta^*(W))^2]^{\frac{1}{2}}=O( \sqrt{E[\zeta^{2}]}+\sqrt{E[\varepsilon^{2}]}+\sqrt{E[\xi^2]})$.
By Minkowski inequality, we have
\begin{align}\label{38}
	&[E(\Delta^*(W))^2]^{\frac{1}{2}}\leq\sum_{i=1}^5T_{3,i},
\end{align}
where
\begin{align*}
	T_{3,1}:=&\left[E\left(\frac{A_{1}A_2}{\pi_{a}^{*}(\bS_1)\rho_{a}^{*}(\bS)}(Y-\nu_{a}^{*}(\bS))\right)^2\right]^{\frac{1}{2}},\\
	T_{3,2}:=&\left[E\left(\frac{A_{1}}{\pi_{a}^{*}(\bS_1)}(\nu_{a}^{*}(\bS)-\mu_{a}^{*}(\bS_1))\right)^2\right]^{\frac{1}{2}},\\
	T_{3,3}:=&\left[E\left(\frac{(1-A_{1})(1-A_2)}{\pi_{a'}^{*}(\bS_1)\rho_{a'}^{*}(\bS)}(Y-\nu_{a'}^{*}(\bS))\right)^2\right]^{\frac{1}{2}},\\
	T_{3,4}:=&\left[E\left(\frac{1-A_{1}}{\pi_{a'}^{*}(\bS_1)}(\nu_{a'}^{*}(\bS)-\mu_{a'}^{*}(\bS_1))\right)^2\right]^{\frac{1}{2}},\\
	T_{3,5}:=&\left[E\left(\mu_{a}^{*}(\bS_1)-\mu_{a'}^{*}(\bS_1)-\theta\right)^2\right]^{\frac{1}{2}}.
\end{align*}
We bound each of the above terms in turn. Under Assumption 4 and recall the equation \eqref{23}, we have
\begin{align}
	T_{3,1}&\leq\frac{1}{c_0^2}[E(A_{1}A_2(Y-\nu_{a}^{*}(\bS))^2)]^{\frac{1}{2}}\leq\frac{1}{c_0^2}\sqrt{E[\zeta^2]}.\label{33}
\end{align}
Similarly, since $E[\varepsilon^{2}]\geq E[A_{1}\varepsilon^{2}]=E[\varepsilon_{a}^{2}]=E[A_{1}(\nu_{a}^{*}(\bS)-\mu_{a}^{*}(\bS_1))^{2}]$, we have
\begin{align}\label{34}
	T_{3,2}&\leq\frac{1}{c_0}[E(A_{1}(\nu_{a}^{*}(\bS)-\mu_{a}^{*}(\bS_1))^2)]^{\frac{1}{2}}\leq\frac{1}{c_0}\sqrt{E[\varepsilon^2]}.
\end{align}
Repeating the same process for $T_{3,3}$ and $T_{3,4}$, we also have
\begin{align}
	T_{3,3}\leq\frac{1}{c_0^2}\sqrt{E[\zeta^2]},\quad T_{3,4}\leq\frac{1}{c_0}\sqrt{E[\varepsilon^2]}.\label{36}
\end{align}
Additionally,
\begin{align*}
	&\frac{2}{c_0}E[\zeta^2]+2E[\varepsilon^2]\overset{(i)}{\geq}E[A_{1}(\mu_{a}^*(\bS_1)-\mu_{a}(\bS_1))^2]\overset{(ii)}{=}E[\pi_{a}(\bS_1)(\mu_{a}^*(\bS_1)-\mu_{a}(\bS_1))^2]\\
	&\qquad\overset{(iii)}{\geq}c_0E[(\mu_{a}^*(\bS_1)-\mu_{a}(\bS_1))^2],
\end{align*}
where (i) holds by \eqref{28}; (ii) holds by the tower rule; (iii) holds under Assumption 1. Similarly, we also have
$$\frac{2}{c_0}E[\zeta^2]+2E[\varepsilon^2]\geq c_0E[(\mu_{a'}^*(\bS_1)-\mu_{a'}(\bS_1))^2].$$
By Minkowski inequality,
\begin{align}
	T_{3,5}&\leq [E(\mu_{a}^{*}(\bS_1)-\mu_{a}(\bS_1))^2]^{\frac{1}{2}}+[E(\mu_{a'}^{*}(\bS_1)-\mu_{a'}(\bS_1))^2]^{\frac{1}{2}}+[E[\xi^2]]^{\frac{1}{2}}\nonumber\\
	&\leq 2\sqrt{\frac{2}{c_0^2}E[\zeta^2]+\frac{2}{c_0}E[\varepsilon^2]}+\sqrt{E[\xi^2]}\leq\frac{2\sqrt2}{c_0}\sqrt{E[\zeta^{2}]}+\frac{2\sqrt2}{\sqrt{c_0}}\sqrt{E[\varepsilon^{2}]}+\sqrt{E[\xi^2]}.\label{37}
\end{align}
Plugging \eqref{33}-\eqref{37} into \eqref{38}, we have
\begin{align*}
	[E(\Delta^*(W))^2]^{\frac{1}{2}}=O\left( \sqrt{E[\zeta^{2}]}+\sqrt{E[\varepsilon^{2}]}+\sqrt{E[\xi^2]}\right).
\end{align*}

(b) When all the models are correctly specified, Assumption 1 implies Assumption 4. Hence, by part (a), we also have \eqref{39}.
\end{proof}

\begin{proof}[Proof of Lemma \ref{lem_T4_general}]
In this proof, the expectations are taken w.r.t. the distribution of new observations $\bS$ (or only $\bS_1$ if $\bS_2$ is not involved).
Additionally, we condition on the event $\mathcal E_4$, defined as \eqref{def:E4}.
 Under Assumption 5, such an event occurs with probability approaching one. 

(a) We first show \eqref{40}. Recall the representation \eqref{eq:psi_Q}, by Minkowski inequality, we have 
\begin{align*}
	[E(\widehat \Delta(W)-\Delta^*(W))^2]^{\frac{1}{2}}\leq \sum_{i=1}^{6}[E(Q_{i}^2)]^{\frac{1}{2}},
\end{align*}	
where $Q_i$ $(i\in\{1,\dots,6\})$ are defined as\eqref{Q_1}-\eqref{Q_6}. Then, by Chebyshev's inequality, it suffices to show 
$$\sum_{i=1}^{6}[E(Q_{i}^2)]^{\frac{1}{2}}=O_p\left(a_N+b_N+\sqrt{E[\zeta^{2}]}+\sqrt{E[\varepsilon^{2}]}\right).$$ 
Additionally, under Assumption 4, we also have
$$P(c_0\leq \pi_{a}^*(\bS_1)\leq1-c_0)=1,\;\;P(c_0\leq \rho_{a}^*(\bS)\leq1-c_0)=1.$$
For the first term $[E(Q_1^2)]^{\frac{1}{2}}$, under Assumption 4 and on the event $\mathcal E_4$,
\begin{align}
&\ \ \ \ [E(Q_1^2)]^{\frac{1}{2}}\nonumber\\
&\leq \frac{1}{c_0^4}\{E[A_{1}A_2\pi_{a}^{*}(\bS_1)\rho_{a}^{*}(\bS)(Y-\widehat{\nu}_{a}(\bS)) -A_{1}A_2\widehat{\pi}_{a}(\bS_1)\widehat{\rho}_{a}(\bS)(Y-\nu_{a}^{*}(\bS))]^2\}^{\frac{1}{2}}\nonumber\\
&\overset{(i)}{\leq}\frac{1}{c_0^4}\{E[\pi_{a}^{*}(\bS_1)\rho_{a}^{*}(\bS)(\nu_{a}^{*}(\bS)+\zeta-\widehat{\nu}_{a}(\bS))-\widehat{\pi}_{a}(\bS_1)\widehat{\rho}_{a}(\bS)\zeta]^2\}^{\frac{1}{2}}\nonumber\\
&\overset{(ii)}{\leq}\frac{1}{c_0^4}\{E[\widehat{\nu}_{a}(\bS)-\nu_{a}^{*}(\bS)]^2\}^{\frac{1}{2}}+\frac{1}{c_0^4}\{E[(\widehat{\pi}_{a}(\bS_1)\widehat{\rho}_{a}(\bS)-\pi_{a}^{*}(\bS_1)\rho_{a}^{*}(\bS))\zeta]^2\}^{\frac{1}{2}},\label{41}
\end{align}
where (i) holds by the fact that $|A_{1}|\leq 1$, $|A_2|\leq 1$ and 
$A_{1}A_2Y=A_{1}A_2\nu_{a}^{*}(\bS)+A_{1}A_2\zeta;$ (ii) holds from Minkowski inequality and the fact that $P(\pi_{a}^{*}(\bS_1)\rho_{a}^{*}(\bS)\leq 1)=1$. Since $P(0\leq\pi_{a}^{*}(\bS_1)\rho_{a}^{*}(\bS)\leq 1)=1$ and $P(0\leq\widehat{\pi}_{a}(\bS_1)\widehat{\rho}_{a}(\bS)\leq 1)=1$ under $\mathcal E_4$, we have
\begin{align}
[E(Q_1^2)]^{\frac{1}{2}}
&\leq \frac{1}{c_0^4}[E(\widehat{\nu}_{a}(\bS)-\nu_{a}^{*}(\bS))^2]^{\frac{1}{2}}+\frac{1}{c_0^4}[E(\zeta^2)]^{\frac{1}{2}}=O_p\left(b_N+\sqrt{E[\zeta^2]}\right).\label{B.27}
\end{align}
Similarly, for the second term $[E(Q_2^2)]^{\frac{1}{2}}$, under Assumption 4 and on the event $\mathcal E_4$,
\begin{align}
[E(Q_2^2)]^{\frac{1}{2}}&\leq \frac{1}{c_0^2}\{E[A_{1}\pi_{a}^{*}(\bS_1)(\widehat{\nu}_{a}(\bS)-\widehat{\mu}_{a}(\bS_1))-A_{1}\widehat{\pi}_{a}(\bS_1)(\nu_{a}^{*}(\bS)-\mu_{a}^{*}(\bS_1))]^2\}^{\frac{1}{2}}\nonumber\\
&\overset{(i)}{\leq}\frac{1}{c_0^2}\{E[\pi_{a}^{*}(\bS_1)(\widehat{\nu}_{a}(\bS)-\widehat{\mu}_{a}(\bS_1)) -\widehat{\pi}_{a}(\bS_1)\varepsilon]^2\}^{\frac{1}{2}}\nonumber\\
&\overset{(ii)}{\leq}\frac{1}{c_0^2}[E(\widehat{\nu}_{a}(\bS)-\nu_{a}^{*}(\bS))^2]^{\frac{1}{2}}+\frac{1}{c_0^2}[E(\widehat{\mu}_{a}(\bS_1)-\mu_{a}^{*}(\bS_1))^2]^{\frac{1}{2}}\nonumber\\
&\qquad +\frac{1}{c_0^2}\{E[(\widehat{\pi}_{a}(\bS_1)-\pi_{a}^{*}(\bS_1))\varepsilon]^2\}^{\frac{1}{2}}\label{42}\\
&\overset{(iii)}{\leq}\frac{1}{c_0^2}[E(\widehat{\nu}_{a}(\bS)-\nu_{a}^{*}(\bS))^2]^{\frac{1}{2}}+\frac{1}{c_0^2}[E(\widehat{\mu}_{a}(\bS_1)-\mu_{a}^{*}(\bS_1))^2]^{\frac{1}{2}}+\frac{1}{c_0^2}\{E[\varepsilon^2]\}^{\frac{1}{2}}\nonumber\\
&=O_p\left(a_N+b_N+\sqrt{E[\varepsilon^2]}\right),\label{B.29}
\end{align}
where (i) holds from the fact that $|A_{1}|\leq 1$ and 
$A_{1}\nu_{a}^{*}(\bS)=A_{1}\mu_{a}^{*}(\bS_1)+A_{1}\varepsilon;$ (ii) holds from Minkowski inequality and $P(\pi_{a}^{*}(\bS_1)\leq 1)=1$; (iii) holds by the fact that $P(0\leq\pi_{a}^{*}(\bS_1)\leq 1)=1$ and $P(0\leq\widehat{\pi}_{a}(\bS_1)\leq 1)=1$ on $\mathcal E_4$. For the third term $[E(Q_{3}^2)]^{\frac{1}{2}}$, we have
\begin{align}
[E(Q_{3}^2)]^{\frac{1}{2}}=O_p\left(b_N\right),\label{B.30}
\end{align}
under Assumption 5. Combining \eqref{B.27}, \eqref{B.29} and \eqref{B.30}, we obtain that
\begin{align*}
[E(Q_1^2)]^{\frac{1}{2}}+[E(Q_2^2)]^{\frac{1}{2}}+[E(Q_{3}^2)]^{\frac{1}{2}}=O_p\left(a_N+b_N+\sqrt{E[\zeta^{2}]}+\sqrt{E[\varepsilon^{2}]}\right).
\end{align*}
Repeating the same procedure above, we also have the same result for $[E(Q_{4}^2)]^{\frac{1}{2}}+[E(Q_{5}^2)]^{\frac{1}{2}}+[E(Q_{6}^2)]^{\frac{1}{2}}$. Then, \eqref{40} follows.

(b) Now, we show \eqref{B.31}. By \eqref{41}, under Assumption 5, we have
\begin{align*}
[E(Q_1^2)]^{\frac{1}{2}}&\leq \frac{1}{c_0^4}[E(\widehat{\nu}_{a}(\bS)-\nu_{a}(\bS))^2]^{\frac{1}{2}}\\
&\qquad+\frac{1}{c_0^4}\{E[\zeta^{2}|\bS]\}^{\frac{1}{2}}\{E[(\widehat{\pi}_{a}(\bS_1)\widehat{\rho}_{a}(\bS)-\pi_{a}(\bS_1)\rho_{a}(\bS))]^2\}^{\frac{1}{2}}\\
&\leq \frac{1}{c_0^4}[E(\widehat{\nu}_{a}(\bS)-\nu_{a}(\bS))^2]^{\frac{1}{2}}+\frac{\sqrt{CE[\zeta^2]}}{c_0^4}\{E[(\widehat{\pi}_{a}(\bS_1)\widehat{\rho}_{a}(\bS)-\pi_{a}(\bS_1)\rho_{a}(\bS))]^2\}^{\frac{1}{2}}
\end{align*}
By Minkowski inequality and under $\mathcal E_4$, we have
\begin{align*}
&\{E[\widehat{\pi}_{a}(\bS_1)\widehat{\rho}_{a}(\bS)-\pi_{a}(\bS_1)\rho_{a}(\bS)]^2\}^{\frac{1}{2}}\\
&\qquad\leq\{E[(\widehat{\pi}_{a}(\bS_1)-\pi_{a}(\bS_1))\widehat{\rho}_{a}(\bS)]^2\}^{\frac{1}{2}}+\{E[\pi_{a}(\bS_1)(\widehat{\rho}_{a}(\bS)-\rho_{a}(\bS))]^2\}^{\frac{1}{2}}\\
&\qquad\leq [E(\widehat{\pi}_{a}(\bS_1)-\pi_{a}(\bS_1))^2]^{\frac{1}{2}}+[E(\widehat{\rho}_{a}(\bS)-\rho_{a}(\bS))^2]^{\frac{1}{2}}=O_p\left(c_N+d_N\right).
\end{align*}
Hence,
$$[E(Q_1^2)]^{\frac{1}{2}}=O_p\left(a_N+(c_N+d_N)\sqrt{E[\zeta^{2}]}\right).$$
In addition, by \eqref{42}, under Assumption 5, we have
\begin{align*}
[E(Q_2^2)]^{\frac{1}{2}}&\leq\frac{1}{c_0^2}[E(\widehat{\nu}_{a}(\bS)-\nu_{a}(\bS))^2]^{\frac{1}{2}}+\frac{1}{c_0^2}[E(\widehat{\mu}_{a}(\bS_1)-\mu_{a}(\bS_1)]^2]^{\frac{1}{2}}\\
&\qquad +\frac{1}{c_0^2}\{E[\varepsilon^{2}|\bS_1]\}^{\frac{1}{2}}\{E[(\widehat{\pi}_{a}(\bS_1)-\pi_{a}(\bS_1))]^2\}^{\frac{1}{2}}\\
&\leq\frac{1}{c_0^2}[E(\widehat{\nu}_{a}(\bS)-\nu_{a}(\bS))^2]^{\frac{1}{2}}+\frac{1}{c_0^2}[E(\widehat{\mu}_{a}(\bS_1)-\mu_{a}(\bS_1)]^2]^{\frac{1}{2}}\\
&\qquad+\frac{\sqrt{CE[\varepsilon^2]}}{c_0^2}\{E[(\widehat{\pi}_{a}(\bS_1)-\pi_{a}(\bS_1))]^2\}^{\frac{1}{2}}\\
&=O_p\left(a_N+b_N+c_N\sqrt{E[\varepsilon^{2}]}\right).
\end{align*}
Besides, by Assumption 5,
$$[E(Q_{3}^2)]^{\frac{1}{2}}=O_p\left(b_N\right).$$
Repeating the same procedure above, we also have
\begin{align*}
[E(Q_{4}^2)]^{\frac{1}{2}}=&O_p\left(a_N+(c_N+d_N)\sqrt{E[\zeta^{2}]}\right),\\
[E(Q_{5}^2)]^{\frac{1}{2}}=&O_p\left(a_N+b_N+c_N\sqrt{E[\varepsilon^{2}]}\right),\\
[E(Q_{6}^2)]^{\frac{1}{2}}=&O_p\left(b_N\right).
\end{align*}
Now, we have
$$[E(\widehat \Delta(W)-\Delta^*(W))^2]^{\frac{1}{2}}=O_P\left(a_N+b_N+c_N(\sqrt{E[\zeta^{2}]}+\sqrt{E[\varepsilon^{2}]})+d_N\sqrt{E[\zeta^{2}]}\right).$$
By Chebyshev's inequality, we conclude that \eqref{B.31} holds.
\end{proof}

\begin{proof}[Proof of Lemma \ref{lemma:sigma_mis}]
(a) We notice the following representation:
\begin{align}
\Delta^*(W)-\theta=\sum_{i=1}^8 O_i,\label{eq:psi_O}
\end{align}
where
\begin{align}
O_1&:=\frac{A_{1}A_2(Y-\nu_{a}(\bS))}{\pi_{a}^*(\bS_1)\rho_{a}^*(\bS)},\label{def:O_1}\\
O_2&:=\frac{A_{1}}{\pi_{a}^*(\bS_1)}\left(1-\frac{A_2}{\rho_{a}^*(\bS)}\right)(\nu_{a}^*(\bS)-\nu_{a}(\bS)),\\
O_3&:=\frac{A_{1}(\nu_{a}(\bS)-\mu_{a}(\bS_1))}{\pi_{a}^*(\bS_1)},\\
O_4&:=-\frac{(1-A_{1})(1-A_2)(Y-\nu_{a'}(\bS))}{\pi_{a'}^*(\bS_1)\rho_{a'}^*(\bS)},\\
O_5&:=-\frac{1-A_{1}}{\pi_{a'}^*(\bS_1)}\left(1-\frac{1-A_2}{\rho_{a'}^*(\bS)}\right)(\nu_{a'}^*(\bS)-\nu_{a'}(\bS)),\\
O_6&:=-\frac{(1-A_{1})(\nu_{a'}(\bS)-\mu_{a'}(\bS))}{\pi_{a'}^*(\bS_1)},\\
O_7&:=\left(1-\frac{A_{1}}{\pi_{a}^*(\bS_1)}\right)(\mu_{a}^*(\bS_1)-\mu_{a}(\bS_1))\nonumber\\
&\qquad-\left(1-\frac{1-A_{1}}{\pi_{a'}^*(\bS_1)}\right)(\mu_{a'}^*(\bS_1)-\mu_{a'}(\bS_1)),\\
O_8&:=\mu_{a}(\bS_1)-\mu_{a'}(\bS_1)-\theta=\xi.\label{def:O_8}
\end{align}
In the following, we demonstrate that
\begin{align}\label{eq:sigma_O}
\sigma^2=E(\Delta^*(W)-\theta)^2=\sum_{i=1}^8E[O_i^2].
\end{align}
It suffices to show that $E[O_iO_j]=0$ for all $i\neq j$. Firstly, since $A_{1}(1-A_{1})=0$, we have 
\begin{align}
O_iO_j=0,\;\;\text{for each}\;\;i\in\{1,2,3\},\;\;\text{and}\;\;j\in\{4,5,6\}.\label{O123O456}
\end{align}

\paragraph*{Step 1} We show $E[O_1O_i]=0$ for each $i\geq2$. By \eqref{O123O456}, we know that $O_1O_i=0$ for $i\in\{4,5,6\}$. Note that, $O_3,O_7,O_8$ are all functions of $(\bS,A_{1})$. Hence, for each $i\in\{3,7,8\}$,
$$E[O_1O_i]=E[O_iE[O_1|\bS,A_{1}=a_1]P(A_{1}=a_1|\bS)]=0,$$
since
$$E[O_1|\bS,A_{1}=a_1]\overset{(i)}{=}\frac{E[A_2|\bS,A_{1}=a_1]E[Y(a)-\mu_{a}(\bS_1)|\bS,A_{1}=a_1]}{\pi_{a}^*(\bS_1)\rho_{a}^*(\bS)}\overset{(ii)}{=}0,$$
where (i) holds under Assumption 1; (ii) holds because $E[Y(a)|\bS,A_{1}=a_1]=\mu_{a}(\bS_1)$. Besides, we note that
\begin{align*}
&E[O_1O_2]=E\left[\frac{A_{1}A_2(Y-\nu_{a}(\bS))(\nu_{a}^*(\bS)-\nu_{a}(\bS))(\rho_{a}^*(\bS)-1)}{(\pi_{a}^*(\bS_1)\rho_{a}^*(\bS))^2}\right]\\
&\quad\overset{(i)}{=}E\left[\frac{E[A_2(Y(a)-\nu_{a}(\bS))|\bS,A_{1}=a_1](\nu_{a}^*(\bS)-\nu_{a}(\bS))(\rho_{a}^*(\bS)-1)}{(\pi_{a}^*(\bS_1)\rho_{a}^*(\bS))^2}P(A_{1}=a_1|\bS)\right]\\
&\quad\overset{(ii)}{=}E\left[\frac{\rho_{a}(\bS)E[Y(a)-\nu_{a}(\bS)|\bS,A_{1}=a_1](\nu_{a}^*(\bS)-\nu_{a}(\bS))(\rho_{a}^*(\bS)-1)}{(\pi_{a}^*(\bS_1)\rho_{a}^*(\bS))^2}P(A_{1}=a_1|\bS)\right]\\
&\quad\overset{(iii)}{=}0,
\end{align*}
where (i) holds by the tower rule; (ii) holds under Assumption 1; (iii) holds because $E[Y(a)|\bS,A_{1}=a_1]=\mu_{a}(\bS_1)$.

\paragraph*{Step 2} We show $E[O_2O_i]=0$ for each $i\geq3$. By \eqref{O123O456}, we know that $O_2O_i=0$ for $i\in\{4,5,6\}$. Since $O_3,O_7,O_8$ are all functions of $(\bS,A_{1})$, it follows that, for each $i\in\{3,7,8\}$,
\begin{align*}
E[O_2O_i]=E[O_iE[O_2|\bS,A_{1}=a_1]P(A_{1}=a_1|\bS)]=0,
\end{align*}
since
\begin{align*}
&E[O_2|\bS,A_{1}=a_1]=\frac{\nu_{a}^*(\bS)-\nu_{a}(\bS)}{\pi_{a}^*(\bS_1)}\left(1-\frac{E[A_2|\bS,A_{1}=a_1]}{\rho_{a}^*(\bS)}\right)\\
&\qquad=\frac{\nu_{a}^*(\bS)-\nu_{a}(\bS)}{\pi_{a}^*(\bS_1)}\left(1-\frac{\rho_{a}(\bS)}{\rho_{a}^*(\bS)}\right)\overset{(i)}{=}0,
\end{align*}
where (i) holds because either $\nu_{a}^*(\cdot)=\nu_1(\cdot)$ or $\rho_{a}^*(\cdot)=\rho_{a}(\cdot)$ by assumption. 

\paragraph*{Step 3} We show $E[O_3O_i]=0$ for each $i\geq4$. By \eqref{O123O456}, we know that $O_3O_i=0$ for $i\in\{4,5,6\}$. Since $O_7,O_8$ are all functions of $(\bS_1,A_{1})$, it follows that, for each $i\in\{7,8\}$,
\begin{align*}
E[O_3O_i]=E[O_iE[O_3|\bS_1,A_{1}=a_1]P(A_{1}=a_1|\bS_1)]=0,
\end{align*}
since
\begin{align*}
E[O_3|\bS_1,A_{1}=a_1]=\frac{E[\nu_{a}(\bS)|\bS_1,A_{1}=a_1]-\mu_{a}(\bS_1)}{\pi_{a}^*(\bS_1)}\overset{(i)}{=}0,
\end{align*}
where (i) holds because $E[\nu_{a}(\bS)|\bS_1,A_{1}=a_1]=\mu_{a}(\bS_1)$.

\paragraph*{Step 4} By repeating the same procedure as in Steps 1-3, we also have $E[O_iO_j]=0$ for each $i\in\{4,5,6\}$ and $j\geq i+1$.

\paragraph*{Step 5} We show $E[O_7O_8]=0$. Since $O_8$ is a function of $\bS_1$, we have
$$E[O_7O_8]=E[O_8E[O_7|\bS_1]]=0,$$
since
\begin{align*}
&E[O_7|\bS_1]=\left(1-\frac{\pi_{a}(\bS)}{\pi_{a}^*(\bS_1)}\right)(\mu_{a}^*(\bS_1)-\mu_{a}(\bS_1))-\left(1-\frac{\pi_{a'}(\bS)}{\pi_{a'}^*(\bS_1)}\right)(\mu_{a'}^*(\bS_1)-\mu_{a'}(\bS_1))\\
&\qquad\overset{(i)}{=}0,
\end{align*}
where (i) holds because, by assumption, 1) either $\pi_{a}^*(\cdot)=\pi_{a}(\cdot)$ or $\mu_{a}^*(\cdot)=\mu_{a}(\cdot)$, and 2) either $\pi_{a'}^*(\cdot)=\pi_{a'}(\cdot)$ or $\mu_{a'}^*(\cdot)=\mu_{a'}(\cdot)$.

Based on all Steps 1-5, we conclude that \eqref{eq:sigma_O} holds. Now, note that
\begin{align*}
E[O_1^2]&\geq E[A_{1}A_2(Y(a)-\nu_{a}(\bS))^2],\\
E[O_2^2]&=E\left[\frac{A_{1}((\rho_{a}^*(\bS))^2-2A_2\rho_{a}^*(\bS)+A_2)}{(\pi_{a}^*(\bS_1)\rho_{a}^*(\bS))^2}(\nu_{a}^*(\bS)-\nu_{a}(\bS))^2\right]\\
&=E\left[\frac{A_{1}((\rho_{a}^*(\bS)-\rho_{a}(\bS))^2+\rho_{a}(\bS)(1-\rho_{a}(\bS)))}{(\pi_{a}^*(\bS_1)\rho_{a}^*(\bS))^2}(\nu_{a}^*(\bS)-\nu_{a}(\bS))^2\right]\\
&\geq c_0^2E[A_{1}(\nu_{a}^*(\bS)-\nu_{a}(\bS))^2],\\
E[O_3^2]&=E\left[\frac{A_{1}(\nu_{a}(\bS)-\mu_{a}(\bS_1))^2}{(\pi_{a}^*(\bS_1))^2}\right]\geq E[A_{1}(\nu_{a}(\bS)-\mu_{a}(\bS_1))^2]
\end{align*}
Hence, 
\begin{align}
&E[A_{1}A_2\zeta^2]=E[\zeta_{a}^2]=E[A_{1}A_2(Y(a)-\nu_{a}^*(\bS))^2]\nonumber\\
&\qquad\overset{(i)}{=}E[A_{1}A_2((Y(a)-\nu_{a}(\bS))^{2}+(\nu_{a}(\bS)-\nu_{a}^{*}(\bS))^{2})]\leq E[O_1^2]+\frac{1}{c_0^2}E[O_2^2],\label{bound:zeta1}
\end{align}
where (i) holds as in \eqref{eq:Y-nustar}. Additionally,
\begin{align*}
&E[A_{1}\varepsilon^2]=E[\varepsilon_{a}^2]=E[A_{1}(\nu_{a}^*(\bS)-\mu_{a}^*(\bS_1))^2]\\
&\quad\leq 3\left[E[A_{1}(\nu_{a}^*(\bS)-\nu_{a}(\bS))^2]+E[A_{1}(\nu_{a}(\bS)-\mu_{a}(\bS_1))^2]+E[A_{1}(\mu_{a}(\bS_1)-\mu_{a}^*(\bS_1))^2]\right]\\
&\quad\leq\frac{3}{c_0^2}E[O_2^2]+3E[O_3^2]+3C_\mu\sigma^2.
\end{align*}
Repeating the process above, we also have
\begin{align}
E[(1-A_{1})(1-A_2)\zeta^2]&\leq E[O_4^2]+\frac{1}{c_0^2}E[O_5^2],\label{bound:zeta2}\\
E[(1-A_{1})\varepsilon^2]&\leq\frac{3}{c_0^2}E[O_5^2]+3E[O_6^2]+3C_\mu\sigma^2.\nonumber
\end{align}
Besides, we also have
\begin{align}
E[\xi^2]=E[O_8^2].\label{bound:xi}
\end{align}
Therefore, we conclude that
\begin{align*}
&E[\zeta^2]+E[\varepsilon^2]+E[\xi^2]\\
&\qquad=E[A_{1}A_2\zeta^2]+E[(1-A_{1})(1-A_2)\zeta^2]+E[A_{1}\varepsilon^2]+E[(1-A_{1})\varepsilon^2]+E[\xi^2]\\
&\qquad\leq E[O_1^2+\frac{4}{c_0^2}O_2^2+3O_3^2+O_4^2+\frac{4}{c_0^2}O_5^2+3O_6^2+O_8^2]+6C_\mu\sigma^2\leq\left(\frac{4}{c_0^2}+6C_\mu\right)\sigma^2,
\end{align*}
since $c_0<1$ and \eqref{eq:sigma_O} holds.


(b) Now, we assume Assumption 2 holds. Same as in part (a), we also have \eqref{eq:sigma_O}, \eqref{bound:zeta1}, \eqref{bound:zeta2}, and \eqref{bound:xi} hold. Additionally, under Assumption 2, by Lemma \ref{lemma:psi2norm}, we also have
$$E[\varepsilon^2]\leq2\sigma_\varepsilon^2\sigma^2.$$
Therefore,
\begin{align*}
&E[\zeta^2]+E[\varepsilon^2]+E[\xi^2]\\
&\qquad=E[A_{1}A_2\zeta^2]+E[(1-A_{1})(1-A_2)\zeta^2]+E[\varepsilon^2]+E[\xi^2]\\
&\qquad\leq E[O_1^2+\frac{1}{c_0^2}O_2^2+O_4^2+\frac{1}{c_0^2}O_5^2+O_8^2]+2\sigma_\varepsilon^2\sigma^2\leq\left(\frac{1}{c_0^2}+2\sigma_\varepsilon^2\right)\sigma^2.
\end{align*}
\end{proof}


\begin{proof}[Proof of Lemma \ref{lem_Lyap}]
We first show that \eqref{53} holds. By Lemma \ref{lemma:sigma_mis}, we have
\begin{align*}
\Delta^*(W)-\theta=\sum_{i=1}^8O_i,\quad\sigma^2=E(\Delta^*(W)-\theta)^2=\sum_{i=1}^8E[O_i^2],
\end{align*}
where $\{O_i\}_{i=1}^8$ are defined as \eqref{def:O_1}-\eqref{def:O_8}. Since now we assume that all the models are correctly specified, we have $O_i=0$ for $i\in\{2,5,7\}$ and hence
\begin{align}
\Delta^*(W)-\theta&=O_1+O_3+O_4+O_6+O_8,\label{def:psi_O_correct}\\
\sigma^2&=E[O_1^2]+E[O_3^2]+E[O_4^2]+E[O_6^2]+E[O_8^2]=\sum_{i=1}^5V_i,\nonumber
\end{align}
where
\begin{align*}
V_1:=&E\left[\left(\frac{A_{1}A_2}{\pi_{a}(\bS_1)\rho_{a}(\bS)}(Y-\nu_{a}(\bS))\right)^2\right],\\
V_2:=&E\left[\left(\frac{A_{1}}{\pi_{a}(\bS_1)}(\nu_{a}(\bS)-\mu_{a}(\bS_1))\right)^2\right],\\
V_3:=&E\left[\left(\frac{(1-A_{1})(1-A_2)}{\pi_{a'}(\bS_{1})\rho_{a'}(\bS)}(Y-\nu_{a'}(\bS))\right)^2\right],\\
V_4:=&E\left[\left(\frac{1-A_{1}}{\pi_{a'}(\bS_{1})}(\nu_{a'}(\bS)-\mu_{a'}(\bS_1))\right)^2\right],\\
V_5:=&E\left[\left(\mu_{a}(\bS_1)-\mu_{a'}(\bS_1)-\theta\right)^2\right].
\end{align*}
We lower bound each terms above:
\begin{align*}
&V_1\overset{(i)}{=}E\left[\left(\frac{\zeta_{a}}{\pi_{a}(\bS_1)\rho_{a}(\bS)}\right)^2\right]\overset{(ii)}{=}E\left[\left(\frac{A_{1}A_2}{\pi_{a}(\bS_1)\rho_{a}(\bS)}\zeta\right)^2\right]\overset{(iii)}\geq E[A_{1}A_2\zeta^2],\\
&V_2\overset{(iv)}{=}E\left[\left(\frac{\varepsilon_{a}}{\pi_{a}(\bS_1)}\right)^2\right]\overset{(v)}{=}E\left[\left(\frac{A_{1}}{\pi_{a}(\bS_1)}\varepsilon\right)^2\right]\overset{(vi)}{\geq}E[A_{1}\varepsilon^2],
\end{align*}
where (i) and (iv) hold since $\nu_{a}^*(\cdot)=\nu_{a}(\cdot)$ and $\mu_{a}^*(\cdot)=\mu_{a}(\cdot)$; (ii) and (v) hold since $\zeta_{a}=A_{1}A_2\zeta$ and $\varepsilon_{a}=A_{1}\varepsilon$; (iii) and (vi) hold since $A_{1},A_2\in\{0,1\}$, $\pi_{a}(\bS_1)\leq1$ and $\rho_{a}(\bS)\leq1$ with probability 1 under Assumption 1. Similarly,
$$V_3\geq E[(1-A_{1})(1-A_2)\zeta^2],\quad V_4\geq E[(1-A_{1})\varepsilon^2].$$
Additionally, by definition, $\xi=\mu_{a}(\bS_1)-\mu_{a'}(\bS_1)-\theta$. Hence,
$$V_5=E[\xi^2].$$
Combining all the previous results, we have
\begin{align*}
\sigma^2:&=E[\Delta^*(W)-\theta]^2\\
&\geq E[A_{1}A_2\zeta^2+(1-A_{1})(1-A_2)\zeta^2]+ E[A_{1}\varepsilon^2+(1-A_{1}) \varepsilon^2]+E[\xi^2]\\
&=E[\zeta^2]+E[\varepsilon^2]+E[\xi^2].
\end{align*}
Next, we show that \eqref{B.36} holds. Recall the representation \eqref{def:psi_O_correct}. By the finite form of Jensen's inequality, and note that the function $u\mapsto|u|^{2+t}$ is convex for any $t>0$, we have
\begin{align*}
&\left|\frac{\Delta^*(W)-\theta}{5}\right|^{2+t}=\left|\frac{O_1+O_3+O_4+O_6+O_8}{5}\right|^{2+t}\\
&\qquad\leq\frac{|O_1|^{2+t}+|O_3|^{2+t}+|O_4|^{2+t}+|O_6|^{2+t}+|O_8|^{2+t}}{5}
\end{align*}
Therefore, 
\begin{align*}
E|\Delta^*(W)-\theta|^{2+t}&\leq5^{1+t}E[|O_1|^{2+t}+|O_3|^{2+t}+|O_4|^{2+t}+|O_6|^{2+t}+|O_8|^{2+t}]\\
&=C_t\sum_{i=1}^5V_i',
\end{align*}
where $C_t=5^{1+t}$ and
\begin{align*}
V_1':=&E\left[\left|\frac{A_{1}A_2}{\pi_{a}(\bS_1)\rho_{a}(\bS)}(Y-\nu_{a}(\bS))\right|^{2+t}\right],\\
V_2':=&E\left[\left|\frac{A_{1}}{\pi_{a}(\bS_1)}(\nu_{a}(\bS)-\mu_{a}(\bS_1))\right|^{2+t}\right],\\
V_3':=&E\left[\left|\frac{(1-A_{1})(1-A_2)}{\pi_{a'}(\bS_{1})\rho_{a'}(\bS)}(Y-\nu_{a'}(\bS))\right|^{2+t}\right],\\
V_4':=&E\left[\left|\frac{1-A_{1}}{\pi_{a'}(\bS_{1})}(\nu_{a'}(\bS)-\mu_{a'}(\bS_1))\right|^{2+t}\right],\\
V_5':=&E\left[\left|\mu_{a}(\bS_1)-\mu_{a'}(\bS_1)-\theta\right|^{2+t}\right].
\end{align*}
Now, we upper bound each of the terms above.
\begin{align*}
V_1'\overset{(i)}{=}&E\left[\left|\frac{\zeta_{a}}{\pi_{a}(\bS_1)\rho_{a}(\bS)}\right|^{2+t}\right]\overset{(ii)}{=}E\left[\left|\frac{A_{1}A_2}{\pi_{a}(\bS_1)\rho_{a}(\bS)}\zeta\right|^{2+t}\right]\overset{(iii)}{\leq}\frac{1}{c_0^{4+2t}}E[|\zeta|^{2+t}],\\
V_2'\overset{(iv)}{=}&E\left[\left|\frac{\varepsilon_{a}}{\pi_{a}(\bS_1)}\right|^{2+t}\right]\overset{(v)}{=}E\left[\left|\frac{A_{1}}{\pi_{a}(\bS_1)}\varepsilon\right|^{2+t}\right]\overset{(vi)}{\leq}\frac{1}{c_0^{4+2t}}E[|\varepsilon|^{2+t}],
\end{align*}
where (i) and (iv) hold since $\nu_{a}^*(\cdot)=\nu_{a}(\cdot)$ and $\mu_{a}^*(\cdot)=\mu_{a}(\cdot)$; (ii) and (v) hold since $\zeta_{a}=A_{1}A_2\zeta$ and $\varepsilon_{a}=A_{1}\varepsilon$; (iii) and (vi) hold since $A_{1},A_2\in\{0,1\}$, $\pi_{a}(\bS_1),\rho_{a}(\bS)\in[c_0,1-c_0]$ with probability 1 under Assumption 1. Similarly, we also have
$$V_3'\leq\frac{1}{c_0^{4+2t}}E[|\zeta|^{2+t}],\quad V_4'\leq\frac{1}{c_0^{2+t}}E[|\varepsilon|^{2+t}].$$
In addition, by definition, $\xi=\mu_{a}(\bS_1)-\mu_{a'}(\bS_1)-\theta$. Hence,
$$V_5'=E[|\xi|^{2+t}].$$
Therefore, we conclude that
\begin{align*}
&E|\Delta^*(W)-\theta|^{2+t}\leq C_t\left[\frac{2}{c_0^{4+2t}}E[|\zeta|^{2+t}]+\frac{2}{c_0^{2+t}}E[|\varepsilon|^{2+t}]+E[|\xi|^{2+t}]\right]\\
&\qquad\leq\frac{2C_t}{c_0^{4+2t}}E[|\zeta|^{2+t}+|\varepsilon|^{2+t}+|\xi|^{2+t}],
\end{align*}
since $0<c_0<1$ and $t>0$.
\end{proof}

\begin{proof}[Proof of Lemma \ref{lem_sigmahat}]
We show that for each $k=1,...,K$,
\begin{align}
\frac{1}{n}\sum_{i\in \mathcal I_k}(\Delta^*(W_i)-\theta)^2-\sigma^2&=o_p(\sigma^2),\label{B.40}\\
\frac{1}{n}\sum_{i\in \mathcal I_k}(\widehat \Delta(W_i)-\thetahat_\mathrm{gen})^2-\frac{1}{n}\sum_{i\in \mathcal I_k}(\Delta^*(W_i)-\theta)^2&=o_{p}(\sigma^2),\label{B.41}
\end{align}

We first show \eqref{B.40}. Let $Z_{N,i}:=\sigma^{-1}(\Delta^*(W_i)-\theta)^2-1$, note that $W_i$ and nuisance functions $\nu_c^*(\cdot)$, $\mu_c^*(\cdot)$, $\pi_c^*(\cdot)$, and $\rho_c^*(\cdot)$ are possibly dependent with $(d_1,d_2)=(d_{N,1},d_{N,2})$. Hence, $(Z_{N,i})_{N,i}$ forms a row-wise independent and identically distributed triangular array, and \eqref{B.40} is equivalent to
$$\frac{1}{n}\sum_{i\in \mathcal I_k}Z_i=o(1).$$
By Lemma 3 of \cite{zhang2022high}, it suffices to show that $E(Z_{d,1})=0$ and $E|Z_{d,1}|^q<C'$ with some constants $q>1$ and $C'>0$. By definition,
$$E(Z_{d,1})=E\left[\frac{(\Delta^*(W)-\theta)^2}{\sigma^2}-1\right]=\frac{\sigma^2}{\sigma^2}-1=0.$$
In addition, by Minkowski inequality,
$$\left[E\left|\frac{(\Delta^*(W)-\theta)^2}{\sigma^2}-1\right|^{\frac{2+t}{2}}\right]^{\frac{2}{2+t}}
\leq\left[\frac{E|(\Delta^*(W)-\theta)|^{2+t}}{\sigma^{2+t}}\right]^{\frac{2}{2+t}}+1<C+1.$$
It follows that
$$E|Z_{d,1}|^{\frac{2+t}{2}}=E\left|\frac{(\Delta^*(W)-\theta)^2}{\sigma^2}-1\right|^{\frac{2+t}{2}}<(C+1)^{\frac{2+t}{2}},$$
with $(2+t)/2>1$. Therefore, by Lemma 3 of \cite{zhang2022high}, we conclude that \eqref{B.40} holds.

Next, we show \eqref{B.41}. Let $a_i=\widehat \Delta(W_i)-\Delta^*(W_i)-(\thetahat_\mathrm{gen}-\theta)$ and $b_i=\Delta^*(W_i)-\theta$. Then, it follows from the triangle inequality that
\begin{align*}
&\left|\frac{1}{n}\sum_{i\in \mathcal I_k}(\widehat \Delta(W_i)-\thetahat_\mathrm{gen})^2-\frac{1}{n}\sum_{i\in \mathcal I_k}(\Delta^*(W_i)-\theta)^2\right|\\
&\qquad\leq\frac{1}{n}\sum_{i\in \mathcal I_k}|a_i|\cdot|a_i+2b_i|\overset{(i)}{\leq}\left[\frac{1}{n}\sum_{i\in \mathcal I_k}a_i^2\right]^{\frac{1}{2}}\cdot\left[\frac{1}{n}\sum_{i\in \mathcal I_k}(a_i+2b_i)^2\right]^{\frac{1}{2}}\\
&\qquad\overset{(ii)}{\leq}\left[\frac{1}{n}\sum_{i\in \mathcal I_k}a_i^2\right]^{\frac{1}{2}}\cdot\left[\left(\frac{1}{n}\sum_{i\in \mathcal I_k}a_i^2\right)^{\frac{1}{2}}+2\left(\frac{1}{n}\sum_{i\in \mathcal I_k}b_i^2\right)^{\frac{1}{2}}\right],
\end{align*}
where (i) follows from Cauchy-Schwarz inequality; (ii) follows from Minkowski inequality. Recall the equation \eqref{B.40}, we have
$$\frac{1}{n}\sum_{i\in \mathcal I_k}b_i^2=\frac{1}{n}\sum_{i\in \mathcal I_k}(\Delta^*(W_i)-\theta)^2=\sigma^2(1+o_p(1)).$$
Since, by assumption, $\thetahat_\mathrm{gen}-\theta=O_p(\sigma/\sqrt{N})$ and $\left[\frac{1}{n}\sum_{i\in \mathcal I_k}|\widehat \Delta(W_i)-\Delta^*(W_i)|^2\right]^{\frac{1}{2}}=o_p(\sigma)$, we have
\begin{align*}
\left[\frac{1}{n}\sum_{i\in \mathcal I_k}a_i^2\right]^{\frac{1}{2}}\leq\left[\frac{1}{n}\sum_{i\in \mathcal I_k}|\widehat \Delta(W_i)-\Delta^*(W_i)|^2\right]^{\frac{1}{2}}+|\thetahat_\mathrm{gen}-\theta|=o_p(\sigma).
\end{align*}
Therefore,
\begin{align*}
&\left|\frac{1}{n}\sum_{i\in \mathcal I_k}(\widehat \Delta(W_i)-\thetahat_\mathrm{gen})^2-\frac{1}{n}\sum_{i\in \mathcal I_k}(\Delta^*(W_i)-\theta)^2\right|\\
&\qquad=o_p(\sigma)\cdot\left[o_p(\sigma)+\sigma(1+o_p(1))\right]=o_p(\sigma^2).
\end{align*}
Now, by \eqref{B.40} and \eqref{B.41}, we have
\begin{align*}
&\sigmahat_\mathrm{gen}-\sigma^2=\frac{1}{K}\sum_{k=1}^K\frac{1}{n}\sum_{i\in \mathcal I_k}(\widehat \Delta(W_i)-\thetahat_\mathrm{gen})^2-\sigma\\
&\qquad=\frac{1}{K}\sum_{k=1}^K\left(\frac{1}{n}\sum_{i\in \mathcal I_k}(\widehat \Delta(W_i)-\thetahat_\mathrm{gen})^2-(\Delta^*(W_i)-\theta)^2+(\Delta^*(W_i)-\theta)^2-\sigma\right)\\
&\qquad=o_p(\sigma^2).
\end{align*}
\end{proof}

\section{Asymptotic theory for Sequential Double Robust Lasso (S-DRL) estimator}\label{sec:proof_tdr}
In this section, we provide theoretical results for the S-DRL estimator. We consider the following nuisance estimators: 
$\widehat{\nu}_{c}(\bS)=\bU^\top \balphahat_c$, $\widehat{\mu}_{c}(\bS)=\bV^\top \bbetahat_c$, $\pihat_{c}(\bS_1)=g(\bV^\top\bgammahat_c)$, and $\rhohat_{c}(\bS)=g(\bU^\top\bdeltahat_c)$ for each $c\in\{a,a'\}$, where $\balphahat_c$, $ \bbetahat_c$, $\bgammahat_c$, and $\bdeltahat_c$ are defined in Section 2.1. Then $\widehat{\psi}_{c}(\cdot)$, defined as \eqref{def:DR-scorehat}, can be written as
\begin{align*}
\widehat{\psi}_c(W)=\bV^{\top}\widehat{\bbeta}_c+\mathbbm1_{\{A_{1}=c_1\}}\frac{\bU^{\top}\widehat{\balpha}_c-\bV^{\top}\widehat{\bbeta}_c}{g(\bV^{\top}\widehat{\bgamma}_c)}+\mathbbm1_{\{A_{1}=c_1,A_{2}=c_2\}}\frac{Y-\bU^{\top}\widehat{\balpha}_c}{g(\bV^{\top}\widehat{\bgamma}_c)g(\bU^{\top}\widehat{\bdelta}_c)}.
\end{align*}
We consider the following target nuisance functions:
$\nu_{c}^*(\bS)=\bU^\top \balpha_c^*$, $\mu_{c}^*(\bS)=\bV^\top \bbeta_c^*$, $\pi_{c}^*(\bS_1)=g(\bV^\top\bgamma_c^*)$, and $\rho_c^*(\bS)=g(\bU^\top\bdelta_c^*)$ for each $c\in\{a,a'\}$, where $\balpha_c^*$, $\bbeta_c^*$, $\bgamma_c^*$, and $\bdelta_c^*$ are defined in Section 2.1. Then $ {\psi}_c^*(\cdot)$, defined as \eqref{def:psi-gen-star}, can be written as
\begin{align*}
\psi_c^*(W)&=\bV^{\top}\bbeta_c^*+\mathbbm1_{\{A_{1}=c_1\}}\frac{\bU^{\top}\balpha_c^*-\bV^{\top}\bbeta_c^*}{g(\bV^{\top}\bgamma_c^*)}+\mathbbm1_{\{A_{1}=c_1,A_{2}=c_2\}}\frac{Y-\bU^{\top}\balpha_c^*}{g(\bV^{\top}\bgamma_c^*)g(\bU^{\top}\bdelta_c^*)}.
\end{align*}

\subsection{Auxiliary Lemmas}
\begin{lemma}\label{lem_T2_DDRL}
(a) Suppose that at least one of $\mu_{a}^{*}(\cdot)$ and $\pi_{a}^{*}(\cdot)$ is correctly specified, and at least one of the models $\nu_{a}^{*}(\cdot)$ and $\rho_{a}^{*}(\cdot)$ is correctly specified. Let Assumptions 1-4 hold. Assume that ${\max\{s_{\bm{\alpha}_{a}}, s_{\bm{\beta}_{a}}, s_{\bm{\gamma}_{a}},s_{\bm{\delta}_{a}}\}\log(d)} =o(N)$, and either (a) $\|\bS_1\|_\infty\leq C$ almost surely, with some constant $C>0$, or (b) $s_{\bdelta_a}\log^2(d)=O(N)$. Then,
\begin{align}
	T_{2}=O_p\left(\sigma\frac{s_1\log(d)}{N}+\sigma\sqrt{\frac{s_2\log(d)}{N}}\right),\label{B.3_T_{2}'}
\end{align}
where $T_{2}$ is defined as \eqref{def:T_2} and
\begin{align*} 
	s_1&:=\max\{\sqrt{s_{\bm{\alpha}_{a}}s_{\bm{\delta}_{a}}}, \sqrt{s_{\bm{\beta}_{a}}s_{\bm{\gamma}_{a}}}\},\\
	s_2&:=\max\left\{s_{\bm{\alpha}_{a}}\mathbbm{1}_{\{\rho_{a}^{*}\neq\rho_{a}\}},s_{\bm{\beta}_{a}}\mathbbm{1}_{\{\pi_{a}^{*}\neq\pi_{a}\}},s_{\bm{\gamma}_{a}}\mathbbm{1}_{\{\mu_{a}^{*}\neq\mu_{a}\}}, 
	s_{\bm{\delta}_{a}}\mathbbm{1}_{\{\nu_{a}^{*}\neq\nu_{a}\}}\right\}.
\end{align*}
	
b) Suppose that all the nuisance models $\mu_a^*(\cdot)$, $\nu_a^*(\cdot)$, $\pi_a^*(\cdot)$, and $\rho_a^*(\cdot)$ are correctly specified. Let Assumptions 1-3 hold. Assume that ${\max\{s_{\bm{\alpha}_{a}}, s_{\bm{\beta}_{a}}, s_{\bm{\gamma}_{a}},s_{\bm{\delta}_{a}}\}\log(d)} =o(N)$, and either (a) $\|\bS_1\|_\infty\leq C$ almost surely, with some constant $C>0$, or (b) $s_{\bdelta_a}\log^2(d)=O(N)$. Then, Then, 
\begin{align}
	T_{2}=O_p\left(\sigma\frac{s_1\log(d)}{N}\right).\label{C.28'}
\end{align}
\end{lemma}

\begin{lemma}\label{lem_T4_DDRL}
Suppose that at least one of $\mu_{a}^{*}(\cdot)$ and $\pi_{a}^{*}(\cdot)$ is correctly specified, and at least one of the models $\nu_{a}^{*}(\cdot)$ and $\rho_{a}^{*}(\cdot)$ is correctly specified. Let Assumptions of Lemma \ref{lem_T2_DDRL} (a) hold. Then,
	\begin{align}
		[E(\Deltahat(W) - \Delta^*(W))^2]^{\frac{1}{2}}&=O_p\left(\sigma\sqrt{\frac{\max\{s_{\bm{\alpha}_{a}},s_{\bm{\beta}_{a}},s_{\bm{\gamma}_{a}},s_{\bm{\delta}_{a}}\}\log(d)}{N}}\right),\label{C.26'}\\
		T_{4}&=O_p\biggl(\sigma\frac{\sqrt{\max\{s_{\bm{\alpha}_{a}},s_{\bm{\beta}_{a}},s_{\bm{\gamma}_{a}},s_{\bm{\delta}_{a}}\}\log(d)}}{N}\biggl),\label{B.4_T_{4}'}
	\end{align}
	where $T_{4}$ is defined as \eqref{def:T_4}. 
\end{lemma}

\subsection{Proof of Theorem 2}
Let $\xi:=\mu_{a}(\bS_1)-\mu_{a'}(\bS_1)-\theta$. Recall the representation \eqref{representation}. By Lemmas \ref{lem_T1}, \ref{lem_T2_DDRL}, \ref{lem_T3}, and \ref{lem_T4_DDRL}, we have
\begin{align*}
	T_{1}&=0,\\
	T_{2}^{(k)}&=O_p\left(\sigma\frac{s_1\log(d)}{N}+\sigma\sqrt{\frac{s_2\log(d)}{N}}\right),\nonumber\\
	T_{3}^{(k)}&=O_p\left(\frac{1}{\sqrt{N}}\left[\sqrt{E[\zeta^{2}]}+\sqrt{E[\varepsilon^{2}]}+\sqrt{E[\xi^2]}\right]\right),\\
	T_{4}^{(k)}&=O_p\left(\sigma\frac{\sqrt{\max\{s_{\bm{\alpha}_{a}},s_{\bm{\beta}_{a}},s_{\bm{\gamma}_{a}},s_{\bm{\delta}_{a}}\}\log(d)}}{N}\right).
\end{align*}
for each $k\leq K$. Therefore, by Lemma \ref{lemma:sigma_mis} with ${\max\{s_{\bm{\alpha}_{a}}, s_{\bm{\beta}_{a}}, s_{\bm{\gamma}_{a}},s_{\bm{\delta}_{a}}\}\log(d)} =o(N)$, we obtain 
\begin{align*} 
	\thetahat-\theta&=K^{-1}\sum_{k=1}^K(T_{1}+T_{2}^{(k)}+T_{3}^{(k)}+T_{4}^{(k)})\\
	&=O_p\left(\sigma\frac{s_1\log(d)}{N}+\sigma\sqrt{\frac{s_2\log(d)}{N}}+\frac{1}{\sqrt{N}}\sigma\right),
\end{align*}
with $s_1:=\max\{\sqrt{s_{\bm{\alpha}_{a}}s_{\bm{\delta}_{a}}}, \sqrt{s_{\bm{\beta}_{a}}s_{\bm{\gamma}_{a}}}\}$ and
\begin{align*} 
s_2&:=\max\left\{s_{\bm{\alpha}_{a}}\mathbbm{1}_{\{\rho_{a}^{*}\neq\rho_{a}\}},s_{\bm{\beta}_{a}}\mathbbm{1}_{\{\pi_{a}^{*}\neq\pi_{a}\}},s_{\bm{\gamma}_{a}}\mathbbm{1}_{\{\mu_{a}^{*}\neq\mu_{a}\}}, 
s_{\bm{\delta}_{a}}\mathbbm{1}_{\{\nu_{a}^{*}\neq\nu_{a}\}}\right\}.
\end{align*}

\subsection{Proof of Theorem 3}
In this theorem, we consider the setting where all the nuisance models are correctly specified. Note that, Assumption 4 holds under Assumption 1 when all the nuisance models are correct.
\subsubsection{Consistency}
Let $\xi:=\mu_{a}(\bS_1)-\mu_{a'}(\bS_1)-\theta$. 
Recall the representation \eqref{representation}. By Lemmas \ref{lem_T1}, \ref{lem_T2_DDRL}, \ref{lem_T3}, and \ref{lem_T4_DDRL}, we have
\begin{align*}
	T_{1}&=0,\\
	T_{2}^{(k)}&=O_p\left(\sigma\frac{s_1\log(d)}{N}\right),\\
	T_{3}^{(k)}&=O_p\left(\frac{1}{\sqrt{N}}\left[\sqrt{E[\zeta^{2}]}+\sqrt{E[\varepsilon^{2}]}+\sqrt{E[\xi^2]}\right]\right),\\
	T_{4}^{(k)}&=O_p\biggl(\sigma\frac{\sqrt{\max\{s_{\bm{\alpha}_{a}},s_{\bm{\beta}_{a}},s_{\bm{\gamma}_{a}},s_{\bm{\delta}_{a}}\}\log(d)}}{N}\biggl).
\end{align*}
for each $k\leq K$. By Assumption, $s_1\log(d)=o(\sqrt N)$ and $\max\{s_{\bm{\alpha}_{a}},s_{\bm{\beta}_{a}},s_{\bm{\gamma}_{a}},s_{\bm{\delta}_{a}}\}\log(d)=o(N)$. Together with Lemma \ref{lem_Lyap} , we obtain that 
\begin{align}\label{C.15'}
	\widehat{\theta}-\theta=K^{-1}\sum_{k=1}^K(T_{1}+T_{2}^{(k)}+T_{3}^{(k)}+T_{4}^{(k)})=O_p\left(\frac{1}{\sqrt{N}}\sigma\right).
\end{align}

\subsubsection{Asymptotic Normality}\label{proof:normal_thM4'}
By Lemmas \ref{lem_T1}, \ref{lem_T2_DDRL} and \ref{lem_T4_DDRL} with $s_1\log(d)=o(\sqrt N)$ and $\max\{s_{\bm{\alpha}_{a}},s_{\bm{\beta}_{a}},s_{\bm{\gamma}_{a}},s_{\bm{\delta}_{a}}\}\log(d)=o(N)$, we have 
$$\sqrt{n}\sigma^{-1}(T_{1}+T_{2}^{(k)}+T_{4}^{(k)})=o_p(1)$$
for each $k\leq K$. Hence, to demonstrate 
$$\sqrt{N}\sigma^{-1}(\thetahat-\theta)=\sqrt{N}\sigma^{-1}K^{-1}\sum_{k=1}^K(T_{1}+T_{2}^{(k)}+T_{3}^{(k)}+T_{4}^{(k)})\leadsto N(0,1),$$
we need to show
$$\sqrt{N}\sigma^{-1}K^{-1}\sum_{k=1}^KT_{3}^{(k)}=\sqrt N\sigma^{-1}\left(N^{-1}\sum_{i=1}^N\Delta^*(W_i)-\theta\right)\leadsto N(0,1),$$
where $T_{3}^{(k)}$ is defined as \eqref{def:T_3}. Here, $\Delta_{N,i}:=\Delta^*(W_i)$ is possibly dependent with $N$ since both $W_i$ and nuisance parameters $(\balpha_c^*,\bbeta_c^*,\bgamma_c^*, \bdelta_c^*)$ potentially depend on $(d_1,d_2)$, and $(d_1,d_2)=(d_{1,N},d_{2,N})$ are allowed to grow with $N$. 
Hence, $\{\Delta_{N,i}\}_{N,i}$ forms a triangular array. By Lyapunov's central limit theorem, it suffices to show that, for some $t>0$, the following Lyapunov's condition holds:
\begin{align}\label{S.5}
	\lim_{n \to \infty}\frac{E|\Delta^*(W)-\theta|^{2+t}}{n^{\frac{t}{2}}\sigma^{2+t}}=0.
\end{align}

\paragraph*{Step 1}
In order to check Lyapunov's condition, we show that for some constant $C'$,
\begin{align}\label{C.17}
	\frac{E|\Delta^*(W)-\theta|^{2+t}}{\sigma^{2+t}}<C'.
\end{align}
By Lemma \ref{lem_Lyap}, we have, for some constants $t>0$ and $C_t>0$,
\begin{align*}
	\frac{E|\Delta^*(W)-\theta|^{2+t}}{\sigma^{2+t}}
	&\leq\frac{2C_t}{c_0^{4+2t}}\biggl(\frac{E[|\zeta|^{2+t}}{\sigma^{2+t}}+\frac{E[|\varepsilon|^{2+t}]}{\sigma^{2+t}}+\frac{E|\xi|^{2+t}}{[E|\xi|^2]^{1+\frac{t}{2}}}\biggl).
\end{align*}
Let $\bm{e}_1=(1,\bm{0}_{1\times d_1})^{\top}$, then we write $\xi=\mu_{a}(\bS_1)-\mu_{a'}(\bS_1)-\theta=\bV^{\top}(\bm{\beta}_{a}^{*}-\bm{\beta}_{a'}^{*}-\bm{e}_1\theta)$.
By Lemma \ref{lem_eigen}, under Assumption 3, we have
\begin{align*}
	\|\xi\|_{\psi_2}
	&=\|(\bm{\beta}_{a}^{*}-\bm{\beta}_{a'}^{*}-\bm{e}_1\theta)^{\top}\bV\|_{\psi_2}
	\leq2\sigma_{u}\|\bm{\beta}_{a}^{*}-\bm{\beta}_{a'}^{*}-\bm{e}_1\theta\|_2.
\end{align*}
It follows from Lemma \ref{lemma:psi2norm} that
\begin{align}
	E[|\xi|^{2+t}]
	&\leq 2^{3+t} \sigma_{u}^{2+t}\|\bm{\beta}_{a}^{*}-\bm{\beta}_{a'}^{*}-\bm{e}_1\theta\|_2^{2+t}\Gamma(2+t/2).\label{C.3}
\end{align}
Similarly, by Assumption 2, we have
\begin{align}
	E[|\zeta|^{2+t}]&\leq 2^{3+t}\sigma^{2+t}\sigma_{\zeta}^{2+t}\Gamma(2+t/2), \label{C.4}\\
	E[|\varepsilon|^{2+t}]&\leq 2^{3+t}\sigma^{2+t}\sigma_{\varepsilon}^{2+t}\Gamma(2+t/2).\label{C.5}
\end{align}
By Lemma \ref{lem_eigen}, under Assumption 3, we also have
\begin{align}
	E[|\xi|^2]
	&=E[|\bV^{\top}(\bm{\beta}_{a}^{*}-\bm{\beta}_{a'}^{*}-\bm{e}_1\theta)|^2]
	\geq \|\bm{\beta}_{a}^{*}-\bm{\beta}_{a'}^{*}-\bm{e}_1\theta\|_2^2\cdot \lambda_{\min}(E[\bV\bV^{\top}])\label{C.6}\\
	&\geq \kappa_l\|\bm{\beta}_{a}^{*}-\bm{\beta}_{a'}^{*}-\bm{e}_1\theta\|_2^2.\nonumber
\end{align}
Using \eqref{C.3} and \eqref{C.6}, we get that 
\begin{align}\label{C.7}
	\frac{E|\xi|^{2+t}}{[E|\xi|^2]^{1+\frac{t}{2}}}
	\leq\frac{2^{3+t} \sigma_{u}^{2+t}\|\bm{\beta}_{a}^{*}-\bm{\beta}_{a'}^{*}-\bm{e}_1\theta\|_2^{2+t}\Gamma(2+t/2) }{\kappa_l^{1+t/2}\|\bm{\beta}_{a}^{*}-\bm{\beta}_{a'}^{*}-\bm{e}_1\theta\|_2^{2+t}}
	=\frac{2^{3+t} \sigma_{u}^{2+t}\Gamma(2+t/2) }{\kappa_l^{1+t/2}}.
\end{align}
Using \eqref{C.4}, \eqref{C.5} and \eqref{C.7}, then we obtain that
\begin{align*}
&\frac{E| \Delta^*(W)-\theta|^{2+t}}{\sigma^{2+t}}\\
&\quad\leq\frac{2C_t}{c_0^{4+2t}}\biggl(2^{3+t}\sigma_{\zeta}^{2+t}\Gamma(2+t/2)+2^{3+t}\sigma_{\varepsilon}^{2+t}\Gamma(2+t/2)+\frac{2^{3+t} \sigma_{u}^{2+t}\Gamma(2+t/2) }{\kappa_l^{1+t/2}}\biggl)=:C'.
\end{align*}
That is, \eqref{C.17} holds and hence the Lyapunov's condition is satisfied.

\paragraph*{Step 2}
Now, we show the consistency of the asymptotic variance's estimator. In this step, the expectations are taken w.r.t. the joint distribution of $(W_i)_{i\in \mathcal I_k}$.
By \eqref{C.15'}, we have $ \widehat{\theta}-\theta=O_p(\sigma/\sqrt{N})$. Then, we show, for each $k\leq K$,
\begin{align}\label{C.23}
	\left[\frac{1}{n}\sum_{i\in \mathcal I_k}|\Deltahat(W_i) - \Delta^*(W_i)|^2\right]^{\frac{1}{2}}=o_p(\sigma).
\end{align}
It follows from Jensen's inequality that
\begin{align*}
	&E\biggl[\frac{1}{n}\sum_{i\in \mathcal I_k}|\Deltahat(W_i) - \Delta^*(W_i)|^2\biggl]^{\frac{1}{2}}\leq\biggl\{E\biggl[\frac{1}{n}\sum_{i\in \mathcal I_k}|\Deltahat(W_i) - \Delta^*(W_i)|^2\biggl]\biggl\}^{\frac{1}{2}}\\
	&\qquad=[E|\Deltahat(W) - \Delta^*(W)|^2]^{\frac{1}{2}}
	\overset{(i)}{=}O_p\biggl(\sigma\frac{\sqrt{\max\{s_{\bm{\alpha}_{a}},s_{\bm{\beta}_{a}},s_{\bm{\gamma}_{a}},s_{\bm{\delta}_{a}}\}\log(d)}}{N}\biggl),
\end{align*}
where (i) follows from \eqref{C.26'} in Lemma \ref{lem_T4_DDRL} with correctly specified nuisance models. 
By Markov's inequality with $\max\{s_{\bm{\alpha}_{a}},s_{\bm{\beta}_{a}},s_{\bm{\gamma}_{a}},s_{\bm{\delta}_{a}}\}\log(d)=o(N)$, we have
\begin{align*}
	\left[\frac{1}{n}\sum_{i\in \mathcal I_k}|\Deltahat(W_i) - \Delta^*(W_i)|^2\right]^{\frac{1}{2}}
	=O_p\biggl(\sigma\frac{\sqrt{\max\{s_{\bm{\alpha}_{a}},s_{\bm{\beta}_{a}},s_{\bm{\gamma}_{a}},s_{\bm{\delta}_{a}}\}\log(d)}}{N}\biggl)=o_p(\sigma).
\end{align*}
Therefore, using \eqref{C.15'}, \eqref{C.17} and \eqref{C.23}, we get $\widehat{\sigma}^2-\sigma^2=o_p(\sigma^2)$ by Lemma \ref{lem_sigmahat}.

\subsection{Proofs of Auxiliary Lemmas}

\begin{proof}[Proof of Lemma \ref{lem_T2_DDRL}]
In this proof, the expectations are taken w.r.t. the distribution of a new observation $W$. 
Recall the representation \eqref{eq:T3_Q} that $T_{2}=E[\widehat \Delta(W)-\Delta^*(W)]=\sum_{i=1}^6E[Q_i]$. Here, we first upper bound $E[Q_1+Q_2+Q_3]$. Same as in the proof of Lemma \ref{lem_T2_general}, we also have \eqref{eq:Q_R} holds, with $R_i$s defined in \eqref{def:R_1}-\eqref{def:R_8}. Same as in \eqref{def:R_7} and \eqref{def:R_8}, we have $R_7=R_8=0$.	
Now, we obtain an upper bound for $R_i$ $(i\in\{1,\dots,6\})$. For $R_1+R_2$, since $|A_{1}|\leq1$, $|\pi_{a}^*(\bS_1)|\leq1$ and $|\rho_{a}^*(\bS)|\leq1$, we have
\begin{align}
	R_1+R_2&\overset{(i)}{\leq}\{E|\widehat{\pi}_{a}(\bS_1)|^{-3}\}^{\frac{1}{3}}\left\{E\left|\frac{1}{\widehat{\rho}_{a}(\bS)}-\frac{1}{\rho_{a}^{*}(\bS)}\right|^{3}\right\}^{\frac{1}{3}}\{E|\widehat{\nu}_{a}(\bS)-\nu_{a}^{*}(\bS)|^{3}\}^{\frac{1}{3}}\nonumber\\
	&\qquad+\left\{E\left|\frac{1}{\widehat{\pi}_{a}(\bS_1)}-\frac{1}{\pi_{a}^*(\bS_1)}\right|^{2}\right\}^{\frac{1}{2}}\{E|\widehat{\mu}_{a}(\bS_1)-\mu_{a}^{*}(\bS_1)|^{2}\}^{\frac{1}{2}}\nonumber\\
	&\overset{(ii)}{=}O_p\biggl(\frac{\sigma\sqrt{s_{\balpha_a}s_{\bdelta_a}}\log(d)}{N}+\frac{\sigma\sqrt{s_{\bbeta_a}s_{\bgamma_a}}\log(d)}{N}\label{rate:R1-2}\\
	&\qquad\qquad+\frac{\sigma\sqrt{s_{\balpha_a}s_{\bgamma_a}}\log(d)}{N}\mathbbm{1}_{\{\rho_{a}^{*}(\cdot)\neq\rho_{a}(\cdot)\}}+\frac{\sigma\sqrt{s_{\bdelta_a}s_{\bgamma_a}}\log(d)}{N}\mathbbm{1}_{\{\nu_{a}^{*}(\cdot)\neq\nu_{a}(\cdot)\}}\biggl),\nonumber
\end{align}
where (i) holds by H{\"o}lder's inequality; (ii) follows from Lemmas \ref{cor_mu2}, \ref{lem_conclusion} and Theorem 9 with $s_{\bm{\gamma}_{a}}\log(d) =o(N) $ and $d_1\asymp d$. 
Similarly, for $R_3+R_4$, since $|A_{1}|\leq1$, $|\rho_{a}^*(\bS)-\rho_{a}(\bS)|\leq1$, $|\pi_{a}^*(\bS_1)-\pi_{a}(\bS_1)|\leq1$,
\begin{align}
	R_3+R_4\leq&\{E|\widehat{\pi}_{a}(\bS_1)|^{-3}\}^{\frac{1}{3}}\{E|\widehat{\rho}_{a}(\bS)|^{-3}\}^{\frac{1}{3}}\{E|\widehat{\nu}_{a}(\bS)-\nu_{a}^{*}(\bS)|^{3}\}^{\frac{1}{3}}\mathbbm{1}_{\{\rho_{a}^{*}(\cdot)\neq\rho_{a}(\cdot)\}}\nonumber\\
	&\quad+\{E|\widehat{\pi}_{a}(\bS_1)|^{-2}\}^{\frac{1}{2}}\{E|\widehat{\mu}_{a}(\bS_1)- \mu_{a}^{*}(\bS_1)|^{2}\}^{\frac{1}{2}}\mathbbm{1}_{\{\pi_{a}^{*}(\cdot)\neq\pi_{a}(\cdot)\}}\nonumber\\
	&\overset{(i)}{=}O_p\biggl(\sigma\sqrt\frac{s_{\bbeta_a}\log(d)}{N}\mathbbm{1}_{\{\pi_{a}^{*}(\cdot)\neq\pi_{a}(\cdot)\}}+\sigma \sqrt{\frac{s_{\bm{\alpha}_{a}}\log(d)}{N}}\mathbbm{1}_{\{\rho_{a}^{*}(\cdot)\neq\rho_{a}(\cdot)\}}\label{rate:R3-4}\\
	&\qquad+\frac{\sigma\sqrt{s_{\bdelta_a}s_{\balpha_a}}\log(d)}{N}\mathbbm{1}_{\{\pi_{a}^{*}(\cdot)\neq\pi_{a}(\cdot)\}}+\sigma\sqrt\frac{s_{\bdelta_a}\log(d)}{N}\mathbbm{1}_{\{\pi_{a}^{*}(\cdot)\neq\pi_{a}(\cdot), \nu_{a}^{*}(\cdot)\neq\nu_{a}(\cdot)\}}\biggl),\nonumber
\end{align}
where (i) holds by Lemmas \ref{cor_mu2}, \ref{lem_conclusion} and Theorem 9 with $d_1\asymp d$.
For $R_5+R_6$, since $|\rho_{a}^*(\bS)|\leq1$,
\begin{align}
	R_5+R_6&\leq\{E|\widehat{\pi}_{a}(\bS_1)|^{-4}\}^{\frac{1}{4}}\left\{E\left|\frac{1}{\widehat{\rho}_{a}(\bS)}-\frac{1}{\rho_{a}^{*}(\bS)}\right|^{4}\right\}^{\frac{1}{4}}\{E[A_{1}|\nu_{a}^{*}(\bS)-\nu_{a}(\bS)|^{2}]\}^{\frac{1}{2}}\nonumber\\
	&\qquad+\left\{E\left|\frac{1}{\widehat{\pi}_{a}(\bS_1)}-\frac{1}{\pi_{a}^*(\bS_1)}\right|^{2}\right\}^{\frac{1}{2}}E[A_{1}|\{ \mu_{a}^{*}(\bS_1)-\mu_{a}(\bS_1)|^{2}]\}^{\frac{1}{2}}\nonumber\\
	&\overset{(i)}{=} O_p\left(\sigma\sqrt{\frac{s_{\bm{\gamma}_{a}}\log(d)}{N}}\mathbbm{1}_{\{\mu_{a}^{*}(\cdot)\neq\mu_{a}(\cdot)\}}+\sigma\sqrt{\frac{s_{\bm{\delta}_{a}}\log(d)}{N}}\mathbbm{1}_{\{\nu_{a}^{*}(\cdot)\neq\nu_{a}(\cdot)\}}\right).\label{rate:R5-6}
\end{align}
where (i) follows from Lemma \ref{lem_conclusion}, \eqref{25}, \eqref{28}, and Lemma \ref{lemma:sigma_mis}.
Combining \eqref{rate:R1-2}-\eqref{rate:R5-6} with $s_{\bm{\gamma}_{a}}\log(d) =o(N) $, we have
\begin{align*}
	E[Q_1+Q_2+Q_3]&=O_p\left(\sigma\frac{s_1\log(d)}{N}+\sigma\sqrt{\frac{s_2\log(d)}{N}}\right).
\end{align*}
Analogously to $E[Q_1+Q_2+Q_{3}]$, we have the same result for $E[Q_{4}+Q_{5}+Q_{6}]$. Theorefore, \eqref{B.3_T_{2}'} follows.

(b) When all the models are correctly specified, we have $s_2=0$. Hence, by part (a), \eqref{C.28'} holds.
\end{proof}

\begin{proof}[Proof of Lemma \ref{lem_T4_DDRL}]
In this proof, the expectations are taken w.r.t. a new observation $W$, unless stated otherwise.
We first show that \eqref{C.26'} holds. Recall the representation \eqref{eq:psi_Q}, by Minkowski inequality, we have 
\begin{align*}
	[E(\widehat \Delta(W)-\Delta^*(W))^2]^{\frac{1}{2}}\leq \sum_{i=1}^{6}[E(Q_{i}^2)]^{\frac{1}{2}},
\end{align*}	
where $Q_i$ $(i\in\{1,\dots,6\})$ are defined as\eqref{Q_1}-\eqref{Q_6}. In the following, we show that
$$
\sum_{i=1}^{6}[E(Q_{i}^2)]^{\frac{1}{2}}=O_p\left(\sigma\sqrt{\frac{\max\{s_{\bm{\alpha}_{a}},s_{\bm{\beta}_{a}},s_{\bm{\gamma}_{a}},s_{\bm{\delta}_{a}}\}\log(d)}{N}}\right).$$		
By Minkowski's inequality,
\begin{align}
	[E(Q_1^2)]^{\frac{1}{2}}&\leq \left\{E\left[\frac{A_{1}A_2}{\widehat{\pi}_{a}(\bS_1)\widehat{\rho}_{a}(\bS)}(\widehat{\nu}_{a}(\bS)-\nu_{a}^*(\bS))\right]^2\right\}^{\frac{1}{2}}\nonumber\\
	&\qquad+ \left\{E\left[\left(\frac{A_{1}A_2}{\widehat{\pi}_{a}(\bS_1)\widehat{\rho}_{a}(\bS)}-\frac{A_{1}A_2}{\pi_{a}^{*}(\bS_1)\rho_{a}^{*}(\bS)}\right)(Y-\nu_{a}^{*}(\bS))\right]^2\right\}^{\frac{1}{2}}\nonumber\\	&\overset{(i)}{\leq}\left\{E\left[\frac{1}{\widehat{\pi}_{a}(\bS_1)\widehat{\rho}_{a}(\bS)}(\widehat{\nu}_{a}(\bS)-\nu_{a}^*(\bS))\right]^2\right\}^{\frac{1}{2}}\nonumber\\
	&\qquad+ \left\{E\left[\left(\frac{1}{\widehat{\pi}_{a}(\bS_1)\widehat{\rho}_{a}(\bS)}-\frac{1}{\pi_{a}^{*}(\bS_1)\rho_{a}^{*}(\bS)}\right)\zeta\right]^2\right\}^{\frac{1}{2}}\nonumber\\
	&\overset{(ii)}{\leq} \{E|\widehat{\pi}_{a}(\bS_1)|^{-6}\}^{\frac{1}{6}}\{E|\widehat{\rho}_{a}(\bS)|^{-6}\}^{\frac{1}{6}}\{E|\widehat{\nu}_{a}(\bS)-\nu_{a}^{*}(\bS)|^{6}\}^{\frac{1}{6}}\nonumber\\
	&\qquad+\{E|\zeta|^{4}\}^{\frac{1}{4}}\left\{E\left|\frac{1}{\widehat{\pi}_{a}(\bS_1)\widehat{\rho}_{a}(\bS)}-\frac{1}{\pi_{a}^{*}(\bS_1)\rho_{a}^{*}(\bS)}\right|^{4}\right\}^{\frac{1}{4}}\nonumber\\
	&\overset{(iii)}{=}O_p\left(\sigma\sqrt{\frac{\max\{s_{\bm{\alpha}_{a}},s_{\bm{\gamma}_{a}},s_{\bm{\delta}_{a}}\}\log(d)}{N}}\right),\label{bound:Q1^2}
\end{align}
where (i) holds by the fact that $|A_{1}|\leq1$, $|A_2|\leq1$ and $A_{1}A_2\zeta=\zeta_{a}=A_{1}A_2(Y-\nu_{a}^{*}(\bS))$; (ii) holds by H{\"o}lder's inequality; (iii) follows from Lemmas \ref{cor_mu2}, \ref{lem_conclusion}, and
under Assumption 2, by Lemma \ref{lemma:psi2norm},
\begin{align}\label{B.9'}
	E[|\zeta|^{4}]\leq 8\sigma^{4}\sigma_{\zeta}^{4},\quad E[|\varepsilon|^{4}]\leq 8\sigma^{4}\sigma_{\varepsilon}^{4}.
\end{align}
Then, similarly as above, we obtain
\begin{align}
	[E(Q_2^2)]^{\frac{1}{2}}&\leq \left\{E\left[\frac{A_{1}}{\widehat{\pi}_{a}(\bS_1)}(\widehat{\nu}_{a}(\bS)-\nu_{a}^{*}(\bS))\right]^2\right\}^{\frac{1}{2}}
	+\left\{E\left[\frac{A_{1}}{\widehat{\pi}_{a}(\bS_1)}(\widehat{\mu}_{a}(\bS_1)-\mu_{a}^{*}(\bS_1))\right]^2\right\}^{\frac{1}{2}}\nonumber\\
	&\qquad+\left\{E\left[\left(\frac{A_{1}}{\widehat{\pi}_{a}(\bS_1)}-\frac{A_{1}}{\pi_{a}^{*}(\bS_1)}\right)(\nu_{a}^{*}(\bS)-\mu_{a}^{*}(\bS_1))\right]^2\right\}^{\frac{1}{2}}\nonumber\\
	&\leq \left\{E\left[\frac{1}{\widehat{\pi}_{a}(\bS_1)}(\widehat{\nu}_{a}(\bS)-\nu_{a}^{*}(\bS))\right]^2\right\}^{\frac{1}{2}}
	+\left\{E\left[\frac{1}{\widehat{\pi}_{a}(\bS_1)}(\widehat{\mu}_{a}(\bS_1)-\mu_{a}^{*}(\bS_1))\right]^2\right\}^{\frac{1}{2}}\nonumber\\
	&\qquad+\left\{E\left[\left(\frac{1}{\widehat{\pi}_{a}(\bS_1)}-\frac{1}{\pi_{a}^{*}(\bS_1)}\right)\varepsilon\right]^2\right\}^{\frac{1}{2}}\nonumber\\
	&\leq \{E|\widehat{\pi}_{a}(\bS_1)|^{-4}\}^{\frac{1}{4}}\{E|\widehat{\nu}_{a}(\bS)-\nu_{a}^{*}(\bS)|^{4}\}^{\frac{1}{4}}+ \left\{E\left|\frac{1}{\widehat{\pi}_{a}(\bS_1)}-\frac{1}{\pi_{a}^{*}(\bS_1)}\right|^{4}\right\}^{\frac{1}{4}}\{E|\varepsilon|^{4}\}^{\frac{1}{4}}\nonumber\\
	&\qquad +\{E|\widehat{\pi}_{a}(\bS_1)|^{-4}\}^{\frac{1}{4}}\{E|\widehat{\mu}_{a}(\bS_1)-\mu_{a}^{*}(\bS_1)|^{4}\}^{\frac{1}{4}}\nonumber\\
	&\overset{(i)}{=}O_p\biggl(\sigma \sqrt{\frac{\max\{s_{\bm{\alpha}_{a}},s_{\bm{\beta}_{a}},s_{\bm{\gamma}_{a}}\}\log(d)}{N}}+\sigma\sqrt\frac{s_{\bdelta_a}\log(d)}{N}\mathbbm{1}_{\{\nu_{a}^{*}(\cdot)\neq\nu_{a}(\cdot)\}}\biggl),\label{bound:Q2^2}
\end{align}
where (i) follows from Lemmas \ref{cor_mu2}, \ref{lem_conclusion}, \eqref{B.9'}, and Theorem 9 with $s_{\bm{\delta}_{a}}\log(d) =o(N) $ and $d_1\asymp d$.
By Theorem 9 , we also have
\begin{align}
	[E(Q_{3}^2)]^{\frac{1}{2}}&=\left\{E[\muhat_a(\bS_1)-\mu_{a}^{*}(\bS_1)]^2\right\}^{1/2}\nonumber\\
	&=O_p\biggl(\sigma\sqrt\frac{s_{\bbeta_a}\log(d_1)}{N}+\frac{\sigma\sqrt{s_{\bdelta_a}s_{\balpha_a}}\log(d)}{N}\nonumber\\
	&\qquad+\sigma\sqrt\frac{s_{\balpha_a}\log(d)}{N}\mathbbm{1}_{\{\rho_{a}^{*}(\cdot)\neq\rho_{a}(\cdot)\}}+\sigma\sqrt\frac{s_{\bdelta_a}\log(d)}{N}\mathbbm{1}_{\{\nu_{a}^{*}(\cdot)\neq\nu_{a}(\cdot)\}}\biggl).\label{bound:Q3^2}
\end{align}
Combining \eqref{bound:Q1^2}-\eqref{bound:Q3^2}, we have
\begin{align*}
	[E(Q_1^2)]^{\frac{1}{2}}+[E(Q_2^2)]^{\frac{1}{2}}+[E(Q_{3}^2)]^{\frac{1}{2}}=O_p\left(\sigma\sqrt{\frac{\max\{s_{\bm{\alpha}_{a}},s_{\bm{\beta}_{a}},s_{\bm{\gamma}_{a}},s_{\bm{\delta}_{a}}\}\log(d)}{N}}\right).
\end{align*}
Repeating the procedure above, we obtain the same result for $[E(Q_{4}^2)]^{\frac{1}{2}}+[E(Q_{5}^2)]^{\frac{1}{2}}+[E(Q_{6}^2)]^{\frac{1}{2}}$. 
Therefore, \eqref{C.26'} holds.
Now, we show \eqref{B.4_T_{4}'}. Recall the definition \eqref{def:T_4}, we have $T_{4}:=n^{-1}\sum_{i\in \mathcal I_k}[\widehat \Delta(W_i)-\Delta^*(W_i)]-E[\widehat \Delta(W)-\Delta^*(W)]$. By Chebyshev's inequality, we have for any $t>0$,
\begin{align}
	P(|T_{4}|>t)
	&\leq \frac{1}{t^2}\mathrm{Var}\left[\frac{1}{n}\sum_{i\in \mathcal I_k}(\widehat \Delta(W_i)-\Delta^*(W_i))\right]\label{E_notnew1'}\\
	&\leq \frac{1}{nt^2} E[\widehat \Delta(W)-\Delta^*(W)]^2.\nonumber
\end{align}
In the righ-hand side of \eqref{E_notnew1'}, the variance is taken over the joint distribution of $(W_i)_{i\in\mathcal I_k}$. Note that, based on the sample-splitting, the nuisance estimates are independent of $(W_i)_{i\in\mathcal I_k}$. Together with \eqref{C.26'}, we conclude that \eqref{B.4_T_{4}'} holds.	
\end{proof}

\section{Asymptotic theory for Dynamic Treatment Lasso (DTL) estimator}\label{sec:proof_DTL}
In this section, we provide theoretical results for the DTL estimator. The $\ell_1$-regularized nuisance estimates $\widehat{\balpha}_c$, $\widehat{\bgamma}_c$, $\widehat{\bdelta}_c$ and the target nuisance estimates $\balpha_c^*$, $\bgamma_c^*$, $\bdelta_c^*$ are the same as in Section \ref{sec:proof_tdr}. For the first-time conditional mean function, we consider the nested-regression-based estimator $\bbetahat_{c,{\mbox{\tiny NR}}}$ defined in Section 2.2. With a slight abuse of notation, we consider set the general nuisance estimates as $\widehat{\nu}_{c}(\bS)=\bU^\top \balphahat_c$, $\widehat{\mu}_{c}(\bS)=\widehat{\mu}_{c,{\mbox{\tiny NR}}}(\bS)=\bV^\top \bbetahat_{c,{\mbox{\tiny NR}}}$, $\pihat_{c}(\bS_1)=g(\bV^\top\bgammahat_c)$, and $\rhohat_{c}(\bS)=g(\bU^\top\bdeltahat_c)$ for each $c\in\{a,a'\}$. Then $\widehat{\psi}_{c}(\cdot)$, defined as \eqref{def:DR-scorehat}, can be written as
\begin{align*}
\widehat{\psi}_c(W)=\bV^{\top}\widehat{\bbeta}_{c,{\mbox{\tiny NR}}}+\mathbbm1_{\{A_{1}=c_1\}}\frac{\bU^{\top}\widehat{\balpha}_c-\bV^{\top}\widehat{\bbeta}_{c,{\mbox{\tiny NR}}}}{g(\bV^{\top}\widehat{\bgamma}_c)}+\mathbbm1_{\{A_{1}=c_1,A_{2}=c_2\}}\frac{Y-\bU^{\top}\widehat{\balpha}_c}{g(\bV^{\top}\widehat{\bgamma}_c)g(\bU^{\top}\widehat{\bdelta}_c)}.
\end{align*}
With a slight abuse of notations, we set the general working models as $\nu_{c}^*(\bS)=\bU^\top \balpha_c^*$, $\mu_{c}^*(\bS)=\mu_{c,{\mbox{\tiny NR}}}^*(\bS)=\bV^\top \bbeta_{c,{\mbox{\tiny NR}}}^*$, $\pi_{c}^*(\bS_1)=g(\bV^\top\bgamma_c^*)$, and $\rho_c^*(\bS)=g(\bU^\top\bdelta_c^*)$ for each $c\in\{a,a'\}$. Then $ {\psi}_c^*(\cdot)$, defined as \eqref{def:psi-gen-star}, can be written as
\begin{align*}
\psi_c^*(W)&=\bV^{\top}\bbeta_{c,{\mbox{\tiny NR}}}^*+\mathbbm1_{\{A_{1}=c_1\}}\frac{\bU^{\top}\balpha_c^*-\bV^{\top}\bbeta_{c,{\mbox{\tiny NR}}}^*}{g(\bV^{\top}\bgamma_c^*)}+\mathbbm1_{\{A_{1}=c_1,A_{2}=c_2\}}\frac{Y-\bU^{\top}\balpha_c^*}{g(\bV^{\top}\bgamma_c^*)g(\bU^{\top}\bdelta_c^*)}.
\end{align*}

\subsection{Auxiliary Lemmas}
\begin{lemma}\label{lem_T2_lasso}
(a) Suppose that at least one of $\mu_{a,{\mbox{\tiny NR}}}^*(\cdot)
$ and $\pi_{a}^{*}(\cdot)$ is correctly specified, and at least one of the models $\nu_{a}^{*}(\cdot)$ and $\rho_{a}^{*}(\cdot)$ is correctly specified. Let Assumptions 1-4 hold. Assume that ${\max\{s_{\bm{\alpha}_{a}}, s_{\bm{\beta}_{a}}, s_{\bm{\gamma}_{a}},s_{\bm{\delta}_{a}}\}\log(d)} =o(N)$. Then,
\begin{align}
	T_{2}=O_p\left(\sigma\frac{s'_1\log(d)}{N}+\sigma\sqrt{\frac{s'_2\log(d)}{N}}\right),\label{B.3_T_{2}}
\end{align}
where $T_{2}$ is defined as \eqref{def:T_2} and
\begin{align*} 
	s'_1&:=\max\{\sqrt{s_{\bm{\alpha}_{a}}s_{\bm{\gamma}_{a}}},\sqrt{s_{\bm{\alpha}_{a}}s_{\bm{\delta}_{a}}}, \sqrt{s_{\bm{\beta}_{a}}s_{\bm{\gamma}_{a}}}\},\\
	s_2'&:=\max\left\{s_{\bm{\alpha}_{a}}\mathbbm{1}_{\{\pi_{a}^{*}\neq\pi_{a}\;\text{or}\;\rho_{a}^{*}\neq\rho_{a}\}},s_{\bm{\beta}_{a}}\mathbbm{1}_{\{\pi_{a}^{*}\neq\pi_{a}\}},s_{\bm{\gamma}_{a}}\mathbbm{1}_{\{\mu_{a,{\mbox{\tiny NR}}}^*\neq\mu_{a}\}}, s_{\bm{\delta}_{a}}\mathbbm{1}_{\{\nu_{a}^{*}\neq\nu_{a}\}}\right\}.
\end{align*}

(b) Suppose that all the nuisance models $\mu_{a,{\mbox{\tiny NR}}}^*(\cdot)
$, $\nu_a^*(\cdot)$, $\pi_a^*(\cdot)$, and $\rho_a^*(\cdot)$ are correctly specified. Let Assumptions 1-3 hold. Assume that ${\max\{s_{\bm{\alpha}_{a}}, s_{\bm{\beta}_{a}}, s_{\bm{\gamma}_{a}},s_{\bm{\delta}_{a}}\}\log(d)} =o(N)$ Then, 
\begin{align}
	T_{2}=O_p\left(\sigma\frac{s'_1\log(d)}{N}\right).\label{C.28}
\end{align}
\end{lemma}

\begin{lemma}\label{lem_T4_lasso}
Suppose that at least one of $\mu_{a,{\mbox{\tiny NR}}}^*(\cdot)
$ and $\pi_{a}^{*}(\cdot)$ is correctly specified, and at least one of the models $\nu_{a}^{*}(\cdot)$ and $\rho_{a}^{*}(\cdot)$ is correctly specified. Let Assumptions 1-4 hold. Assume that ${\max\{s_{\bm{\alpha}_{a}}, s_{\bm{\beta}_{a}}, s_{\bm{\gamma}_{a}},s_{\bm{\delta}_{a}}\}\log(d)} =o(N)$. Then, 
\begin{align}
	[E(\Deltahat(W)- \Delta^*(W) )^2]^{\frac{1}{2}}&=O_p\left(\sigma\sqrt{\frac{\max\{s_{\bm{\alpha}_{a}},s_{\bm{\beta}_{a}},s_{\bm{\gamma}_{a}},s_{\bm{\delta}_{a}}\}\log(d)}{N}}\right),\label{C.26}\\
	T_{4}&=O_p\biggl(\sigma\frac{\sqrt{\max\{s_{\bm{\alpha}_{a}},s_{\bm{\beta}_{a}},s_{\bm{\gamma}_{a}},s_{\bm{\delta}_{a}}\}\log(d)}}{N}\biggl),\label{B.4_T_{4}}
\end{align}
where $T_{4}$ is defined as \eqref{def:T_4}. 
\end{lemma}

\subsection{Proof of Theorem 4}
Let $\xi=\mu_{a,{\mbox{\tiny NR}}}(\bS_1)-\mu_{a',{\mbox{\tiny NR}}}(\bS_1)-\theta$.
Recall the representation \eqref{representation}. By Lemmas \ref{lem_T1}, \ref{lem_T2_lasso}, \ref{lem_T3}, and \ref{lem_T4_lasso}, we have 
\begin{align*}
	T_{1}&=0,\\
	T_{2}^{(k)}&=O_p\left(\sigma\frac{s'_1\log(d)}{N}+\sigma\sqrt{\frac{s'_2\log(d)}{N}}\right),\\
	T_{3}^{(k)}&=O_p\left(\frac{1}{\sqrt{N}}\left[\sqrt{E[\zeta^{2}]}+\sqrt{E[\varepsilon^{2}]}+\sqrt{E[\xi^2]}\right]\right),\\
	T_{4}^{(k)}&=O_p\biggl(\sigma\frac{\sqrt{\max\{s_{\bm{\alpha}_{a}},s_{\bm{\beta}_{a}},s_{\bm{\gamma}_{a}},s_{\bm{\delta}_{a}}\}\log(d)}}{N}\biggl).
\end{align*}
for each $k\leq K$. Therefore, by Lemma \ref{lemma:sigma_mis} with ${\max\{s_{\bm{\alpha}_{a}}, s_{\bm{\beta}_{a}}, s_{\bm{\gamma}_{a}},s_{\bm{\delta}_{a}}\}\log(d)} =o(N)$, we obtain that
\begin{align*}
	\widehat{\theta}_{\tiny\mbox{DTL}}-\theta&=K^{-1}\sum_{k=1}^K(T_{1}+T_{2}^{(k)}+T_{3}^{(k)}+T_{4}^{(k)})\\
	&=O_p\left(\sigma\frac{s'_1\log(d)}{N}+\sigma\sqrt{\frac{s'_2\log(d)}{N}}+\frac{1}{\sqrt{N}}\sigma\right),
\end{align*}
where
\begin{align*} 
	s'_1&:=\max\{\sqrt{s_{\bm{\alpha}_{a}}s_{\bm{\gamma}_{a}}},\sqrt{s_{\bm{\alpha}_{a}}s_{\bm{\delta}_{a}}}, \sqrt{s_{\bm{\beta}_{a}}s_{\bm{\gamma}_{a}}}\},\\
	s_2'&:=\max\left\{s_{\bm{\alpha}_{a}}\mathbbm{1}_{\{\pi_{a}^{*}\neq\pi_{a}\;\text{or}\;\rho_{a}^{*}\neq\rho_{a}\}},s_{\bm{\beta}_{a}}\mathbbm{1}_{\{\pi_{a}^{*}\neq\pi_{a}\}},s_{\bm{\gamma}_{a}}\mathbbm{1}_{\{\mu_{a,{\mbox{\tiny NR}}}^*\neq\mu_{a}\}}, s_{\bm{\delta}_{a}}\mathbbm{1}_{\{\nu_{a}^{*}\neq\nu_{a}\}}\right\}.
\end{align*}

\subsection{Proof of Theorem 5}
	In this theorem, we consider the setting where all the nuisance models are correctly specified. Note that, Assumption 4 holds under Assumption 1 when all the nuisance models are correct.
	\subsubsection{Consistency}
	Let $\xi=\mu_{a}(\bS_1)-\mu_{a'}(\bS_1)-\theta$.
	Recall the representation \eqref{representation},
	by Lemmas \ref{lem_T1}, \ref{lem_T2_lasso}, \ref{lem_T3}, and \ref{lem_T4_lasso} in that order we have
	\begin{align*}
		T_{1}&=0,\\
		T_{2}^{(k)}&=O_p\left(\sigma\frac{s'_1\log(d)}{N}\right),\\
		T_{3}^{(k)}&=O_p\left(\frac{1}{\sqrt{N}}\left[\sqrt{E[\zeta^{2}]}+\sqrt{E[\varepsilon^{2}]}+\sqrt{E[\xi^2]}\right]\right),\\
		T_{4}^{(k)}&=O_p\biggl(\sigma\frac{\sqrt{\max\{s_{\bm{\alpha}_{a}},s_{\bm{\beta}_{a}},s_{\bm{\gamma}_{a}},s_{\bm{\delta}_{a}}\}\log(d)}}{N}\biggl).
	\end{align*}
	for each $k\leq K$. By Assumption, $s'_1\log(d)=o(\sqrt N)$ and $\max\{s_{\bm{\alpha}_{a}},s_{\bm{\beta}_{a}},s_{\bm{\gamma}_{a}},s_{\bm{\delta}_{a}}\}\log(d)=o(N)$. Together with Lemma \ref{lem_Lyap} , we obtain that 
	\begin{align}\label{C.15}
		\widehat{\theta}_{\tiny\mbox{DTL}}-\theta=K^{-1}\sum_{k=1}^K(T_{1}+T_{2}^{(k)}+T_{3}^{(k)}+T_{4}^{(k)})=O_p\left(\frac{1}{\sqrt{N}}\sigma\right).
	\end{align}
\subsubsection{Asymptotic Normality}\label{proof:normal_thM4}
By Lemmas \ref{lem_T1}, \ref{lem_T2_lasso}, and \ref{lem_T4_lasso} with $s'_1\log(d)=o(\sqrt N)$ and $\max\{s_{\bm{\alpha}_{a}},s_{\bm{\beta}_{a}},s_{\bm{\gamma}_{a}},s_{\bm{\delta}_{a}}\}\log(d)=o(N)$, we have 
$$\sqrt{n}\sigma^{-1}(T_{1}+T_{2}^{(k)}+T_{4}^{(k)})=o_p(1)$$
for each $k\leq K$.
In addition, repeating Section \ref{proof:normal_thM4'} in the proof of Theorem 3, we also have
$$\sqrt{N}\sigma^{-1}K^{-1}\sum_{k=1}^KT_{3}^{(k)}\leadsto N(0,1),$$
which implies, 
$$\sqrt{N}\sigma^{-1}(\thetahat_{\tiny\mbox{DTL}}-\theta)=\sqrt{N}\sigma^{-1}K^{-1}\sum_{k=1}^K(T_{1}+T_{2}^{(k)}+T_{3}^{(k)}+T_{4}^{(k)})\leadsto N(0,1).$$
By \eqref{C.26} and $\max\{s_{\bm{\alpha}_{a}},s_{\bm{\beta}_{a}},s_{\bm{\gamma}_{a}},s_{\bm{\delta}_{a}}\}\log(d)=o(N)$, we have
\begin{align*}
	\left[\frac{1}{n}\sum_{i\in \mathcal I_k}|\Deltahat(W_i) - \Delta^*(W_i)|^2\right]^{\frac{1}{2}}=o_p(\sigma).
\end{align*}
Together with \eqref{C.15} and \eqref{C.17}, we have 
$\widehat{\sigma}_{\tiny\mbox{DTL}}^2-\sigma^2=o_p(\sigma^2)$ by Lemma \ref{lem_sigmahat}.

\subsection{Proofs of Auxiliary Lemmas}

\begin{proof}[Proof of Lemma \ref{lem_T2_lasso}]
In this proof, the expectations are taken w.r.t. the distribution of a new observation $W$. We repeat the proof of Lemma \ref{lem_T2_DDRL}, except here we consider the nested-regression-based estimate $\widehat{\mu}_{c}(\cdot)=\widehat{\mu}_{c,{\mbox{\tiny NR}}}(\cdot)$ and the corresponding target $\mu_{c}^*(\cdot)=\mu_{c,{\mbox{\tiny NR}}}^*(\cdot)$. Note that
\begin{align}
 &T_{2}=E[\widehat \Delta(W)-\Delta^*(W)]=\sum_{i=1}^6E[Q_i],\label{R1-6}
\end{align}
with $E[Q_1+Q_2+Q_3]=\sum_{i=1}^6 R_i$,
where $R_i$s are defined in \eqref{def:R_1}-\eqref{def:R_6}. 
For $R_1+R_2$, we have
\begin{align}
	R_1+R_2&\overset{(i)}{\leq}\{E|\widehat{\pi}_{a}(\bS_1)|^{-3}\}^{\frac{1}{3}}\left\{E\left|\frac{1}{\widehat{\rho}_{a}(\bS)}-\frac{1}{\rho_{a}^{*}(\bS)}\right|^{3}\right\}^{\frac{1}{3}}\{E|\widehat{\nu}_{a}(\bS)-\nu_{a}^{*}(\bS)|^{3}\}^{\frac{1}{3}}\nonumber\\
	&\qquad+\left\{E\left|\frac{1}{\widehat{\pi}_{a}(\bS_1)}-\frac{1}{\pi_{a}^*(\bS_1)}\right|^{2}\right\}^{\frac{1}{2}}\{E|\widehat{\mu}_{a,{\mbox{\tiny NR}}}(\bS_1)-\mu_{a,{\mbox{\tiny NR}}}^{*}(\bS_1)|^{2}\}^{\frac{1}{2}}\nonumber\\
	&\overset{(ii)}{=}O_p\left(\sigma \frac{s'_1\log(d)}{N}\right),\label{bound:R1-2}
\end{align}
where (i) holds by H{\"o}lder's inequality; (ii) follows from Lemmas \ref{cor_mu2}, \ref{lem_conclusion} and Theorem 10 with $d_1\asymp d$. 
Similarly, for $R_3+R_4$, 
\begin{align}
	R_3+R_4\leq&\{E|\widehat{\pi}_{a}(\bS_1)|^{-3}\}^{\frac{1}{3}}\{E|\widehat{\rho}_{a}(\bS)|^{-3}\}^{\frac{1}{3}}\{E|\widehat{\nu}_{a}(\bS)-\nu_{a}^{*}(\bS)|^{3}\}^{\frac{1}{3}}\mathbbm{1}_{\{\rho_{a}^{*}(\cdot)\neq\rho_{a}(\cdot)\}}\nonumber\\
	&\quad+\{E|\widehat{\pi}_{a}(\bS_1)|^{-2}\}^{\frac{1}{2}}\{E|\widehat{\mu}_{a,{\mbox{\tiny NR}}}(\bS_1)-\mu_{a,{\mbox{\tiny NR}}}^{*}(\bS_1)|^{2}\}^{\frac{1}{2}}\mathbbm{1}_{\{\pi_{a}^{*}(\cdot)\neq\pi_{a}(\cdot)\}}\nonumber\\
	=&O_p\left(\sigma\sqrt{\frac{(s_{\bm{\alpha}_{a}}+s_{\bm{\beta}_{a}})\log(d)}{N}}\mathbbm{1}_{\{\pi_{a}^{*}(\cdot)\neq\pi_{a}(\cdot)\}}+\sigma \sqrt{\frac{s_{\bm{\alpha}_{a}}\log(d)}{N}}\mathbbm{1}_{\{\rho_{a}^{*}(\cdot)\neq\rho_{a}(\cdot)\}}\right).\label{bound:R3-4}
\end{align}
Repeating the similar procedure of \eqref{28}, we have
\begin{align}
	E[A_{1}(\mu_{a,{\mbox{\tiny NR}}}^*(\bS_1)-\mu_{a,{\mbox{\tiny NR}}}(\bS_1))^2]\leq\frac{2}{c_0}E[\zeta^2]+2E[\varepsilon^2].\label{28'}
\end{align}
For $R_5+R_6$, 
\begin{align}
	R_5+R_6&\leq\{E|\widehat{\pi}_{a}(\bS_1)|^{-4}\}^{\frac{1}{4}}\left\{E\left|\frac{1}{\widehat{\rho}_{a}(\bS)}-\frac{1}{\rho_{a}^{*}(\bS)}\right|^{4}\right\}^{\frac{1}{4}}\{E[A_{1}|\nu_{a}^{*}(\bS)-\nu_{a}(\bS)|^{2}]\}^{\frac{1}{2}}\nonumber\\
	&\qquad+\left\{E\left|\frac{1}{\widehat{\pi}_{a}(\bS_1)}-\frac{1}{\pi_{a}^*(\bS_1)}\right|^{2}\right\}^{\frac{1}{2}}\{E[A_{1}|\mu_{a,{\mbox{\tiny NR}}}^{*}(\bS_1)-\mu_{a,{\mbox{\tiny NR}}}(\bS_1)|^{2}]\}^{\frac{1}{2}}\nonumber\\
	&\overset{(i)}{=} O_p\left(\sigma\sqrt{\frac{s_{\bm{\gamma}_{a}}\log(d)}{N}}\mathbbm{1}_{\{\mu_{a,{\mbox{\tiny NR}}}^{*}(\cdot)\neq\mu_{a,{\mbox{\tiny NR}}}(\cdot)\}}+\sigma\sqrt{\frac{s_{\bm{\delta}_{a}}\log(d)}{N}}\mathbbm{1}_{\{\nu_{a}^{*}(\cdot)\neq\nu_{a}(\cdot)\}}\right),\label{bound:R5-6}
\end{align}
where (i) follows from Lemma \ref{lem_conclusion}, \eqref{25}, \eqref{28'}, and Lemma \ref{lemma:sigma_mis}.
Combining \eqref{bound:R1-2}-\eqref{bound:R5-6}, we have
\begin{align*}
	E[Q_1+Q_2+Q_{3}]=\sum_{i=1}^6R_i=O_p\left(\sigma\frac{s'_1\log(d)}{N}+\sigma\sqrt{\frac{s'_2\log(d)}{N}}\right).
\end{align*}
Note that $E[Q_4+Q_5+Q_6]$ can be controlled similarly as $E[Q_1+Q_2+Q_3]$. By \eqref{R1-6}, we have \eqref{B.3_T_{2}} holds.
	
	(b) When all the models are correctly specified, we have $s'_2=0$. Hence, by part (a), \eqref{C.28} holds.
\end{proof}

\begin{proof}[Proof of Lemma \ref{lem_T4_lasso}]
In this proof, the expectations are taken w.r.t. a new observation $W$, unless stated otherwise. We repeat the proof of Lemma \ref{lem_T4_DDRL}, except here we consider the nested-regression-based estimate $\muhat_c(\cdot)=\widehat{\mu}_{c,{\mbox{\tiny NR}}}(\cdot)$ and the corresponding target $\mu_c^*(\cdot)=\mu_{c,{\mbox{\tiny NR}}}^*(\cdot)$. Note that the estimation error of $\mu_c^*(\cdot)$ only appears in steps \eqref{bound:Q2^2} and \eqref{bound:Q3^2} when controlling the terms $[E(Q_2^2)]^{1/2}$ and $[E(Q_3^2)]^{1/2}$.
By Lemmas \ref{cor_mu2}, \ref{lem_conclusion}, \eqref{B.9'}, and Theorem 10 with with $d_1\asymp d$, we have
\begin{align*}
&[E(Q_2^2)]^{\frac{1}{2}}\\
&\quad\leq \{E|\widehat{\pi}_{a}(\bS_1)|^{-4}\}^{\frac{1}{4}}\{E|\widehat{\nu}_{a}(\bS)-\nu_{a}^{*}(\bS)|^{4}\}^{\frac{1}{4}}+ \left\{E\left|\frac{1}{\widehat{\pi}_{a}(\bS_1)}-\frac{1}{\pi_{a}^{*}(\bS_1)}\right|^{4}\right\}^{\frac{1}{4}}\{E|\varepsilon|^{4}\}^{\frac{1}{4}}\nonumber\\
&\quad\qquad+\{E|\widehat{\pi}_{a}(\bS_1)|^{-4}\}^{\frac{1}{4}}\{E|\widehat{\mu}_{a,{\mbox{\tiny NR}}}(\bS_1)-\mu_{a,{\mbox{\tiny NR}}}^{*}(\bS_1)|^{4}\}^{\frac{1}{4}}\nonumber\\
&\quad=O_p\left(\sigma\sqrt{\frac{\max\{s_{\bm{\alpha}_{a}},s_{\bm{\beta}_{a}},s_{\bm{\gamma}_{a}}\}\log(d)}{N}}\right).
\end{align*}
By Theorem 10, we also have
\begin{align*}
	[E(Q_{3}^2)]^{\frac{1}{2}}&=\left\{E[\muhat_{a,{\mbox{\tiny NR}}}(\bS_1)-\mu_{a,{\mbox{\tiny NR}}}^{*}(\bS_1)]^2\right\}^{1/2}=O_p\left(\sigma\sqrt{\frac{\max\{s_{\bm{\alpha}_{a}},s_{\bm{\beta}_{a}}\}\log(d)}{N}}\right).
\end{align*}
Repeating the remaining steps of the proof of Lemma \ref{lem_T4_DDRL}, we have
\begin{align*}
	[E(\Deltahat(W)- \Delta^*(W) )^2]^{\frac{1}{2}}&=O_p\left(\sigma\sqrt{\frac{\max\{s_{\bm{\alpha}_{a}},s_{\bm{\beta}_{a}},s_{\bm{\gamma}_{a}},s_{\bm{\delta}_{a}}\}\log(d)}{N}}\right),\\
	T_{4}&=O_p\biggl(\sigma\frac{\sqrt{\max\{s_{\bm{\alpha}_{a}},s_{\bm{\beta}_{a}},s_{\bm{\gamma}_{a}},s_{\bm{\delta}_{a}}\}\log(d)}}{N}\biggl).
\end{align*}

\end{proof}

\section{Proof of the results for multi-stage treatment estimation with DR methods}

\begin{proof}[Proof of Theorem 11]
By construction, we have $\mu_{T+1}^{*}(\bSbar_{T+1},\abar_{T})=Y$, and it follows that
\begin{align}
&E\left[\sum_{r=t+1}^{T}\frac{\prod_{l=t+1}^{r}\mathbbm{1}_{\{A_l = a_l\}}}{\prod_{l=t+1}^{r}\pi_{l}^{*}(\bSbar_l,\abar_{l})}(\mu_{r+1}^{*}(\bSbar_{r+1},\abar_{T})-\mu_{r}^{*}(\bSbar_{r},\abar_{T}))\mid \bSbar_{t},\Abar_t=\abar_t\right]\nonumber\\
&\qquad:=H_T+\sum_{r=t+1}^{T-1}(H_{r,1}+H_{r,2}+H_{r,3}),\label{G}
\end{align}
where 
\begin{align*}
H_T&=E\left[\frac{\prod_{l=t+1}^{T}\mathbbm{1}_{\{A_l = a_l\}}}{\prod_{l=t+1}^{T}\pi_{l}^{*}(\bSbar_l,\abar_{l})}(Y-\mu_{T}^{*}(\bSbar_{T},\abar_{T}))\mid \bSbar_{t},\Abar_t=\abar_t\right],
\end{align*}
and for any $r\in\{t+1,\dots,T\}$,
\begin{align*}
H_{r,1}&=E\left[\frac{\prod_{l=t+1}^{r}\mathbbm{1}_{\{A_l = a_l\}}}{\prod_{l=t+1}^{r}\pi_{l}^{*}(\bSbar_l,\abar_{l})}(\mu_{r+1}^{*}(\bSbar_{r+1},\abar_{T})-\mu_{r+1}(\bSbar_{r+1},\abar_{T}))\mid \bSbar_{t},\Abar_t=\abar_t\right],\\
H_{r,2}&=E\left[\frac{\prod_{l=t+1}^{r}\mathbbm{1}_{\{A_l = a_l\}}}{\prod_{l=t+1}^{r}\pi_{l}^{*}(\bSbar_l,\abar_{l})}(\mu_{r+1}(\bSbar_{r+1},\abar_{T})-\mu_{r}(\bSbar_{r},\abar_{T}))\mid \bSbar_{t},\Abar_t=\abar_t\right],\\
H_{r,3}&=E\left[\frac{\prod_{l=t+1}^{r}\mathbbm{1}_{\{A_l = a_l\}}}{\prod_{l=t+1}^{r}\pi_{l}^{*}(\bSbar_l,\abar_{l})}(\mu_{r}(\bSbar_{r},\abar_{T})-\mu_{r}^{*}(\bSbar_{r},\abar_{T}))\mid \bSbar_{t},\Abar_t=\abar_t\right].
\end{align*}
Define $\Atil_{r}=(A_{t+1},A_{t+2},\dots,A_{r})$ and $\atil_{r}=(a_{t+1},a_{t+2},\dots,a_{r})$ for $t+1\leq r \leq T$.
For $H_T$, by the tower rule with $Y(\Abar_T)=Y$ under Assumption 6, we have 
\begin{align}
H_T&=E\biggl[E\biggl[\frac{\prod_{l=t+1}^{T}\mathbbm{1}_{\{A_l = a_l\}}}{\prod_{l=t+1}^{T}\pi_{l}^{*}(\bSbar_l,\abar_{l})}(Y(\abar_T)-\mu_{T}^{*}(\bSbar_{T},\abar_{T})) \mid \bSbar_{T},\Abar_{T-1}=\abar_{T-1}\biggl]\nonumber\\
&\qquad\qquad\cdot P(\Atil_{T-1}=\atil_{T-1}\mid \bSbar_{T},\Abar_t=\abar_t)\mid \bSbar_{t},\Abar_t=\abar_t\biggl]\nonumber\\
&\overset{(i)}{=}E\biggl[\frac{E[\mathbbm{1}_{\{A_{T} = a_T\}}\mid \bSbar_{T},\Abar_{T-1}=\abar_{T-1}]}{\prod_{l=t+1}^{T}\pi_{l}^{*}(\bSbar_l,\abar_{l})}(E[Y(\abar_T)\mid \bSbar_{T},\Abar_{T-1}=\abar_{T-1} ]-\mu_{T}^{*}(\bSbar_{T},\abar_{T}))\nonumber\\ 
&\qquad\qquad\cdot E(\mathbbm{1}_{\{\Atil_{T-1} = \atil_{T-1}\}}\mid \bSbar_{T},\Abar_{t}=\abar_{t})\mid \bSbar_{t},\Abar_t=\abar_t\biggl]\nonumber\\
&\overset{(ii)}{=}E\left[\frac{\prod_{l=t+1}^{T-1}\mathbbm{1}_{\{A_l = a_l\}}\pi_{T}(\bSbar_T,\abar_{T})}{\prod_{l=t+1}^{T-1}\pi_{l}^{*}(\bSbar_l,\abar_{l})\pi_{T}^{*}(\bSbar_T,\abar_{T})}(\mu_{T}(\bSbar_{T},\abar_{T})-\mu_{T}^{*}(\bSbar_{T},\abar_{T}))\mid \bSbar_{t},\Abar_t=\abar_t\right]\nonumber\\
&\overset{(iii)}{=}E\left[\frac{\prod_{l=t+1}^{T-1}\mathbbm{1}_{\{A_l = a_l\}}}{\prod_{l=t+1}^{T-1}\pi_{l}^{*}(\bSbar_l,\abar_{l})}(\mu_{T}(\bSbar_{T},\abar_{T})-\mu_{T}^{*}(\bSbar_{T},\abar_{T}))\mid \bSbar_{t},\Abar_t=\abar_t\right]=-H_{T-1,1},\label{G_t}
\end{align}
where (i) holds since $Y(\abar_T)\perp \!\!\! \perp A_T \mid \bSbar_T, \Abar_{T-1}=\abar_{T-1}$ under the Assumption 6; (ii) holds since $\pi_{T}(\bSbar_T,\abar_{T})=P[A_{T} = a_T\mid \bSbar_{T},\Abar_{T-1}=\abar_{T-1}]$ and $\mu_{T}(\bSbar_{T},\abar_{T})=E[Y(\abar_T)\mid \bSbar_{T}, \Abar_{T-1}=\abar_{T-1} ]$ under the Assumption 6; (iii) holds since either $\pi_{T}^{*}(\cdot,\abar_{T})=\pi_{T}(\cdot,\abar_{T})$ or $\mu_{T}^{*}(\cdot,\abar_{T})=\mu_{T}(\cdot,\abar_{T})$ by assumption. For $H_{r,2}$, by the tower rule, we have 
\begin{align}
H_{r,2}&=E\biggl[E\biggl[\frac{\prod_{l=t+1}^{r}\mathbbm{1}_{\{A_l = a_l\}}}{\prod_{l=t+1}^{r}\pi_{l}^{*}(\bSbar_l,\abar_{l})}(\mu_{r+1}(\bSbar_{r+1},\abar_{T})-\mu_{r}(\bSbar_{r},\abar_{T})) \mid \bSbar_{r},\Abar_r=\abar_r\biggl]\nonumber\\
&\qquad\qquad\cdot P(\Atil_r=\atil_r\mid \bSbar_{r},\Abar_t=\abar_t)\mid \bSbar_{t},\Abar_t=\abar_t\biggl]\nonumber\\
&=E\biggl[\frac{E[\mu_{r+1}(\bSbar_{r+1},\abar_{T})\mid \bSbar_{r},\Abar_r=\abar_r]-\mu_{r}(\bSbar_{r},\abar_{T})}{\prod_{l=t+1}^{r}\pi_{l}^{*}(\bSbar_l,\abar_{l-1})}\nonumber \\
&\qquad\qquad\cdot P(\Atil_r=\atil_r\mid \bSbar_{r},\Abar_t=\abar_t)\mid \bSbar_{t},\Abar_t=\abar_t\biggl].\label{eq:Hr2}
\end{align}
For any $r\in\{1,\dots,T\}$, we have
\begin{align}
\mu_{r}(\bSbar_{r},\abar_{T})&=E[Y(\abar_T)\mid\bSbar_r,\Abar_{r-1}=\abar_{r-1}]\overset{(i)}{=}E[Y(\abar_T)\mid\bSbar_r,\Abar_{r}=\abar_{r}]\nonumber\\
&\overset{(ii)}{=}E[E[Y(\abar_T)\mid\bSbar_{r+1},\Abar_{r}=\abar_{r}]\mid\bSbar_r,\Abar_{r}=\abar_{r}]\nonumber\\
&=E[\mu_{r+1}(\bSbar_{r+1},\abar_{T})\mid\bSbar_{r},\Abar_{r}=\abar_{r}],\label{eq:nest-rep-multi}
\end{align}
where (i) holds since $Y(\abar_T)\perp \!\!\! \perp A_r \mid \bSbar_r, \Abar_{r-1}=\abar_{r-1}$ under the Assumption 6; (ii) holds by the tower rule. Together with \eqref{eq:Hr2}, we conclude that
\begin{align}\label{G_r2}
H_{r,2}=0.
\end{align}
For $H_{r,3}$, by the tower rule, we have 
\begin{align}
H_{r,3}&=E\biggl[E\biggl[\frac{\prod_{l=t+1}^{r}\mathbbm{1}_{\{A_l = a_l\}}}{\prod_{l=t+1}^{r}\pi_{l}^{*}(\bSbar_l,\abar_{l})}(\mu_{r}(\bSbar_{r},\abar_{T})-\mu_{r}^{*}(\bSbar_{r},\abar_{T})) \mid \bSbar_{r},\Abar_{r-1}=\abar_{r-1}\biggl]\nonumber\\
&\qquad\qquad\cdot P(\Atil_{r-1}=\atil_{r-1}\mid \bSbar_{r},\Abar_t=\abar_t)\mid \bSbar_{t},\Abar_t=\abar_t\biggl]\nonumber\\
&=E\biggl[\frac{E[\mathbbm{1}_{\{A_{r} = a_r\}}\mid \bSbar_{r},\Abar_{r-1}=\abar_{r-1}]}{\prod_{l=t+1}^{r}\pi_{l}^{*}(\bSbar_l,\abar_{l})}(\mu_{r}(\bSbar_{r},\abar_{T})-\mu_{r}^{*}(\bSbar_{r},\abar_{T}))\nonumber\\
&\qquad\qquad\cdot E(\mathbbm{1}_{\{\Atil_{r-1} = \atil_{r-1}\}} \mid \bSbar_{r},\Abar_t=\abar_t)\mid \bSbar_{t},\Abar_t=\abar_t\biggl]\nonumber\\
&\overset{(i)}{=}E\left[\frac{\prod_{l=t+1}^{r-1}\mathbbm{1}_{\{A_l = a_l\}}\pi_{r}(\bSbar_r,\abar_{r})}{\prod_{l=t+1}^{r-1}\pi_{l}^{*}(\bSbar_l,\abar_{l})\pi_{r}^{*}(\bSbar_r,\abar_{r})}(\mu_{r}(\bSbar_{r},\abar_{T})-\mu_{r}^{*}(\bSbar_{r},\abar_{T}))\mid \bSbar_{t},\Abar_t=\abar_t \right]\nonumber\\
&\overset{(ii)}{=}E\left[\frac{\prod_{l=t+1}^{r-1}\mathbbm{1}_{\{A_l = a_l\}}}{\prod_{l=t+1}^{r-1}\pi_{l}^{*}(\bSbar_l,\abar_{l})}(\mu_{r}(\bSbar_{r},\abar_{T})-\mu_{r}^{*}(\bSbar_{r},\abar_{T}))\mid \bSbar_{t},\Abar_t=\abar_t \right]=-H_{r-1,1},\label{G_r3}
\end{align}
where (i) holds by the tower rule and that $\pi_{r}(\bSbar_r,\abar_{r-1})=P[A_{r} = a_r\mid \bSbar_{r},\Abar_{r-1}=\abar_{r-1}]$; (ii) holds since either $\pi_{r}^{*}(\cdot,\abar_{T})=\pi_{r}(\cdot,\abar_{T})$ or $\mu_{r}^{*}(\cdot,\abar_{T})=\mu_{r}(\cdot,\abar_{T})$ by assumption. 
Combining \eqref{G_t}-\eqref{G_r3} with \eqref{G}, we have
\begin{align*}
&E\biggl[\sum_{r=t+1}^{T}\frac{\prod_{l=t+1}^{r}\mathbbm{1}_{\{A_l = a_l\}}}{\prod_{l=t+1}^{r}\pi_{l}^{*}(\bSbar_l,\abar_{l})}(\mu_{r+1}^{*}(\bSbar_{r+1},\abar_{T})-\mu_{r}^{*}(\bSbar_{r},\abar_{T}))\\
&\qquad\qquad+\mu_{t+1}^{*}(\bSbar_{t+1},\abar_{T})\mid \bSbar_{t}=\bsbar_{t},\Abar_t=\abar_{t}\biggl]\\
&\qquad=H_T+\sum_{r=t+1}^{T-1}(H_{r,1}+H_{r,2}+H_{r,3})+E[\mu_{t+1}^{*}(\bSbar_{t+1},\abar_{T})\mid \bSbar_{t},\Abar_t=\abar_t]\\
&\qquad=-H_{T-1,1}+\sum_{r=t+1}^{T-1}(H_{r,1}-H_{r-1,1})+E[\mu_{t+1}^{*}(\bSbar_{t+1},\abar_{T})\mid \bSbar_{t},\Abar_t=\abar_t]\\
&\qquad=-H_{t,1}+E[\mu_{t+1}^{*}(\bSbar_{t+1},\abar_{T})\mid \bSbar_{t},\Abar_t=\abar_t]\\
&\qquad\overset{(i)}{=}E\left[\mu_{t+1}(\bSbar_{t+1},\abar_{T})\mid \bSbar_{t},\Abar_t=\abar_t\right]\overset{(ii)}{=}\mu_{t}(\bSbar_{t},\abar_{T}),\;\;\text{for any}\;\;t\in\{1,\dots,T\},
\end{align*}
where (i) holds since $H_{t,1}=E\left[\mu_{t+1}^{*}(\bSbar_{t+1},\abar_{T})-\mu_{t+1}(\bSbar_{t+1},\abar_{T})\mid\bSbar_{t},\Abar_t=\abar_t\right]$; (ii) holds since \eqref{eq:nest-rep-multi} holds for any $r\in\{1,\dots,T\}$.
\end{proof}

\begin{proof}[Proof of Proposition 1]
By Theorem 11 with $t=0$, we have
\begin{align*}
	&E\left[\sum_{t=1}^{T}\frac{\mathbbm{1}_{\{\Abar_{t} = \abar_t\}}}{\prod_{l=1}^{t}\pi_{l}^{*}(\bSbar_l,\abar_{l})}(\mu_{t+1}^{*}(\bSbar_{t+1},\abar_{T})-\mu_{t}^{*}(\bSbar_{t},\abar_{T}))+\mu_{1}^{*}(\bS_1,\abar_T)\right]\\
	&\qquad=\mu_0(\bsbar_0,\abar_T)=E[Y(\abar_T)|\bSbar_0=\bsbar_0]=\theta_{\abar_T}.
\end{align*}
\end{proof}


\end{document}